\documentclass[]{llncs}

\usepackage[utf8]{inputenc}
\usepackage{amsfonts,amssymb} 
\usepackage{stmaryrd}
\usepackage{pgf}
\usepackage{tikz}
	\usetikzlibrary{trees,decorations,arrows,automata,positioning,plotmarks}
\usepackage{xspace}
\usepackage{extarrows}
\usepackage{color}
\usepackage[left=2cm, right=2cm, top=2cm, bottom=2cm]{geometry}

\usepackage{encodingSynPi}

\title{Is it a ``Good'' Encoding of Mixed Choice? (Technical Report) %
\thanks{This work was supported by the DFG (German Research Foundation), grant {NE-1505/2-1}.}}
\author{Kirstin Peters \and Uwe Nestmann}
\institute{TU Berlin, Germany}

\begin{document}

\maketitle

\begin{abstract}
	This technical report contains the proofs to the lemmata and theorems of \cite{petersNestmann12} as well as some additional material. As main contributions \cite{petersNestmann12} presents an encoding of mixed choice in the context of the \piCal-calculus and a criterion to measure whether the degree of distribution in process networks is preserved.
\end{abstract}

\section{Technical Preliminaries}
\label{sec:techPre}

\subsection{The \piCal-Calculus}
\label{sec:piCalculus}

Our source language is the monadic $ \pi $-calculus as described for instance in \cite{sangiorgiWalker01}. As already demonstrated in \cite{palamidessi03} the most interesting operator for a comparison of the expressive power between the full \piCal-calculus and its asynchronous variant is mixed choice, i.e., choice between input and output capabilities. Thus we denote the full $ \pi $-calculus also by \piMix. Let $ \names $ denote a countably infinite set of names with $ \tau \notin \names $ and $ \coNames $ the set of co-names, i.e., $ \coNames = \Set{ \Out{n} \mid n \in \names} $. We use lower case letters $ a, a', a_1, \ldots, x, y, \ldots $ to range over names.

\begin{definition}[\piMix]
\label{def:piMix}
  The set of process terms of the \emph{synchronous \piCal-calculus (with mixed choice)}, denoted by $ \piMixProc $, is given by
	\begin{align*}
		P & \;\mathop{::=}\;
                \RestrictedTerm{n}{P}
                \sep P_1 \mid P_2
                \sep \Match{a}{b}P
                \sep \ReplicateInput{y}{x}.P
                \sep \sum_{i \in \indexSet} \guard_i.P_i 
	\end{align*}
	where $ \guard \;\mathop{::=}\; \Input{y}{x} \; \mid \; \Output{y}{z} \; \mid \; \tau $ for some names $ n, a, b, x, y, z \in \names $ and a finite index set $ \indexSet $.
\end{definition}
\noindent
The interpretation of the defined process terms is as usual. Moreover we consider two subcalculi of \piMix. The process terms $ \piSepProc $ of \piSep|the \emph{\piCal-calculus with separate choice}|are obtained by restricting the choice primitive, such that in each choice either no input guarded or no output guarded alternatives appear.

\begin{definition}[\piSep]
  The set of process terms of the \emph{\piCal-calculus (with separate choice)}, denoted by $ \piSepProc $, is given by
	\begin{align*}
		P & \;\mathop{::=}\;
                \RestrictedTerm{n}{P}
                \sep P_1 \mid P_2
                \sep \Match{a}{b}P
                \sep \ReplicateInput{y}{x}.P
                \sep \sum_{i \in \indexSet} \guard^I_i.P_i
                \sep \sum_{i \in \indexSet} \guard^O_i.P_i
	\end{align*}
	where $ \guard^I \;\mathop{::=}\; \Input{y}{x} \; \mid \; \tau $ and $ \guard^O \;\mathop{::=}\; \Output{y}{z} \; \mid \; \tau $ for some names $ n, a, b, x, y, z \in \names $ and a finite index set $ \indexSet $.
\end{definition}

Finally, the process terms $ \piAsynProc $ of the \emph{asynchronous \piCal-calculus} \piAsyn \cite{boudol92,hondaTokoro91} are obtained by restricting each sum to be of length zero or one and requiring, that outputs can only guard the empty sum.
 
\begin{definition}[\piAsyn]
  The set of process terms of the \emph{asynchronous \piCal-calculus}, denoted by $ \piAsynProc $, is given by
	\begin{align*}
		P & \;\mathop{::=}\;
				\nullTerm
                \sep \RestrictedTerm{n}{P}
                \sep P_1 \mid P_2
                \sep \Match{a}{b}P
                \sep \Output{y}{\tilde{z}}.P
                \sep \Input{y}{\tilde{x}}.P
                \sep \tau.P
                \sep \ReplicateInput{y}{\tilde{x}}.P
	\end{align*}
	for some names $ n, a, b, y \in \names $ and some finite sequences of names $ \tilde{x}, \tilde{z} \subset \names $.
\end{definition}

Note that we augment all three variants of the \piCal-Calculus with matching, because we need it at least in \piAsyn to encode mixed choice. Of course, the presence of match influences the expressive power of \piAsyn. However, we do not know, whether the use of match in the encoding of mixed choice can be circumvented, although there are reasons indicating that this is indeed not possible. We left the consideration of this problem to further research.

We use capital letters $ P , P' , P_1, \ldots, Q, R, \ldots $ to range over processes. Let $ \FreeNames{P} $, $ \BoundNames{P} $, and $ \Names{P} $ denotes the sets of free names, bound names and all names occurring in P, respectively. Their definitions are completely standard. Given an input prefix $ \Input{y}{x} $ or an output prefix $ \Output{y}{x} $ we call $ y $ the subject and $ x $ the object of the action. Moreover we denote the subject of an action also as link or channel name, while we denote the object as value or parameter. Note that in case the object does not matter we omit it, i.e., we abbreviate an input guarded term $ \Input{y}{x}.P $ or an output guarded term $ \Output{y}{x}.P $ such that $ x \not \in \FreeNames{P} $ by $ \In{y}.P $ or $ \Out{y}.P $, respectively. Moreover we denote the empty sum with $ \nullTerm $ and often omit it in continuations. As usual we sometimes notate a sum $ \sum_{i \in \Set{ i_1, \ldots, i_n }} \guard_i.P_i $ by $ \guard_{i_1}.P_{i_1} + \ldots + \guard_{i_n}.P_{i_n} $.

We use $ \sigma, \sigma', \sigma_1, \ldots $ to range over substitutions. A substitution is a mapping $ \Set{ \Subst{x_1}{y_1}, \ldots, \Subst{x_n}{y_n} } $ from names to names. The application of a substitution on a term $ \Set{ \Subst{x_1}{y_1}, \ldots, \Subst{x_n}{y_n} }\left( P \right) $ is defined as the result of simultaneously replacing all free occurrences of $ y_i $ by $ x_i $ for $ i \in \Set{ 1, \ldots, n } $, possibly applying alpha-conversion to avoid capture or name clashes. For all names $ \names \setminus \Set{ y_1, \ldots, y_n } $ the substitution behaves as the identity mapping. Let $ \id $ denote identity, i.e. $ \id $ is the empty substitution. We naturally extend substitutions to co-names, i.e. $ \forall \overline{n} \in \coNames \logdot \sigma\left( \overline{n} \right) = \overline{\sigma\left( n \right)} $ for all substitutions $ \sigma $.

Moreover, let $ \tilde{x} $ denote a sequence of names. For simplicity, we occasionally treat $ \tilde{x} $ as a set. So $ \tilde{x} \subset M $ denotes a sequence of names of the set $ M $, $ |\tilde{x}| $ is the length of the sequence $ \tilde{x} $, and $ y \in \tilde{x} $ denotes that $ y $ is one of the names within the sequence $ \tilde{x} $. Accordingly, we use the notion of sequence to abbreviate multiple restrictions, i.e., $ \RestrictedTerm{\tilde{x}}{P} \deff \RestrictedTerm{x_1}{\ldots \RestrictedTerm{x_n}{P}} $ for a sequence of names $ \tilde{x} = x_1, \ldots, x_n $. Moreover we naturally extend substitutions to sequences of names, i.e., $ \sigma\left( \tilde{x} \right) \deff \sigma\left( x_1 \right), \ldots, \sigma\left( x_n \right) $ and $ \Set{ \Subst{\tilde{x}}{\tilde{y}} }P \deff \Set{ \Subst{x_1}{y_1}, \ldots, \Subst{x_n}{y_n} }P $ for two sequences of names $ \tilde{x} = x_1, \ldots, x_n $ and $ \tilde{y} = y_1, \ldots, y_n $.

\begin{figure}[htp]
	\begin{align*}
		\begin{array}{|c|}
			\hline
			\\
			\quad P \equiv Q \begin{aligned}[t]
					& \text{ if } Q \text{ can be obtained from } P \text{ by renaming one or more of the bound names in P,} \quad\\
					& \text{ silently avoiding name clashes }
				\end{aligned}\\
			\\
			\quad P \mid \nullTerm \equiv P \hspace*{2em} P \mid Q \equiv Q \mid P \hspace*{2em} P \mid \left( Q \mid R \right) \equiv \left( P \mid Q \right) \mid R \hspace*{2em} \Match{a}{a}P \equiv P \quad\\
			\\
			\RestrictedTerm{n}{\nullTerm} \equiv \nullTerm \hspace*{2em} \RestrictedTerm{n}{\RestrictedTerm{m}{P}} \equiv \RestrictedTerm{m}{\RestrictedTerm{n}{P}} \hspace*{2em} P \mid \RestrictedTerm{n}{Q} \equiv \RestrictedTerm{n}{\left( P \mid Q \right)} \text{ if } n \notin \FreeNames{P}\\
			\\
			\hline
		\end{array}
	\end{align*}
	\caption{Structural congruence.} \label{fig:SC}
\end{figure}

The \emph{reduction semantics} of \piMix, \piSep, and \piAsyn are jointly given by the transition rules in Figure \ref{fig:concurrentReductionSemantics}, where \emph{structural congruence}, denoted by $ \equiv $, is given by the rules in Figure \ref{fig:SC}. Note that the rule $ \textsc{Com}_{\indexAsyn} $ for communication in \piAsyn is a simplified version of the rule $ \textsc{Com}_{\indexMix, \indexSep} $ for communication in \piMix or \piSep. The differences between these two rules result from the differences in the syntax, i.e. the lack of choice and the fact that only input can be used as guard in \piAsyn. The same can be observed for the rules $ \textsc{Tau}_{\indexAsyn} $ and $ \textsc{Rep}_{\indexAsyn} $. As usual, we use $ \equivAlpha $ if we refer to alpha-conversion (the first rule of Figure \ref{fig:SC}) only.

\begin{figure}[htp]
	\begin{align*}
		\begin{array}{|c|}
			\hline
			\\
			\quad \textsc{Tau}_{\indexMix, \indexSep} \quad \ldots + \tau.P + \ldots \step P \hspace*{3em} \textsc{Tau}_{\indexAsyn} \quad \tau.P \step P \quad\\
			\\
			\quad \textsc{Com}_{\indexMix, \indexSep} \quad \left( \ldots + \Input{y}{x}.P + \ldots \right) \mid \left( \ldots + \Output{y}{z}.Q + \ldots \right) \step \Set{ \Subst{z}{x} }P \mid Q \quad\\
			\\
			\quad \textsc{Com}_{\indexAsyn} \quad \Input{y}{\tilde{x}}.P \mid \Output{y}{\tilde{z}} \step \Set{ \Subst{\tilde{z}}{\tilde{x}} }P \quad \text{ if } |\tilde{x}| = |\tilde{z}| \quad\\
			\\
			\quad \textsc{Rep}_{\indexMix, \indexSep} \quad \ReplicateInput{y}{x}.P \mid \left( \ldots + \Output{y}{z}.Q + \ldots \right) \step \Set{ \Subst{z}{x} }P \mid Q \mid \ReplicateInput{y}{x}.P \quad\\
			\\
			\quad \textsc{Rep}_{\indexAsyn} \quad \ReplicateInput{y}{\tilde{x}}.P \mid \Output{y}{\tilde{z}} \step \Set{ \Subst{\tilde{z}}{\tilde{x}} }P \mid \ReplicateInput{y}{\tilde{x}}.P \quad \text{ if } |\tilde{x}| = |\tilde{z}| \quad\\
			\\
			\quad \textsc{Par} \quad \dfrac{P \step P'}{P \mid Q \step P' \mid Q} \hspace*{3em} \textsc{Res} \quad \dfrac{P \step P'}{\RestrictedTerm{x}{P} \step \RestrictedTerm{x}{P'}} \hspace*{3em} \textsc{Cong} \quad \dfrac{P \equiv P' \quad P' \step Q' \quad Q' \equiv Q}{P \step Q} \quad\\
			\\
			\hline
		\end{array}
	\end{align*}
	\caption{Concurrent reduction semantics.} \label{fig:concurrentReductionSemantics}
\end{figure}

Let $ P \step $ ($ P \not\step $) denote existence (non-existence) of a step from $ P $, i.e. there is (no) $ P' \in \proc $ such that $ P \step P' $. Moreover, let $ \steps $ be the reflexive and transitive closure of $ \step $ and let $ \inftySteps $ define an infinite sequence of steps.

In Section \ref{sec:qualityCriteria}, we present several criteria to measure the quality of an encoding. The first of these criteria relies on the notion of a context. A \emph{context} $ \Context{}{}{\hole_1, \ldots, \hole_n} $ is a \piCal-term, i.e., a \piAsyn-term in case of Definition \ref{def:compositionality}, with~$ n $ so-called holes. Plugging the \piAsyn-terms $ P_1, \ldots, P_n $ in this order into the holes $ \hole_1, \ldots, \hole_n $ of the context, respectively, yields a term denoted $ \Context{}{}{P_1, \ldots, P_n} $. Therefore, we consider a context as a function from terms into terms, e.g., the context $ \Context{}{}{\hole_1, \ldots, \hole_n} \in \piAsynProc \times \ldots \times \piAsynProc \to \piAsynProc $ maps $ n $ \piAsyn-terms onto a \piAsyn-term. Sometimes, we refer to $ P_1, \ldots, P_n $ as the parameters of the context $ \context $. Note that a context may bind some free names of $ P_1, \ldots, P_n $. The \emph{arity} of a context is the number of its holes.

As usual we will use equivalence relations to compare \piCal-terms by means of their behaviour. Moreover, as explained in Section \ref{sec:qualityCriteria}, we use an equivalence to abstract from junk, i.e., remains of encoded terms that are no longer of any use. Since we use a reduction semantics, a standard equivalence to compare \piMix-terms is barbed congruence, denoted by $ \barbCong $. Its definition relies on the notion of an \emph{observable} or \emph{barb} (we refer to \cite{sangiorgiWalker01} for a detailed explanation).

\begin{definition}[Observable] \label{def:barb}
	Let $ P \in \piMixProc $. Then $ P $ has an \emph{input observable} $ y $, denoted by $ P\Barb{\In{y}} $, if $ P $ can perform an input on $ y $, i.e.,
	\begin{align*}
		\exists P', P'' \in \piMixProc \logdot \exists \tilde{x} \subset \names \logdot \exists z \in \names \logdot P \equiv \RestrictedTerm{\tilde{x}}{\left( P' \mid \Input{y}{z}.P'' \right)} \wedge y \notin \tilde{x}
	\end{align*}
	and $ P $ has an \emph{output observable} $ y $, denoted by $ P\Barb{\Out{y}} $, if $ P $ can perform an output on $ y $, i.e.,
	\begin{align*}
		\exists P', P'' \in \piMixProc \logdot \exists \tilde{x} \subset \names \logdot \exists z \in \names \logdot P \equiv \RestrictedTerm{\tilde{x}}{\left( P' \mid \Output{y}{z}.P'' \right)} \wedge y \notin \tilde{x}.
	\end{align*}
\end{definition}

\subsection{Abbreviations}
\label{sec:abbreviations}

To shorten the presentation and ease the readability of the rather lengthy encoding function in the next section, we use some abbreviations on \piAsyn-terms. 
First note that we defined only monadic versions of the calculi \piMix, \piSep, and \piAsyn, where over any link exactly one value is transmitted. However, within the presented encoding functions, we treat the target language \piAsyn as if it allows for polyadic communication. More precisely, we allow asynchronous links to carry any number of values from zero to five, of course under the requirement that within each \piAsyn-term no link name is used twice with different multiplicities. Note that these polyadic actions can be simply translated into monadic actions by a standard encoding as given in \cite{sangiorgiWalker01}. Thus, we silently use the polyadic version of \piAsyn in the following.
Second, as already done in \cite{nestmann00}, we use the following abbreviations to define boolean values and a conditional construct.

\begin{definition}[Tests on Booleans]
\label{def:testBoolean}
	Let $ \bool \deff \Set{ \true, \false } $ be the set of \emph{boolean values}, where $ \true $ denotes \emph{true} and $ \false $ denotes \emph{false}.
	
	Let $ \sumLock, t, f \in \names $ and $ P, Q \in \piAsynProc $. Then a \emph{boolean instantiation} of $ \sumLock $, i.e., the allocation of a boolean value to a link $ \sumLock $, and a \emph{test-statement} on a boolean instantiation are defined by
	\begin{align*}
		\Output{\sumLock}{\true} & \deff \Input{\sumLock}{t, f}.\Out{t}\\
		\Output{\sumLock}{\false} & \deff \Input{\sumLock}{t, f}.\Out{f}\\
		\Test{\sumLock}{P}{Q} & \deff \RestrictedTerm{t, f}{\left( \Output{\sumLock}{t, f} \mid \In{t}.P \mid \In{f}.Q \right)}
	\end{align*}
	for some $ t, f \notin \FreeNames{P} \cup \FreeNames{Q} $.
\end{definition}

Finally, we define forwarders, i.e., a simple process to forward each received message along some specified set of links.

\begin{definition}[Forwarder]
\label{def:forwarder}
	Let $ \indexSet $ be a finite index set and for all $ i \in \indexSet $ let $ y $ and $ y_i $ be channel names with multiplicity $ n \in \nat $, then a \emph{forwarder} is given by:
	\begin{align*}
		\Forward{y}{\Set{ y_i \mid i \in \indexSet }} & \deff \ReplicateInput{y}{x_1, \ldots, x_n}.\left( \prod_{i \in \indexSet} \Output{y_i}{x_1, \ldots, x_n} \right)
	\end{align*}
	In case of a singleton set we omit the brackets, i.e., $ \Forward{y}{y'} \deff \Forward{y}{\Set{ y' }} $.
\end{definition}

\subsection{Quality Criteria for Encodings}
\label{sec:qualityCriteria}

Within this paper we consider two encodings, (1) an encoding from \piSep into \piAsyn presented in \cite{nestmann00}, denoted by $ \encodingSepAsyn $, and (2) a new encoding from \piMix into \piAsyn, denoted by $ \encodingMixAsyn $. To measure the quality of such an encoding, Gorla \cite{gorla10} suggested five criteria well suited for language comparison. Accordingly, we consider an encoding to be ``good'', if it satisfies Gorla's five criteria.

As in \cite{gorla10}, an encoding is a mapping from a source into a target language; in our case, \piMix and \piSep are source languages and \piAsyn is the target language. To distinguish terms on these languages or definitions for the respective encodings, we use $ \indexMix $, $ \indexSep $, and $ \indexAsyn $ as super- and subscripts. Thereby, the superscript usually refers to the source and the subscript the target language. Moreover, we use $ S, S', S_1, \ldots $ to range over terms of the source languages and $ T, T', T_1, \ldots $ to range over terms of the target language.

The five conditions are divided into two structural and three semantic criteria. The structural criteria include (1) \emph{compositionality} and (2) \emph{name invariance}. The semantic criteria include (3) \emph{operational correspondence}, (4) \emph{divergence reflection} and (5) \emph{success sensitiveness}. Note that for the definition of name invariance and operational correspondence a behavioural equivalence $ \asymp $ on the target language is assumed. Its purpose is to describe the abstract behaviour of a target process, where abstract basically means with respect to the behaviour of the source term.

Intuitively, an encoding is compositional if the translation of an operator depends only on the translation of its parameters. To mediate between the translations of the parameters the encoding defines a unique context for each operator, whose arity is the arity of the operator. Moreover, the context can be parametrised on the free names of the corresponding source term.

\begin{definition}[Criterion 1: Compositionality]
\label{def:compositionality}
	The encoding $ \arbitraryEncoding $ is \emph{compositional} if, for every k-ary operator $ \mathbf{op} $ of $ \piMix $ and for every subset of names $ N $, there exists a k-ary context $ \Context{N}{\mathbf{op}}{\hole_1, \ldots , \hole_k} $ such that, for all $ S_1, \ldots, S_k $ with $ \FreeNames{S_1} \cup \ldots \cup \FreeNames{S_k} = N $, it holds that
	\begin{align*}
		\ArbitraryEncoding{\mathbf{op}\left( S_1, \ldots, S_k \right)} = \Context{N}{\mathbf{op}}{\ArbitraryEncoding{S_1}, \ldots , \ArbitraryEncoding{S_k}}.
	\end{align*}
\end{definition}

\noindent
If the context is again the original operator, i.e., if an operator is translated by encoding its parameters and apply the renaming policy, as in $ \ArbitraryEncoding{\RestrictedTerm{x}{P}} = \RestrictedTerm{\ArbitraryRenamingPolicy{x}}{\ArbitraryEncoding{P}} $, we call this encoding \emph{\clean}. Note that Gorla requires that the parallel composition operator ``$ \mid $'' is binary and unique in the source as well as in the target language. Thus, compositionality prevents from introducing a global coordinator or to use global knowledge, i.e., knowledge about surrounding source terms or the structure of the parameters.

The second structural criterion states that the encoding should not depend on specific names used in the source term. This is important, since sometimes it is necessary to translate a source term name into a sequences of names or reserve some names for the encoding function. To ensure that there are no conflicts between these reserved names and the source term names, the encoding is equipped with a renaming policy $ \arbitraryRenamingPolicy $, i.e., a substitution from names into sequences of names\footnote{To keep distinct names distinct Gorla assumes that $ \forall n, m \in \names \logdot n \neq m $ implies $ \ArbitraryRenamingPolicy{n} \cap \ArbitraryRenamingPolicy{m} = \emptyset $, where $ \ArbitraryRenamingPolicy{x} $ is simply considered as set here.}. Since we translate source term names only into single names, the renaming policies introduced by $ \encodingSepAsyn $ and $ \encodingMixAsyn $ are injective substitutions from names into names. Based on such a renaming policy an encoding is independent of specific names if it preserves all substitutions $ \sigma $ on source terms by a substitution $ \sigma' $ on target terms such that $ \sigma' $ respects the changes made by the renaming policy.

\begin{definition}[Criterion 2: Name Invariance]
\label{def:nameInvariance}
	The encoding $ \arbitraryEncoding $ is \emph{name invariant} if, for every $ S $ and $ \sigma $, it holds that
	\begin{align*}
		\ArbitraryEncoding{\sigma\left( S \right)} \begin{cases} \equivAlpha \sigma'\left( \ArbitraryEncoding{S} \right) & \quad \text{if } \sigma \text{ is injective}\\ \asymp \sigma'\left( \ArbitraryEncoding{S} \right) & \quad \text{otherwise} \end{cases}
	\end{align*}
	where $ \sigma' $ is such that $ \ArbitraryRenamingPolicy{\sigma\left( n \right)} = \sigma'\left( \ArbitraryRenamingPolicy{n} \right) $ for every $ n \in \names $.
\end{definition}

The first semantic criterion and usually the most elaborate one to prove is operational correspondence, which consists of a soundness and a completeness condition. \emph{Completeness} requires that every computation of a source term can be \simulated by its translation, i.e., the translation does not reduce the computations of the source term. Note that encodings often translate single source term steps into a sequence of target term steps. We call such a sequence an \emph{\simulation}\!\! of the corresponding source term step. \emph{Soundness} requires that every computation of a target term corresponds to some computation of the corresponding source term, i.e., the translation does not introduce new computations.

\begin{definition}[Criterion 3: Operational Correspondence]
\label{def:operationalCorrespondence}
	Let $ \arbitraryEncoding $ be an arbitrary encoding.  Then, two operational criteria are defined as follows. 
        \begin{center}
          \begin{tabular}{ll}
            \emph{Completeness}: & For all $ S \steps S' $, it holds that $ \ArbitraryEncoding{S} \steps \asymp \ArbitraryEncoding{S'} $.\\
            \emph{Soundness}: & For all $ \ArbitraryEncoding{S} \steps T $, there exists an $ S' $ such that\\
            & $ S \steps S' $ and $ T \steps \asymp \ArbitraryEncoding{S'} $.
          \end{tabular}
        \end{center}
\end{definition}

\noindent
Note that the definition of operational correspondence relies on the
equivalence $ \asymp $ to get rid of junk possibly left over within
computations of target terms (compare to Section \ref{sec:transBarbBisim} for a discussion of that equivalence). Sometimes, we refer to the completeness criterion of operational correspondence as \emph{operational completeness} and, accordingly, for the soundness criterion as \emph{operational soundness}.

The next criterion concerns the role of infinite computations in encodings.

\begin{definition}[Criterion 4: Divergence Reflection]
\label{def:divergenceReflection}
  The encoding $ \arbitraryEncoding $ \emph{reflects divergence} if, for every $ S $, $ \ArbitraryEncoding{S} \step^{\omega} $ implies $ S \step^{\omega} $.
\end{definition}

\noindent
The last criterion links the behaviour of source terms to the behaviour of their encodings.
With Gorla \cite{gorla10}, we assume a \emph{success} operator $ \success $ as part of the syntax of both the source and the target language, i.e., of 
\piMix, \piSep, and \piAsyn. Since $ \success $ can not be further reduced, the operational semantics is left unchanged in all three cases. Moreover, note that $ \Names{\success} = \FreeNames{\success} = \BoundNames{\success} = \emptyset $, so also interplay of $\success$ with the rules of structural congruence is smooth and does not require explicit treatment. The test for reachability of success is standard.

\begin{definition}[Success]
\label{def:success}
	A process $ P \in \proc $ \emph{may lead to success}, denoted as $ P \reachSuccess $, if (and only if) it is reducible to a process containing a top-level unguarded occurrence of $ \success $, i.e. $ \exists P', P'' \in \proc \logdot P \steps P' \wedge P' \equiv P'' \mid \success $.
\end{definition}

\noindent
Note that we choose may-testing here. Finally, an encoding preserves the abstract behaviour of the source term if it and its encoding answer the tests for success in exactly the same way.

\begin{definition}[Criterion 5: Success Sensitiveness]
\label{def:succesSensitiveness}
  The encoding $ \arbitraryEncoding $ is \emph{success sensitive} if, for every $ S $, $ S \reachSuccess $ if and only if $ \ArbitraryEncoding{S} \reachSuccess $.
\end{definition}

\noindent
This criterion only links the behaviours of source terms and their literal translations, but not of their continuations. To do so, Gorla relates success sensitiveness and operational correspondence by requiring that the equivalence on the target language never relates two processes 
with different success behaviours.

\begin{definition}[Success Respecting]
\label{def:successRespecting}
	$ \asymp \; \subseteq \piAsynProc \times \piAsynProc $ is \emph{success respecting} if, for every $ P $ and $ Q $ with $ P \reachSuccess $ and $ Q \not\reachSuccess $, it holds that $ P \not\asymp Q $.
\end{definition}

\section{Correctness of the Encodings} \label{sec:CorrectnessEncodings}

Let us first present the full representations of the encodings $ \encodingSepAsyn $ in Figure \ref{fig:encodingSepAsyn} and $ \encodingMixAsyn $ in Figure \ref{fig:encodingMixAsyn}. Note that in \cite{nestmann00,nestmannPierce00} slightly different version of \piSep and \piAsyn are used, namely $ \tau $ is no prefix and there are neither a match operator nor a success operator in the syntax of \piSep and \piAsyn. We choose the respective encodings to be \clean \ except for source terms guarded by $ \tau $. Since $ \tau $ guarded terms can reduce without a communication partner, we implement their translation by a simple test-statement on their sum lock in both encodings.

\begin{figure}[ht]
	\begin{align*}
		\EncodingSepAsyn{\RestrictedTerm{x}{P}} & \deff \RestrictedTerm{\RenamingPolicySepAsyn{x}}{\EncodingSepAsyn{P}}\\
		\EncodingSepAsyn{P \mid Q} & \deff \EncodingSepAsyn{P} \mid \EncodingSepAsyn{Q}\\
		\EncodingSepAsyn{\Match{a}{b}P} & \deff \Match{\RenamingPolicySepAsyn{a}}{\RenamingPolicySepAsyn{b}}\EncodingSepAsyn{P}\\
		\EncodingSepAsyn{\sum_{i \in \indexSet} \guard_i.P_i} & \deff \RestrictedTerm{\sumLock}{\left( \Output{\sumLock}{\true} \mid \prod_{i \in \indexSet} \EncodingSepAsyn{\guard_i.P_i} \right)}\\
		\EncodingSepAsyn{\tau.P} & \deff \Test{\sumLock}{\left( \Output{\sumLock}{\false} \mid \EncodingSepAsyn{P} \right)}{\Output{\sumLock}{\false}}\\
		\EncodingSepAsyn{\Output{y}{z}.P} & \deff \RestrictedTerm{\senderLock}{\left( \Output{\RenamingPolicySepAsyn{y}}{\sumLock, \senderLock, \RenamingPolicySepAsyn{z}} \mid \In{\senderLock}.\EncodingSepAsyn{P} \right)}\\
		\EncodingSepAsyn{\Input{y}{x}.P} & \deff \RestrictedTerm{\receiverLock}{\big( \Out{\receiverLock} \mid \ReplicateIn{\receiverLock}.\Input{\RenamingPolicySepAsyn{y}}{\sumLock', \senderLock, \RenamingPolicySepAsyn{x}}.\BigTest{\sumLock}{\BigTest{\sumLock'}{\Output{\sumLock}{\false} \mid \Output{\sumLock'}{\false} \mid \Out{\senderLock} \mid \EncodingSepAsyn{P}}{\Output{\sumLock}{\true} \mid \Output{\sumLock'}{\false} \mid \Out{\receiverLock}}}{\Output{\sumLock}{\false} \mid \Output{\RenamingPolicySepAsyn{y}}{\sumLock', \senderLock, \RenamingPolicySepAsyn{x}} \big)}}\\
		\EncodingSepAsyn{\ReplicateInput{y}{x}.P} & \deff \ReplicateInput{\RenamingPolicySepAsyn{y}}{\sumLock, \senderLock, \RenamingPolicySepAsyn{x}}.\Test{\sumLock}{\Output{\sumLock}{\false} \mid \Out{\senderLock} \mid \EncodingSepAsyn{P}}{\Output{\sumLock}{\false}}\\
		\EncodingSepAsyn{\success} & \deff \success
	\end{align*}
	\begin{center}
		Here $ \renamingPolicySepAsyn $ is some arbitrary injective substitution such that $ \forall n \in \names \logdot \RenamingPolicySepAsyn{n} \cap \Set{ \sumLock, l', \receiverLock, \senderLock } = \emptyset $.
	\end{center}
	\caption{Encoding from \piSep into \piAsyn.} \label{fig:encodingSepAsyn}
\end{figure}

\begin{figure}[htp]
	\begin{align*}
		\EncodingMixAsyn{\RestrictedTerm{x}{P}} & \deff \RestrictedTerm{\RenamingPolicyMixAsyn{x}}{\EncodingMixAsyn{P}}\\
		\EncodingMixAsyn{P \mid Q} & \deff \RestrictedTerm{\matchingCoordinatorOut, \matchingCoordinatorIn, \coordinatorUpOut, \coordinatorUpIn, \coordinatorMatchingOut, \coordinatorMatchingIn}{\big(\\
				& \hspace*{1em} \RestrictedTerm{\parallelChannelOut, \parallelChannelIn}{\big( \EncodingMixAsyn{P} \mid \Forward{\parallelChannelIn}{\Set{ \matchingCoordinatorIn, \coordinatorUpIn }} \mid \Forward{\parallelChannelOut}{\Set{ \matchingCoordinatorOut, \coordinatorUpOut }} \big)}\\
				& \hspace*{1em} \mid \RestrictedTerm{\parallelChannelOut, \parallelChannelIn}{\big( \begin{aligned}[t]
						& \EncodingMixAsyn{Q}\\
						& \mid \Output{\coordinatorMatchingOut}{\matchingCoordinatorIn} \mid \ReplicateInput{\coordinatorMatchingOut}{m_i}.\Input{\parallelChannelOut}{y, l_s, s, z}.\big(\\
						& \hspace*{2em} \RestrictedTerm{\matchingUpIn}{( \begin{aligned}[t]
								& \ReplicateInput{\matchingCoordinatorIn}{y', l_r, r}.\left( \Match{y'}{y}\Output{r}{l_r, l_s, l_s, s, z} \mid \Output{\matchingUpIn}{y', l_r, r} \right)\\
								& \mid \RestrictedTerm{\matchingCoordinatorIn}{\left( \Forward{\matchingUpIn}{\matchingCoordinatorIn} \mid \Output{\coordinatorMatchingOut}{\matchingCoordinatorIn} \right)} )
							\end{aligned}\\
						& \hspace*{2em} \mid \Output{\coordinatorUpOut}{y, l_s, s, z}} \big)\\
						& \mid \Output{\coordinatorMatchingIn}{\matchingCoordinatorOut} \mid \ReplicateInput{\coordinatorMatchingIn}{m_o}.\Input{\parallelChannelIn}{y, l_r, r}.\big(\\
						& \hspace*{2em} \RestrictedTerm{\matchingUpOut}{( \begin{aligned}[t]
								& \ReplicateInput{\matchingCoordinatorOut}{y', l_s, s, z}.\left( \Match{y'}{y}\Output{r}{l_s, l_r, l_s, s, z} \mid \Output{\matchingUpOut}{y', l_s, s, z} \right)\\
								& \mid \RestrictedTerm{\matchingCoordinatorOut}{\left( \Forward{\matchingUpOut}{\matchingCoordinatorOut} \mid \Output{\coordinatorMatchingIn}{\matchingCoordinatorOut} \right)} )
							\end{aligned}\\
						& \hspace*{2em} \mid \Output{\coordinatorUpIn}{y, l_r, r}} \big) \big)
					\end{aligned}}\\
				& \hspace*{1em} \mid \Forward{\coordinatorUpOut}{\parallelChannelOut} \mid \Forward{\coordinatorUpIn}{\parallelChannelIn} \big)}\\
		\EncodingMixAsyn{\Match{a}{b}P} & \deff \Match{\RenamingPolicyMixAsyn{a}}{\RenamingPolicyMixAsyn{b}}\EncodingMixAsyn{P}\\
		\EncodingMixAsyn{\sum_{i \in \indexSet} \guard_i.P_i} & \deff \RestrictedTerm{\sumLock}{\left( \Output{l}{\true} \mid \prod_{i \in \indexSet} \EncodingMixAsyn{\guard_i.P_i} \right)}\\
		\EncodingMixAsyn{\tau.P} & \deff \Test{\sumLock}{\Output{\sumLock}{\false} \mid \EncodingMixAsyn{P}}{\Output{\sumLock}{\false}}\\
		\EncodingMixAsyn{\Output{y}{z}.P} & \deff \RestrictedTerm{\senderLock}{\left( \Output{\parallelChannelOut}{\RenamingPolicyMixAsyn{y}, \sumLock, \senderLock, \RenamingPolicyMixAsyn{z}} \mid \In{\senderLock}.\EncodingMixAsyn{P} \right)}\\
		\EncodingMixAsyn{\Input{y}{x}.P} & \deff \RestrictedTerm{\receiverLock}{\big( \Output{\parallelChannelIn}{\RenamingPolicyMixAsyn{y}, \sumLock, \receiverLock} \mid \ReplicateInput{\receiverLock}{l_1, l_2, -, s, \RenamingPolicyMixAsyn{x}}.\\
			& \hspace*{3em} \Test{l_1}{\left( \Test{l_2}{\Output{l_1}{\false} \mid \Output{l_2}{\false} \mid \Out{\senderLock} \mid \EncodingMixAsyn{P}}{\Output{l_1}{\true} \mid \Output{l_2}{\false}} \right)}{\Output{l_1}{\false} \big)}}\\
		\EncodingMixAsyn{\ReplicateInput{y}{x}.P} & \deff \RestrictedTerm{\sumLock, \receiverLock, \coordinatorRepA, \coordinatorRepB, \matchingReceiverOut, \matchingReceiverIn}{\big(\\
				& \hspace*{1em} \Output{\parallelChannelIn}{\RenamingPolicySepAsyn{y}, \sumLock, \receiverLock} \mid \ReplicateInput{\receiverLock}{-, -, l_s, s, \RenamingPolicySepAsyn{x}}.\Test{l_s}{\Output{l_s}{\false} \mid \Out{s} \mid \Output{\coordinatorRepA}{\RenamingPolicyMixAsyn{x}}}{\Output{l_s}{\false}}\\
				& \hspace*{1em} \mid \Output{\matchingReceiverIn}{\RenamingPolicySepAsyn{y}, \sumLock, \receiverLock} \mid \Output{\sumLock}{\true} \mid \Output{\coordinatorRepB}{\matchingReceiverOut, \matchingReceiverIn}\\
				& \hspace*{1em} \mid \ReplicateInput{\coordinatorRepA}{\RenamingPolicyMixAsyn{x}}.\Input{\coordinatorRepB}{\matchingReceiverOut, \matchingReceiverIn}.\RestrictedTerm{\matchingCoordinatorOut, \matchingCoordinatorIn, \coordinatorUpOut, \coordinatorUpIn, \matchingReceiverUpOut, \matchingReceiverUpIn, \coordinatorMatchingOut, \coordinatorMatchingIn, \matchingUpOut, \matchingUpIn}{\big(\\
				& \hspace*{2em} \begin{aligned}[t]
						& \Forward{\matchingReceiverOut}{\Set{ \matchingCoordinatorOut, \matchingReceiverUpOut }} \mid \Forward{\matchingReceiverIn}{\Set{ \matchingCoordinatorIn, \matchingReceiverUpIn }}\\
						& \mid \RestrictedTerm{\parallelChannelOut, \parallelChannelIn}{\big( \begin{aligned}[t]
								& \EncodingMixAsyn{P}\\
								& \mid \Output{\coordinatorMatchingOut}{\matchingCoordinatorIn} \mid \ReplicateInput{\coordinatorMatchingOut}{m_i}.\Input{\parallelChannelOut}{y, l_s, s, z}.\big(\\
								& \hspace*{1em} \begin{aligned}[t]
										& \RestrictedTerm{\matchingUpIn}{\big( \begin{aligned}[t]
												& \ReplicateInput{m_i}{y', l_r, r}.\left( \Match{y}{y'}\Output{r}{l_r, l_s, l_s, s, z} \mid \Output{\matchingUpIn}{y', l_r, r} \right)\\
												& \mid \RestrictedTerm{\matchingCoordinatorIn}{\left( \Forward{\matchingUpIn}{\matchingCoordinatorIn} \mid \Output{\coordinatorMatchingOut}{\matchingCoordinatorIn} \right)} \big)
											\end{aligned}\\
										& \mid \Output{\coordinatorUpOut}{y, l_s, s, z} \big) \big)
									\end{aligned}}\\
								& \mid \Output{\coordinatorMatchingIn}{\matchingCoordinatorOut} \mid \ReplicateInput{\coordinatorMatchingIn}{m_o}.\Input{\parallelChannelIn}{y, l_r, r}.\big(\\
								& \hspace*{1em} \begin{aligned}[t]
										& \RestrictedTerm{\matchingUpOut}{\big( \begin{aligned}[t]
												& \ReplicateInput{m_o}{y', l_s, s, z}.\left( \Match{y}{y'}\Output{r}{l_s, l_r, l_s, s, z} \mid \Output{\matchingUpOut}{y', l_s, s, z} \right)\\
												& \mid \RestrictedTerm{\matchingCoordinatorOut}{\left( \Forward{\matchingUpOut}{\matchingCoordinatorOut} \mid \Output{\coordinatorMatchingIn}{\matchingCoordinatorOut} \right)} \big)
											\end{aligned}\\
										& \mid \Output{\coordinatorUpIn}{y, l_r, r} \big)
									\end{aligned}}
							\end{aligned}\\
						& \mid \RestrictedTerm{\matchingReceiverOut, \matchingReceiverIn}{\left( \Output{\coordinatorRepB}{\matchingReceiverOut, \matchingReceiverIn} \mid \Forward{\coordinatorUpOut}{\Set{ \parallelChannelOut, \matchingReceiverOut}} \mid \Forward{\matchingReceiverUpOut}{\matchingReceiverOut} \mid \Forward{\coordinatorUpIn}{\Set{ \parallelChannelIn, \matchingReceiverIn }} \mid \Forward{\matchingReceiverUpIn}{\matchingReceiverIn} \right) \big) \big)}
					\end{aligned}}}}\\
		\EncodingSepAsyn{\success} & \deff \success	
	\end{align*}
	\begin{center}
		Here $ \renamingPolicyMixAsyn $ is some arbitrary injective substitution such that $ \forall n \in \names \logdot \RenamingPolicyMixAsyn{n} \cap N = \emptyset $, where $ N $ is the set of reserved names, i.e., $ N = \Set{ \parallelChannelOut, \parallelChannelIn, \coordinatorUpOut, \coordinatorUpIn, \matchingCoordinatorOut, \matchingCoordinatorIn, \matchingUpOut, \matchingUpIn, \coordinatorMatchingOut, \coordinatorMatchingIn, \sumLock, \sumLock_s, \sumLock_r, \sumLock_1, \sumLock_2, \senderLock, \receiverLock, \coordinatorRepA, \coordinatorRepB, \matchingReceiverIn, \matchingReceiverOut, \matchingReceiverUpIn, \matchingReceiverUpOut, y, y', z, t, f } $.
	\end{center}
	\caption{Encoding $ \encodingMixAsyn $ from \piMix into \piAsyn.} \label{fig:encodingMixAsyn}
\end{figure}

In the following we will argue for the correctness of these encodings with respect to the criteria of Gorla presented in Section \ref{sec:qualityCriteria}.

\subsection{Structural Criteria} \label{sec:structuralCriteria}

The first two criteria to prove are the structural criteria; compositionality and name invariance. An encoding is compositional if it defines a fixed context for each operator including holes for the translation of its parameters. By Definition \ref{def:compositionality} of compositionality the context is allowed to depend on the free names of the parameters. However, both presented encodings, $ \encodingSepAsyn $ and $ \encodingMixAsyn $, do not use that feature, i.e., the contexts do not depend on any names. By Figure \ref{fig:encodingSepAsyn} and Figure \ref{fig:encodingMixAsyn} both encodings are obviously compositional.

Let us have a closer look at the contexts. In the encodings of restriction, matching and success the context is used only to translate source term names according to the renaming policy. Apart from that the encodings are \clean. The encoding of the sum operator inserts a positive instantiation of a fresh sum lock and splits up the encodings of the summands in parallel because there is no sum operator in the target language. Therefore of course we have to consider the sum operator as binary operator with an index set and a set of its summands as parameters, or as unary operator with the set or list\footnote{Usually an unordered set of summands suffice to describe a sum since usually we consider sums as being reflexive and symmetric, i.e., $ A + A = A $ and $ A + B = B + A $. If for some reasons we have to abandon reflexivity and/or symmetry, e.g. in case of a randomised version of the calculi, an ordered list might be the better choice to describe a sum.} of its summands as parameter. We left the question whether there is an encoding from \piMix into \piAsyn with the binary sum operator as an open question to further research. The encodings of input and output guarded terms and the encoding of terms guarded by $ \tau $ introduce rather small contexts. However, in case of $ \encodingMixAsyn $ the contexts introduced to translate the binary parallel operator and replicated input are rather complicated and huge. Remember that we claim in Section \ref{sec:qualityCriteria} that the parallel operator is binary. Comparing its encoding with the encoding of the sum operator we observe that this claim may be crucial because it forbids the introduction of a global coordinator for all parallel terms as the sum lock is for all the summands of a sum.

Name invariance follows by the fact that names are translated into single names again and that conflicts between names used by the encoding functions and translated source term names are ruled out by the renaming policy.

\begin{lemma} \label{lem:nameInvarianceSepAsyn}
	The encoding $ \encodingSepAsyn $ is name invariant.
\end{lemma}

\begin{proof}
	By Definition \ref{def:nameInvariance} it suffice to show, that:
	\begin{align*}
		\forall S \in \piSepProc \logdot \forall \sigma \subseteq \names \logdot \exists \sigma' \subseteq \names \logdot \EncodingSepAsyn{\sigma\left( S \right)} \equivAlpha \sigma'\left( \EncodingSepAsyn{S} \right) \wedge \forall z \in \names \logdot \RenamingPolicySepAsyn{\sigma\left( z \right)} = \sigma'\left( \RenamingPolicySepAsyn{z} \right)
	\end{align*}
	Without loss of generality let $ \sigma = \Set{ \Subst{y_1}{x_1}, \ldots \Subst{y_n}{x_n} } $ for some $ n \in \nat $. We choose
	\begin{align*}
		\sigma' \deff \Set{ \Subst{\RenamingPolicySepAsyn{y_1}}{\RenamingPolicySepAsyn{x_1}}, \ldots, \Subst{\RenamingPolicySepAsyn{y_n}}{\RenamingPolicySepAsyn{x_n}} }.
	\end{align*}
	So $ \forall z \in \names \logdot \RenamingPolicySepAsyn{\sigma\left( z \right)} = \sigma'\left( \RenamingPolicySepAsyn{z} \right) $. We proceed with an induction over the structure of $ S $.
	\begin{description}
		\item[Base Case:] Since $ \Names{\success} = \emptyset = \Names{\nullTerm} $ and $ \FreeNames{\RestrictedTerm{\sumLock}{\Output{\sumLock}{\true}}} = \emptyset $, we have $ \EncodingSepAsyn{\sigma\left( \success \right)} = \EncodingSepAsyn{\success} = \success = \sigma'\left( \success \right) = \sigma'\left( \EncodingSepAsyn{\success} \right) $ and $ \EncodingSepAsyn{\sigma\left( \nullTerm \right)} = \EncodingSepAsyn{\nullTerm} = \RestrictedTerm{\sumLock}{\Output{\sumLock}{\true}} = \sigma'\left( \RestrictedTerm{\sumLock}{\Output{\sumLock}{\true}} \right) = \sigma'\left( \EncodingSepAsyn{\nullTerm} \right) $.
		\item[Induction Hypothesis:] $ \forall S \in \piSepProc \logdot \forall \sigma \subseteq \names \logdot \exists \sigma' \subseteq \names \logdot \EncodingSepAsyn{\sigma\left( S \right)} \equivAlpha \sigma'\left( \EncodingSepAsyn{S} \right) $
		\item[Induction Step:] Let $ x' \in \names $ be such that $ x' \notin \Names{\sigma} \cup \FreeNames{P} $. Then $ \RenamingPolicyMixAsyn{x'} \notin \Names{\sigma} $. Moreover, since $ \forall n \in \names \logdot \RenamingPolicyMixAsyn{n} \notin \Set{ \sumLock, \sumLock', \receiverLock, \senderLock } $, we have $ \Names{\sigma'} \cap \Set{ \sumLock, \sumLock', \receiverLock, \senderLock } = \emptyset $. Note that:
			\begin{align*}
				\EncodingSepAsyn{\sigma\left( \Set{ \Subst{x'}{x} }\left( P \right) \right)} & = \EncodingSepAsyn{\Set{ \Subst{x'}{x}, \Subst{y_i}{x_i} \mid \Subst{y_i}{x_i} \in \sigma \wedge x_i \neq x }\left( P \right)}\\
				& \equivAlpha \Set{ \Subst{\RenamingPolicySepAsyn{x'}}{\RenamingPolicySepAsyn{x}}, \Subst{\RenamingPolicySepAsyn{y_i}}{\RenamingPolicySepAsyn{x_i}} \mid \Subst{y_i}{x_i} \in \sigma \wedge x_i \neq x }\left( \EncodingSepAsyn{P} \right) & \text{ by IH}\\
				& = \Set{ \Subst{\RenamingPolicySepAsyn{x'}}{\RenamingPolicySepAsyn{x}}, \Subst{\RenamingPolicySepAsyn{y_i}}{\RenamingPolicySepAsyn{x_i}} \mid \Subst{\RenamingPolicySepAsyn{y_i}}{\RenamingPolicySepAsyn{x_i}} \in \sigma' \wedge \RenamingPolicySepAsyn{x_i} \neq \RenamingPolicySepAsyn{x} }\left( \EncodingSepAsyn{P} \right)\\
				& = \sigma'\left( \Set{ \Subst{\RenamingPolicySepAsyn{x'}}{\RenamingPolicySepAsyn{x}} }\left( \EncodingSepAsyn{P} \right) \right)
			\end{align*}
			We proceed by a case split.
			\begin{description}
				\item[Case of $ S = \RestrictedTerm{x}{P} $:] Then
					\begin{align*}
						\EncodingSepAsyn{\sigma\left( S \right)} & = \EncodingSepAsyn{ \sigma\left( \RestrictedTerm{x}{P} \right)} \equivAlpha \EncodingSepAsyn{ \sigma\left( \RestrictedTerm{x'}{\Set{ \Subst{x'}{x} }\left( P \right)} \right) } = \EncodingSepAsyn{ \RestrictedTerm{x'}{\sigma\left( \Set{ \Subst{x'}{x} }\left( P \right) \right)} }\\
						& = \RestrictedTerm{\RenamingPolicySepAsyn{x'}}{\EncodingSepAsyn{ \sigma\left( \Set{ \Subst{x'}{x} }\left( P \right) \right)} } \equivAlpha \RestrictedTerm{\RenamingPolicySepAsyn{x'}}{\sigma'\left( \Set{ \Subst{\RenamingPolicySepAsyn{x'}}{\RenamingPolicySepAsyn{x}} }\left( \EncodingSepAsyn{P} \right) \right)}\\
						& = \sigma'\left( \RestrictedTerm{\RenamingPolicySepAsyn{x'}}{\Set{ \Subst{\RenamingPolicySepAsyn{x'}}{\RenamingPolicySepAsyn{x}} }\left( \EncodingSepAsyn{P} \right)} \right) \equivAlpha \sigma'\left( \RestrictedTerm{\RenamingPolicySepAsyn{x}}{\EncodingSepAsyn{P}} \right) = \sigma'\left( \EncodingSepAsyn{S} \right).
					\end{align*}
				\item[Case of $ S = P \mid Q $:] Then
					\begin{align*}
						\EncodingSepAsyn{\sigma\left( S \right)} & = \EncodingSepAsyn{\sigma\left( P \mid Q \right)} = \EncodingSepAsyn{\sigma\left( P \right) \mid \sigma\left( Q \right)} = \EncodingSepAsyn{\sigma\left( P \right)} \mid \EncodingSepAsyn{\sigma\left( Q \right)} \equivAlpha \sigma'\left( \EncodingSepAsyn{P} \right) \mid \sigma'\left( \EncodingSepAsyn{Q} \right)\\
						& = \sigma'\left( \EncodingSepAsyn{P} \mid \EncodingSepAsyn{Q} \right) = \sigma'\left( \EncodingSepAsyn{P \mid Q} \right).
					\end{align*}
				\item[Case of $ S = \Match{a}{b}P $:] Then
					\begin{align*}
						\EncodingSepAsyn{\sigma\left( S \right)} & = \EncodingSepAsyn{\sigma\left( \Match{a}{b}P \right)} = \EncodingSepAsyn{\Match{\sigma\left( a \right)}{\sigma\left( b \right)}\sigma\left( P \right)} = \Match{\RenamingPolicySepAsyn{\sigma\left( a \right)}}{\RenamingPolicySepAsyn{\sigma\left( b \right)}}\EncodingSepAsyn{\sigma\left( P \right)}\\
						& \equivAlpha \Match{\sigma'\left( \RenamingPolicySepAsyn{a} \right)}{\sigma'\left( \RenamingPolicySepAsyn{b} \right)}\sigma'\left( \EncodingSepAsyn{P} \right) = \sigma'\left( \Match{\RenamingPolicySepAsyn{a}}{\RenamingPolicySepAsyn{b}}\EncodingSepAsyn{P} \right) = \sigma'\left( \EncodingSepAsyn{S} \right).
					\end{align*}
				\item[Case of $ S = \sum_{i \in \indexSet} \guard_i.P_i $:] Then
					\begin{align*}
						\EncodingSepAsyn{\sigma\left( S \right)} & = \EncodingSepAsyn{\sigma\left( \sum_{i \in \indexSet} \guard_i.P_i \right)} = \EncodingSepAsyn{\sum_{i \in \indexSet} \sigma\left( \guard_i.P_i \right)} = \RestrictedTerm{\sumLock}{\left( \Output{\sumLock}{\true} \mid \prod_{i \in \indexSet} \EncodingSepAsyn{\sigma\left( \guard_i.P_i \right)} \right)}\\
						& \equivAlpha \RestrictedTerm{\sumLock}{\left( \Output{\sumLock}{\true} \mid \prod_{i \in \indexSet} \sigma'\left( \EncodingSepAsyn{\guard_i.P_i} \right) \right)} = \sigma'\left( \RestrictedTerm{\sumLock}{\left( \Output{\sumLock}{\true} \mid \prod_{i \in \indexSet} \EncodingSepAsyn{\guard_i.P_i} \right)} \right) = \sigma'\left( \EncodingSepAsyn{S} \right).
					\end{align*}
				\item[Case of $ S = \tau.P $:] Then
					\begin{align*}
						\EncodingSepAsyn{\sigma\left( S \right)} & = \EncodingSepAsyn{\sigma\left( \tau.P \right)} = \EncodingSepAsyn{\tau.\sigma\left( P \right)} = \Test{\sumLock}{\Output{\sumLock}{\false} \mid \EncodingSepAsyn{\sigma\left( P \right)}}{\Output{\sumLock}{\false}}\\
						& \equivAlpha \Test{\sumLock}{\Output{\sumLock}{\false} \mid \sigma'\left( \EncodingSepAsyn{P} \right)}{\Output{\sumLock}{\false}} = \sigma'\left( \Test{\sumLock}{\Output{\sumLock}{\false} \mid \EncodingSepAsyn{P}}{\Output{\sumLock}{\false}} \right)\\
						& = \sigma'\left( \EncodingSepAsyn{S} \right).
					\end{align*}
				\item[Case of $ S = \Output{y}{z}.P $:] Then
					\begin{align*}
						\EncodingSepAsyn{\sigma\left( S \right)} & = \EncodingSepAsyn{\sigma\left( \Output{y}{z}.P \right)} = \EncodingSepAsyn{\Output{\sigma\left( y \right)}{\sigma\left( z \right)}.\sigma\left( P \right)}\\
						& = \RestrictedTerm{\senderLock}{\left( \Output{\RenamingPolicySepAsyn{\sigma\left( y \right)}}{\sumLock, \senderLock, \RenamingPolicySepAsyn{\sigma\left( x \right)}} \mid \Test{\senderLock}{\EncodingSepAsyn{\sigma\left( P \right)}}{\nullTerm} \right)}\\
						& \equivAlpha \RestrictedTerm{\senderLock}{\left( \Output{\sigma'\left( \RenamingPolicySepAsyn{y} \right)}{\sumLock, \senderLock, \sigma'\left( \RenamingPolicySepAsyn{x} \right)} \mid \Test{\senderLock}{\sigma'\left( \EncodingSepAsyn{P} \right)}{\nullTerm} \right)}\\
						& = \sigma'\left( \RestrictedTerm{\senderLock}{\left( \Output{\RenamingPolicySepAsyn{y}}{\sumLock, \senderLock, \RenamingPolicySepAsyn{x}} \mid \Test{\senderLock}{\EncodingSepAsyn{P}}{\nullTerm} \right)} \right) = \sigma'\left( \EncodingSepAsyn{S} \right).
					\end{align*}
				\item[Case of $ S = \Input{y}{x}.P $:] Then
					\begin{align*}
						\EncodingSepAsyn{\sigma\left( S \right)} & = \EncodingSepAsyn{\sigma\left( \Input{y}{x}.P \right)} \equivAlpha \EncodingSepAsyn{\sigma\left( \Input{y}{x'}.\Set{ \Subst{x'}{x} }\left( P \right) \right)} = \EncodingSepAsyn{\Input{\sigma\left( y \right)}{x'}.\sigma\left( \Set{ \Subst{x'}{x} }\left( P \right) \right)}\\
						& = \RestrictedTerm{\receiverLock}{\big( \Out{\receiverLock} \mid \ReplicateIn{\receiverLock}.\Input{\RenamingPolicySepAsyn{\sigma\left( y \right)}}{\sumLock', \senderLock, \RenamingPolicySepAsyn{x'}}.\\
						& \hspace*{3em} \BigTest{\sumLock}{\Test{\sumLock'}{\Output{\sumLock}{\false} \mid \Output{\sumLock'}{\false} \mid \Output{\senderLock}{\true} \mid \EncodingSepAsyn{\sigma\left( \Set{ \Subst{x'}{x} }\left( P \right) \right)}}{\Output{\sumLock}{\true} \mid \Output{\sumLock'}{\false} \mid \Output{\senderLock}{\false} \mid \Out{\receiverLock}}}{\Output{\sumLock}{\false} \mid \Output{\RenamingPolicySepAsyn{\sigma\left( y \right)}}{\sumLock', \senderLock, \RenamingPolicySepAsyn{x'}} \big)}}\\
						& \equivAlpha \RestrictedTerm{\receiverLock}{\big( \Out{\receiverLock} \mid \ReplicateIn{\receiverLock}.\Input{\sigma'\left( \RenamingPolicySepAsyn{y} \right)}{\sumLock', \senderLock, \RenamingPolicySepAsyn{x'}}.\\
						& \hspace*{3em} \BigTest{\sumLock}{\BigTest{\sumLock'}{\Output{\sumLock}{\false} \mid \Output{\sumLock'}{\false} \mid \Output{\senderLock}{\true} \mid \sigma'\left( \Set{ \Subst{\RenamingPolicySepAsyn{x'}}{\RenamingPolicySepAsyn{x}} }\left( \EncodingSepAsyn{P} \right) \right)}{\Output{\sumLock}{\true} \mid \Output{\sumLock'}{\false} \mid \Output{\senderLock}{\false} \mid \Out{\receiverLock}}}{\Output{\sumLock}{\false} \mid \Output{\sigma'\left( \RenamingPolicySepAsyn{y} \right)}{\sumLock', \senderLock, \RenamingPolicySepAsyn{x'}} \big)}}\\
						& = \sigma'\big( \RestrictedTerm{\receiverLock}{\big( \Out{\receiverLock} \mid \ReplicateIn{\receiverLock}.\Input{\RenamingPolicySepAsyn{y}}{\sumLock', \senderLock,\RenamingPolicySepAsyn{x'}}.\\
						& \hspace*{3em} \BigTest{\sumLock}{\BigTest{\sumLock'}{\Output{\sumLock}{\false} \mid \Output{\sumLock'}{\false} \mid \Output{\senderLock}{\true} \mid \Set{ \Subst{\RenamingPolicySepAsyn{x'}}{\RenamingPolicySepAsyn{x}} }\left( \EncodingSepAsyn{P} \right)}{\Output{\sumLock}{\true} \mid \Output{\sumLock'}{\false} \mid \Output{\senderLock}{\false} \mid \Out{\receiverLock}}}{\Output{\sumLock}{\false} \mid \Output{\RenamingPolicySepAsyn{y}}{\sumLock', \senderLock, \RenamingPolicySepAsyn{x'}} \big) \big)}}\\
						& \equivAlpha \sigma'\big( \RestrictedTerm{\receiverLock}{\big( \Out{\receiverLock} \mid \ReplicateIn{\receiverLock}.\Input{\RenamingPolicySepAsyn{y}}{\sumLock', \senderLock,\RenamingPolicySepAsyn{x}}.\\
						& \hspace*{3em} \BigTest{\sumLock}{\Test{\sumLock'}{\Output{\sumLock}{\false} \mid \Output{\sumLock'}{\false} \mid \Output{\senderLock}{\true} \mid \EncodingSepAsyn{P}}{\Output{\sumLock}{\true} \mid \Output{\sumLock'}{\false} \mid \Output{\senderLock}{\false} \mid \Out{\receiverLock}}}{\Output{\sumLock}{\false} \mid \Output{\RenamingPolicySepAsyn{y}}{\sumLock', \senderLock, \RenamingPolicySepAsyn{x}} \big) \big)}}\\
						& = \sigma'\left( \EncodingSepAsyn{S} \right).
					\end{align*}
				\item[Case of $ S = \ReplicateInput{y}{x}.P $:] Then
					\begin{align*}
						\EncodingSepAsyn{\sigma\left( S \right)} & = \EncodingSepAsyn{\sigma\left( \ReplicateInput{y}{x}.P \right)} \equivAlpha \EncodingSepAsyn{\sigma\left( \ReplicateInput{y}{x'}.\Set{ \Subst{x'}{x} }\left( P \right) \right)} = \EncodingSepAsyn{\ReplicateInput{\sigma\left( y \right)}{x'}.\sigma\left( \Set{ \Subst{x'}{x} }\left( P \right) \right)}\\
						& = \ReplicateInput{\RenamingPolicySepAsyn{\sigma\left( y \right)}}{\sumLock, \senderLock, \RenamingPolicySepAsyn{x'}}.\Test{\sumLock}{\Output{\sumLock}{\false} \mid \Output{\senderLock}{\false} \mid \EncodingSepAsyn{\sigma\left( \Set{ \Subst{x'}{x} }\left( P \right) \right)}}{\Output{\sumLock}{\false} \mid \Output{\senderLock}{\false}}\\
						& \equivAlpha \ReplicateInput{\sigma'\left( \RenamingPolicySepAsyn{y} \right)}{\sumLock, \senderLock, \RenamingPolicySepAsyn{x'}}.\Test{\sumLock}{\Output{\sumLock}{\false} \mid \Output{\senderLock}{\false} \mid \sigma'\left( \Set{ \Subst{\RenamingPolicySepAsyn{x'}}{\RenamingPolicySepAsyn{x}} }\left( \EncodingSepAsyn{P} \right) \right)}{\Output{\sumLock}{\false} \mid \Output{\senderLock}{\false}}\\
						& = \sigma'\left( \ReplicateInput{\RenamingPolicySepAsyn{y}}{\sumLock, \senderLock, \RenamingPolicySepAsyn{x'}}.\Test{\sumLock}{\Output{\sumLock}{\false} \mid \Output{\senderLock}{\false} \mid \Set{ \Subst{\RenamingPolicySepAsyn{x'}}{\RenamingPolicySepAsyn{x}} }\left( \EncodingSepAsyn{P} \right)}{\Output{\sumLock}{\false} \mid \Output{\senderLock}{\false}} \right)\\
						& \equivAlpha \sigma'\left( \ReplicateInput{\RenamingPolicySepAsyn{y}}{\sumLock, \senderLock, \RenamingPolicySepAsyn{x}}.\Test{\sumLock}{\Output{\sumLock}{\false} \mid \Output{\senderLock}{\false} \mid \EncodingSepAsyn{P}}{\Output{\sumLock}{\false} \mid \Output{\senderLock}{\false}} \right)\\
						& = \sigma'\left( \EncodingSepAsyn{S} \right).
					\end{align*}
			\end{description}
	\end{description}
	\qed
\end{proof}

\begin{lemma}
	The encoding $ \encodingMixAsyn $ is name invariant. \label{lem:nameInvarianceMixAsyn}
\end{lemma}

\begin{proof}
	By Definition \ref{def:nameInvariance} it suffice to show, that:
	\begin{align*}
		\forall S \in \piMixProc \logdot \forall \sigma \subseteq \names \logdot \exists \sigma' \subseteq \names \logdot \EncodingMixAsyn{\sigma\left( S \right)} \equivAlpha \sigma'\left( \EncodingMixAsyn{S} \right) \wedge \forall z \in \names \logdot \RenamingPolicyMixAsyn{\sigma\left( z \right)} = \sigma'\left( \RenamingPolicyMixAsyn{z} \right)
	\end{align*}
	Without loss of generality let $ \sigma = \Set{ \Subst{y_1}{x_1}, \ldots \Subst{y_n}{x_n} } $ for some $ n \in \nat $. We choose
	\begin{align*}
		\sigma' \deff \Set{ \Subst{\RenamingPolicyMixAsyn{y_1}}{\RenamingPolicyMixAsyn{x_1}}, \ldots, \Subst{\RenamingPolicyMixAsyn{y_n}}{\RenamingPolicyMixAsyn{x_n}} }.
	\end{align*}
	So $ \forall z \in \names \logdot \RenamingPolicyMixAsyn{\sigma\left( z \right)} = \sigma'\left( \RenamingPolicyMixAsyn{z} \right) $. We proceed with an induction over the structure of $ S $.
	\begin{description}
		\item[Base Case:] Since $ \Names{\success} = \emptyset = \Names{\nullTerm} $ and $ \FreeNames{\RestrictedTerm{\sumLock}{\Output{\sumLock}{\true}}} = \emptyset $, we have $ \EncodingMixAsyn{\sigma\left( \success \right)} = \EncodingMixAsyn{\success} = \success = \sigma'\left( \success \right) = \sigma'\left( \EncodingMixAsyn{\success} \right) $ and $ \EncodingMixAsyn{\sigma\left( \nullTerm \right)} = \EncodingMixAsyn{\nullTerm} = \RestrictedTerm{\sumLock}{\Output{\sumLock}{\true}} = \sigma'\left( \RestrictedTerm{\sumLock}{\Output{\sumLock}{\true}} \right) = \sigma'\left( \EncodingMixAsyn{\nullTerm} \right) $.
		\item[Induction Hypothesis:] $ \forall S \in \piMixProc \logdot \forall \sigma \subseteq \names \logdot \exists \sigma' \subseteq \names \logdot \EncodingMixAsyn{\sigma\left( S \right)} \equivAlpha \sigma'\left( \EncodingMixAsyn{S} \right) $
		\item[Induction Step:] Let $ x' \in \names $ be such that $ x' \notin \Names{\sigma} \cup \FreeNames{P} $. Then $ \RenamingPolicyMixAsyn{x'} \notin \Names{\sigma} $. Moreover, since $ \forall z \in \names \logdot \RenamingPolicyMixAsyn{z} \notin N $, where
			\begin{align*}
				N = \Set{ \parallelChannelOut, \parallelChannelIn, \coordinatorUpOut, \coordinatorUpIn, \matchingCoordinatorOut, \matchingCoordinatorIn, \matchingUpOut, \matchingUpIn, \coordinatorMatchingOut, \coordinatorMatchingIn, \sumLock, \sumLock_s, \sumLock_r, \sumLock_1, \sumLock_2, \senderLock, \receiverLock, \coordinatorRepA, \coordinatorRepB, \matchingReceiverIn, \matchingReceiverOut, \matchingReceiverUpIn, \matchingReceiverUpOut, y, y', z, t, f },
			\end{align*}
			we have $ \Names{\sigma'} \cap N = \emptyset $. Note that:
			\begin{align*}
				\EncodingMixAsyn{\sigma\left( \Set{ \Subst{x'}{x} }\left( P \right) \right)} & = \EncodingMixAsyn{\Set{ \Subst{x'}{x}, \Subst{y_i}{x_i} \mid \Subst{y_i}{x_i} \in \sigma \wedge x_i \neq x }\left( P \right)}\\
				& \equivAlpha \Set{ \Subst{\RenamingPolicyMixAsyn{x'}}{\RenamingPolicyMixAsyn{x}}, \Subst{\RenamingPolicyMixAsyn{y_i}}{\RenamingPolicyMixAsyn{x_i}} \mid \Subst{y_i}{x_i} \in \sigma \wedge x_i \neq x }\left( \EncodingMixAsyn{P} \right) & \text{ by IH}\\
				& = \Set{ \Subst{\RenamingPolicyMixAsyn{x'}}{\RenamingPolicyMixAsyn{x}}, \Subst{\RenamingPolicyMixAsyn{y_i}}{\RenamingPolicyMixAsyn{x_i}} \mid \Subst{\RenamingPolicyMixAsyn{y_i}}{\RenamingPolicyMixAsyn{x_i}} \in \sigma' \wedge \RenamingPolicyMixAsyn{x_i} \neq \RenamingPolicyMixAsyn{x} }\left( \EncodingMixAsyn{P} \right)\\
				& = \sigma'\left( \Set{ \Subst{\RenamingPolicyMixAsyn{x'}}{\RenamingPolicyMixAsyn{x}} }\left( \EncodingMixAsyn{P} \right) \right)
			\end{align*}
			We proceed by a case split.
			\begin{description}
				\item[Case of $ S = \RestrictedTerm{x}{P} $:] Then
					\begin{align*}
						\EncodingMixAsyn{\sigma\left( S \right)} & = \EncodingMixAsyn{ \sigma\left( \RestrictedTerm{x}{P} \right)} \equivAlpha \EncodingMixAsyn{ \sigma\left( \RestrictedTerm{x'}{\Set{ \Subst{x'}{x} }\left( P \right)} \right) } = \EncodingMixAsyn{ \RestrictedTerm{x'}{\sigma\left( \Set{ \Subst{x'}{x} }\left( P \right) \right)} }\\
						& = \RestrictedTerm{\RenamingPolicyMixAsyn{x'}}{\EncodingMixAsyn{ \sigma\left( \Set{ \Subst{x'}{x} }\left( P \right) \right)} } \equivAlpha \RestrictedTerm{\RenamingPolicyMixAsyn{x'}}{\sigma'\left( \Set{ \Subst{\RenamingPolicyMixAsyn{x'}}{\RenamingPolicyMixAsyn{x}} }\left( \EncodingMixAsyn{P} \right) \right)}\\
						& = \sigma'\left( \RestrictedTerm{\RenamingPolicyMixAsyn{x'}}{\Set{ \Subst{\RenamingPolicyMixAsyn{x'}}{\RenamingPolicyMixAsyn{x}} }\left( \EncodingMixAsyn{P} \right)} \right) \equivAlpha \sigma'\left( \RestrictedTerm{\RenamingPolicyMixAsyn{x}}{\EncodingMixAsyn{P}} \right) = \sigma'\left( \EncodingMixAsyn{S} \right).
					\end{align*}
				\item[Case of $ S = P \mid Q $:] Then
					\begin{align*}
						\EncodingMixAsyn{\sigma\left( S \right)} & = \EncodingMixAsyn{\sigma\left( P \mid Q \right)} = \EncodingMixAsyn{\sigma\left( P \right) \mid \sigma\left( Q \right)}\\
						& = \begin{aligned}[t]
								& \RestrictedTerm{\coordinatorLock, \matchingCoordinatorOut, \matchingCoordinatorIn, \coordinatorUpOut, \coordinatorUpIn, \coordinatorMatchingOut, \coordinatorMatchingIn, \matchingUpOut, \matchingUpIn}{\big(\\
								& \hspace*{1em} \Out{\coordinatorLock} \mid \RestrictedTerm{\parallelChannelOut, \parallelChannelIn}{\big( \EncodingMixAsyn{\sigma\left( P \right)} \mid \processLeftOutputRequests \mid \processLeftInputRequests \big)}\\
								& \hspace*{1em} \mid \RestrictedTerm{\parallelChannelOut, \parallelChannelIn}{\big( \begin{aligned}[t]
										& \EncodingMixAsyn{\sigma\left( Q \right)} \mid \processRightOutputRequests \mid \processRightInputRequests \big)
									\end{aligned}}\\
								& \hspace*{1em} \mid \pushRequests \big)}
							\end{aligned}\\
						& \equivAlpha \begin{aligned}[t]
								& \RestrictedTerm{\coordinatorLock, \matchingCoordinatorOut, \matchingCoordinatorIn, \coordinatorUpOut, \coordinatorUpIn, \coordinatorMatchingOut, \coordinatorMatchingIn, \matchingUpOut, \matchingUpIn}{\big(\\
								& \hspace*{1em} \Out{\coordinatorLock} \mid \RestrictedTerm{\parallelChannelOut, \parallelChannelIn}{\big( \sigma\left( \EncodingMixAsyn{P} \right) \mid \processLeftOutputRequests \mid \processLeftInputRequests \big)}\\
								& \hspace*{1em} \mid \RestrictedTerm{\parallelChannelOut, \parallelChannelIn}{\big( \begin{aligned}[t]
										& \sigma'\left( \EncodingMixAsyn{Q} \right) \mid \processRightOutputRequests \mid \processRightInputRequests \big)
									\end{aligned}}\\
								& \hspace*{1em} \mid \pushRequests \big)}
							\end{aligned}\\
						& = \sigma'\big( \begin{aligned}[t]
								& \RestrictedTerm{\coordinatorLock, \matchingCoordinatorOut, \matchingCoordinatorIn, \coordinatorUpOut, \coordinatorUpIn, \coordinatorMatchingOut, \coordinatorMatchingIn, \matchingUpOut, \matchingUpIn}{\big(\\
								& \hspace*{1em} \Out{\coordinatorLock} \mid \RestrictedTerm{\parallelChannelOut, \parallelChannelIn}{\big( \EncodingMixAsyn{P} \mid \processLeftOutputRequests \mid \processLeftInputRequests \big)}\\
								& \hspace*{1em} \mid \RestrictedTerm{\parallelChannelOut, \parallelChannelIn}{\big( \begin{aligned}[t]
										& \EncodingMixAsyn{Q} \mid \processRightOutputRequests \mid \processRightInputRequests \big)
									\end{aligned}}\\
								& \hspace*{1em} \mid \pushRequests \big)} \big)
							\end{aligned}\\
						& = \sigma'\left( \EncodingMixAsyn{S} \right).
					\end{align*}
				\item[Case of $ S = \Match{a}{b}P $:] Then
					\begin{align*}
						\EncodingMixAsyn{\sigma\left( S \right)} & = \EncodingMixAsyn{\sigma\left( \Match{a}{b}P \right)} = \EncodingMixAsyn{\Match{\sigma\left( a \right)}{\sigma\left( b \right)}\sigma\left( P \right)} = \Match{\RenamingPolicyMixAsyn{\sigma\left( a \right)}}{\RenamingPolicyMixAsyn{\sigma\left( b \right)}}\EncodingMixAsyn{\sigma\left( P \right)}\\
						& \equivAlpha \Match{\sigma'\left( \RenamingPolicyMixAsyn{a} \right)}{\sigma'\left( \RenamingPolicyMixAsyn{b} \right)}\sigma'\left( \EncodingMixAsyn{P} \right) = \sigma'\left( \Match{\RenamingPolicyMixAsyn{a}}{\RenamingPolicyMixAsyn{b}}\EncodingMixAsyn{P} \right) = \sigma'\left( \EncodingMixAsyn{S} \right).
					\end{align*}
				\item[Case of $ S = \sum_{i \in \indexSet} \guard_i.P_i $:] Then
					\begin{align*}
						\EncodingMixAsyn{\sigma\left( S \right)} & = \EncodingMixAsyn{\sigma\left( \sum_{i \in \indexSet} \guard_i.P_i \right)} = \EncodingMixAsyn{\sum_{i \in \indexSet} \sigma\left( \guard_i.P_i \right)} = \RestrictedTerm{\sumLock}{\left( \Output{\sumLock}{\true} \mid \prod_{i \in \indexSet} \EncodingMixAsyn{\sigma\left( \guard_i.P_i \right)} \right)}\\
						& \equivAlpha \RestrictedTerm{\sumLock}{\left( \Output{\sumLock}{\true} \mid \prod_{i \in \indexSet} \sigma'\left( \EncodingMixAsyn{\guard_i.P_i} \right) \right)} = \sigma'\left( \RestrictedTerm{\sumLock}{\left( \Output{\sumLock}{\true} \mid \prod_{i \in \indexSet} \EncodingMixAsyn{\guard_i.P_i} \right)} \right) = \sigma'\left( \EncodingMixAsyn{S} \right).
					\end{align*}
				\item[Case of $ S = \tau.P $:] Then
					\begin{align*}
						\EncodingMixAsyn{\sigma\left( S \right)} & = \EncodingMixAsyn{\sigma\left( \tau.P \right)} = \EncodingMixAsyn{\tau.\sigma\left( P \right)} = \Test{\sumLock}{\Output{\sumLock}{\false} \mid \EncodingMixAsyn{\sigma\left( P \right)}}{\Output{\sumLock}{\false}}\\
						& \equivAlpha \Test{\sumLock}{\Output{\sumLock}{\false} \mid \sigma'\left( \EncodingMixAsyn{P} \right)}{\Output{\sumLock}{\false}} = \sigma'\left( \Test{\sumLock}{\Output{\sumLock}{\false} \mid \EncodingMixAsyn{P}}{\Output{\sumLock}{\false}} \right)\\
						& = \sigma'\left( \EncodingMixAsyn{S} \right).
					\end{align*}
				\item[Case of $ S = \Output{y}{z}.P $:] Then
					\begin{align*}
						\EncodingMixAsyn{\sigma\left( S \right)} & = \EncodingMixAsyn{\sigma\left( \Output{y}{z}.P \right)} = \EncodingMixAsyn{\Output{\sigma\left( y \right)}{\sigma\left( z \right)}.\sigma\left( P \right)}\\
						& = \RestrictedTerm{\senderLock}{\left( \Output{\parallelChannelOut}{\RenamingPolicyMixAsyn{\sigma\left( y \right)}, \sumLock, \senderLock, \RenamingPolicyMixAsyn{\sigma\left( z \right)}} \mid \In{\senderLock}.\EncodingMixAsyn{\sigma\left( P \right)} \right)}\\
						& \equivAlpha \RestrictedTerm{\senderLock}{\left( \Output{\parallelChannelOut}{\sigma'\left( \RenamingPolicyMixAsyn{y} \right), \sumLock, \senderLock, \sigma'\left( \RenamingPolicyMixAsyn{z} \right)} \mid \In{\senderLock}.\sigma'\left( \EncodingMixAsyn{P} \right) \right)}\\
						& = \sigma'\left( \RestrictedTerm{\senderLock}{\left( \Output{\parallelChannelOut}{\RenamingPolicyMixAsyn{y}, \sumLock, \senderLock, \RenamingPolicyMixAsyn{z}} \mid \In{\senderLock}.\EncodingMixAsyn{P} \right)} \right) = \sigma'\left( \EncodingMixAsyn{S} \right).
					\end{align*}
				\item[Case of $ S = \Input{y}{x}.P $:] Then
					\begin{align*}
						\EncodingMixAsyn{\sigma\left( S \right)} & = \EncodingMixAsyn{\sigma\left( \Input{y}{x}.P \right)} \equivAlpha \EncodingMixAsyn{\sigma\left( \Input{y}{x'}.\Set{ \Subst{x'}{x} }\left( P \right) \right)} = \EncodingMixAsyn{\Input{\sigma\left( y \right)}{x'}.\sigma\left( \Set{ \Subst{x'}{x} }\left( P \right) \right)}\\
						& = \RestrictedTerm{\receiverLock}{\big( \Output{\parallelChannelIn}{\RenamingPolicyMixAsyn{\sigma\left( y \right)}, \sumLock, \receiverLock} \mid \ReplicateInput{\receiverLock}{l_1, l_2, -, s, \RenamingPolicyMixAsyn{x'}}.\\
						& \hspace*{3em} \BigTest{l_1}{\Test{l_2}{\Output{l_1}{\false} \mid \Output{l_2}{\false} \mid \Out{\senderLock} \mid \EncodingMixAsyn{\sigma\left( \Set{ \Subst{x'}{x} }\left( P \right) \right)}}{\Output{l_1}{\true} \mid \Output{l_2}{\false}}}{\Output{l_1}{\false} \big)}}\\
						& \equivAlpha \RestrictedTerm{\receiverLock}{\big( \Output{\parallelChannelIn}{\sigma'\left( \RenamingPolicyMixAsyn{y} \right), \sumLock, \receiverLock} \mid \ReplicateInput{\receiverLock}{l_1, l_2, -, s, \RenamingPolicyMixAsyn{x'}}.\\
						& \hspace*{3em} \BigTest{l_1}{\Test{l_2}{\Output{l_1}{\false} \mid \Output{l_2}{\false} \mid \Out{\senderLock} \mid \sigma'\left( \Set{ \Subst{\RenamingPolicyMixAsyn{x'}}{\RenamingPolicyMixAsyn{x}} }\left( \EncodingMixAsyn{P} \right) \right)}{\Output{l_1}{\true} \mid \Output{l_2}{\false}}}{\Output{l_1}{\false} \big)}}\\
						& = \sigma'\big( \RestrictedTerm{\receiverLock}{\big( \Output{\parallelChannelIn}{\RenamingPolicyMixAsyn{y}, \sumLock, \receiverLock} \mid \ReplicateInput{\receiverLock}{l_1, l_2, -, s, \RenamingPolicyMixAsyn{x'}}.\\
						& \hspace*{3em} \BigTest{l_1}{\Test{l_2}{\Output{l_1}{\false} \mid \Output{l_2}{\false} \mid \Out{\senderLock} \mid \Set{ \Subst{\RenamingPolicyMixAsyn{x'}}{\RenamingPolicyMixAsyn{x}} }\left( \EncodingMixAsyn{P} \right)}{\Output{l_1}{\true} \mid \Output{l_2}{\false}}}{\Output{l_1}{\false} \big) \big)}}\\
						& \equivAlpha \sigma'\big( \RestrictedTerm{\receiverLock}{\big( \Output{\parallelChannelIn}{\RenamingPolicyMixAsyn{y}, \sumLock, \receiverLock} \mid \ReplicateInput{\receiverLock}{l_1, l_2, -, s, \RenamingPolicyMixAsyn{x}}.\\
						& \hspace*{3em} \Test{l_1}{\left( \Test{l_2}{\Output{l_1}{\false} \mid \Output{l_2}{\false} \mid \Out{\senderLock} \mid \EncodingMixAsyn{P}}{\Output{l_1}{\true} \mid \Output{l_2}{\false}} \right)}{\Output{l_1}{\false} \big) \big)}}\\
						& = \sigma'\left( \EncodingMixAsyn{S} \right).
					\end{align*}
				\item[Case of $ S = \ReplicateInput{y}{x}.P $:] Then
					{\allowdisplaybreaks
					\begin{align*}
						\EncodingMixAsyn{\sigma\left( S \right)} & = \EncodingMixAsyn{\sigma\left( \ReplicateInput{y}{x}.P \right)} \equivAlpha \EncodingMixAsyn{\sigma\left( \ReplicateInput{y}{x'}.\Set{ \Subst{x'}{x} }\left( P \right) \right)} = \EncodingMixAsyn{\ReplicateInput{\sigma\left( y \right)}{x'}.\sigma\left( \Set{ \Subst{x'}{x} }\left( P \right) \right)}\\
						& = \RestrictedTerm{\sumLock, \receiverLock, \coordinatorRepA, \coordinatorRepB, \matchingReceiverOut, \matchingReceiverIn}{\big(\\
							& \hspace*{3em} \Output{\parallelChannelIn}{\RenamingPolicyMixAsyn{\sigma\left( y \right)}, \sumLock, \receiverLock} \mid \ReplicateInput{\receiverLock}{-, -, l_s, s, \RenamingPolicyMixAsyn{x'}}.\Test{l_s}{\Output{l_s}{\false} \mid \Out{s} \mid \Output{\coordinatorRepA}{\RenamingPolicyMixAsyn{x'}}}{\Output{l_s}{\false}}\\
							& \hspace*{3em} \mid \Output{\matchingReceiverIn}{\RenamingPolicyMixAsyn{\sigma\left( y \right)}, \sumLock, \receiverLock} \mid \Output{\sumLock}{\true} \mid \Output{\coordinatorRepB}{\matchingReceiverOut, \matchingReceiverIn}\\
							& \hspace*{3em} \mid \ReplicateInput{\coordinatorRepA}{\RenamingPolicyMixAsyn{x'}}.\Input{\coordinatorRepB}{\matchingReceiverOut, \matchingReceiverIn}.\RestrictedTerm{\matchingCoordinatorOut, \matchingCoordinatorIn, \coordinatorUpOut, \coordinatorUpIn, \matchingReceiverUpOut, \matchingReceiverUpIn, \coordinatorMatchingOut, \coordinatorMatchingIn, \matchingUpOut, \matchingUpIn}{\big(\\
							& \hspace*{6em} \begin{aligned}[t]
									& \pushRequestsIn\\
									& \mid \RestrictedTerm{\parallelChannelOut, \parallelChannelIn}{\Big( \EncodingMixAsyn{\sigma\left( \Set{ \Subst{x'}{x} }\left( P \right) \right)}\\
									& \hspace*{3em} \mid \processRightOutputRequests \mid \processRightInputRequests \Big)\\
									& \mid \RestrictedTerm{\matchingReceiverOut, \matchingReceiverIn}{\left( \Output{\coordinatorRepB}{\matchingReceiverOut, \matchingReceiverIn} \mid \pushRequestsOut \right) \big) \big)}
								\end{aligned}}}}\\
						& \equivAlpha \RestrictedTerm{\sumLock, \receiverLock, \coordinatorRepA, \coordinatorRepB, \matchingReceiverOut, \matchingReceiverIn}{\big(\\
							& \hspace*{3em} \Output{\parallelChannelIn}{\sigma'\left( \RenamingPolicyMixAsyn{y} \right), \sumLock, \receiverLock} \mid \ReplicateInput{\receiverLock}{-, -, l_s, s, \RenamingPolicyMixAsyn{x'}}.\Test{l_s}{\Output{l_s}{\false} \mid \Out{s} \mid \Output{\coordinatorRepA}{\RenamingPolicyMixAsyn{x'}}}{\Output{l_s}{\false}}\\
							& \hspace*{3em} \mid \Output{\matchingReceiverIn}{\sigma'\left( \RenamingPolicyMixAsyn{y} \right), \sumLock, \receiverLock} \mid \Output{\sumLock}{\true} \mid \Output{\coordinatorRepB}{\matchingReceiverOut, \matchingReceiverIn}\\
							& \hspace*{3em} \mid \ReplicateInput{\coordinatorRepA}{\RenamingPolicyMixAsyn{x'}}.\Input{\coordinatorRepB}{\matchingReceiverOut, \matchingReceiverIn}.\RestrictedTerm{\matchingCoordinatorOut, \matchingCoordinatorIn, \coordinatorUpOut, \coordinatorUpIn, \matchingReceiverUpOut, \matchingReceiverUpIn, \coordinatorMatchingOut, \coordinatorMatchingIn, \matchingUpOut, \matchingUpIn}{\big(\\
							& \hspace*{6em} \begin{aligned}[t]
									& \pushRequestsIn\\
									& \mid \RestrictedTerm{\parallelChannelOut, \parallelChannelIn}{\Big( \sigma'\left( \Set{ \Subst{\RenamingPolicyMixAsyn{x'}}{\RenamingPolicyMixAsyn{x}} }\left( \EncodingMixAsyn{P} \right) \right)\\
									& \hspace*{3em} \mid \processRightOutputRequests \mid \processRightInputRequests \Big)\\
									& \mid \RestrictedTerm{\matchingReceiverOut, \matchingReceiverIn}{\left( \Output{\coordinatorRepB}{\matchingReceiverOut, \matchingReceiverIn} \mid \pushRequestsOut \right) \big) \big)}
								\end{aligned}}}}\\
						\hspace*{1em} & = \sigma'\Big( \RestrictedTerm{\sumLock, \receiverLock, \coordinatorRepA, \coordinatorRepB, \matchingReceiverOut, \matchingReceiverIn}{\big(\\
							& \hspace*{3em} \Output{\parallelChannelIn}{\RenamingPolicyMixAsyn{y}, \sumLock, \receiverLock} \mid \ReplicateInput{\receiverLock}{-, -, l_s, s, \RenamingPolicyMixAsyn{x'}}.\Test{l_s}{\Output{l_s}{\false} \mid \Out{s} \mid \Output{\coordinatorRepA}{\RenamingPolicyMixAsyn{x'}}}{\Output{l_s}{\false}}\\
							& \hspace*{3em} \mid \Output{\matchingReceiverIn}{\RenamingPolicyMixAsyn{y}, \sumLock, \receiverLock} \mid \Output{\sumLock}{\true} \mid \Output{\coordinatorRepB}{\matchingReceiverOut, \matchingReceiverIn}\\
							& \hspace*{3em} \mid \ReplicateInput{\coordinatorRepA}{\RenamingPolicyMixAsyn{x'}}.\Input{\coordinatorRepB}{\matchingReceiverOut, \matchingReceiverIn}.\RestrictedTerm{\matchingCoordinatorOut, \matchingCoordinatorIn, \coordinatorUpOut, \coordinatorUpIn, \matchingReceiverUpOut, \matchingReceiverUpIn, \coordinatorMatchingOut, \coordinatorMatchingIn, \matchingUpOut, \matchingUpIn}{\big(\\
							& \hspace*{6em} \begin{aligned}[t]
									& \pushRequestsIn\\
									& \mid \RestrictedTerm{\parallelChannelOut, \parallelChannelIn}{\Big( \Set{ \Subst{\RenamingPolicyMixAsyn{x'}}{\RenamingPolicyMixAsyn{x}} }\left( \EncodingMixAsyn{P} \right)\\
									& \hspace*{3em} \mid \processRightOutputRequests \mid \processRightInputRequests \Big)\\
									& \mid \RestrictedTerm{\matchingReceiverOut, \matchingReceiverIn}{\left( \Output{\coordinatorRepB}{\matchingReceiverOut, \matchingReceiverIn} \mid \pushRequestsOut \right) \big) \big) \Big)}
								\end{aligned}}}}\\
						& \equivAlpha \sigma'\Big( \RestrictedTerm{\sumLock, \receiverLock, \coordinatorRepA, \coordinatorRepB, \matchingReceiverOut, \matchingReceiverIn}{\big(\\
							& \hspace*{3em} \Output{\parallelChannelIn}{\RenamingPolicyMixAsyn{y}, \sumLock, \receiverLock} \mid \ReplicateInput{\receiverLock}{-, -, l_s, s, \RenamingPolicyMixAsyn{x}}.\Test{l_s}{\Output{l_s}{\false} \mid \Out{s} \mid \Output{\coordinatorRepA}{\RenamingPolicyMixAsyn{x}}}{\Output{l_s}{\false}}\\
							& \hspace*{3em} \mid \Output{\matchingReceiverIn}{\RenamingPolicyMixAsyn{y}, \sumLock, \receiverLock} \mid \Output{\sumLock}{\true} \mid \Output{\coordinatorRepB}{\matchingReceiverOut, \matchingReceiverIn}\\
							& \hspace*{3em} \mid \ReplicateInput{\coordinatorRepA}{\RenamingPolicyMixAsyn{x}}.\Input{\coordinatorRepB}{\matchingReceiverOut, \matchingReceiverIn}.\RestrictedTerm{\matchingCoordinatorOut, \matchingCoordinatorIn, \coordinatorUpOut, \coordinatorUpIn, \matchingReceiverUpOut, \matchingReceiverUpIn, \coordinatorMatchingOut, \coordinatorMatchingIn, \matchingUpOut, \matchingUpIn}{\big(\\
							& \hspace*{6em} \begin{aligned}[t]
									& \pushRequestsIn\\
									& \mid \RestrictedTerm{\parallelChannelOut, \parallelChannelIn}{\Big( \EncodingMixAsyn{P} \mid \processRightOutputRequests \mid \processRightInputRequests \Big)\\
									& \mid \RestrictedTerm{\matchingReceiverOut, \matchingReceiverIn}{\left( \Output{\coordinatorRepB}{\matchingReceiverOut, \matchingReceiverIn} \mid \pushRequestsOut \right) \big) \big) \Big)}
								\end{aligned}}}}\\
						& = \sigma'\left( \EncodingMixAsyn{S} \right).
					\end{align*}}
			\end{description}
	\end{description}
	\qed
\end{proof}

Analysing these proofs we observe (1) that $ \sigma' $ depends only on $ \sigma $ and the respective renaming policy, and (2) that we can prove the first case of name invariance (compare to Definition \ref{def:nameInvariance}) for all kinds of substitutions $ \sigma $, i.e., it suffice to consider equivalence modulo alpha conversion.

\begin{corollary}[Encoding substitutions] \label{col:encodingSubstitutions}
	For all substitutions $ \sigma = \Set{ \Subst{y_1}{x_1}, \ldots, \Subst{y_n}{x_n} } $ it holds that
	\begin{align*}
		\forall S \in \piSepProc \logdot \EncodingSepAsyn{\sigma\left( S \right)} \equivAlpha \RenamingPolicySepAsyn{\sigma}\left( \EncodingSepAsyn{S} \right) \quad \text{ and } \quad \forall S \in \piMixProc \logdot \EncodingMixAsyn{\sigma\left( S \right)} \equivAlpha \RenamingPolicyMixAsyn{\sigma}\left( \EncodingMixAsyn{S} \right),
	\end{align*}
	where $ \RenamingPolicySepAsyn{\sigma} = \Set{ \Subst{\RenamingPolicySepAsyn{y_1}}{\RenamingPolicySepAsyn{x_1}}, \ldots, \Subst{\RenamingPolicySepAsyn{y_n}}{\RenamingPolicySepAsyn{x_n}} } $ and $ \RenamingPolicyMixAsyn{\sigma} = \Set{ \Subst{\RenamingPolicyMixAsyn{y_1}}{\RenamingPolicyMixAsyn{x_1}}, \ldots, \Subst{\RenamingPolicyMixAsyn{y_n}}{\RenamingPolicyMixAsyn{x_n}} } $.
\end{corollary}

\subsection{Basic Properties}

In the following we prove correctness with respect to the three semantical criteria. We observe, that in order to do so we do not have to prove conditions on arbitrary \piAsyn-terms but on encoded source terms and their derivatives. To simplify the argumentation we will denote such terms as \emph{target terms}.

\begin{definition}[Target Terms] \label{def:targetTerm}
	Let $ T \in \piAsynProc $. Then $ T $ is a \emph{target term}, denoted by $ T \in \targetTermsMixAsyn $ (or $ T \in \targetTermsSepAsyn $), if $ \exists S \in \piMixProc \logdot T \equiv \EncodingMixAsyn{S} \vee \EncodingMixAsyn{S} \steps T $ (or $ \exists S \in \piSepProc \logdot T \equiv \EncodingSepAsyn{S} \vee \EncodingSepAsyn{S} \steps T $).
\end{definition}

\paragraph*{Requests.} Note that the encoding $ \encodingSepAsyn $ translates source term observables by adding an instantiation of a sum lock (except from observables due to replicated inputs) to keep track of the information, whether this observable is still active, i.e., whether the corresponding in- or output can still be used to \simulate a source term step. Besides that additional information $ \encodingSepAsyn $ does not change the observables. In contrast, the encoding $ \encodingMixAsyn $ translates source term observables into requests, which are again augmented by sum locks. Requests are outputs with either three or four parameters. Input requests, i.e., requests that originate from the translation of an input guarded term or replicated input, are outputs with three parameters. Output requests, i.e., requests that originate from the translation of an output guarded term, are outputs with four parameters. Note that we can indeed consider any output of three or four parameters as request, because the encoding function does not use these multiplicities for other purposes (compare to Figure \ref{fig:encodingMixAsyn}).

\begin{definition}[Request] \label{def:request}
	An \emph{input request} is an unguarded output with three parameters, i.e., an output of kind $ \Output{y}{x_1, x_2, x_3} $ for some $ y, x_1, x_2, x_3 \in \names $, and an \emph{output request} is an unguarded output with four parameters, i.e., an output of kind $ \Output{y}{x_1, x_2, x_3, x_4} $ for some $ y, x_1, x_2, x_3, x_4 \in \names $. We refer to guarded variants of those outputs as \emph{guarded requests} and to $ y $ as \emph{request channel}.
\end{definition}
\noindent
Note that the channels introduced by the encoding function are somehow well typed in the sense, that each name once used as link with multiplicity $ n $ will never appear as link with a multiplicity different from $ n $. Because of that, it make sense to denote the channel $ y $ here as request channel, because whenever it is used as link name a request is transferred above that link. Moreover note, that the first parameter $ x_1 $ of a request is always the translation of the respective source term channel and the second parameter $ x_2 $ always refers to the sum lock that is connected to that requests, i.e., that covers the information about the liveness of the corresponding observable. Note that in case of an input request that originate from an replicated input the second parameter refers to a fake sum lock, which is never checked. In case of an input request the third parameter $ x_3 $ refers to the corresponding receiver lock and in case of an output request the third parameter $ x_3 $ refers to the corresponding sender lock and the fourth parameter $ x_4 $ is the translation of the send value.

An interesting fact is, that requests are preserved by the encoding function, i.e., each derivative of a target term has all the requests of its predecessor. Note that we consider here two requests that only differ by their link name but not their values as the same request. Requests are pushed upwards along and from right to left within the parallel structure of the term but they are never completely consumed. If the message refers to an inactive observable the respective sum lock is instantiated false to ensure that such a request can no longer be used to \simulate a source term step. The corresponding output messages of the encoding, i.e., the requests, remain as junk (compare to Lemma \ref{lem:junkRequestsOnFalseSumLocks}).

\begin{lemma}[$ \encodingMixAsyn $ preserves requests] \label{lem:encodingMixAsynPreserveRequests}
	\begin{align*}
		& \forall T_1, T_2 \in \targetTermsMixAsyn \logdot \forall \parallelChannelIn, y, \sumLock, \receiverLock \in \names \logdot \left( \exists T_1' \in \piAsynProc \logdot \exists \tilde{x} \subset \names \logdot T_1 \equiv \RestrictedTerm{\tilde{x}}{\left( T_1' \mid \Output{\parallelChannelIn}{y, \sumLock, \receiverLock} \right)} \right) \wedge T_1 \steps T_2\\
		& \hspace*{1em} \text{ implies } \left( \exists T_2' \in \piAsynProc \logdot \exists \tilde{x} \subset \names \logdot \exists \parallelChannelIn' \in \names \logdot T_2 \equiv \RestrictedTerm{\tilde{x}}{\left( T_2' \mid \Output{\parallelChannelIn'}{y, \sumLock, \receiverLock} \right)} \right)
	\end{align*}
	and
	\begin{align*}
		& \forall T_1, T_2 \in \targetTermsMixAsyn \logdot \forall \parallelChannelOut, y, \sumLock, \senderLock, z \in \names \logdot \left( \exists T_1' \in \piAsynProc \logdot \exists \tilde{x} \subset \names \logdot T_1 \equiv \RestrictedTerm{\tilde{x}}{\left( T_1' \mid \Output{\parallelChannelOut}{y, \sumLock, \senderLock, z} \right)} \right) \wedge T_1 \steps T_2\\
		& \hspace*{1em} \text{ implies } \left( \exists T_2' \in \piAsynProc \logdot \exists \tilde{x} \subset \names \logdot \exists \parallelChannelOut' \in \names \logdot T_2 \equiv \RestrictedTerm{\tilde{x}}{\left( T_2' \mid \Output{\parallelChannelOut'}{y, \sumLock, \senderLock, z} \right)} \right)
	\end{align*}
\end{lemma}

\begin{proof}
	First note, that due to Figure \ref{fig:encodingMixAsyn} the translations of source term names are used as values only. So any in- or output of a target term is generated by the encoding function on special names reserved for the encoding. In case $ T_1 = T_2 $, i.e., in case the sequence $ T_1 \steps T_2 $ is empty, the lemma holds trivially. Let us consider the case of a single step, i.e., $ T_1 \step T_2 $. The Lemma then follows by induction over the number of steps in the sequence $ T_1 \steps T_2 $.
	
	Analysing the encoding function in Figure \ref{fig:encodingMixAsyn} we observe, that any input with three or four parameters is due to the encoding of the parallel operator or a replicated input. In case of a forwarder the lemma again trivially holds, because each forwarder immediately restores each consumed message. The only remain inputs in the encoding of a parallel operator or a replicated input are due to the processing of right requests, i.e., due to \processRightOutputRequests \ and \processRightInputRequests.
	\begin{align*}
		\processRightOutputRequests & \deff \Output{\coordinatorMatchingOut}{\matchingCoordinatorIn} \mid \ReplicateInput{\coordinatorMatchingOut}{\matchingCoordinatorIn}.\Input{\parallelChannelOut}{y, \sumLock_s, \senderLock, z}.\big(\\
		& \hspace*{2em} \RestrictedTerm{\matchingUpIn}{\big( \begin{aligned}[t]
				& \ReplicateInput{\matchingCoordinatorIn}{y', \sumLock_r, \receiverLock}.\left( \Match{y'}{y}\Output{\receiverLock}{\sumLock_r, \sumLock_s, \sumLock_s, \senderLock, z} \mid \Output{\matchingUpIn}{y', \sumLock_r, \receiverLock} \right)\\
				& \mid \RestrictedTerm{\matchingCoordinatorIn}{\left( \Forward{\matchingUpIn}{\matchingCoordinatorIn} \mid \Output{\coordinatorMatchingOut}{\matchingCoordinatorIn} \right)} \big)
			\end{aligned}\\
		& \hspace*{2em} \mid \Output{\coordinatorUpOut}{y, \sumLock_s, \senderLock, z}} \big)
	\end{align*}
	In case of \processRightOutputRequests \ there are two inputs on request channels, namely $ \Input{\parallelChannelOut}{y, \sumLock_s, \senderLock, z} $ and $ \ReplicateInput{\matchingCoordinatorIn}{y', \sumLock_r, \receiverLock} $. In the first case, whenever a request is consumed by $ \Input{\parallelChannelOut}{y, \sumLock_s, \senderLock, z} $ it is immediately restored by $ \Output{\coordinatorUpOut}{y, \sumLock_s, \senderLock, z} $. In the second case any consumed request is immediately restored by $ \Output{\matchingUpIn}{y', \sumLock_r, \receiverLock} $. So the Lemma holds. The argumentation for \processRightInputRequests \ is similar.
	\qed
\end{proof}

A closer look at this proof and the encoding function in Figure \ref{fig:encodingMixAsyn} reveals, that (1) any initial request is due to the encoding of a guarded term or a replicated input and (2) any other request is a copy of an existing request. Because of that, as long as we are only interested in the values a request may carry and do not concern the link over it is currently transmitted, then we can conclude that any request originate to the encoding of a guarded term or a replicated input.
\begin{corollary} \label{col:originRequests}
	Any request originates to the encoding of a guarded term or a replicated input.
\end{corollary}

\paragraph*{Sum locks.} Sum locks|for both considered encodings|are channels carrying a boolean value. They are used by the encoding functions to ensure that at most one summand of each sum is chosen for communication. Note that any channel used to transport a boolean value is a sum lock. However, since by Definition \ref{def:testBoolean} at page \pageref{def:testBoolean} booleans and test-statements are just abbreviations, we use some simple type informations to unambiguous identify sum locks in both encodings. So, to be precise, instantiations on sum locks are inputs carrying two values, that are links with multiplicity zero. So sum locks are the only channels of multiplicity two, that carry only values of multiplicity zero.

\begin{definition}[Sum lock] \label{def:sumLock}
	Let $ T \in \targetTermsMixAsyn $ (or $ T \in \targetTermsSepAsyn $). A \emph{sum lock} of $ T $ is a name $ \sumLock $ that is used in $ T $ as link with multiplicity two carrying two links with multiplicity zero.
	
	Let $ \sumLock $ be a sum lock. Then we refer to unguarded occurrences of $ \Output{\sumLock}{\true} $ as \emph{positive instantiation} and accordingly to unguarded occurrences of $ \Output{\sumLock}{\false} $ as \emph{negative instantiation} of $ \sumLock $. An \emph{instantiation} of a sum lock $ \sumLock $ is either a positive or negative instantiation of $ \sumLock $.
\end{definition}

Note that in most of the following definitions and proofs we hide the definition of booleans as well as of the corresponding test-statement. To show that sum locks meet our intuition we prove that in each target term $ T $ there is at most one instantiation of each sum lock.

\begin{lemma} \label{lem:instantiationSumLocks}
	For each target term each sum lock is instantiated at most once, i.e.,
	\begin{align*}
		\forall T \in \targetTermsMixAsyn \left( \text{or } T \in \targetTermsSepAsyn \right) \logdot \forall \sumLock \in \names \logdot \forall T' \in \piAsynProc \logdot \forall \tilde{x} \subset \names \logdot \forall b_1, b_2 \in \bool \logdot T \not\equiv \RestrictedTerm{\tilde{x}}{\left( T' \mid \Output{\sumLock}{b_1} \mid \Output{\sumLock}{b_2} \right)}.
	\end{align*}
\end{lemma}

\begin{proof}
	By Figure \ref{fig:encodingSepAsyn} and Figure \ref{fig:encodingMixAsyn} this condition holds for all encoded source terms, i.e., for all target terms $ \EncodingSepAsyn{S} $ for some $ S \in \piSepProc $ or $ S \in \piMixProc $, because for each sum there is exactly one positive instantiation of each sum lock and, since all sum locks appear restricted, the sum locks of different sums are different. All remainig instantiations of sum locks are guarded by a test-statement. To prove the condition for arbitrary target terms we take a closer look on these test-statements. We observe that for both encodings for each test-statement and for each of its possible outcomes the reduction of a test-statement unguards exactly one instantiation of a each sum lock that has to be consumed to reduce the respective test-statement. So for each new unguarded instantiation of a sum lock a former instantiation of the same lock was consumed.
	
	 \cite{nestmann00} proves that the encoding $ \encodingSepAsyn $ does not introduce deadlock, i.e., whenever a test-statement consumes a sum lock a new instantiation of the same lock is eventually unguarded. Moreover, it shows that a complete ordering of the sum locks as implemented in $ \encodingMixAsyn $ suffice to ensure that even in the case of source terms from $ \piMixProc $ the test-statements can not cause a deadlock. So again for each consumed instantiation of a sum lock a new instantiation of the same lock is eventually unguarded.
	\qed
\end{proof}

Note that, analysing the encoding functions obviously any instantiation of a sum lock is a positive or negative instantiation. The prove of Lemma \ref{lem:instantiationSumLocks} also shows that (1) all instantiations of sum locks in encoded source terms are positive|negative instantiations are only due to reduction steps, (2) all sum locks are initially instantiated, and (3) for each consumed instantiation of a sum lock eventually a new instantiation is unguarded.
\begin{corollary} \label{col:initialSumLocksArePositive}
	Any sum lock is initially instantiated positive, i.e.,
	\begin{align*}
		\forall S \in \piSepProc \logdot \forall \sumLock \in \names \logdot \forall b \in \bool \logdot \forall T \in \piAsynProc \logdot \forall \tilde{x} \subset \names \logdot \EncodingSepAsyn{S} \equiv \RestrictedTerm{\tilde{x}}{\left( T \mid \Output{\sumLock}{b} \right)} \text{ implies } b = \true
	\end{align*}
	and
	\begin{align*}
		\forall S \in \piMixProc \logdot \forall \sumLock \in \names \logdot \forall b \in \bool \logdot \forall T \in \piAsynProc \logdot \forall \tilde{x} \subset \names \logdot \EncodingMixAsyn{S} \equiv \RestrictedTerm{\tilde{x}}{\left( T \mid \Output{\sumLock}{b} \right)} \text{ implies } b = \true.
	\end{align*}
\end{corollary}

\begin{corollary} \label{col:instantiationsSumLocksAreRestored}
	Let $ T $ be a target term, i.e., $ T \in \targetTermsSepAsyn $ or $ T \in \targetTermsMixAsyn $, and let $ L \subset \names $ be the set of all sum locks of $ T $. Then
	\begin{align*}
		\forall \sumLock \in L \logdot \exists T' \in \piAsynProc \logdot \exists \tilde{x} \subset \names \logdot \exists b \in \bool \logdot T \steps \RestrictedTerm{\tilde{x}}{\left( T' \mid \Output{\sumLock}{b} \right)}.
	\end{align*}
\end{corollary}

In the proof of Lemma \ref{lem:instantiationSumLocks} we observe that new instantiations of sum locks are unguarded by test-statements. So reducing a test-statement is the only possibility to change an instantiation of sum lock. A closer look reveals that positive instantiations can be changed into negative but never the other way around.

\begin{lemma} \label{lem:changeInstantiationSumLock}
	A negative instantiation of a sum lock can not be changed into a positive instantiation, i.e.,
	\begin{align*}
		& \forall T_1, T_2 \in \targetTermsMixAsyn \left( \text{or } T_1, T_2 \in \targetTermsSepAsyn \right) \logdot \forall \sumLock \in \names \logdot \forall T_1', T_2' \in \piAsynProc \logdot \forall \tilde{x}_1, \tilde{x}_2 \subset \names \logdot \forall b \in \bool \logdot\\
		& \hspace*{1em} T_1 \equiv \RestrictedTerm{\tilde{x}_1}{\left( T_1' \mid \Output{\sumLock}{\false} \right)} \wedge T_1 \steps T_2 \wedge T_2 \equiv \RestrictedTerm{\tilde{x}_2}{\left( T_2' \mid \Output{\sumLock}{b} \right)} \text{ implies } b = \false
	\end{align*}
\end{lemma}

\begin{proof}
	Revisiting the argumentation in the proof of Lemma \ref{lem:instantiationSumLocks} we observe that the reduction of a test-statement is the only way to change the instantiation of a sum lock. A closer look at the test-statements in Figure \ref{fig:encodingMixAsyn} and Figure \ref{fig:encodingSepAsyn} reveals that for both encoding functions and for each test-statements the then-case can change a positive instantiation of sum lock into negative instantiation but the else-cases always simply restore all consumed instances. Note that in case of the nested test-statement of the encoding of an input guarded term a positive instantiation of the first tested lock is changed only if the second tested lock is again positive instantiated. In this case both instantiations are changed into negative instantiations. Moreover note that all the other (single) test-statements change any consumed positive instantiation into a negative one. Because of that, it is possible to change a positive instantiation into a negative one but not the other way around.
	\qed
\end{proof}

\paragraph*{Sender and Receiver Locks.} The third parameter of an output request refers to a sender lock, while the third parameter of an input request refers to a receiver lock. In both encodings, sender and receiver locks are used to guard the encoded continuation of output or input guarded source terms or replicated inputs. Sender locks are links with multiplicity zero in both encodings, in opposite receiver locks are links with multiplicity zero in $ \encodingSepAsyn $ and five in $ \encodingMixAsyn $. Note that in $ \encodingMixAsyn $ only sender locks appear as third parameter of an output request. Since by Lemma \ref{lem:encodingMixAsynPreserveRequests} the encoding function $ \encodingMixAsyn $ preserves requests, they unambiguous define sender locks. In the encoding $ \encodingSepAsyn $ also receiver locks are links that never transport any values (compare to Figure \ref{fig:encodingSepAsyn}). Moreover, the reserved names $ t $ and $ f $|necessary to implement booleans|are used as links without parameters in both encodings. To distinguish them in $ \encodingSepAsyn $ note, that sender locks are used as input links, while receiver locks are used as replicated inputs only, and that in each case of a test-statement there is an unrestricted instantiation of a sum lock, whereas all instantiations of sum locks within an encoded source term appear restricted.

\begin{definition}[Sender Lock] \label{def:senderLock}
	Let $ T \in \targetTermsSepAsyn $. Then any name $ \senderLock $ is a \emph{sender lock} of $ T $ if
	\begin{align*}
		\exists T', T'' \in \piAsynProc \logdot \exists \tilde{x} \subset \names \logdot T \steps \RestrictedTerm{\tilde{x}}{\left( T' \mid \In{\senderLock}.T'' \right)} \quad \text{and} \quad \forall \sumLock \in \FreeNames{T''} \logdot \forall b \in \bool \logdot \forall T''' \in \piAsynProc \logdot T'' \not\equiv \left( T''' \mid \Output{\sumLock}{b} \right).
	\end{align*}
	Let $ T \in \targetTermsMixAsyn $. Then any name $ \senderLock $ is a \emph{sender lock} of $ T $ if
	\begin{align*}
		\exists T' \in \piAsynProc \logdot \exists \tilde{x} \subset \names \logdot \exists \parallelChannelOut, y, \sumLock, z \in \names \logdot T \equiv \RestrictedTerm{\tilde{x}}{\left( T' \mid \Output{\parallelChannelOut}{y, \sumLock, \senderLock, z} \right)}.
	\end{align*}
	An \emph{instantiation} of a sender lock is an output on a sender lock.
\end{definition}

Beside the blocking of the encoded continuation, in $ \encodingMixAsyn $ the receiver lock is used by the encoding of the parallel operator and its pendant in the encoding of a replicated input to transmit the order of the sum locks back to the encoding of an input guarded source term. Remember, that the encoding of an input guarded source term tests these sum locks to \simulate a communication step of the source term and that the ordering of sum locks is necessary to avoid deadlock. In case of $ \encodingMixAsyn $, receiver locks are again unambiguous identified by input requests. However they are also the only links in $ \encodingMixAsyn $ carrying five parameters. In $ \encodingSepAsyn $ receiver locks are the only links of multiplicity zero, that are used as replicated inputs.

\begin{definition}[Receiver lock] \label{def:receiverLock}
	Let $ T \in \targetTermsSepAsyn $. Then any name $ \receiverLock $ is a \emph{receiver lock} of $ T $ if
	\begin{align*}
		\exists T', T'' \in \piAsynProc \logdot \exists \tilde{x} \subset \names \logdot T \equiv \RestrictedTerm{\tilde{x}}{\left( T' \mid \ReplicateIn{\receiverLock}.T'' \right)}.
	\end{align*}
	Let $ T \in \targetTermsMixAsyn $. Then any name $ \receiverLock $ is a \emph{receiver lock} of $ T $ if
	\begin{align*}
		\exists T' \in \piAsynProc \logdot \exists \tilde{x} \subset \names \logdot \exists \parallelChannelIn, y, \sumLock \in \names \logdot T \equiv \RestrictedTerm{\tilde{x}}{\left( T' \mid \Output{\parallelChannelIn}{y, \sumLock, \receiverLock} \right)}.
	\end{align*}
	An \emph{instantiation} of a receiver lock is an output on a receiver lock.
\end{definition}

Note that|for both encodings|for each input or output guarded source term exactly one receiver or sender lock is generated. Similarly, for each sum a unique sum lock is generated. However, since a sum may contain several input and/or output guarded summands, the encodings of different input or output guarded terms may share the same sum lock. But each encoded guarded term is connected to exactly one sum lock. The encoding $ \encodingMixAsyn $ is obviously more complex than the encoding $ \encodingSepAsyn $. But on the other hand the existence of requests outlines the connection between sum locks and sender or receiver locks|and with it to the corresponding encodings of guarded terms|more clearly.

\begin{definition} \label{def:connectionSumLockSenderReceiver}
	Let $ T \in \targetTermsMixAsyn $ and let $ \sumLock, \receiverLock, \senderLock \in \names $. If $ T $ contains an unguarded input request with $ \sumLock $ as second and $ \receiverLock $ as third parameter, i.e., if
	\begin{align*}
		\exists T' \in \piAsynProc \logdot \exists \tilde{x} \subset \names \logdot \exists \parallelChannelIn, y \in \names \logdot T \equiv \RestrictedTerm{\tilde{x}}{\left( T' \mid \Output{\parallelChannelIn}{y, \sumLock, \receiverLock} \right)},
	\end{align*}
	then we call the sum lock $ \sumLock $ and the receiver lock $ \receiverLock $ \emph{connected}. Sometime we also say that $ \sumLock $ is the sum lock of the receiver lock $ \receiverLock $.
	
	Accordingly, if $ T $ contains an unguarded output request with $ \sumLock $ as second and $ \senderLock $ as third parameter, i.e., if
	\begin{align*}
		\exists T' \in \piAsynProc \logdot \exists \tilde{x} \subset \names \logdot \exists \parallelChannelOut, y, z \in \names \logdot T \equiv \RestrictedTerm{\tilde{x}}{\left( T' \mid \Output{\parallelChannelOut}{y, \sumLock, \senderLock, z} \right)},
	\end{align*}
	then we call the sum lock $ \sumLock $ and the sender lock $ \senderLock $ \emph{connected}. Sometime we also say that $ \sumLock $ is the sum lock of the sender lock $ \senderLock $.
\end{definition}

Of course, the connection of sum locks to sender or receiver locks is unambiguous.

\begin{lemma} \label{lem:sumLockOfReceiverOrSender}
	Let $ T \in \targetTermsMixAsyn $ be a target term. Then each receiver lock $ \receiverLock $ and each sender lock $ \senderLock $ of $ T $ is connected to exactly one sum lock $ \sumLock $ of $ T $, i.e.,
	\begin{align*}
		& \forall T_1, T_2 \in \targetTermsMixAsyn \logdot \forall \receiverLock, \parallelChannelIn_1, \parallelChannelIn_2, y_1, y_2, \sumLock_1, \sumLock_2 \in \names \logdot \forall T_1', T_2' \in \piAsynProc \logdot \forall \tilde{x}_1, \tilde{x}_2 \subset \names \logdot\\
		& \hspace*{1em} T_1 \equiv \RestrictedTerm{\tilde{x}_1}{\left( T_1' \mid \Output{\parallelChannelIn_1}{y_1, \sumLock_1, \receiverLock} \right)} \wedge T_1 \steps T_2 \wedge T_2 \equiv \RestrictedTerm{\tilde{x}_2}{\left( T_2' \mid \Output{\parallelChannelIn_2}{y_2, \sumLock_2, \receiverLock} \right)} \text{ implies } y_1 = y_2 \wedge \sumLock_1 = \sumLock_2
	\end{align*}
	and
	\begin{align*}
		& \forall T_1, T_2 \in \targetTermsMixAsyn \logdot \forall \senderLock, \parallelChannelOut_1, \parallelChannelOut_2, y_1, y_2, \sumLock_1, \sumLock_2, z_1, z_2 \in \names \logdot \forall T_1', T_2' \in \piAsynProc \logdot \forall \tilde{x}_1, \tilde{x}_2 \subset \names \logdot\\
		& \hspace*{1em} T_1 \equiv \RestrictedTerm{\tilde{x}_1}{\left( T_1' \mid \Output{\parallelChannelOut_1}{y_1, \sumLock_1, \senderLock, z_1} \right)} \wedge T_1 \steps T_2 \wedge T_2 \equiv \RestrictedTerm{\tilde{x}_2}{\left( T_2' \mid \Output{\parallelChannelOut_2}{y_2, \sumLock_2, \senderLock, z_2} \right)}\\
		& \hspace*{1em} \text{implies } y_1 = y_2 \wedge \sumLock_1 = \sumLock_2 \wedge z_1 = z_2.
	\end{align*}
\end{lemma}

\begin{proof}
	Analysing the encoding function $ \encodingMixAsyn $ in Figure \ref{fig:encodingMixAsyn} we observe that initially, i.e., for each target term $ \EncodingMixAsyn{S} $ for some source term $ S \in \piMixProc $, for each receiver lock, i.e., for each encoding of an input guarded term or a replicated input, and for each sender lock, i.e., for each encoding of an output guarded term, exactly one request is generated. Since the receiver and sender locks are generated under restriction, for each encoded source term there are not two requests with the same receiver or the same sender lock, i.e., no two requests share their third parameter. Because of that the lemma holds for all $ T_1 = T_2 = \EncodingMixAsyn{S} $ for some source term $ S \in \piMixProc $. By Corollary \ref{col:originRequests} then the lemma holds for all target terms $ T_1 $ and $ T_2 $ such that $ T_1 \steps T_2 $.
	\qed
\end{proof}

Note that this lemma does not only shows that the connection between sum locks and sender or receiver locks is unambiguous but moreover that the information carried by requests is persistent. It does not change while requests wander to the structure of the encoded term generated by the parallel operator nesting of the corresponding source term.

Since, in case of $ \encodingMixAsyn $ receiver locks are not only used to guard the encoding of the continuation of an input guarded term or replicated input but also to send the order of the locks back to the corresponding test-statement, they carry all necessary informations to perform the test.

\begin{lemma} \label{lem:parametersReceiverLock}
	Let $ T \in \targetTermsMixAsyn $ and $ \receiverLock \in \names $ be a receiver lock of $ T $. Then the first three parameters of $ \receiverLock $ are sum locks, the fourth parameter is a sender lock $ \senderLock $, and the last parameter is a translated source term name. Moreover, the third sum lock belongs to $ \senderLock $ and among the first to sum locks one belongs to $ \receiverLock $ and the other one is again the sum lock of $ \senderLock $.
\end{lemma}

\begin{proof}
	There are four different outputs on receiver locks, i.e., outputs of five parameters, one in each of \processRightOutputRequests \ and \processRightInputRequests \ in the encoding of a parallel operator and in the encoding of a replicated input (compare to Figure \ref{fig:encodingMixAsyn}).
	\begin{align*}
		\processRightOutputRequests & \deff \Output{\coordinatorMatchingOut}{\matchingCoordinatorIn} \mid \ReplicateInput{\coordinatorMatchingOut}{\matchingCoordinatorIn}.\Input{\parallelChannelOut}{y, \sumLock_s, \senderLock, z}.\big(\\
		& \hspace*{2em} \RestrictedTerm{\matchingUpIn}{\big( \begin{aligned}[t]
				& \ReplicateInput{\matchingCoordinatorIn}{y', \sumLock_r, \receiverLock}.\left( \Match{y'}{y}\Output{\receiverLock}{\sumLock_r, \sumLock_s, \sumLock_s, \senderLock, z} \mid \Output{\matchingUpIn}{y', \sumLock_r, \receiverLock} \right)\\
				& \mid \RestrictedTerm{\matchingCoordinatorIn}{\left( \Forward{\matchingUpIn}{\matchingCoordinatorIn} \mid \Output{\coordinatorMatchingOut}{\matchingCoordinatorIn} \right)} \big)
			\end{aligned}\\
		& \hspace*{2em} \mid \Output{\coordinatorUpOut}{y, \sumLock_s, \senderLock, z}} \big)
	\end{align*}
	In \processRightOutputRequests \ the output $ \Output{\receiverLock}{\sumLock_r, \sumLock_s, \sumLock_s, \senderLock, z} $ is guarded by a replicated input $ \ReplicateInput{\matchingCoordinatorIn}{y', \sumLock_r, \receiverLock} $ of three parameters which is in turn guarded by an input $ \Input{\parallelChannelOut}{y, \sumLock_s, \senderLock, z} $ of four parameters. So to unguard the output on the receiver lock two requests|first an output request and then an input request|have to consumed. The values of the parameters of the output $ \Output{\receiverLock}{\sumLock_r, \sumLock_s, \sumLock_s, \senderLock, z} $ are completely determined by these two requests. Because of that the first three parameters are sum locks, the fourth parameter is a sender lock $ \senderLock $, and the last parameter is a translated source term name. Moreover, the first parameter is the sum lock of the receiver lock and the second and third parameter are the sum lock of the sender lock.
	\begin{align*}
		\processRightInputRequests & \deff \Output{\coordinatorMatchingIn}{\matchingCoordinatorOut} \mid \ReplicateInput{\coordinatorMatchingIn}{\matchingCoordinatorOut}.\Input{\parallelChannelIn}{y, \sumLock_r, \receiverLock}.\big(\\
		& \hspace*{2em} \RestrictedTerm{\matchingUpOut}{\big( \begin{aligned}[t]
				& \ReplicateInput{\matchingCoordinatorOut}{y', \sumLock_s, \senderLock, z}.\left( \Match{y'}{y}\Output{\receiverLock}{\sumLock_s, \sumLock_r, \sumLock_s, \senderLock, z} \mid \Output{\matchingUpOut}{y', \sumLock_s, \senderLock, z} \right)\\
				& \mid \RestrictedTerm{\matchingCoordinatorOut}{\left( \Forward{\matchingUpOut}{\matchingCoordinatorOut} \mid \Output{\coordinatorMatchingIn}{\matchingCoordinatorOut} \right)} \big)
			\end{aligned}\\
		& \hspace*{2em} \mid \Output{\coordinatorUpIn}{y, \sumLock_r, \receiverLock}} \big)
	\end{align*}
	The case of \processRightInputRequests is similar but here the input request has to be consumed first, and the first and the third parameter are the sum lock of the sender lock and the second parameter is the sum lock of the receiver lock.
	\qed
\end{proof}

In order to ease the proof of Lemma \ref{lem:nonConflictingStepsMixAsyn} at page \pageref{lem:nonConflictingStepsMixAsyn} we introduce another kind of lock. $ \encodingMixAsyn $ translates source term observables into requests, which are then combined to search for potential communication partners. In order to avoid divergence, requests can not be copied arbitrary often. To ensure that indeed each left request is combined with each possible matching right request, the right requests|in the encoding of a parallel operator as well as in the encoding of a replicated input|are linked within some kind of chain or list, along which the left requests are forwarded. Again to avoid divergence these chains or lists can not be infinitely long, so the links $ \coordinatorMatchingOut $ and $ \coordinatorMatchingIn $ are introduced by the encoding function to extent these chain or list by a new right request as soon as its last place is occupied. We will denote these links as \emph{chain locks}. Similarly, the chain lock $ \coordinatorRepA $ in the encoding of a replicated input is used to establish some kind of chain on encoded source terms|the encoded continuations of that replicated input|instead of right requests. Note that chain locks are the only links in target terms with multiplicity one.

\begin{definition}[Chain Lock] \label{def:chainLock}
	Let $ T \in \targetTermsMixAsyn $. Then any name $ \coordinatorMatchingOut $ is a \emph{chain lock} of $ T $ if $ \exists T' \in \piAsynProc \logdot \exists \tilde{x} \subset \names \logdot \exists \matchingCoordinatorIn \in \names \logdot T \steps \RestrictedTerm{\tilde{x}}{\left( T' \mid \Output{\coordinatorMatchingOut}{\matchingCoordinatorIn} \right)} $.
	
	An \emph{instantiation} of a chain lock is an output on a chain lock.
\end{definition}

Note that the value of a chain lock used to establish a chain of right requests is always a request channel, while the value of a chain lock used to establish a chain of encoded source terms is always a translated source term name, i.e., a value never used as link. Because of that, we can easily distinguish these two kinds of chain locks by a simple type information.

\subsection{Steps of an \Simulation} \label{sec:stepsSimulation}

Before we formally define receiver and sender locks in the last section we argue that they are introduced by the encoding function to guard the encoding of the continuation of a guarded source term. To mimic the behaviour of the source term, the encodings of continuations have to stay guarded until the source term step unguarding them in the source term is \simulated by its encoding. As already mentioned both encodings translate a single source term step into a sequence of target term steps called \simulation. However, in both encodings we can unambiguous allocate the main responsibility for an \simulation to a single step of that \simulation and call all the other steps pre- or postprocessing steps of that \simulation. In the following we will refer to the first kind of target term steps as \emph{\nonAdmin steps}, since they perform the main task of an \simulation and because of that constitute the transition from the \simulation of a source term to the \simulation of its successor. It turns out, that for both encodings the \nonAdmin steps are connected to the test-statements. More precisely, in case of an \simulation of a step on a term guarded by $ \tau $ or a replicated input, it is the consumption of the positive instantiation of a sum lock by the corresponding test-statement that performs the main task of the \simulation. In case of an \simulation of a step on an input guarded source term, it is the consumption of the second positive instantiation in the nested test-statement that we call \nonAdmin step.

\begin{definition}[\NonAdmin Step] \label{def:nonAdminStep}
	Let $ T_1, T_2 \in \targetTermsMixAsyn $ (or $ T_1, T_2 \in \targetTermsSepAsyn $). A step $ T_1 \step T_2 $ is a \emph{\nonAdmin step}, denoted by $ T_1 \nonAdminStep T_2 $, if this step consumes a positive instantiation of a sum lock either within a single test-statement or within the second test of a nested test-statement.
\end{definition}

Note that, since negative instantiations of sum locks refer to encoded in- or outputs, that remains of a former \simulation as junk, we do not consider any consummation of a negative sum lock as \nonAdmin step. According, test-statements consuming a negative instantiation of a sum lock only restore all consumed information. In the following we prove our intuition of \nonAdmin steps by showing that encoded continuations can only be unguarded after a \nonAdmin step.

\begin{lemma} \label{lem:nonAdminStepContinuation}
	Only \nonAdmin steps may lead to the unguarding of the encoding of a continuation of some source term.
\end{lemma}

\begin{proof}
	Analysing the encoding functions in Figure \ref{fig:encodingSepAsyn} and Figure \ref{fig:encodingMixAsyn} we observe that the encoded continuations of output guarded source terms appear guarded by the respective sender lock. Moreover we observe that all instantiations of sender locks are guarded by test-statements. More precisely, they are guarded by the then-case of a single test-statement (due to the encoding of a replicated input) or the then-case of the second test in a nested test-statement (due to the encoding of an input guarded term). Note that in case of $ \encodingMixAsyn $, by Lemma \ref{lem:parametersReceiverLock} the input on the receiver lock that guards the test-statements unambiguous identifies the following outputs of multiplicity zero as sender locks. So a \nonAdmin step is necessary to unguard them.
	
	In case of a source term guarded by $ \tau $ or an input guarded source term the respective encoded continuations appear directly as required in the respective then-cases of the test-statements. So, in these cases, the lemma holds directly by the Definition of the encoding functions.
	\qed
\end{proof}

Note that to our intuition for each \simulation of a source term step, there is exactly one \nonAdmin step. However, we will need some further information to prove that statement (compare to Lemma \ref{lem:simulationVSNonAdminStep} at page \pageref{lem:simulationVSNonAdminStep}). So let us have a closer look at the remaining pre- and post processing steps. There is one step of an \simulation, namely the unguarding of the encoded continuation of an output guarded source term by communication over a sender lock, that for certainty has to be performed after the \nonAdmin step of the corresponding \simulation. Because of that we can call reductions on sender locks postprocessing steps. There are also steps, as for instance reductions over receiver locks, that for certainty have to be performed before the corresponding \nonAdmin step. So we can call all reductions on receiver locks preprocessing steps. However, there are some steps that may be performed before or after the corresponding \nonAdmin step. Moreover, the fact whether a non \nonAdmin step was performed before or after the corresponding \nonAdmin step is usually not important and often hard to prove. So pre- or postprocessing steps is not a good characterisation for our purposes. Instead we will refer to those steps as \emph{\admin steps}, since they perform administrative tasks, that are necessary to perform an \simulation, but they do not carry the main responsibility for the \simulation, i.e., they do|at least for the general case|not inevitably implement the decision to \simulate a specific source term step.

\begin{definition}[\Admin Step] \label{def:adminStep}
	Let $ T_1, T_2 \in \targetTermsMixAsyn $ (or $ T_1, T_2 \in \targetTermsSepAsyn $). A step $ T_1 \step T_2 $ is a \emph{\admin step}, denoted by $ T_1 \adminStep T_2 $, if it is no \nonAdmin step.
	
	Let $ T_1 \adminSteps T_2 $ denote a sequence of \admin steps, i.e., $ \adminSteps $ is the transitive and reflexive closure of $ \adminStep $.
\end{definition}

In the easiest case, none of the \admin steps influences the decision to \simulate a specific source term steps. That means: Consider a source term that can perform alternative but conflicting steps. So the encoding can perform different but conflicting \simulations. In this case none of the \admin steps should influences which of the \simulations may be completed, i.e., no sequence of \admin steps should be able to rule out the completion of one of these \simulations.

Unfortunately, for both of the presented encodings this turns out to be wrong. It is always a \nonAdmin step that finally decides which of the conflicting \simulations is completed by preventing any other conflicting \simulation from completion. However, in case there are more than two possible conflicting \simulations, a sequence of \admin steps may rule out one alternative while allowing for different still possible \simulations.

\begin{example} \label{exa:intermediateStates}
	Let us consider the source term $ S = \left( \In{a}.S_1 + \In{a}.S_2 \right) \mid \Out{a}.S_3 \mid \In{a}.S_4 $ for some $ S_1, S_2, S_3, S_4 \in \piSepProc $. So $ S \in \piSepProc $ as well as $ S \in \piMixProc $. The encoding of $ S $, regardless whether it is encoded by $ \encodingSepAsyn $ or $ \encodingMixAsyn $, generates three sum locks, one for each sum. Let us assume, the sum lock generated for $ \In{a}.S_1 + \In{a}.S_2 $ is $ \sumLock_1 $, the sum lock for $ \Out{a}.S_3 $ is $ \sumLock_2 $, and $ \sumLock_3 $ is generated for $ \In{a}.S_4 $. $ S $ can perform three conflicting steps leading to $ S_1 \mid S_3 $, $ S_2 \mid S_3 $, or $ S_3 \mid S_4 $, respectively. Each of these steps can be \simulated by both of the encodings. To \simulate the step to $ S_1 \mid S_3 $ for both encodings the positive instantiation of the sum lock $ \sumLock_1 $ is consumed first and to complete the \simulation a positive instantiation of the sum lock $ \sumLock_2 $ has to be consumed. However, it is possible, that the encoded term instead performs another (nested) test-statement and consume both positive instantiations of $ \sumLock_2 $ and $ \sumLock_3 $ to \simulate the step to $ S_3 \mid S_4 $ instead. Because of that, in case of a nested test-statement, we do not consider the first consumption of a positive instantiation of a sum lock as a \nonAdmin step. Even, if at this point the sum lock, that is tested next by this test-statement, is still instantiated positive, it can become negative by an interleaving other test-statement.
	
	So after the consumption of the positive instantiation of $ \sumLock_1 $ in order to \simulate the source term step to $ S_1 \mid S_3 $, there is still the possibility to complete instead the \simulation of the source term step to $ S_3 \mid S_4 $. However, there is no possibility to complete the \simulation of the source term step to $ S_2 \mid S_3 $. Note that, to \simulate the step to $ S_2 \mid S_3 $, a positive instantiation of each of the locks $ \sumLock_1 $ and $ \sumLock_2 $ is necessary. Here, the instantiation of $ \sumLock_1 $ is consumed. The only possibility to restore the positive instantiation is to consume a negative instantiation of $ \sumLock_2 $ (compare to the nested test-statements in the encodings of input guarded terms in Figure \ref{fig:encodingSepAsyn} and Figure \ref{fig:encodingMixAsyn}). By Lemma \ref{lem:changeInstantiationSumLock}, then there is no possibility to change that negative instantiation back into a positive one. So, as soon as $ \sumLock_1 $ is consumed, one of the three possible \simulations is ruled out while there are still two possible \simulations left.
\end{example}

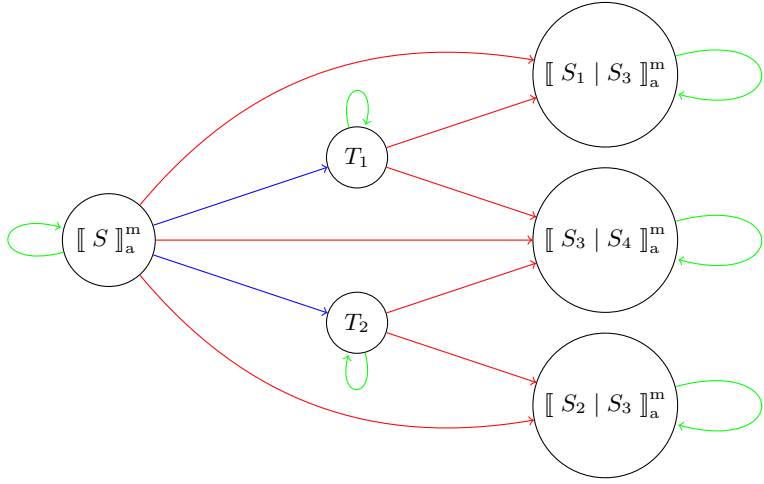
\begin{figure}[ht]
	\centering
	\begin{tikzpicture}[auto,node distance=1.1cm]
		\node[state]	(S)					{$ \EncodingMixAsyn{S} $};
		\node			(d1) [right of=S]	{};
		\node			(d2) [right of=d1]	{};
		\node			(d3) [right of=d2]	{};
		\node[state]	(T1) [above of=d3]	{$ T_1 $};
		\node[state]	(T2) [below of=d3]	{$ T_2 $};
		\node			(d6) [right of=d3]	{};
		\node			(d7) [right of=d6]	{};
		\node[state]	(S2) [right of=d7]	{$ \EncodingMixAsyn{S_3 \mid S_4} $};
		\node			(d8) [above of=S2]	{};
		\node[state]	(S1) [above of=d8]	{$ \EncodingMixAsyn{S_1 \mid S_3} $};
		\node			(d9) [below of=S2]	{};
		\node[state]	(S3) [below of=d9]	{$ \EncodingMixAsyn{S_2 \mid S_3} $};
		
		\path[->,green]	(S)		edge [loop left]	node {} ();
		\path[->,blue]	(S)		edge				node {} (T1);
		\path[->,green]	(T1)	edge [loop above]	node {} ();
		\path[->,red]	(T1)	edge				node {} (S1);
		\path[->,red]	(T1)	edge				node {} (S2);
		\path[->,blue]	(S)		edge				node {} (T2);
		\path[->,green]	(T2)	edge [loop below]	node {} ();
		\path[->,red]	(T2)	edge				node {} (S2);
		\path[->,red]	(T2)	edge				node {} (S3);
		\path[->,red]	(S)		edge [bend left]	node {} (S1);
		\path[->,green]	(S1)	edge [loop right]	node {} ();
		\path[->,red]	(S)		edge				node {} (S2);
		\path[->,green]	(S2)	edge [loop right]	node {} ();
		\path[->,red]	(S)		edge [bend right]	node {} (S3);
		\path[->,green]	(S3)	edge [loop right]	node {} ();
	\end{tikzpicture}
	\caption{Intermediate States.} \label{fig:intermediateStates}
\end{figure}
	
We visualise this phenomenon in Figure \ref{fig:intermediateStates} (for the case of $ \encodingMixAsyn $). Here the red lines denote sequences of steps with (exactly) one \nonAdmin step, the blue lines denote sequences without \nonAdmin steps but with at least one \impure \admin step, and the green lines denote sequences of only \pure \admin steps (compare to Definition \ref{def:pureImpureAdminStep}). Between the encoding of the source term $ S $ and the encodings of its reducts $ S_1 \mid S_3 $, $ S_2 \mid S_3 $, and $ S_3 \mid S_4 $ there are two \emph{intermediate states}, i.e., two states that differ from each encoded source term within that picture. Note that the observability of such an intermediate state depends on the chosen equivalence on target terms.

\begin{definition}[Intermediate State] \label{def:intermediateState}
	A term $ T' $ is an \emph{intermediate state}, if
	\begin{align*}
		\exists T, T_1, T_2, T_3 \logdot T \steps T_1 \wedge T \steps T_2 \wedge T \steps T_3 \wedge T \steps T' \wedge T' \steps T_1 \wedge T' \steps T_2 \wedge T' \not\steps T_3,
	\end{align*}
	i.e. if
	
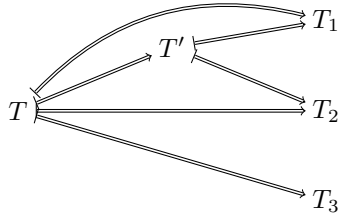
\begin{center}
		\begin{tikzpicture}[auto,node distance=1.2cm]
			\node (T)						{$ T $};
			\node (d1)	[right of=T]		{};
			\node (T')	[above right of=d1]	{$ T' $};
			\node (d2)	[below right of=T']	{};
			\node (T2)	[right of=d2]		{$ T_2 $};
			\node (T1)	[above of=T2]		{$ T_1 $};
			\node (T3)	[below of=T2]		{$ T_3 $};
			
			\def\myArrow(#1)#2(#3){
				\draw[|->,shorten >=-1pt,shorten <=-0.5pt] (#1) #2 (#3);
				\draw[double] (#1) #2 (#3);
			}
			
			\myArrow (T) -- (T');
			\draw[|->,shorten >=-1pt,shorten <=-0.5pt] (T) edge [bend left] (T1);
			\draw (T) edge [double,bend left] (T1);
			\myArrow (T) -- (T2);
			\myArrow (T) -- (T3);
			\myArrow (T') -- (T1);
			\myArrow (T') -- (T2);
		\end{tikzpicture}
	\end{center}
\end{definition}

Remarkably, the existence of intermediate states in $ \encodingSepAsyn $ is independent of structural congruence|if two source terms are structural congruent, then their encodings have the same intermediate states|while this is not true for $ \encodingMixAsyn $. Here the locks are always tested according to a total ordering created along the structure induced by the nesting of parallel operators of the source term. Because of that, this ordering of sum locks differ for source terms, that are structural congruent but differ in the order of their subprocesses, i.e. differ by rule $ P \mid Q \equiv Q \mid P $. So structural congruent source terms can differ in the number and nature of reachable intermediate states; e.\,g.\ $ S' = \Out{a}.S_3 \mid \left( \In{a}.S_1 + \In{a}.S_2 \right) \mid \In{a}.S_4 $, which is structural congruent to $ S $ from Example \ref{exa:intermediateStates}, does not reach any of the above intermediate states. Instead, in $ S' $ the first consumption of a positive instantiation of a sum lock, which is no \nonAdmin step here, completely determines which \simulation can be completed.

To capture that fact we further distinguish \admin steps, into \pure \admin steps|that never rule out the completion of any possible \simulation\!\!|and \impure \admin steps|that due to the consumption of a positive instantiation of a sum lock may possibly rule out the completion of an \simulation. Note that in case of $ \encodingMixAsyn $ requests are copied to ensure that each possible combination of input and output requests is checked exactly once. Because of that steps on requests are \pure \admin steps. In opposite, $ \encodingSepAsyn $ does not translate source term observables into requests. To check for a potential pair of translated communication partners a communication of the translated channel names is performed. Similar to the consumption of positive instantiations of sum locks this might rule out alternative \simulations. Because of that we consider steps on translated source term names as \impure \admin steps of $ \encodingSepAsyn $. Also note, that in $ \encodingMixAsyn $ translated source term names are used as values only; so there are no steps on translated source term names.

\begin{definition}[\Pure and \Impure \Admin Step] \label{def:pureImpureAdminStep}
	Let $ T_1, T_2 \in \targetTermsMixAsyn $ (or $ T_1, T_2 \in \targetTermsSepAsyn $). A step $ T_1 \adminStep T_2 $ is a \emph{\pure \admin step}, denoted by $ T_1 \pureAdminStep T_2 $, if it is neither on a sum lock nor on a translated source term name, else it is an \emph{\impure \admin step}, denoted by $ T_1 \impureAdminStep T_2 $.
	
	Let $ T_1 \pureAdminSteps T_2 $ denote a sequence of \pure \admin steps, i.e., $ \pureAdminSteps $ is the transitive and reflexive closure of $ \pureAdminStep $.
\end{definition}

To show that the definition of \pure \admin steps meets our intuition, we prove some kind of local confluence property for most of the \pure \admin steps. Intuitively, it states that indeed none of the \pure \admin steps can rule out the completion of any \simulation, because they are (in most cases) not conflicting, i.e., does not rule out any other sequence of steps.

\begin{lemma} \label{lem:nonConflictingStepsSepAsyn}
	Within target terms \pure \admin steps are not conflicting, i.e.,
	\begin{align*}
		\forall T, T_1, T_2 \in \targetTermsSepAsyn \logdot T \pureAdminStep T_1 \wedge T \steps T_2 \text{ implies } \exists T' \in \targetTermsSepAsyn \logdot T_1 \steps T' \wedge T_2 \pureAdminStep T'.
	\end{align*}
\end{lemma}

\begin{proof}
	In comparison to $ \encodingMixAsyn $ the encoding $ \encodingSepAsyn $ introduces only a few \pure \admin steps and all of them are not conflicting. First note that, since in \piAsyn there are no sums, two target term steps can only be in conflict if one of it consumes some input or output necessary to perform the other step. So it suffice to concentrate on steps on the same channel. Analysing the encoding function in Figure \ref{fig:encodingSepAsyn} we discover that steps on receiver or sender locks (compare to Definitions \ref{def:receiverLock} and \ref{def:senderLock}) are \pure \admin steps.
	
	In case of receiver lock, since they are generated under restriction, for each receiver lock there is exactly one replicated input and no other input. Because of that it does not matter how many other steps on (the same) receiver lock may appear within the sequence $ T \steps T_2 $ the step $ T \pureAdminStep T_1 $ can be performed before or after that sequence and in both cases the same term $ T' $ is reached.
	
	In case of sender locks, there is exactly one input and no replicated input for each sender lock. Moreover, we observe that initially there is no instantiation of a sender lock, and there is exactly one output (on a translated source term name) which carries the sender lock as value. This output is consumed to unguard a test-statement and that test-statement can then ungard again at most one output which carries the sender lock as value and which can itself unguard another test-statement. The only instantiations of sender locks are due to the then-case of a single test-statement in the encoding of a replicated input or the then-case of the second test-statement in the encoding of an input guarded term, and in both cases only a single instantiation of a sender lock is unguarded. So for each target term and each sender lock there can be at most one instantiation of a sender lock. Moreover note, that the output on the translated source term names does not only carry the sender lock, but also a sum lock. In order to obtain an instantiation of the sender lock this sum lock has to be instantiated positive. But whenever a sender lock is instantiated the encoding also generates a negative instantiation of that sum lock. By Lemma \ref{lem:changeInstantiationSumLock} that negative instantiation can never be turned into a positive one again; so there is no chance to generate a second instantiation on the same sender lock. Because of that, if $ T \pureAdminStep T_1 $ is a step on a sender lock, then none of the steps in $ T \steps T_2 $ is a step on that sender lock. So the lemma holds.
	
	Besides these to steps there is another kind of \pure \admin steps that is not that obvious in Figure \ref{fig:encodingSepAsyn} because that kind of steps is hidden by the abbreviation used to introduce booleans and test-statements (compare to Definition \ref{def:testBoolean}). We observe, that a test-statement is reduced within two steps. The first consumes the instantiation of the sum lock and is thus no \pure \admin step. The second step unguards the corresponding then- or else-case. In both cases that step is a \pure \admin step. However the names $ t $ and $ f $ are restricted for each test-statement, are not used any there else by the encoding function, and thanks to the renaming policy are different from each translated source term name. Because of that, we can again conclude that, if $ T \pureAdminStep T_1 $ is a step on a version of $ t $ or $ f $, then none of the steps in $ T \steps T_2 $ is a step on the same name. So the lemma holds.
	\qed
\end{proof}

\begin{lemma} \label{lem:nonConflictingStepsMixAsyn}
	Within target terms \pure \admin steps, that either are on sender locks or booleans, or do not unguard an instantiation of a chain lock carrying a request channel, are not conflicting, i.e.,
	\begin{align*}
		& \forall T, T_1, T_2 \in \targetTermsMixAsyn \logdot T \pureAdminStep T_1 \wedge T \steps T_2\\
		& \hspace*{3em} \wedge \big( \begin{aligned}[t]
				& T \step T_1 \text{ is a step on a sender lock or booleans}\\
				& \vee \big( \forall T_1' \in \piAsynProc \logdot \forall \tilde{x} \subset \names \logdot \forall \coordinatorMatchingOut, \matchingCoordinatorIn \in \names \logdot T_1 \equiv \RestrictedTerm{\tilde{x}}{\left( T_1' \mid \Output{\coordinatorMatchingOut}{\matchingCoordinatorIn} \right)} \text{, where } \matchingCoordinatorIn \text{ is a request channel}\\
				& \hspace*{1em} \text{ implies } \left( \exists T' \in \piAsynProc \logdot \exists \tilde{x}_1 \subset \names \logdot T \equiv \RestrictedTerm{\tilde{x}_1}{\left( T' \mid \Output{\coordinatorMatchingOut}{\matchingCoordinatorIn} \right)} \right) \big) \big)
			\end{aligned}\\
		& \hspace*{1em} \text{ implies } \exists T_3 \in \targetTermsMixAsyn \logdot T_1 \steps T_3 \wedge T_2 \pureAdminStep T_3.
	\end{align*}
\end{lemma}

\begin{proof}
	$ \encodingMixAsyn $ relies on much more \pure \admin steps than $ \encodingSepAsyn $. First note that, the encoding of a parallel operator generates unguarded instantiations of chain locks carrying a request channel. Since a step on a sender lock and a step on $ t $ to unguard the then-case of a test-statement unguards an encoded continuation, they unguard instantiations of such chain locks, if the corresponding source term in the continuation contains a parallel operator. Nevertheless, we want to prove the condition also for steps on sender locks and booleans.
	
	Since source term names are translated into values, never used as links, it suffice to consider steps on names introduced by the encoding function. A look at the definition of the corresponding renaming policy in Figure \ref{fig:encodingMixAsyn} suggests the following case split\footnote{Note that in most cases the considered names are restricted, so a simple alpha conversion may change them. Because of that the use of concrete names in the following case split should not imply that we consider steps on these specific names. Instead the names refer to the meaning which is related to them by the encoding function.} on the subject of the step $ T \step T_1 $.
	\begin{description}
		\item[Case of $ \parallelChannelOut, \parallelChannelIn, \coordinatorUpOut, \coordinatorUpIn, \matchingCoordinatorOut, \matchingCoordinatorIn, \matchingUpOut, \matchingUpIn, \matchingReceiverIn, \matchingReceiverOut, \matchingReceiverUpIn, \matchingReceiverUpOut $:] All these names are request channels, i.e., there are introduced by $ \encodingMixAsyn $ to transport requests. Note that the encoding function puts much effort in the direction of requests. Usually there is exactly one way for them, namely: (1) upwards in the structure generated by the nesting of parallel operators in the corresponding source term, (2) within the encoding of a parallel operator or each branch of the encoding of a replicated input from the left side to each right request, which are linked within a chain, and (3) within the encoding of a replicated input from each branch to each next branch, where each branch represents an encoding of the continuation of that replicated input and the branches are again linked within some kind of chain.
			
			Indeed there is only one point, at which the way of a requests is not completely determined. That is the point at which right requests are linked within a chain (compare to the encoding of a parallel operator or a replicated input). The order in which the right requests are consumed determines their order in the chain, so these steps are conflicting. They are steps on the first input on output requests in \processRightOutputRequests \ and on the first input on input requests in \processLeftOutputRequests.
			\begin{align*}
				\processRightOutputRequests & \deff \Output{\coordinatorMatchingOut}{\matchingCoordinatorIn} \mid \ReplicateInput{\coordinatorMatchingOut}{\matchingCoordinatorIn}.\Input{\parallelChannelOut}{y, \sumLock_s, \senderLock, z}.\big(\\
					& \hspace*{2em} \RestrictedTerm{\matchingUpIn}{\big( \begin{aligned}[t]
							& \ReplicateInput{\matchingCoordinatorIn}{y', \sumLock_r, \receiverLock}.\left( \Match{y'}{y}\Output{\receiverLock}{\sumLock_r, \sumLock_s, \sumLock_s, \senderLock, z} \mid \Output{\matchingUpIn}{y', \sumLock_r, \receiverLock} \right)\\
							& \mid \RestrictedTerm{\matchingCoordinatorIn}{\left( \Forward{\matchingUpIn}{\matchingCoordinatorIn} \mid \Output{\coordinatorMatchingOut}{\matchingCoordinatorIn} \right)} \big)
						\end{aligned}\\
					& \hspace*{2em} \mid \Output{\coordinatorUpOut}{y, \sumLock_s, \senderLock, z}} \big)\\
				\processRightInputRequests & \deff \Output{\coordinatorMatchingIn}{\matchingCoordinatorOut} \mid \ReplicateInput{\coordinatorMatchingIn}{\matchingCoordinatorOut}.\Input{\parallelChannelIn}{y, \sumLock_r, \receiverLock}.\big(\\
					& \hspace*{2em} \RestrictedTerm{\matchingUpOut}{\big( \begin{aligned}[t]
							& \ReplicateInput{\matchingCoordinatorOut}{y', \sumLock_s, \senderLock, z}.\left( \Match{y'}{y}\Output{\receiverLock}{\sumLock_s, \sumLock_r, \sumLock_s, \senderLock, z} \mid \Output{\matchingUpOut}{y', \sumLock_s, \senderLock, z} \right)\\
							& \mid \RestrictedTerm{\matchingCoordinatorOut}{\left( \Forward{\matchingUpOut}{\matchingCoordinatorOut} \mid \Output{\coordinatorMatchingIn}{\matchingCoordinatorOut} \right)} \big)
						\end{aligned}\\
					& \hspace*{2em} \mid \Output{\coordinatorUpIn}{y, \sumLock_r, \receiverLock}} \big)
			\end{align*}
			In both cases immediately an instantiation of a chain lock carrying a request channel is unguarded. So the lemma holds, because its precondition is violated.
			
			Let us have a look at the remaining request channels. Within the left side of the encoding of a parallel operator there are two restricted different replicated inputs on requests channels. Note that, the encoding function places all inputs or replicated inputs under restriction, i.e., for all source terms their encoding has no inputs or replicated inputs on free names. Because of that, for the request channels restricted at the left side of the parallel operator encoding there is exactly one replicated input and no other input. Thus it does not matter how many other steps on the same request channel may appear within the sequence $ T \steps T_2 $, the step $ T \pureAdminStep T_1 $ can be performed before or after that sequence and in both cases the same term $ T_3 $ is reached.
			
			Beside the already considered possibly conflicting input, there are two different replicated inputs on request channels and no other input on request channels within \processRightOutputRequests. The link of the first is bounded by a guarding replicated input on a chain lock. The link of the second is restricted. So we can again apply the argumentation of the case before. The same applies to \processRightInputRequests. The two links of replicated inputs in \pushRequests \ and the four links of replicated inputs \pushRequestsOut \ are restricted, apart from that, this case is similar to the cases before. In opposite, in case of \pushRequestsIn \ the links of both replicated inputs are bounded by a guarding input, apart from that, this case is similar to the cases before.
		\item[Case of $ \coordinatorMatchingOut, \coordinatorMatchingIn $:] These names is used as chain locks carrying a request channel (compare to Definition \ref{def:chainLock}). Those chain locks are used by the encoding function to direct the combinations of left and right requests in the encoding of a parallel operator as well as in the encoding of a replicated input. In order to avoid divergence, requests can not be copied arbitrary often. To ensure that indeed each left request is combined which each possible matching right request, the right requests are linked within some kind of chain or list, along which the left requests are forwarded. Again to avoid divergence these chains can not be infinitely long, so $ \coordinatorMatchingOut $ and $ \coordinatorMatchingIn $ allow to extent these chains by a new right request as soon as its last place is occupied. Since these names are generated under restriction, for each of them there is exactly one replicated input and no other input. Because of that it does not matter how many other steps on the same names may appear within the sequence $ T \steps T_2 $, the step $ T \pureAdminStep T_1 $ can be performed before or after that sequence and in both cases the same term $ T_3 $ is reached.
		\item[Case of $ \sumLock, \sumLock_s, \sumLock_r, \sumLock_1, \sumLock_2 $:] Any of these names refer to a sum lock in the encoding given in Figure \ref{fig:encodingMixAsyn}. By Definition \ref{def:pureImpureAdminStep} steps on sum locks are no \pure \admin steps.
		\item[Case of $ \senderLock $:] $ \senderLock $ is used by the encoding function to introduce sender locks. For steps on sender locks it suffice to repeat the argumentation given in the proof of Lemma \ref{lem:nonConflictingStepsSepAsyn} for sender locks.
		\item[Case of $ \receiverLock $:] In case of receiver lock, since they are generated under restriction, for each receiver lock there is exactly one replicated input and no other input. Because of that it does not matter how many other steps on (the same) receiver lock may appear within the sequence $ T \steps T_2 $, the step $ T \pureAdminStep T_1 $ can be performed before or after that sequence and in both cases the same term $ T_3 $ is reached.
		\item[Case of $ \coordinatorRepA $:] $ \coordinatorRepA $ again is a chain lock, this time carrying a translated source term name. Intuitively, it is used by the encoding function for a purpose similar to the usage of chain locks carrying a request channel. Instead of a chain of right requests, $ \coordinatorRepA $ as well as $ \coordinatorRepB $ are used to build up a some kind of chain of the encoded continuations of several reductions on the same encoded replicated input. As already explained our encoding $ \encodingMixAsyn $ relies on the structure with is build by the parallel operators in the corresponding source term. Each reduction of a replicated input changes this structure. To allow for different encoded continuations to communicate among each other or with the encoded replicated input we link the encodings of the continuations. This time we have to add a new member, i.e., an encoded continuation, to the chain whenever the encoded replicated input is used to \simulate a step. Therefore $ \coordinatorRepA $ is instantiated within the then-case of the test-statement in the encoding of a replicated input. Since $ \coordinatorRepA $ is generated under restriction, for each $ \coordinatorRepA $ there is exactly one replicated input and no other input. Because of that, it does not matter how many other steps on the same name may appear within the sequence $ T \steps T_2 $, the step $ T \pureAdminStep T_1 $ can be performed before or after that sequence and in both cases the same term $ T_3 $ is reached.
		\item[Case of $ \coordinatorRepB $:] To link the members in the chain of right requests for each new member a new $ \matchingCoordinatorOut $ or $ \matchingCoordinatorIn $ is restricted and transmitted over $ \coordinatorMatchingOut $ or $ \coordinatorMatchingIn $. The encoded continuations of a replicated input are linked over $ \matchingReceiverIn $ and $ \matchingReceiverOut $, which are again restricted for each encoded continuation. $ \coordinatorRepB $ is used to transmit these new restricted names to the respective next member of the chain. Note that, this kind of link is generated always under restriction. Initially there is exactly one unguarded output and one input, guarded by a replicated input on a chain lock. Reducing this input immediately unguards some instantiations of chain locks, so the lemma holds, because its precondition is violated. Also note that, due to several \simulations on the encoded replicated input, there may be several unguarded inputs on $ \coordinatorRepB $. The order in which these inputs are consumed determines the order of the encoded continuations within the constructed chain. Because of that steps on $ \coordinatorRepB $ can be indeed conflicting.
		\item[Case of $ y, y', z $:] These names are used by the encoding function as values only, but never as links. So there are no (\pure \admin) steps on these names.
		\item[Case of $ t, f $:] These names are used only to implement the test-statements and the instantiation of sum locks. The reduction of a test-statement is performed in two steps. The first consumes the instantiation of a sum lock and is thus not a \pure \admin step. For the second one|necessary to unguard the then- or else-case of a test-statement|it suffice to repeat the argumentation given in the proof of Lemma \ref{lem:nonConflictingStepsSepAsyn} for these kind of steps.
	\end{description}
	\qed
\end{proof}

Note that the proof above provides a detailed explanation of the purposes of the names reserved by the encoding function $ \encodingMixAsyn $. Moreover note that the \pure \admin steps that are conflicting are exactly the steps that introduce additional causal dependencies (compare to \cite{petersSchickeNestmann11} and thus prevent the preservation of the degree of distribution (compare to Section 4 in \cite{petersNestmann12}). In the following we strongly rely on the Lemmata \ref{lem:nonConflictingStepsSepAsyn} and \ref{lem:nonConflictingStepsMixAsyn}, because they basically allow us to ignore all not conflicting steps while considering the reachability of success or translated observables in the next section.

In order to show that the remaining \pure \admin steps do not cause any problems but in fact, as described in the proof of Lemma \ref{lem:nonConflictingStepsMixAsyn}, do only influence the order of right requests in the chain of right request or the order of encoded continuations of a replicated input generated by its encoding, we prove that those steps never cause any deadlock. Note that, \cite{nestmann00} proves that the encoding $ \encodingSepAsyn $ does not introduce deadlocks.

\begin{lemma} \label{lem:pureAdminStepsNoDeadlock}
	\Pure \admin steps do not introduce deadlocks.
\end{lemma}

\begin{proof}
	In case of a step on a sender lock or a step which does not unguard an instantiation of a chain lock, this lemma directly follows by Lemma \ref{lem:nonConflictingStepsMixAsyn}.
	
	In case of a step on a request channel which does unguard an instantiation of a chain lock (compare to the first case of the proof of Lemma \ref{lem:nonConflictingStepsMixAsyn}), a step on the unguarded chain lock unguards another input on the same request channel. Moreover the corresponding continuations differ only on a single name free to that continuation. Let us denote the first such continuation by $ A $ and the second by $ B $, and the free link name by $ p $. Then it turns out, that $ B $ is connected to $ A $ by $ p $ and either $ A $ is connected in a similar way to another such continuation or $ A $ is connected to the requests from the left. The only requests travelling along $ p $ are left requests. However as soon as they arrive $ A $ they are copied and transmitted to $ B $. Because of that the order of $ A $ and $ B $ in that chain does not matter (at least as long as we do not consider causal dependencies). Of course, the step on the chain lock unguarding $ B $ is not forced to be performed immediately after the step on the request channel, but by Lemma \ref{lem:nonConflictingStepsMixAsyn} it will eventually happen. So the step on the request channel blocks alternative steps on this request channel for some time but not for ever, i.e., it does not introduce deadlock.
	
	In case of a step on a channel of multiplicity two|let us denote it $ c $|which does unguard an instantiation of a chain lock (compare to case $ \coordinatorRepB $ of the proof of Lemma \ref{lem:nonConflictingStepsMixAsyn}), immediately an other output on $ c $ is unguarded. As in the case before by the communication on the chain lock, the communication over $ c $ links the encodings of the continuations of the respective replicated input, in this case by two request channels. Apart from that, the situation is exactly the same as before. If there is an other encoded continuation of that replicated input|caused by another \simulation\!\!|it will eventually be linked within the chain and all requests that arrive at a previous member of that chain are forwarded to that member. Note that the encoded continuations initially are all equal. So the chains resulting from different orders of two members, that were available at the same time, are equal. We conclude that a step on $ c $ can not introduce deadlock.
	\qed
\end{proof}

\begin{lemma} \label{lem:impureAdminStepsNoDeadlock}
	\Impure \admin steps do not introduce deadlocks.
\end{lemma}

\begin{proof}
	In $ \encodingMixAsyn $ any \impure \admin step reduces a test-statement by consuming an instantiation of a sum lock (compare to Definition \ref{def:pureImpureAdminStep}). Deadlock occurs, if|due to interleaving of these test-statements|the instantiations of a subset of the sum locks are consumed, such that none of the involved test-statements can be resolved. As already stated in \cite{nestmann00} a total ordering on the sum locks suffice to circumvent any potential deadlock. Note that the encodings of parallel operator and replicated input implement such a total ordering on sum locks. They somehow reuse the structure generated by the parallel operators of the corresponding source term to force the nested test-statements to always test the lock first, which is according to that parallel structure left to the other one. Since the parallel operator is binary, this structure is a binary tree. So testing always the left lock first, indeed implements a total ordering.
	\qed
\end{proof}

Note that, by Lemma \ref{lem:simulationVSNonAdminStep} at page \pageref{lem:simulationVSNonAdminStep}, \nonAdmin steps and the source term steps of the corresponding \simulations coincide. So Lemma \ref{lem:simulationVSNonAdminStep} in combination with the Lemma \ref{lem:pureAdminStepsNoDeadlock} and Lemma \ref{lem:impureAdminStepsNoDeadlock} proves that the encoding $ \encodingMixAsyn $ does not introduce deadlocks.

\subsection{Translated Observables and Choosing a Bisimulation} \label{sec:transBarbBisim}

In order to prove the presented encodings correct with respect to the criteria of Gorla we have to choose an equivalence $ \asymp $ for operational correspondence (compare to Definition \ref{def:operationalCorrespondence}). In \cite{gorla10} Gorla describes $ \asymp $ as follows:
\begin{quotation}
	``$ \asymp $ is a behavioural equivalence needed to describe the abstract behaviour of a process. Usually, $ \asymp $ is a congruence at least with respect to parallel composition; it is often defined in the form of a barbed equivalence or can be derived directly from the reduction semantics.''
\end{quotation}
Moreover, by the criteria in Section \ref{sec:qualityCriteria}, we know that $ \asymp $ should be success respecting (compare to Definition \ref{def:successRespecting}). The main purpose of $ \asymp $ in the definition of operational correspondence is to abstract from junk, i.e., remains left over by former \simulations that do not influence the abstract behaviour of a target term. Usually, two kinds of junks are distinguished inactive junk, i.e., remains that neither can perform further reductions on its own nor interact with the surrounding target term, and active junk, i.e., remains that may by reduced or even interact with the surrounding target term. Of course, proving an encoding to be good requires to prove that its active junk does no harm, i.e., does not influence the abstract behaviour of the target term. However, the presented encodings $ \encodingSepAsyn $ and $ \encodingMixAsyn $ induce the consideration of a second dimension of junk, namely observable and inobservable junk. In most cases developers of encodings make sure that all produced junk is inobservable, i.e., using the standard notions of observables for the target language neither the steps on junk nor the junk itself is observable. Unfortunately, as for the presented encodings, it is not always possible to define the encoding function such that all produced junk is inobservable.

In the \piCal-calculus observables are usually defined to be the unguarded input or output guards of a term, whose channel name is not restricted (compare to Definition \ref{def:barb}). To encode sums both encodings split up the summands into parallel. Of course, while doing so, the information which of these summands originally belongs to the same sum gets lost. To recover it, the encodings introduce sum locks, which cover a boolean value to indicate whether the respective summands of that sum can still be used to complete an \simulation or whether a former \simulation already consume one of the summands and thus no other can be used any more. Thus the encoding functions translate a source term observable into an observable|in case of $ \encodingSepAsyn $|or a request|in case of $ \encodingMixAsyn $|both times augment with the information covered by the sum lock. So source term observables are not translated into single observables again.

\begin{definition}[Translated Observables] \label{def:transBarb}
	Let $ T \in \targetTermsSepAsyn $. Then $ T $ has a \emph{translated input observables} $ y $, denoted by $ T\TransBarbSepAsyn{\In{y}} $, if
	\begin{align*}
		& \exists T', T_1', T_2' \in \piAsynProc \logdot \exists \tilde{x} \subset \names \logdot \exists \sumLock, \sumLock', \senderLock, z \in \names \logdot\\
		& \hspace*{1em} \left( T \equiv \RestrictedTerm{\tilde{x}}{\left( T' \mid \Input{y}{\sumLock', \senderLock, z}.\Test{\sumLock}{T_1'}{T_2'} \mid \Output{\sumLock}{\true} \right)} \vee T \equiv \RestrictedTerm{\tilde{x}}{\left( T' \mid \ReplicateInput{y}{\sumLock', \senderLock, z} \right)} \right) \wedge y \notin \tilde{x}
	\end{align*}
	and $ T $ has a \emph{translated output observables} $ y $, denoted by $ T\TransBarbSepAsyn{\Out{y}} $, if
	\begin{align*}
		\exists T' \in \piAsynProc \logdot \exists \tilde{x} \subset \names \logdot \exists \sumLock, \senderLock, z \in \names \logdot T \equiv \RestrictedTerm{\tilde{x}}{\left( T' \mid \Output{y}{\sumLock, \senderLock, z} \mid \Output{\sumLock}{\true} \right)} \wedge y \notin \tilde{x}.
	\end{align*}
	Let $ T \in \targetTermsMixAsyn $. Then $ T $ has a \emph{translated input observables} $ y $, denoted by $ T\TransBarbMixAsyn{\In{y}} $, if
	\begin{align*}
		\exists T' \in \piAsynProc \logdot \tilde{x} \subset \names \logdot \exists \parallelChannelIn, \sumLock, \receiverLock \in \names \logdot T \equiv \RestrictedTerm{\tilde{x}}{\left( T' \mid \Output{\parallelChannelIn}{y, \sumLock, \receiverLock} \mid \Output{\sumLock}{\true} \right)} \wedge y \notin \tilde{x}
	\end{align*}
	and $ T $ has a \emph{translated output observables} $ y $, denoted by $ T\TransBarbMixAsyn{\Out{y}} $, if
	\begin{align*}
		\exists T' \in \piAsynProc \logdot \tilde{x} \subset \names \logdot \exists \parallelChannelOut, \sumLock, \senderLock, z \in \names \logdot T \equiv \RestrictedTerm{\tilde{x}}{\left( T' \mid \Output{\parallelChannelOut}{y, \sumLock, \senderLock, z} \mid \Output{\sumLock}{\true} \right)} \wedge y \notin \tilde{x}.
	\end{align*}
	Moreover for some input or output observable $ \mu $ we define $ T\WeakTransBarbSepAsyn{\mu} \deff \exists T' \in \piAsyn \logdot T \steps T' \wedge T'\TransBarbSepAsyn{\mu} $ and $ T\WeakTransBarbMixAsyn{\mu} \deff \exists T' \in \piAsyn \logdot T \steps T' \wedge T'\TransBarbMixAsyn{\mu} $.
\end{definition}

Note that for all target terms $ T \in \targetTermsMixAsyn $ any output with three parameters is an input request and any output with four parameters is an output request. So a simple typing suffice to securely identify requests. Since requests already contain a reference to their related sum lock, the identification of the related sum lock instantiation is unambiguous as well. The condition $ y \notin \tilde{x} $ is necessary to rule out translated observables that corresponds to invisible in- or outputs of the source term (compare to Definition \ref{def:barb} at page \pageref{def:barb}). To show that the notion of translated observables indeed captures our intuition we prove that the set of observables reachable for a source term coincides with the set of translated observables reachable for its encoding.

\begin{lemma}
	The set of reachable observables of a source term and of reachable translated observables of its encoding coincide, i.e.,
	\begin{align*}
		\forall S \in \piSepProc \logdot \forall \mu \in \names \cup \coNames \logdot S\WeakBarb{\mu} \text{ iff } \; \EncodingSepAsyn{S}\WeakTransBarbSepAsyn{\mu}
	\end{align*}
	and
	\begin{align*}
		\forall S \in \piMixProc \logdot \forall \mu \in \names \cup \coNames \logdot S\WeakBarb{\mu} \text{ iff } \; \EncodingMixAsyn{S}\WeakTransBarbMixAsyn{\mu}.
	\end{align*}
\end{lemma}

\begin{proof}
	By Corollary \ref{col:initialSumLocksArePositive} stating that initially all sum locks are instantiated positive and Figure \ref{fig:encodingMixAsyn} the set of observables of a source term $ S \in \piMixProc $ and the set of translated observables of $ \EncodingMixAsyn{S} $ coincide, i.e.,
	\begin{align*}
		\forall S \in \piMixProc \logdot \forall \mu \in \names \cup \coNames \logdot S\Barb{\mu} \text{ iff } \; \EncodingMixAsyn{S}\TransBarbMixAsyn{\mu}.
	\end{align*}
	In case of $ \encodingSepAsyn $ we obtain a similar result after reducing all instantiation of receiver locks, since they guard the respective inputs (compare to Figure \ref{fig:encodingSepAsyn}). The lemma then follows by operational correspondence, i.e. by the Lemmata \ref{lem:operationalCompletenessSepAsyn} at page \pageref{lem:operationalCompletenessSepAsyn}, \ref{lem:operationalCompletenessMixAsyn} at page \pageref{lem:operationalCompletenessMixAsyn},  and \ref{lem:operationalSoundness} at page \pageref{lem:operationalSoundness}\footnote{Note that we present this fact just to visualize our intuition. We do not use it within another proof.}.
	\qed
\end{proof}

The problem now is, that the completion of an \simulation changes positive instantiations of sum locks into negative ones and so obviously influences the translated observables, but the corresponding requests|in case of $ \encodingMixAsyn $|or in- and outputs of other summands|in case of $ \encodingSepAsyn $|remain as observable junk. While active junk often aggravates the proof of correctness of an encoding, due to intricate proofs to show that it does no harm, observable junk turns out to be even worse for an encoding, because it prevents for the use of standard equivalences to describe the abstract behaviour of a target term.

Since the target language is the asynchronous \piCal-calculus, it seems natural to choose weak asynchronous bisimilarity $ \weakAsynBisim $ or asynchronous barbed congruence $ \asynBarbCong $. Unfortunately for both choices the presented encodings are not good. Consider for example the source term $ S = \RestrictedTerm{x}{\left( \In{x} + \In{y} \mid \Out{x} \right)} $. It can perform a reduction to $ \nullTerm $. But, all derivatives of its encoding, i.e., all $ T \in \piAsynProc $ with $ \EncodingSepAsyn{S} \steps T $ or $ \EncodingMixAsyn{S} \steps T $, are neither asynchronous bisimilar nor synchronous barbed congruent to the encoding of $ \nullTerm $, i.e., $ T \not\weakAsynBisim \RestrictedTerm{\sumLock}{\left( \Output{\sumLock}{\true} \right)} $ and $ T \not\asynBarbCong \RestrictedTerm{\sumLock}{\left( \Output{\sumLock}{\true} \right)} $, where $ \EncodingSepAsyn{\nullTerm} = \RestrictedTerm{\sumLock}{\left( \Output{\sumLock}{\true} \right)} = \EncodingMixAsyn{\nullTerm} $. Note that this is not due to the encoding of $ \nullTerm $, which is indeed weak asynchronous bisimilar to $ \nullTerm $ again, but to the observable junk, which suffice to distinguish the remains of \simulations from $ \nullTerm $. Because of this, a proof of the correctness of these encodings with respect to $ \weakAsynBisim $ or $ \asynBarbCong $ fails due to operational correspondence (compare to the Definition \ref{def:operationalCorrespondence}). Of course, you might argue that an encoding that can not get rid of observable junk is no good encoding. On the other side, Nestmann in \cite{nestmann00} gives some good reasons to accept $ \encodingSepAsyn $ as a good encoding. Moreover, the translation of observables into something different seems to be a quite natural manner of encoding functions. And indeed rephrasing a standard equivalence to take instead of observables translated observables into account suffice to turn it into an equivalence that describes the abstract behaviour of encoded terms. The same holds if we do not consider observables at all, but e.g. only reachability of success.

Note that to test a sum lock it has to be consumed first. Analysing the encoding function we observe that in each case an instantiation of sum lock is consumed, another instantiation of that lock is restored as soon as the respective test statement is completed. However, since there may lay many steps between the start and the completion of a test statement, instantiations of locks may temporally not be available. Because of that, we will use the notion of $ P\WeakTransBarbMixAsyn{\mu} $ instead of $ P\TransBarbMixAsyn{\mu} $ in the following.

\begin{definition}[Translated Barbed Bisimilarity] \label{def:transBarbBisim}
	Let $ P, Q \in \piAsynProc $. Then $ P $ and $ Q $ are \emph{translated barbed bisimilar} with respect to $ \encodingSepAsyn $, denoted by $ P \transBarbBisimSepAsyn Q $, if
	\begin{enumerate}
		\item $ P\reachSuccess $ iff $ Q\reachSuccess $,
		\item for all $ \mu \in \names \cup \coNames $, $ P\WeakTransBarbSepAsyn{\mu} $ iff $ Q\WeakTransBarbSepAsyn{\mu} $,
		\item for all $ P' \in \piAsynProc $, $ P \step P' $ implies $ Q \steps \transBarbBisimSepAsyn P' $, and
		\item for all $ Q' \in \piAsynProc $, $ Q \step Q' $ implies $ P \steps \transBarbBisimSepAsyn Q' $.
	\end{enumerate}
	And $ P $ and $ Q $ are \emph{translated barbed bisimilar} with respect to $ \encodingMixAsyn $, denoted by $ P \transBarbBisimMixAsynA Q $, if
	\begin{enumerate}
		\item $ P\reachSuccess $ iff $ Q\reachSuccess $,
		\item for all $ \mu \in \names \cup \coNames $, $ P\WeakTransBarbMixAsyn{\mu} $ iff $ Q\WeakTransBarbMixAsyn{\mu} $,
		\item for all $ P' \in \piAsynProc $, $ P \step P' $ implies $ Q \steps \transBarbBisimMixAsynA P' $, and
		\item for all $ Q' \in \piAsynProc $, $ Q \step Q' $ implies $ P \steps \transBarbBisimMixAsynA Q' $.
	\end{enumerate}
\end{definition}

Note that the first condition of each equivalence ensures that it is success respecting as required in Section \ref{sec:qualityCriteria} by Definition \ref{def:successRespecting}. Conditions 2. to 4. than define a version of weak barbed bisimilarity which utilises translated observables instead of standard barbs. Note that we consider the translation of input as well as output observables, although our target language is asynchronous. However, since in case of $ \encodingMixAsyn $ both kinds of source term actions are translated into requests and instantiations of sum locks, i.e., into outputs, the presented kind of barbed bisimilarity does consider barbs on outputs only. So it is an asynchronous variant of barbed bisimulation. Moreover note, that due to the definition above we do not consider any barbs except for translated barbs. However, analysing the encoding function $ \encodingMixAsyn $, we observe, that for all target terms all free in- or outputs are on the request channels $ \parallelChannelIn $ and $ \parallelChannelOut $. So, since we do only consider target terms, requests are indeed the only interesting barbs of $ \encodingMixAsyn $.

Alternatively, we could decide not to consider barbs at all, by omitting the second condition of $ \transBarbBisimSepAsyn $ and $ \transBarbBisimMixAsynA $. We result then in an equivalence that considers only reachability of $ \success $ as abstract behaviour of a term. Note that this intuition goes along very well with the criteria defined by Gorla as they also do only require a similar reachability of $ \success $, because reachability of success is defined independent of a specific source or target language. An advantage of such an equivalence is the fact, that it is independent of the considered encoding function. However, the resulting equivalence is obviously strictly weaker then $ \transBarbBisimSepAsyn $ and $ \transBarbBisimMixAsynA $. Moreover, $ \transBarbBisimSepAsyn $ and $ \transBarbBisimMixAsynA $ much better describe how the encoding function proceeds source terms and \simulate source term steps. So we will use these equivalences in the following.

In case of $ \encodingMixAsyn $ we are faced with an other problem concerning the choice of an appropriate equivalence, although that problem is by far not that crucial as observable junk. As explained above, the encodings of structural congruent source terms can differ in the number and nature of reachable intermediate states (compare to Example \ref{exa:intermediateStates}, Definition \ref{def:intermediateState}, and the following discussion above). Operational Soundness explicitly allows for intermediate states, i.e., target term states that due not map to the encodings of any of the corresponding source terms. However, if $ \asymp $ does distinguish target terms by reachability of intermediate states, we have a problem with the \textsc{Cong} rule of Figure \ref{fig:concurrentReductionSemantics} and operational completeness of Definition \ref{def:operationalCorrespondence}. Let us consider the source terms $ S = \left( \In{a}.S_1 + \In{a}.S_2 \right) \mid \Out{a}.S_3 \mid \In{a}.S_4  $ and $ S' = \Out{a}.S_3 \mid \left( \In{a}.S_1 + \In{a}.S_2 \right) \mid \In{a}.S_4 $ again. The source term $ \In{b}.S \mid \Out{b} $ can reduce to $ S $ but by the \textsc{Cong} rule it can reduce to $ S' $ as well. $ \encodingMixAsyn $ can \simulate the first step modulo $ \transBarbBisimMixAsynA $ but not the second step. Note that the \textsc{Cong} rule is used to shorten the presentation of the reduction semantics, but it is neither necessary nor was it the originally choice. So the most natural way to circumvent this problem is to rephrase the rules of the reduction semantics by avoiding the \textsc{Cong} rule and with it the possibility to arbitrary reorder the subprocesses during reductions. However we can also circumvent this problem by using an equivalence which does not distinguishes terms by the reachability of intermediate states.

\begin{definition}[Translated Barbed Bisimilarity] \label{def:transBarbBisimB}
	Let $ P, Q \in \piAsynProc $. Then $ P $ and $ Q $ are \emph{translated barbed bisimilar} with respect to $ \encodingMixAsyn $, denoted by $ P \transBarbBisimMixAsynB Q $, if
	\begin{enumerate}
		\item $ P\reachSuccess $ iff $ Q\reachSuccess $,
		\item for all $ \mu \in \names \cup \coNames $, $ P\WeakTransBarbMixAsyn{\mu} $ iff $ Q\WeakTransBarbMixAsyn{\mu} $,
		\item for all $ P' \in \piAsynProc $, $ P \steps P' $ implies that there exists some $ P'' \in \piAsynProc $ such that $ Q \steps \transBarbBisimMixAsynB P'' $ and $ P' \steps P'' $, and
		\item for all $ Q' \in \piAsynProc $, $ Q \steps Q' $ implies that there exists some $ Q'' \in \piAsynProc $ such that $ P \steps \transBarbBisimMixAsynB Q'' $ and $ Q' \steps Q'' $.
	\end{enumerate}
\end{definition}

Note that the second version of $ \transBarbBisimMixAsyn $ is strictly weaker then the first version and that we only use it to circumvent the problem with the \textsc{Cong} rule in operational completeness (and the therefore necessary Lemma \ref{lem:preservesSCModuloTransBarbBisimMixAsyn}). Because of this we prove the remaining results using the stricter equivalence; silently omitting the subscript $ 1 $.

Before we use these relations, we prove that they are indeed equivalences.

\begin{lemma} \label{lem:transBarbBisimIsEquivalence}
	All presented translated barbed bisimulations are equivalence relations.
\end{lemma}

\begin{proof}
	We have to show that $ \transBarbBisimSepAsyn $, $ \transBarbBisimMixAsynA $, and $ \transBarbBisimMixAsynB $ are reflexive, symmetric, and transitive. Reflexivity and transitivity follow directly by definition. For transitivity of $ \transBarbBisimMixAsynA $ assume $ P, Q, R \in \piAsynProc $ such that $ P \transBarbBisimMixAsynA Q $ and $ Q \transBarbBisimMixAsynA R $. By the first condition we have $ P\reachSuccess $ iff $ Q\reachSuccess $ iff $ R\reachSuccess $. And by the second condition for all $ \mu \in \names \cup \coNames $ we have $ P\WeakTransBarbMixAsyn{\mu} $ iff $ Q\WeakTransBarbMixAsyn{\mu} $ iff $ R\WeakTransBarbMixAsyn{\mu} $. So we can conclude that $ P\reachSuccess $ iff $ R\reachSuccess $ and $ P\WeakTransBarbMixAsyn{\mu} $ iff $ R \WeakTransBarbMixAsyn{\mu} $.
	
	By the third condition for all $ P' \in \piAsynProc $ with $ P \step P' $ there is some $ Q' \in \piAsynProc $ such that $ Q \steps Q' $ and $ P' \transBarbBisimMixAsynA Q' $. Without loss of generality let us assume, that the sequence $ Q \steps Q' $ is of length $ n $, i.e., there are $ Q_1, \ldots, Q_n \in \piAsynProc $ such that $ Q_0 \step Q_1 \step \ldots \step Q_n $, $ Q_0 = Q $, and $ Q_n = Q' $. Let $ R = R_0 $ and $ R' = R_n $. Then, by the third condition, for each step in $ Q \steps Q' $, i.e., for each $ Q_{i - 1} \step Q_i $ with $ 0 < i \leq n $, there is some $ R_i \in \piAsynProc $ such that $ R_{i - 1} \steps R_i $ with $ R_i \transBarbBisimMixAsynA Q_i $. So we conclude that for all $ P' \in \piAsynProc $ with $ P \step P' $ there is some $ R' \in \piAsynProc $ such that $ R \steps R' $ and $ P' \transBarbBisimMixAsynA R' $. The argumentation for the last condition is similar.
	
	The argumentation for $ \transBarbBisimSepAsyn $ and $ \transBarbBisimMixAsynB $ is similar.
	\qed
\end{proof}

The observable junk does not only rule out standard equivalences but also congruences with respect to contexts, that allow for interaction with observable junk. In both encodings such an interaction can for instance lead to a positive instantiation of a formerly negative instantiation of a sum lock and so turn observable junk into a translated observable, or it can instantiate a sender lock and so complete \simulations on junk.
\begin{example}
	Let us consider the target terms $ T_1 = \RestrictedTerm{\sumLock}{\left( \EncodingSepAsyn{\Out{y}.\success} \mid \Output{\sumLock}{\false} \right)} $ and $ T_2 = \RestrictedTerm{\sumLock}{\left( \EncodingMixAsyn{\Out{y}.\success} \mid \Output{\sumLock}{\false} \right)} $. By the Lemmata \ref{lem:junkRemainsOfSumsSepAsyn} and \ref{lem:junkRemainsOfSumsMixAsyn} in the next section, we prove that both terms are junk. They be produced as remains of \simulations (or a part of such a remain), e.g. for a source term $ \Out{x} + \Out{y}.\success \mid \In{x} $. Since neither $ T_1 $ nor $ T_2 $ reaches any translated observables or unguarded occurrence of $ \success $, we have $ T_1 \transBarbBisimSepAsyn \EncodingSepAsyn{\nullTerm} $ and $ T_2 \transBarbBisimMixAsyn \EncodingMixAsyn{\nullTerm} $. However, we can distinguish $ T_1 $ from $ \EncodingSepAsyn{\nullTerm} $ by the context $ \Context{}{1}{\hole} = \hole \mid \Input{\RenamingPolicySepAsyn{y}}{-, \senderLock, -}.\Out{\senderLock} $, because $ \Context{}{1}{T_1}\reachSuccess $ but $ \Context{}{1}{\nullTerm}\not\reachSuccess $. So $ \Context{}{1}{T_1} \notTransBarbBisimSepAsyn \Context{}{1}{\nullTerm} $. Accordingly, we can distinguish $ T_2 $ from $ \EncodingMixAsyn{\nullTerm} $ by the context $ \Context{}{2}{\hole} = \hole \mid \Input{\parallelChannelOut}{-, -, \senderLock, -}.\Out{\senderLock} $, where $ \parallelChannelOut $ is the free output request channel of $ \EncodingMixAsyn{\Out{y}.\success} $. Again we have $ \Context{}{2}{T_2}\reachSuccess $ but $ \Context{}{2}{\nullTerm}\not\reachSuccess $, i.e., $ \Context{}{2}{T_2} \notTransBarbBisimMixAsyn \Context{}{2}{\nullTerm} $.
\end{example}
Because of this, in order to prove operational completeness, we have to reduce the number of contexts we consider to obtain a congruence. Intuitively, we consider only contexts that respect the protocol of the encoding function. Thus, we consider only contexts that, if their argument is a target term as for instance the encoding of $ \nullTerm $, result in a target term.

\begin{definition}[Translated Barbed Congruence] \label{def:transBarbCong}
	Two terms $ T_1, T_2 \in \piAsynProc $ are \emph{translated barbed congruent} with respect to $ \encodingSepAsyn $, denoted as $ T_1 \transBarbCongSepAsyn T_2 $, if
	\begin{align*}
		\forall \Context{}{}{\hole} \in \piAsynProc \to \piAsynProc \logdot \Context{}{}{\EncodingSepAsyn{\nullTerm}} \in \targetTermsSepAsyn \text{ implies } \Context{}{}{T_1} \transBarbBisimSepAsyn \Context{}{}{T_2}.
	\end{align*}
	Two target terms $ T_1, T_2 \in \piAsynProc $ are \emph{translated barbed congruent} with respect to $ \encodingMixAsyn $, denoted as $ T_1 \transBarbCongMixAsynA T_2 $, if
	\begin{align*}
		\forall \Context{}{}{\hole} \in \piAsynProc \to \piAsynProc \logdot \Context{}{}{\EncodingMixAsyn{\nullTerm}} \in \targetTermsMixAsyn \text{ implies } \Context{}{}{T_1} \transBarbBisimMixAsynA \Context{}{}{T_2}.
	\end{align*}
	Two target terms $ T_1, T_2 \in \piAsynProc $ are \emph{translated barbed congruent} with respect to $ \encodingMixAsyn $, denoted as $ T_1 \transBarbCongMixAsynB T_2 $, if
	\begin{align*}
		\forall \Context{}{}{\hole} \in \piAsynProc \to \piAsynProc \logdot \Context{}{}{\EncodingMixAsyn{\nullTerm}} \in \targetTermsMixAsyn \text{ implies } \Context{}{}{T_1} \transBarbBisimMixAsynB \Context{}{}{T_2}.
	\end{align*}
\end{definition}
\noindent
Note that we again usually only consider the stricter first variant of the congruence $ \transBarbCongMixAsyn $, while silently omitting the subscript $ 1 $. Operational correspondence considers only target terms, so it would suffice to define the congruence over target terms only. However, in defining it over all terms of the target language we gain more flexibility. We will use these flexibility in the proof of operational completeness to stepwise reduce junk which in some cases leads to non target terms. Since these non target terms are behavioural equivalent to the considered target terms, they serve as connecting pieces to link the target terms modulo $ \transBarbCongSepAsyn $ or $ \transBarbCongMixAsyn $. Moreover note that the respective congruence relations are strictly weaker than their corresponding equivalences.
\begin{example}
	Let us consider the target terms $ T_1 = \EncodingSepAsyn{\Out{y}.\success} $, $ T_1' = \EncodingSepAsyn{\Out{y}.\nullTerm} $, $ T_2 = \EncodingMixAsyn{\Out{y}.\success} $, and $ T_2' = \EncodingMixAsyn{\Out{y}.\nullTerm} $. Obviously $ T_1 \transBarbBisimSepAsyn T_1' $ and $ T_2 \transBarbBisimMixAsyn T_2' $. But neither $ T_1 \transBarbCongSepAsyn T_1' $ nor $ T_2 \transBarbCongMixAsyn T_2' $, because in both cases the context has only to provide a translated input observable on $ \RenamingPolicySepAsyn{y} $ or $ \RenamingPolicyMixAsyn{y} $, respectively. So in case of $ \encodingSepAsyn $ the context $ \Context{}{}{\hole} = \hole \mid \EncodingSepAsyn{\In{x}} $ suffice to distinguish $ T_1 $ and $ T_1' $, because $ \Context{}{}{T_1}\reachSuccess $ but $ \Context{}{}{T_1'}\not\reachSuccess $. The argumentation for $ \encodingMixAsyn $ is similar, but due to the complex encoding of the parallel operator the respective distinguishing context is rather large. Because of that, intuitively, two equivalent target terms are congruent, only if the encoded continuations of their translated observables are again equivalent.
\end{example}

Of course, all presented congruences are again equivalences.

\begin{lemma} \label{lem:transBarbCongIsEquivalence}
	All presented translated barbed congruences are equivalence relations.
\end{lemma}

\begin{proof}
	Let $ T_1, T_2, T_3 \in \piAsynProc $. Then $ \Context{}{}{T_1} \transBarbBisimSepAsyn \Context{}{}{T_1} $ for all contexts $ \Context{}{}{\hole} \in \piAsynProc \to \piAsynProc $; so $ T_1 \transBarbCongSepAsyn T_1 $. Moreover, if $ T_1 \transBarbCongSepAsyn T_2 $, then $ \Context{}{}{T_1} \transBarbBisimSepAsyn \Context{}{}{T_2} $ for all contexts $ \Context{}{}{\hole} \in \piAsynProc \to \piAsynProc $ such that $ \Context{}{}{\EncodingSepAsyn{\nullTerm}} \in \targetTermsSepAsyn $. Since $ \transBarbBisimSepAsyn $ is an equivalence, then also $ \Context{}{}{T_2} \transBarbBisimSepAsyn \Context{}{}{T_1} $ for all such contexts, i.e., $ T_2 \transBarbCongSepAsyn T_1 $. Finally, if $ T_1 \transBarbCongSepAsyn T_2 $ and $ T_2 \transBarbCongSepAsyn T_3 $, then $ \Context{}{}{T_1} \transBarbBisimSepAsyn \Context{}{}{T_2} $ for all contexts $ \Context{}{}{\hole} \in \piAsynProc \to \piAsynProc $ such that $ \Context{}{}{\EncodingSepAsyn{\nullTerm}} \in \targetTermsSepAsyn $, and $ \Context{}{}{T_2} \transBarbBisimSepAsyn \Context{}{}{T_3} $ for all such contexts. So also $ \Context{}{}{T_1} \transBarbBisimSepAsyn \Context{}{}{T_3} $ for all such contexts, i.e., $ T_1 \transBarbCongSepAsyn T_3 $. We conclude that $ \transBarbCongSepAsyn $ is an equivalence.
	
	The argumentation for $ \transBarbCongMixAsynA $ and $ \transBarbCongMixAsynB $ is similar.
	\qed
\end{proof}

Moreover, the presented congruences include the structural congruence on the target language, because it is already included in the respective bisimulations.

\begin{lemma} \label{lem:transBarbBisimIncludesSC}
	Translated barbed bisimulation includes structural congruence, i.e.,
	\begin{align*}
		\forall T_1, T_2 \in \piAsynProc \logdot T_1 \equiv T_2 \text{ implies } T_1 \transBarbBisimSepAsyn T_2 \wedge T_1 \transBarbBisimMixAsyn T_2.
	\end{align*}
\end{lemma}

\begin{proof}
	Let us assume $ T_1 \equiv T_2 $. Then, by rule \textsc{Cong} in Figure \ref{fig:concurrentReductionSemantics}, $ T_1 $ and $ T_2 $ can perform exactly the same steps such that the successors are again structural congruent. Note that this holds even in case of $ \encodingMixAsyn $ and is obviously a stronger feature than the third and fourth condition of $ \transBarbBisimMixAsynB $ (compare to Definition \ref{def:transBarbBisimB}). Since, by Definition \ref{def:success}, reachability of success is defined modulo structural congruence, $ T_1 $ and $ T_2 $ have the same chance to reach success, i.e., $ T_1 \reachSuccess $ iff $ T_2 \reachSuccess $. Similarly, translated observables are defined modulo structural congruence for both encodings in Definition \ref{def:transBarb}. Note that we do not consider translated observables on restricted names, since the corresponding in- and outputs in the source terms are no observables as well. Because of that translated observables can not be changed by alpha conversion. So $ T_1 $ and $ T_2 $ have the same set of translated observables and the same chance to reach a translated observable, i.e., $ \left( T_1\TransBarbSepAsyn{\mu} \text{ iff } T_2\TransBarbSepAsyn{\mu} \right) $, $ \left( T_1\TransBarbMixAsyn{\mu} \text{ iff } T_2\TransBarbMixAsyn{\mu} \right) $, $ \left( T_1\WeakTransBarbSepAsyn{\mu} \text{ iff } T_2\WeakTransBarbSepAsyn{\mu} \right) $, and $ \left( T_1\WeakTransBarbMixAsyn{\mu} \text{ iff } T_2\WeakTransBarbMixAsyn{\mu} \right) $ for all $ \mu \in \names \cup \coNames $. So $ T_1 \transBarbBisimSepAsyn T_2 $ and $ T_1 \transBarbBisimMixAsyn T_2 $.
	\qed
\end{proof}

\begin{lemma} \label{lem:transBarbCongIncludesSC}
	Weak translated barbed congruence includes structural congruence, i.e.,
	\begin{align*}
		\forall T_1, T_2 \in \piAsynProc \logdot T_1 \equiv T_2 \text{ implies } T_1 \transBarbCongSepAsyn T_2 \wedge T_1 \transBarbCongMixAsyn T_2.
	\end{align*}
\end{lemma}

\begin{proof}
	By Definition \ref{def:transBarbCong}, $ \transBarbCongSepAsyn $ is the largest congruence on contexts restricted to target terms included in $ \transBarbBisimSepAsyn $, and $ \transBarbCongMixAsyn $ is the largest congruence on contexts restricted to target terms included in $ \transBarbBisimMixAsyn $. Note that Definition \ref{def:transBarbCong} restricts only the contexts but not the considered terms. Thus, since by Lemma \ref{lem:transBarbBisimIncludesSC} structural congruence $ \equiv $ is included in $ \transBarbBisimSepAsyn $ and $ \transBarbBisimMixAsyn $, it is included in $ \transBarbCongSepAsyn $ and $ \transBarbCongMixAsyn $.
	\qed
\end{proof}

Remember, that to our intuition \pure \admin steps are only pre- or postprocessing steps that do not influence which \simulations can be completed. To underpin that intuition, we prove that \pure \admin steps do not change the state of a target term modulo the considered equivalences and congruences.

\begin{lemma} \label{lem:pureAdminStepsTransBarbBisimSepAsyn}
	\Pure \admin steps do not influence the state of a target term modulo translated barbed bisimilarity or translated barbed congruence with respect to $ \encodingSepAsyn $, i.e.,
	\begin{align*}
		\forall T, T' \in \targetTermsSepAsyn \logdot T \pureAdminSteps T' \text{ implies } T \transBarbBisimSepAsyn T' \wedge T \transBarbCongSepAsyn T'.
	\end{align*}
\end{lemma}

\begin{proof}
	Translated barbed bisimilarity is some kind of weak bisimilarity that takes instead of observables the reachability of $ \success $ and the reachability of translated observables into account. Note that it is not possible to reduce $ \success $. So, in case $ T \steps T' $, the only way that leads to $ \neg\left( T\reachSuccess \text{ iff } T'\reachSuccess \right) $ is that in the sequence of steps from $ T $ to $ T' $ there is a step that rules out a former possible way to unguard some occurrence of $ \success $. Since by Definition \ref{def:transBarb} \pure \admin steps can not consume translated observables, the same holds for the consideration of translated observables. We have to show, that it not possible when using only \pure \admin steps to rule out a way to a translated observable or an unguarded occurrence of $ \success $.
	
	Obviously, in case none of the \pure \admin steps rules out any other sequence of steps, i.e. if none of the \admin steps is in conflict to any other sequence of step, this condition holds. Because of that, by Lemma \ref{lem:nonConflictingStepsSepAsyn}, $ T \pureAdminSteps T' $ implies $ \left( T\reachSuccess \text{ iff } T'\reachSuccess \right) $.
	
	For the same reason, and since \pure \admin steps do neither consumes positive instantiations of sum locks nor outputs on translated source terms, they do not influence the set of reachable translated observables, i.e., $ T \WeakTransBarbSepAsyn{\mu} $ iff $ T' \WeakTransBarbSepAsyn{\mu} $ for all $ \mu \in \names \cup \coNames $. Note, that such a step can restore a positive or negative instantiation of a sum lock by resolving a test on a negative instantiated sum lock or can unguard new requests and sum lock instantiations by a step on a sender lock, so \pure \admin steps influence the set of translated observables. But, since they do not rule out a run that leads to a translated observable, they do not influence the set of reachable translated observables. So $ T \transBarbBisimSepAsyn T' $.
	
	Since \pure \admin steps do not influence the state of arbitrary target terms modulo $ \transBarbBisimSepAsyn $ and since the congruence $ \transBarbCongSepAsyn $ does only consider target term contexts (compare to Definition \ref{def:transBarbCong}), \pure \admin steps do not influence the state of target terns modulo $ \transBarbCongSepAsyn $, i.e., $ T \transBarbCongSepAsyn T' $.
	\qed
\end{proof}

\begin{lemma} \label{lem:pureAdminStepsTransBarbBisimMixAsyn}
	\Pure \admin steps do not influence the state of a target term modulo translated barbed bisimilarity or translated barbed congruence with respect to $ \encodingMixAsyn $, i.e.,
	\begin{align*}
		\forall T, T' \in \targetTermsMixAsyn \logdot T \pureAdminSteps T' \text{ implies } T \transBarbBisimMixAsyn T' \wedge T \transBarbCongMixAsyn T'.
	\end{align*}
\end{lemma}

\begin{proof}
	Translated barbed bisimilarity is some kind of weak bisimilarity that takes instead of observables the reachability of $ \success $ and the reachability of translated observables into account. Note that it is not possible to reduce $ \success $. So, in case $ T \steps T' $, the only way that leads to $ \neg\left( T\reachSuccess \text{ iff } T'\reachSuccess \right) $ is that in the sequence of steps from $ T $ to $ T' $ there is a step that rules out a former possible way to unguard some occurrence of $ \success $. Since by Definition \ref{def:transBarb} \pure \admin steps can not consume translated observables, the same holds for the consideration of translated observables. We have to show, that it not possible when using only \pure \admin steps to rule out a way to a translated observable or an unguarded occurrence of $ \success $.
	
	Obviously, in case none of the \pure \admin steps rules out any other sequence of steps, i.e. if none of the \pure \admin steps is in conflict to any other sequence of steps, this condition holds. Fortunately, indeed most of the \pure \admin steps are not conflicting. By Lemma \ref{lem:nonConflictingStepsMixAsyn}, the condition $ T \transBarbBisimMixAsyn T' $ holds for all steps that are on a sender lock or do not unguard an instantiation of a chain lock carrying a request channel.
	
	Revisiting the argumentation in the proof of Lemma \ref{lem:nonConflictingStepsMixAsyn} we observe that the remaining steps either influence the order of requests in chains of right requests (compare to \processRightOutputRequests \ and \processRightInputRequests) or the order of encoded continuations in the chain build up by the encoding of a replicated input. By Lemma \ref{lem:pureAdminStepsNoDeadlock}, these steps do not introduce deadlock, moreover revisiting the argumentation of the proof of this lemma we observe, that their impact on the ordering within the chain is indeed their only impact on the behaviour of target terms. Since all encoded continuations of a replicated input are initially the same, their order does not matter for the reachability of $ \success $ or translated observables. The same holds for the order of requests, because regardless of their order eventually each combination is checked. Indeed, a different order may only lead to more or less necessary invisible steps on requests channels to combine a specific pair of requests. Because of that, even the \pure \admin steps that unguards an instantiation of a chain lock do not influence the state of the target term modulo $ \transBarbBisimMixAsyn $. So $ T \transBarbBisimMixAsyn T' $.
	
	Since \pure \admin steps do not influence the state of arbitrary target terms modulo $ \transBarbBisimMixAsyn $ and since the congruence $ \transBarbCongMixAsyn $ does only consider target term contexts (compare to Definition \ref{def:transBarbCong}), \pure \admin steps do not influence the state of target terns modulo $ \transBarbCongMixAsyn $, i.e., $ T \transBarbCongMixAsyn T' $.
	\qed
\end{proof}

Note that due to these two lemmata we can mostly ignore \pure \admin steps in the following proofs, since they are invisible modulo the considered equivalence and congruence relations. To handle the \textsc{Cong} rule in the proof of operational completeness we prove that both encodings preserve structural congruence of source terms modulo the presented equivalences and congruences.

\begin{lemma} \label{lem:preservesSCModuloTransBarbBisimSepAsyn}
	The encoding $ \encodingSepAsyn $ preserves structural congruence of source terms modulo translated barbed bisimilarity and translated barbed congruence, i.e.,
	\begin{align*}
		\forall S, S' \in \piSepProc \logdot S \equiv S' \text{ implies } \EncodingSepAsyn{S} \transBarbBisimSepAsyn \EncodingSepAsyn{S'} \wedge \EncodingSepAsyn{S} \transBarbCongSepAsyn \EncodingSepAsyn{S'}.
	\end{align*}
\end{lemma}

\begin{proof}
	The strict use of the renaming policy $ \renamingPolicySepAsyn $, i.e., the fact that source term names are translated into single names not used by the encoding function for special purposes, ensures the preservation of equality modulo alpha conversion. Since the parallel operator, the match operator, and restriction are translated \cleanly, the encoding $ \encodingSepAsyn $ preserves structural congruence of source terms for all rules except for $ P \mid \nullTerm \equiv P $, i.e., if $ S $ and $ S' $ are structural congruent without using the rule $ P \mid \nullTerm \equiv P $, then $ \EncodingSepAsyn{S} \equiv \EncodingSepAsyn{S'} $. By Lemma \ref{lem:transBarbBisimIncludesSC}, then $ \EncodingSepAsyn{S} \transBarbBisimSepAsyn \EncodingSepAsyn{S'} $ and, by Lemma \ref{lem:transBarbCongIncludesSC}, then $ \EncodingSepAsyn{S} \transBarbCongSepAsyn \EncodingSepAsyn{S'} $.
	
	The rule $ P \mid \nullTerm \equiv P $ is not preserved, because the empty sum $ \nullTerm $ is not translated \cleanly, so e.g. $ \nullTerm \mid \nullTerm \equiv \nullTerm $ but $ \EncodingSepAsyn{\nullTerm \mid \nullTerm} = \RestrictedTerm{\sumLock}{\Output{\sumLock}{\true}} \mid \RestrictedTerm{\sumLock}{\Output{\sumLock}{\true}} \not\equiv \RestrictedTerm{\sumLock}{\Output{\sumLock}{\true}} = \EncodingSepAsyn{\nullTerm} $. Note that, because of the renaming policy $ \renamingPolicySepAsyn $ and the \clean translation of restriction, the rule $ \RestrictedTerm{n}{\nullTerm} \equiv \nullTerm $ is preserved, i.e., since $ \RenamingPolicySepAsyn{n} \notin \FreeNames{\EncodingSepAsyn{\nullTerm}} $, we have $ \EncodingSepAsyn{\RestrictedTerm{n}{\nullTerm}} = \RestrictedTerm{\RenamingPolicySepAsyn{n}}{\RestrictedTerm{\sumLock}{\Output{\sumLock}{\true}}} \equiv \RestrictedTerm{\sumLock}{\Output{\sumLock}{\true}} = \EncodingSepAsyn{\nullTerm} $. However, since $ \nullTerm $ is translated into a closed term that can not perform any step, its encoding behaves as $ \nullTerm $. In particular $ \EncodingSepAsyn{\nullTerm} $ can not interact with any context and does not reach success or any translated observable on its own. So, even in this case, we have $ \EncodingSepAsyn{S} \transBarbBisimSepAsyn \EncodingSepAsyn{S'} $ and $ \EncodingSepAsyn{S} \transBarbCongSepAsyn \EncodingSepAsyn{S'} $.
	\qed
\end{proof}

Since $ \encodingMixAsyn $ does not translate the parallel operator \cleanly, it does not directly preserve structural congruence of source terms. But, since the encoding preserves the abstract behaviour of source terms, the encodings of structural congruent source terms are similar modulo equivalences measuring only these abstract behaviour. As already explained, to prove the following statement, the equivalence must not distinguish terms by their reachable intermediate states.

\begin{lemma} \label{lem:preservesSCModuloTransBarbBisimMixAsyn}
	The encoding $ \encodingMixAsyn $ preserves structural congruence of source terms modulo translated barbed bisimilarity and translated barbed congruence, i.e.,
	\begin{align*}
		\forall S, S' \in \piMixProc \logdot S \equiv S' \text{ implies } \EncodingMixAsyn{S} \transBarbBisimMixAsynB \EncodingMixAsyn{S'} \wedge \EncodingMixAsyn{S} \transBarbCongMixAsynB \EncodingMixAsyn{S'}.
	\end{align*}
\end{lemma}

\begin{proof}
	Again, the strict use of the renaming policy $ \renamingPolicyMixAsyn $, i.e., the fact that source term names are translated into single names not used by the encoding function for special purposes, ensures the preservation of equality modulo alpha conversion. So $ S \equivAlpha S' $ implies $ \EncodingMixAsyn{S} \equivAlpha \EncodingMixAsyn{S'} $. Also, the \clean translation of the match operator and restriction ensures the preservation of structural congruence modulo the rules $ \Match{a}{a}P \equiv P $, $ \RestrictedTerm{n}{\nullTerm} = \nullTerm $, $ \RestrictedTerm{n}{\RestrictedTerm{m}{P}} \equiv \RestrictedTerm{m}{\RestrictedTerm{n}{P}} $, and $ P \mid \RestrictedTerm{n}{Q} \equiv \RestrictedTerm{n}{\left( P \mid Q \right)} $ if $ n \notin \FreeNames{P} $. So, if $ S $ and $ S' $ do only differ due to one or more of these four rules, then $ \EncodingMixAsyn{S} \equiv \EncodingMixAsyn{S'} $. By Lemma \ref{lem:transBarbBisimIncludesSC}, we conclude $ \EncodingMixAsyn{S} \transBarbBisimMixAsynB \EncodingMixAsyn{S'} $ and, by Lemma \ref{lem:transBarbCongIncludesSC}, we conclude $ \EncodingMixAsyn{S} \transBarbCongMixAsynB \EncodingMixAsyn{S'} $ for both of the above cases.
	
	With the preservation of these rules in mind we show the lemma by an induction over the proof tree of $ S \equiv S' $, i.e., over the number of structural congruence rules which are applied to show $ S \equiv S' $.
	\begin{description}
		\item[Base Case:] If $ S = S' $, then $ \EncodingMixAsyn{S} = \EncodingMixAsyn{S'} $. So, by reflexivity, $ \EncodingMixAsyn{S} \transBarbBisimMixAsynB \EncodingMixAsyn{S'} $ and $ \EncodingMixAsyn{S} \transBarbCongMixAsynB \EncodingMixAsyn{S'} $.
		\item[Induction Hypothesis:] If $ S $ and $ S' $ can be proved to be structural congruent within $ n $ applications of structural congruence rules, then $ \EncodingMixAsyn{S} \transBarbBisimMixAsynB \EncodingMixAsyn{S'} $ and $ \EncodingMixAsyn{S} \transBarbCongMixAsynB \EncodingMixAsyn{S'} $.
		\item[Induction Step:] $ S $ and $ S' $ can be proved to be structural congruent within $ n + 1 $ applications of structural congruence rules. Let $ S'' \in \piMixProc $ be such that $ S $ and $ S'' $ can be proved to be structural congruent within $ n $ applications of structural congruence rules and $ S'' $ and $ S' $ can be proved to be structural congruent directly by one application of a structural congruence rule, i.e., $ S \equiv S'' \equiv S' $. By the induction hypothesis, $ \EncodingMixAsyn{S} \transBarbBisimMixAsynB \EncodingMixAsyn{S''} $ and $ \EncodingMixAsyn{S} \transBarbCongMixAsynB \EncodingMixAsyn{S''} $. We proceed with a case split over the rule necessary to prove $ S'' \equiv S' $.
		\begin{description}
			\item[Case of Rule $ P \mid \nullTerm \equiv P $:] In this case $ S'' = P \mid \nullTerm $ and $ S' = P $ for some $ P \in \piMixProc $. By Figure \ref{fig:encodingMixAsyn},
				\begin{align*}
					\EncodingMixAsyn{S''} & = \RestrictedTerm{\matchingCoordinatorOut, \matchingCoordinatorIn, \coordinatorUpOut, \coordinatorUpIn, \coordinatorMatchingOut, \coordinatorMatchingIn, \matchingUpOut, \matchingUpIn}{\big(\\
						& \hspace*{3em} \begin{aligned}[t]
								& \RestrictedTerm{\parallelChannelOut, \parallelChannelIn}{\left( \EncodingMixAsyn{P} \mid \processLeftOutputRequests \mid \processLeftInputRequests \right)}\\
								& \mid \RestrictedTerm{\parallelChannelOut, \parallelChannelIn}{\left( \RestrictedTerm{\sumLock}{\Output{\sumLock}{\true}} \mid \processRightOutputRequests \mid \processRightInputRequests \right)}\\
								& \mid \pushRequests \big)
							\end{aligned}}
				\end{align*}
				and $ \EncodingMixAsyn{S'} = \EncodingMixAsyn{P} $. Obviously $ \EncodingMixAsyn{S''} $ and $ \EncodingMixAsyn{S'} $ are not structural congruent. However, $ \EncodingMixAsyn{P} $ appears unguarded within $ \EncodingMixAsyn{S''} $, so if $ \EncodingMixAsyn{S'} $ reaches $ \success $ or a translated observable then so does $ \EncodingMixAsyn{S''} $. Moreover we observe, that, since the encoding of $ \nullTerm $ does not emit any requests, the hole right branch of $ \EncodingMixAsyn{S''} $
				\begin{align*}
					\RestrictedTerm{\parallelChannelOut, \parallelChannelIn}{\left( \RestrictedTerm{\sumLock}{\Output{\sumLock}{\true}} \mid \processRightOutputRequests \mid \processRightInputRequests \right)}
				\end{align*}
				can do nothing but two steps on chain locks. Because of that requests of $ \EncodingMixAsyn{P} $ are prepared to be transmitted to the right side by $ \processLeftOutputRequests $ and $ \processLeftInputRequests $ but they are never received at the right side. What remains is the upward pushing of all requests of $ \EncodingMixAsyn{P} $ by the interplay of $ \processLeftOutputRequests $, $ \processLeftInputRequests $, and $ \pushRequests $. Because of that, for all target term contexts $ \EncodingMixAsyn{P \mid \nullTerm} $ behaves as $ \EncodingMixAsyn{P} $, i.e., $ \EncodingMixAsyn{S''} \transBarbBisimMixAsynB \EncodingMixAsyn{S'} $ and $ \EncodingMixAsyn{S''} \transBarbCongMixAsynB \EncodingMixAsyn{S'} $. Since $ \transBarbBisimMixAsynB $ and $ \transBarbCongMixAsynB $ are equivalences (compare to Lemmata \ref{lem:transBarbBisimIsEquivalence} and \ref{lem:transBarbCongIsEquivalence}), by transitivity, we conclude $ \EncodingMixAsyn{S} \transBarbBisimMixAsynB \EncodingMixAsyn{S'} $ and $ \EncodingMixAsyn{S} \transBarbCongMixAsynB \EncodingMixAsyn{S'} $.
			\item[Case of Rule $ P \mid Q \equiv Q \mid P $:] In this case $ S'' = P \mid Q $ and $ S' = Q \mid P $ for some $ P, Q \in \piMixProc $. Their encodings are given by:
				\begin{align*}
					\EncodingMixAsyn{S''} & = \RestrictedTerm{\matchingCoordinatorOut, \matchingCoordinatorIn, \coordinatorUpOut, \coordinatorUpIn, \coordinatorMatchingOut, \coordinatorMatchingIn, \matchingUpOut, \matchingUpIn}{\big(\\
						& \hspace*{3em} \begin{aligned}[t]
								& \RestrictedTerm{\parallelChannelOut, \parallelChannelIn}{\left( \EncodingMixAsyn{P} \mid \processLeftOutputRequests \mid \processLeftInputRequests \right)}\\
								& \mid \RestrictedTerm{\parallelChannelOut, \parallelChannelIn}{\left( \EncodingMixAsyn{Q} \mid \processRightOutputRequests \mid \processRightInputRequests \right)}\\
								& \mid \pushRequests \big)
							\end{aligned}}\\
					\EncodingMixAsyn{S'} & = \RestrictedTerm{\matchingCoordinatorOut, \matchingCoordinatorIn, \coordinatorUpOut, \coordinatorUpIn, \coordinatorMatchingOut, \coordinatorMatchingIn, \matchingUpOut, \matchingUpIn}{\big(\\
						& \hspace*{3em} \begin{aligned}[t]
								& \RestrictedTerm{\parallelChannelOut, \parallelChannelIn}{\left( \EncodingMixAsyn{Q} \mid \processLeftOutputRequests \mid \processLeftInputRequests \right)}\\
								& \mid \RestrictedTerm{\parallelChannelOut, \parallelChannelIn}{\left( \EncodingMixAsyn{P} \mid \processRightOutputRequests \mid \processRightInputRequests \right)}\\
								& \mid \pushRequests \big)
							\end{aligned}}
				\end{align*}
				Since all combinations of left and right requests are checked, $ \EncodingMixAsyn{S''} $ can \simulate the same source term steps as $ \EncodingMixAsyn{S'} $. However, since $ \EncodingMixAsyn{P} $ and $ \EncodingMixAsyn{Q} $ are exchanged at the outermost parallel operator the roles of left and right requests are exchanged. As a consequence, if a combination of requests from $ \EncodingMixAsyn{P} $ and $ \EncodingMixAsyn{Q} $ leads to a test on the respective sum locks, the order in which these locks are tested is different in $ \EncodingMixAsyn{S''} $ and $ \EncodingMixAsyn{S'} $. So $ \EncodingMixAsyn{S''} $ and $ \EncodingMixAsyn{S'} $ differ in their total ordering of sum locks. The ordering in $ \EncodingMixAsyn{S''} $ is based on the structure induced by the nesting of parallel operators in $ P \mid Q $; while the ordering in $ \EncodingMixAsyn{S'} $ is based on the structure induced by the parallel operator nesting in $ Q \mid P $. Note that, since in both cases this structure is a binary tree, by Lemma \ref{lem:impureAdminStepsNoDeadlock}, the encoding does not introduce deadlock. But as explained in Example \ref{exa:intermediateStates} the different orderings may lead to different reachable intermediate states. Apart from intermediate states $ \EncodingMixAsyn{S''} $ and $ \EncodingMixAsyn{S'} $ are similar, i.e., they have the same chance to reach success or translated observables.
				
				By Definition \ref{def:transBarbBisimB}, $ \transBarbBisimMixAsynB $ explicitly allows for different reachable intermediate states. Because of that, $ \EncodingMixAsyn{S''} \transBarbBisimMixAsynB \EncodingMixAsyn{S'} $. Analysing the encoding function in Figure \ref{fig:encodingMixAsyn} we observe that any encoded source term has at most two free names that are used as channels|remember that translated source term names are never used as channels within $ \encodingMixAsyn $. Because of that, when placed within a target term context, $ \EncodingMixAsyn{S''} $ and $ \EncodingMixAsyn{S'} $ can start an interaction with the context only by transmitting their requests. Because we consider only target term contexts, i.e., contexts $ \Context{}{}{\hole} \in \piAsynProc \to \piAsynProc $ such that $ \Context{}{}{\EncodingMixAsyn{\nullTerm}} \in \targetTermsMixAsyn $, the context respects the protocol implemented by the encoding function. So, if $ \EncodingMixAsyn{S''} $ provides a translated observable and the context has the matching translated observable, then the context can interact with $ \EncodingMixAsyn{S''} $ to \simulate a source term step. Doing so, an encoded continuation, i.e., an encoded source term, is unguarded within the continuation of $ \EncodingMixAsyn{S''} $. Since $ S'' \equiv S' $, the same context can \simulate the same source term step when interacting with $ \EncodingMixAsyn{S'} $. Moreover, doing so, again an encoded source term is unguarded within the continuation of $ \EncodingMixAsyn{S'} $ and the respective source terms of these continuations in case of $ \EncodingMixAsyn{S''} $ and $ \EncodingMixAsyn{S'} $ are again structural congruent. Because of this we can prove the preservation of structural congruence of source terms is also preserved modulo $ \transBarbCongMixAsynB $ by assuming an arbitrary context and perform an induction over the number of \simulations resulting from an interaction of the context with $ \EncodingMixAsyn{S''} $. So we conclude $ \EncodingMixAsyn{S''} \transBarbCongMixAsynB \EncodingMixAsyn{S'} $.
				
				Since $ \transBarbBisimMixAsynB $ and $ \transBarbCongMixAsynB $ are equivalences (compare to Lemmata \ref{lem:transBarbBisimIsEquivalence} and \ref{lem:transBarbCongIsEquivalence}), by transitivity, we conclude $ \EncodingMixAsyn{S} \transBarbBisimMixAsynB \EncodingMixAsyn{S'} $ and $ \EncodingMixAsyn{S} \transBarbCongMixAsynB \EncodingMixAsyn{S'} $.
			\item[Case of Rule $ P \mid \left( Q \mid R \right) \equiv \left( P \mid Q \right) \mid R $:] In this $ S'' = P \left( Q \mid R \right) $ and $ S' = \left( P \mid Q \right) \mid R $ for some $ P, Q, R \in \piMixProc $. Their encodings are given by
				\begin{align*}
					\EncodingMixAsyn{S''} & = \RestrictedTerm{\matchingCoordinatorOut, \matchingCoordinatorIn, \coordinatorUpOut, \coordinatorUpIn, \coordinatorMatchingOut, \coordinatorMatchingIn, \matchingUpOut, \matchingUpIn}{\big(\\
						& \hspace*{2em} \begin{aligned}[t]
								& \RestrictedTerm{\parallelChannelOut, \parallelChannelIn}{\left( \EncodingMixAsyn{P} \mid \processLeftOutputRequests \mid \processLeftInputRequests \right)}\\
								& \mid \RestrictedTerm{\parallelChannelOut, \parallelChannelIn}{\big( \begin{aligned}[t]
										& \RestrictedTerm{\matchingCoordinatorOut, \matchingCoordinatorIn, \coordinatorUpOut, \coordinatorUpIn, \coordinatorMatchingOut, \coordinatorMatchingIn, \matchingUpOut, \matchingUpIn}{\big(\\
						& \hspace*{1em} \begin{aligned}[t]
								& \RestrictedTerm{\parallelChannelOut, \parallelChannelIn}{\left( \EncodingMixAsyn{Q} \mid \processLeftOutputRequests \mid \processLeftInputRequests \right)}\\
								& \mid \RestrictedTerm{\parallelChannelOut, \parallelChannelIn}{\left( \EncodingMixAsyn{R} \mid \processRightOutputRequests \mid \processRightInputRequests \right)}\\
								& \mid \pushRequests \big)
							\end{aligned}}\\
										& \mid \processRightOutputRequests \mid \processRightInputRequests \big)
									\end{aligned}}\\
								& \mid \pushRequests \big)
							\end{aligned}}
				\end{align*}
				and
				\begin{align*}
					\EncodingMixAsyn{S'} & = \RestrictedTerm{\matchingCoordinatorOut, \matchingCoordinatorIn, \coordinatorUpOut, \coordinatorUpIn, \coordinatorMatchingOut, \coordinatorMatchingIn, \matchingUpOut, \matchingUpIn}{\big(\\
						& \hspace*{2em} \begin{aligned}[t]
								& \RestrictedTerm{\parallelChannelOut, \parallelChannelIn}{\big( \begin{aligned}[t]
										& \RestrictedTerm{\matchingCoordinatorOut, \matchingCoordinatorIn, \coordinatorUpOut, \coordinatorUpIn, \coordinatorMatchingOut, \coordinatorMatchingIn, \matchingUpOut, \matchingUpIn}{\big(\\
						& \hspace*{1em} \begin{aligned}[t]
								& \RestrictedTerm{\parallelChannelOut, \parallelChannelIn}{\left( \EncodingMixAsyn{P} \mid \processLeftOutputRequests \mid \processLeftInputRequests \right)}\\
								& \mid \RestrictedTerm{\parallelChannelOut, \parallelChannelIn}{\left( \EncodingMixAsyn{Q} \mid \processRightOutputRequests \mid \processRightInputRequests \right)}\\
								& \mid \pushRequests \big)
							\end{aligned}}\\
										& \mid \processLeftOutputRequests \mid \processLeftInputRequests \big)
									\end{aligned}}\\
								& \mid \RestrictedTerm{\parallelChannelOut, \parallelChannelIn}{\left( \EncodingMixAsyn{R} \mid \processRightOutputRequests \mid \processRightInputRequests \right)}\\
								& \mid \pushRequests \big).
							\end{aligned}}
				\end{align*}
				In $ \EncodingMixAsyn{S''} $ the encoding of $ Q $ appears left and the encoding of $ R $ appears right within the encoding of a parallel operator. Together they form the right branch of a surrounding encoding of a parallel operator, there the left branch is filled with $ \EncodingMixAsyn{P} $. In opposite in $ \EncodingMixAsyn{S'} $ the terms $ \EncodingMixAsyn{P} $ and $ \EncodingMixAsyn{Q} $ are left and right of a parallel operator encoding which is the left branch of a surrounding parallel operator encoding, where $ \EncodingMixAsyn{R} $ appears right. However, since all requests are pushed upwards to each surrounding parallel operator encoding, again all combinations of requests among the three encoded subterms $ \EncodingMixAsyn{P} $, $ \EncodingMixAsyn{Q} $, and $ \EncodingMixAsyn{R} $ are checked in $ \EncodingMixAsyn{S''} $ as well as in $ \EncodingMixAsyn{S'} $. Moreover, we observe that in both encodings $ \EncodingMixAsyn{S''} $ and $ \EncodingMixAsyn{S'} $ the encoding of $ P $ is always left to the encodings of $ Q $ and $ R $, and the encoding of $ Q $ is always left to the encoding of $ R $. So in this case $ \EncodingMixAsyn{S''} $ and $ \EncodingMixAsyn{S'} $ do not differ by the underlying total ordering of sum locks, i.e., they reach the same intermediate states. So the behaviour of $ \EncodingMixAsyn{S''} $ and $ \EncodingMixAsyn{S'} $ does only differ by \pure \admin steps on requests but they have the same chance to reach $ \success $ and translated observables, i.e., $ \EncodingMixAsyn{S''} \transBarbBisimMixAsynB \EncodingMixAsyn{S'} $. Revisiting the argumentation of the case before we also get $ \EncodingMixAsyn{S''} \transBarbCongMixAsynB \EncodingMixAsyn{S'} $.
				
				Since $ \transBarbBisimMixAsynB $ and $ \transBarbCongMixAsynB $ are equivalences (compare to Lemmata \ref{lem:transBarbBisimIsEquivalence} and \ref{lem:transBarbCongIsEquivalence}), by transitivity, we conclude $ \EncodingMixAsyn{S} \transBarbBisimMixAsynB \EncodingMixAsyn{S'} $ and $ \EncodingMixAsyn{S} \transBarbCongMixAsynB \EncodingMixAsyn{S'} $.
			\item[Else:] For the reaming rules we can apply the above argumentation to show that $ \EncodingSepAsyn{S''} \equiv \EncodingSepAsyn{S'} $. By the Lemmata \ref{lem:transBarbBisimIncludesSC} and \ref{lem:transBarbCongIncludesSC}, we have $ \EncodingMixAsyn{S''} \transBarbBisimMixAsynB \EncodingMixAsyn{S'} $ and $ \EncodingMixAsyn{S''} \transBarbCongMixAsynB \EncodingMixAsyn{S'} $. By Lemma \ref{lem:transBarbBisimIsEquivalence} and Lemma \ref{lem:transBarbCongIsEquivalence}, $ \transBarbBisimMixAsynB $ and $ \transBarbCongMixAsynB $ are equivalences. Thus, by transitivity, we conclude $ \EncodingMixAsyn{S} \transBarbBisimMixAsynB \EncodingMixAsyn{S'} $ and $ \EncodingMixAsyn{S} \transBarbCongMixAsynB \EncodingMixAsyn{S'} $.
		\end{description}
	\end{description}
	\qed
\end{proof}

These two lemmata finally prove that the intermediate states in combination with the application of the \textsc{Cong} rule on source terms do not falsify the criterion on operational completeness modulo $ \transBarbBisimMixAsynB $.

\subsection{Junk}

We consider remains of \simulations that behave modulo $ \transBarbBisimSepAsyn $ and $ \transBarbBisimMixAsyn $ like $ \nullTerm $ and do not influence the possibility or inability to \simulate further source term steps as junk. The \simulation of source term steps may leave different kinds of junk. So, e.g. in order to show operational completeness, we have to prove that junk does no harm.

Of course we are only interested in kinds of junk that appear in target terms, i.e., that are pieces of target terms. However, to ease the argumentation on the proof of operational completeness we want to allow to stepwise reduce junk. Unfortunately, as soon as we reduce a target term by the first piece of junk it is often no target term any more. So, in order to allow for a stepwise reduction of junk, we give a recursive definition of what it means to be a piece of a target term.

\begin{definition}[Piece of a Target Term] \label{def:pieceTargetTerm}
	A term $ T \in \piAsynProc $ is a \emph{piece of a target term} of $ \encodingSepAsyn $, denoted by $ T \in \pieceTargetTermsSepAsyn $, if
	\begin{align*}
		T \in \targetTermsSepAsyn \vee \left( \exists T', J \in \piAsynProc \logdot \exists \tilde{x} \subset \names \logdot T \equiv \RestrictedTerm{\tilde{x}}{T'} \wedge T \transBarbCongSepAsyn \RestrictedTerm{\tilde{x}}{\left( T' \mid J \right)} \wedge \RestrictedTerm{\tilde{x}}{\left( T' \mid J \right)} \in \pieceTargetTermsSepAsyn \right).
	\end{align*}
	Accordingly, a term $ T \in \piAsynProc $ is a \emph{piece of a target term} of $ \encodingMixAsyn $, denoted by $ T \in \pieceTargetTermsMixAsyn $, if
	\begin{align*}
		T \in \targetTermsMixAsyn \vee \left( \exists T', J \in \piAsynProc \logdot \exists \tilde{x} \subset \names \logdot T \equiv \RestrictedTerm{\tilde{x}}{T'} \wedge T \transBarbCongMixAsyn \RestrictedTerm{\tilde{x}}{\left( T' \mid J \right)} \wedge \RestrictedTerm{\tilde{x}}{\left( T' \mid J \right)} \in \pieceTargetTermsMixAsyn \right).
	\end{align*}
\end{definition}
\noindent
Intuitively, the definition above allows for a piece of a target term to recover the corresponding target term by stepwise restoring the reduced junk. Moreover note, that, although the relations $ \transBarbCongSepAsyn $ and $ \transBarbCongMixAsyn $ are not sensitive to divergence, they are sensitive to deadlock. That is why we require $ T \transBarbCongSepAsyn \RestrictedTerm{\tilde{x}}{\left( T' \mid J \right)} $ or $ T \transBarbCongMixAsyn \RestrictedTerm{\tilde{x}}{\left( T' \mid J \right)} $ to ensure that indeed only junk is removed and so, no deadlock is introduced.

\begin{definition}[Junk] \label{def:junk}
	A term $ J \in \piAsynProc $ is called \emph{junk} of the encoding $ \encodingSepAsyn $ modulo $ \transBarbBisimSepAsyn $, if $ J $ behaves modulo $ \transBarbBisimSepAsyn $ similar to $ \nullTerm $ for all pieces of target terms, i.e.,
	\begin{align*}
		\forall \Context{}{}{\hole} \in \piAsynProc \to \piAsynProc \logdot \Context{}{}{J} \in \pieceTargetTermsSepAsyn \text{ implies } \Context{}{}{J} \transBarbBisimSepAsyn \Context{}{}{\nullTerm}.
	\end{align*}
	A term $ J \in \piAsynProc $ is called \emph{junk} of the encoding $ \encodingMixAsyn $ modulo $ \transBarbBisimMixAsyn $, if $ J $ behaves modulo $ \transBarbBisimMixAsyn $ similar to $ \nullTerm $ for all pieces of target terms, i.e.,
	\begin{align*}
		\forall \Context{}{}{\hole} \in \piAsynProc \to \piAsynProc \logdot \Context{}{}{J} \in \pieceTargetTermsMixAsyn \text{ implies } \Context{}{}{J} \transBarbBisimMixAsyn \Context{}{}{\nullTerm}.
	\end{align*}
\end{definition}
\noindent
Since we do not consider junk modulo equivalences different from $ \transBarbBisimSepAsyn $ and $ \transBarbBisimMixAsyn $, we omit the equivalence in the following. Moreover we omit the encoding function if the considered junk appears within both encodings.

Of course, whenever we reduce a piece of a target term by removing junk, the result is again a piece of a target term. Moreover, reducing junk does not change the behaviour of such a term modulo $ \transBarbCongSepAsyn $ or $ \transBarbCongMixAsyn $.

\begin{lemma} \label{lem:removeJunk}
	Let $ T $ be a piece of a target term including some junk $ J \in \piAsynProc $. Then removing this junk results a piece of a target term which is congruent to $ T $, i.e.,
	\begin{align*}
		\forall T \in \pieceTargetTermsSepAsyn \logdot \forall T' \in \piAsynProc \logdot \forall \tilde{x} \subset \names \logdot T \equiv \RestrictedTerm{\tilde{x}}{\left( T' \mid J \right)} \text{ implies } \RestrictedTerm{\tilde{x}}{T'} \in \pieceTargetTermsSepAsyn \wedge T \transBarbCongSepAsyn \RestrictedTerm{\tilde{x}}{T'}
	\end{align*}
	and
	\begin{align*}
		\forall T \in \pieceTargetTermsMixAsyn \logdot \forall T' \in \piAsynProc \logdot \forall \tilde{x} \subset \names \logdot T \equiv \RestrictedTerm{\tilde{x}}{\left( T' \mid J \right)} \text{ implies } \RestrictedTerm{\tilde{x}}{T'} \in \pieceTargetTermsMixAsyn \wedge T \transBarbCongMixAsyn \RestrictedTerm{\tilde{x}}{T'}.
	\end{align*}
\end{lemma}

\begin{proof}
	Let $ J \in \piAsynProc $ be junk. And let $ T \in \pieceTargetTermsSepAsyn $, $ T' \in \piAsynProc $, and $ \tilde{x} \subset \names $ such that $ T \equiv \RestrictedTerm{\tilde{x}}{\left( T' \mid J \right)} $. We show the lemma for $ \encodingSepAsyn $. The argumentation for $ \encodingMixAsyn $ is then similar.
	
	Since, by Lemma \ref{lem:transBarbCongIncludesSC}, $ \transBarbCongSepAsyn $ includes structural congruence, $ T \equiv \RestrictedTerm{\tilde{x}}{\left( T' \mid J \right)} $ implies $ T \transBarbCongSepAsyn \RestrictedTerm{\tilde{x}}{\left( T' \mid J \right)} $. By Definition \ref{def:transBarbCong}, then $ \Context{}{}{T} \transBarbBisimSepAsyn \Context{}{}{\RestrictedTerm{\tilde{x}}{\left( T' \mid J \right)}} $ for all contexts $ \Context{}{}{\hole} \in \piAsynProc \to \piAsynProc $ such that $ \Context{}{}{\EncodingSepAsyn{\nullTerm}} \in \targetTermsSepAsyn $. Let $ \Context{'}{}{\hole} = \RestrictedTerm{\tilde{x}}{\left( T' \mid \hole \right)} $. Then $ \context' $ is a context, i.e., $ \Context{'}{}{\hole} \in \piAsynProc \to \piAsynProc $, and, since $ T \in \pieceTargetTermsMixAsyn $ and $ T = \Context{'}{}{J} $, we have $ \Context{'}{}{J} \in \pieceTargetTermsMixAsyn $. Thus, $ \Context{}{}{T} \transBarbBisimSepAsyn \Context{}{}{\Context{'}{}{J}} $ for all contexts $ \Context{}{}{\hole} \in \piAsynProc \to \piAsynProc $ such that $ \Context{}{}{\EncodingSepAsyn{\nullTerm}} \in \targetTermsSepAsyn $. Moreover, $ \Context{'}{}{J} \in \pieceTargetTermsMixAsyn $ implies $ \Context{}{}{\Context{'}{}{J}} \in \pieceTargetTermsMixAsyn $ for all such contexts $ \context $. By Definition \ref{def:junk} of junk, then $ \Context{}{}{\Context{'}{}{J}} \transBarbBisimSepAsyn \Context{}{}{\Context{'}{}{\nullTerm}} $ for all such contexts $ \context $. Since $ \Context{}{}{\Context{'}{}{\nullTerm}} = \Context{}{}{\RestrictedTerm{\tilde{x}}{\left( T' \mid \nullTerm \right)}} \equiv \Context{}{}{\RestrictedTerm{\tilde{x}}{T'}} $ and since, by Lemma \ref{lem:transBarbBisimIncludesSC}, $ \transBarbBisimSepAsyn $ includes structural congruence, we deduce $ \Context{}{}{\Context{'}{}{J}} \transBarbBisimSepAsyn \Context{}{}{\RestrictedTerm{\tilde{x}}{T'}} $ for all such contexts $ \context $. Because $ \transBarbBisimSepAsyn $ is, by Lemma \ref{lem:transBarbBisimIsEquivalence}, an equivalence, $ \Context{}{}{T} \transBarbBisimSepAsyn \Context{}{}{\Context{'}{}{J}} $ for all such contexts $ \context $ and $ \Context{}{}{\Context{'}{}{J}} \transBarbBisimSepAsyn \Context{}{}{\RestrictedTerm{\tilde{x}}{T'}} $ for all such contexts $ \context $ implies $ \Context{}{}{T} \transBarbBisimSepAsyn \Context{}{}{\RestrictedTerm{\tilde{x}}{T'}} $ for all such contexts $ \context $. Thus, by Definition \ref{def:transBarbCong}, we conclude $ T \transBarbCongSepAsyn \RestrictedTerm{\tilde{x}}{T'} $.
	
	Finally, since $ \transBarbCongSepAsyn $ includes structural congruence and is an equivalence, $ T \transBarbCongSepAsyn \RestrictedTerm{\tilde{x}}{T'} $ implies $ \RestrictedTerm{\tilde{x}}{T'} \transBarbCongSepAsyn \RestrictedTerm{\tilde{x}}{\left( T' \mid J \right)} $. Thus, since $ \RestrictedTerm{\tilde{x}}{\left( T' \mid J \right)} \in \pieceTargetTermsSepAsyn $, we conclude $ \RestrictedTerm{\tilde{x}}{T'} \in \pieceTargetTermsSepAsyn $.
	\qed
\end{proof}

Using this lemma we can remove junk from a target term $ T $. As result we obtain a piece of a target term $ T' $ such that $ T \transBarbCongSepAsyn T' $ or $ T \transBarbCongMixAsyn T' $, respectively. Then we can further reduce $ T' $ by removing junk such that we result in a piece of a target term $ T'' $ with $ T' \transBarbCongSepAsyn T'' $ or $ T' \transBarbCongMixAsyn T'' $ and so forth. Note that, we spend some effort in defining the notion of a piece of a target term to allow the stepwise remove of junk. This allows us to consider different kinds of junks separately. If we instead consider all possible junk of a target term at one go, then for the definition of junk it suffice to require that the context $ \context $ is such that $ \Context{}{}{J} $ is a target term. However, it seems to be more efficient and more descriptive to consider the different kinds of junk separated.

In the simplest case junk is a closed process that can not perform any step, i.e., junk is invisible and inactive. Such kind of junk is produced e.g. as remain of a test-statement. By Definition \ref{def:testBoolean}, a test-statement and the corresponding instantiations of booleans are defined as:
\begin{align*}
	\Output{\sumLock}{\true} & \deff \Input{\sumLock}{t, f}.\Out{t}\\
	\Output{\sumLock}{\false} & \deff \Input{\sumLock}{t, f}.\Out{f}\\
	\Test{\sumLock}{P}{Q} & \deff \RestrictedTerm{t, f}{\left( \Output{\sumLock}{t, f} \mid \In{t}.P \mid \In{f}.Q \right)} \quad \text{for some } t, f \notin \FreeNames{P} \cup \FreeNames{Q}
\end{align*}
Depending on whether we test a positive or a negative instantiation we result either in the then-case $ \RestrictedTerm{t, f}{\left( P \mid \In{f}.Q \right)} $ or the else-case $ \RestrictedTerm{t, f}{\left( \In{t}.P \mid Q \right)} $. Due to the renaming policy within target terms $ t $ and $ f $ are neither free in $ P $ nor in $ Q $. So we can pull out the interesting cases $ P $ or $ Q $ and $ \RestrictedTerm{t, f}{\In{f}.Q} $ or $ \RestrictedTerm{t, f}{\In{t}.P} $ remain as inobservable and inactive junk.

\begin{lemma} \label{lem:junkTestStatement}
	For any $ t, f \in \names $ and any $ P, Q \in \piAsynProc $ the terms $ \RestrictedTerm{t, f}{\In{f}.Q} $ and $ \RestrictedTerm{t, f}{\In{t}.P} $ are junk.
\end{lemma}

\begin{proof}
	Let $ J_1 = \RestrictedTerm{t, f}{\In{f}.Q} $ and $ J_2 = \RestrictedTerm{t, f}{\In{t}.P} $. $ J_1 $ as well as $ J_2 $ are closed terms, which can not perform any step. Moreover, they reach neither success nor any translated observable, i.e., $ J_1\not\reachSuccess $, $ J_2\not\reachSuccess $, $ J_1\not\WeakTransBarbSepAsyn{\mu} $, $ J_2\not\WeakTransBarbSepAsyn{\mu} $, $ J_1\not\WeakTransBarbMixAsyn{\mu} $, and $ J_2\not\WeakTransBarbMixAsyn{\mu} $ for all $ \mu \in \names \cup \coNames $. Because of that, for all contexts $ \Context{}{}{\hole} \in \piAsynProc \to \piAsynProc $ we have $ \Context{}{}{J_1} \transBarbBisimSepAsyn \Context{}{}{\nullTerm} \transBarbBisimSepAsyn \Context{}{}{J_2} $ and $ \Context{}{}{J_1} \transBarbBisimMixAsyn \Context{}{}{\nullTerm} \transBarbBisimMixAsyn \Context{}{}{J_2} $. Thus, by Definition \ref{def:junk}, $ J_1 $ and $ J_2 $ are junk.
	\qed
\end{proof}

Note that, due to this lemma, we can securely omit the remains of test-statements in the following. An other kind of inobservable and inactive junk is produced by the translation of the empty sum $ \nullTerm $. It results in a positive instantiation of a sum lock that is not used anywhere. However, let us generalise this case a little bit to an arbitrary instantiation of a sum lock (either positive or negative) that is not used anywhere.

\begin{lemma} \label{lem:junkEmptySum}
	For any name $ \sumLock $ the term $ \RestrictedTerm{\sumLock}{\left( \Output{\sumLock}{z} \right)} $, where $ z \in \bool $, is junk.
\end{lemma}

\begin{proof}
	Let $ J =  \RestrictedTerm{\sumLock}{\left( \Output{\sumLock}{z} \right)} $. $ J $ is a closed term, which can not perform any step. Moreover, this term can reach neither success nor any translated observable, i.e., $ J\not\reachSuccess $, $ J\not\WeakTransBarbSepAsyn{\mu} $, and $ J\not\WeakTransBarbMixAsyn{\mu} $ for all $ \mu \in \names \cup \coNames $. So for all contexts $ \Context{}{}{\hole} \in \piAsynProc \to \piAsynProc $ we have $ \Context{}{}{J} \transBarbBisimSepAsyn \Context{}{}{\nullTerm} $ and $ \Context{}{}{J} \transBarbBisimMixAsyn \Context{}{}{\nullTerm} $. Thus, by Definition \ref{def:junk}, $ J $ is junk.
	\qed
\end{proof}
\noindent
Note that this lemma especially shows that the translation of the empty sum is junk, i.e., we translate nothing into nothing but junk. Moreover we will use it to reduce the remains left over by \simulations. In the following lemma we prove, that requests on negative instantiations of sum locks are junk of the encoding from \piMix into \piAsyn. Note that in this case we consider potentially observable and active junk.

\begin{lemma} \label{lem:junkRequestsOnFalseSumLocks}
	Requests on negative instantiations of sum locks are junk of $ \encodingMixAsyn $, i.e.,
	\begin{align*}
		& \forall \Context{}{1}{\hole} \in \piAsynProc \to \piAsynProc \logdot \forall \parallelChannelIn, y, \sumLock, \receiverLock \in \names \logdot\\
		& \hspace*{1em} \Context{}{1}{\Output{\parallelChannelIn}{y, \sumLock, \receiverLock}} \in \pieceTargetTermsMixAsyn \wedge \left( \exists T \in \piAsynProc \logdot \exists \tilde{x} \subset \names \logdot \Context{}{1}{\nullTerm} \equiv \RestrictedTerm{\tilde{x}}{\left( T \mid \Output{\sumLock}{\false} \right)} \right)\\
		& \hspace*{1em} \text{implies } \Context{}{1}{\Output{\parallelChannelIn}{y, \sumLock, \receiverLock}} \transBarbBisimMixAsyn \Context{}{1}{\nullTerm}
	\end{align*}
	and
	\begin{align*}
		& \forall \Context{}{2}{\hole} \in \piAsynProc \to \piAsynProc \logdot \forall \parallelChannelOut, y, \sumLock, \senderLock, z \in \names \logdot\\
		& \hspace*{1em} \Context{}{2}{\Output{\parallelChannelOut}{y, \sumLock, \senderLock, z}} \in \pieceTargetTermsMixAsyn \wedge \left( \exists T \in \piAsynProc \logdot \exists \tilde{x} \subset \names \logdot \Context{}{2}{\nullTerm} \equiv \RestrictedTerm{\tilde{x}}{\left( T \mid \Output{\sumLock}{\false} \right)} \right)\\
		& \hspace*{1em} \text{implies } \Context{}{2}{\Output{\parallelChannelOut}{y, \sumLock, \senderLock, z}} \transBarbBisimMixAsyn \Context{}{2}{\nullTerm}.
	\end{align*}
\end{lemma}

\begin{proof}
	Let $ J_1 = \Output{\parallelChannelIn}{y, \sumLock, \receiverLock} $ and $ J_2 = \Output{\parallelChannelOut}{y, \sumLock, \senderLock, z} $.
	
	Since we require $ \Context{}{1}{J_1}, \Context{}{2}{J_2} \in \pieceTargetTermsMixAsyn $, we consider only contexts that respect the protocol implemented by the encoding function $ \encodingMixAsyn $. By analysing the encoding function in Figure \ref{fig:encodingMixAsyn}, we observe that there are many forwarders for requests, i.e., many replicated inputs that consume requests and immediately restore a request on a different channel but with exactly the same values. Note that by such forwarders requests are multiplied. Most of these copies turn out to be junk. However, note that the encoding function restricts the request channels|except for the outermost|and provides for each such channel exactly one (replicate) input and no inputs or replicated inputs for unrestricted request channels. Because of that, the way a request may travel is completely determined, i.e., given a target term there is no choice about which way a request may take.
	
	Moreover, if the context puts the requests at the outermost position, then they can not be consumed at all, i.e., in this case the requests are observable but inactive junk. By Lemma \ref{lem:instantiationSumLocks}, the negative instantiation of the sum lock $ \sumLock $ is the only instantiation of that lock and, by Lemma \ref{lem:changeInstantiationSumLock}, it can not be changed by the context into a positive instantiation. So, by Definition \ref{def:transBarb}, the requests are not considered as translated observables, i.e., $ J_1\not\WeakTransBarbMixAsyn{\mu} $ and $ J_2\not\WeakTransBarbMixAsyn{\mu} $ for all $ \mu \in \names \cup \coNames $. Obviously, $ J_1 $ and $ J_2 $ also have no possibility to reach $ \success $, i.e., $ J_1\not\reachSuccess $ and $ J_2\not\reachSuccess $. Thus for all contexts $ \Context{}{1}{\hole}, \Context{}{2}{\hole} \in \piAsynProc \to \piAsynProc $ such that $ \Context{}{1}{J_1}, \Context{}{2}{J_2} \in \pieceTargetTermsMixAsyn $, $ \parallelChannelIn \in \FreeNames{\Context{}{1}{J_1}} $, and $ \parallelChannelOut \in \FreeNames{\Context{}{2}{J_2}} $, we have $ \Context{}{1}{J_1} \transBarbBisimMixAsyn \Context{}{1}{\nullTerm} $ and $ \Context{}{2}{J_2} \transBarbBisimMixAsyn \Context{}{2}{\nullTerm} $.
	
	Note that we consider only contexts $ \context_1 $ and $ \context_2 $, such that $ \Context{}{1}{J_1} $ and $ \Context{}{2}{J_2} $ are pieces of target terms. So all requests|including $ J_1 $ and $ J_2 $|origin from the translation of input or output guarded terms or replicated inputs. All inputs on request channels, i.e., channels that can transport either three or four values, origin from the translation of the parallel operator or a replicated input. We observe that, for each request, the encoding of a parallel operator receives from one of its encoded parameters, one copy of that request is pushed upwards. Similarly, for each request of each branch of an encoded replicated input one copy is pushed to the right. Because of that, requests never vanish. As soon as a request in a target term is unguarded for the first time, there will always be a copy of that request possibly on another channel but with the same values (compare to Lemma \ref{lem:encodingMixAsynPreserveRequests}).
	
	Note that all inputs on request channels are replicated inputs except for the first inputs in the processing of right in- or output requests in the encoding of a parallel operator or a replicated input (compare to \processRightOutputRequests \ and \processRightInputRequests). However, even in these two cases the consumption of a request enables|after one internal step on a chain lock|the processing of another request in exactly the same way modulo some forwarding processes. So the order in which requests are consumed does not matter (compare to Lemma \ref{lem:pureAdminStepsNoDeadlock}). Moreover note, that the step on the chain lock is completely determined again, i.e., there is no choice about eventually doing it or how to do it. Because of that, the fact, that a request was already pushed further or not, does neither influence the reachability of success nor the reachability of translated observables.
	
	Except from transferring them the encoding function can proceed requests only by combining them. $ \encodingMixAsyn $ ensures, that each pair of requests is combined at most once and that each pair of an input and an output requests, which do not both origin from the same leaf concerning the structure of parallel operator encodings, is eventually combined. Because of that, again the order in which those combinations are performed does not matter. However, the order in which tests on sum locks are induced does indeed matter, because a test induced on only positive instantiations of sum locks turns them into negative instantiations and so a former such test may influences later tests. Since the test is always induced on the sum locks related to the respective request, all test that are induced because of the requests $ J_1 $ or $ J_2 $ are on at least one negative instantiation of a sum lock, i.e., on $ \sumLock $ in the current case. Note that, by Lemma \ref{lem:instantiationSumLocks}, the negative instantiation of the sum lock $ \sumLock $ is the only instantiation of that lock and, by Lemma \ref{lem:changeInstantiationSumLock}, it can not be changed by the context into a positive instantiation. If we now analyse the test-statements in the encodings of an input guarded term or a replicated input, we observe that this false instantiation of $ \sumLock $ reduces the test-statements to some kind of forwarders that consumes one or two instantiations of sum locks and since there are no deadlocks on these test-statements eventually restores them on exactly the same channel and with exactly the same value. So again the processing of $ J_1 $ or $ J_2 $ do neither influence the reachability of success nor the reachability of translated observables.
	
	So for all contexts $ \Context{}{1}{\hole}, \Context{}{2}{\hole} \in \piAsynProc \to \piAsynProc $, such that $ \Context{}{1}{J_1}, \Context{}{2}{J_2} \in \pieceTargetTermsMixAsyn $, we have
	\begin{align*}
		& \left( \Context{}{1}{J_1}\reachSuccess \text{ iff } \Context{}{1}{\nullTerm}\reachSuccess \right) \wedge \left( \Context{}{2}{J_2}\reachSuccess \text{ iff } \Context{}{2}{\nullTerm}\reachSuccess \right) \text{ and}\\
		& \left( \Context{}{1}{J_1}\WeakTransBarbMixAsyn{\mu} \text{ iff } \Context{}{1}{\nullTerm}\WeakTransBarbMixAsyn{\mu} \right) \wedge \left( \Context{}{2}{J_2}\WeakTransBarbMixAsyn{\mu} \text{ iff } \Context{}{2}{\nullTerm}\WeakTransBarbMixAsyn{\mu} \right)
	\end{align*}
	for all $ \mu \in \names \cup \coNames $. Since all steps that results from an interaction with $ J_1 $ and $ J_2 $ are \admin steps, that neither rule out a way leading to an unguarded occurrence of $ \success $ nor to a translated observable, they do not influence the state of a term modulo $ \transBarbBisimMixAsyn $. Thus, we conclude $ \Context{}{1}{J_1} \transBarbBisimMixAsyn \Context{}{1}{\nullTerm} $ and $ \Context{}{2}{J_2} \transBarbBisimMixAsyn \Context{}{2}{\nullTerm} $.
	\qed
\end{proof}

Next we show, that|for both encodings|encoded guarded terms linked to negative instantiations of sum locks are junk, as well. Note that such encoded guarded terms linked to negative instantiations of sum locks result from encoded sums of which already one summand was used to \simulate a source term step.

\begin{lemma} \label{lem:junkRemainsOfSumsSepAsyn}
	For any name $ \sumLock $, any finite index set $ \indexSet $, all guards $ \pi_i $, and all processes $ P_i \in \piSepProc $ the term
	\begin{align*}
		\RestrictedTerm{\sumLock}{\left( \prod_{i \in \indexSet} \EncodingSepAsyn{\pi_i.P_i} \mid \Output{\sumLock}{\false} \right)}
	\end{align*}
	is junk of $ \encodingSepAsyn $.
\end{lemma}

\begin{proof}
	Let $ J = \RestrictedTerm{\sumLock}{\left( \prod_{i \in \indexSet} \EncodingSepAsyn{\pi_i.P_i} \mid \Output{\sumLock}{\false} \right)} $. By Definition \ref{def:junk}, we have to show that for all contexts $ \Context{}{}{\hole} \in \piAsynProc \to \piAsynProc $ such that $ \Context{}{}{J} \in \pieceTargetTermsSepAsyn $ we have $ \Context{}{}{J} \transBarbBisimSepAsyn \Context{}{}{\nullTerm} $.
	
	By Lemma \ref{lem:instantiationSumLocks}, the negative instantiation of the sum lock $ \sumLock $ is the only instantiation of that lock and, by Lemma \ref{lem:changeInstantiationSumLock}, it can not be changed by the context into a positive instantiation. Analysing the encoding function in Figure \ref{fig:encodingSepAsyn} we observe, that all unguarded in- or outputs with three parameters in $ J $ are connected to $ \Output{\sumLock}{\false} $. Thus $ J $ has no translated observables, i.e., $ J\not\TransBarbSepAsyn{\mu} $ for all $ \mu \in \names \cup \coNames $. Moreover, because of the guards $ \guard_i $, $ J $ has no unguarded occurrence of $ \success $ and can not reach some on its own, i.e., $ J\not\reachSuccess $.
	
	$ J $ can perform a step on its own only if for some $ j \in \indexSet $ the guard $ \guard_j $ is equal to $ \tau $ or is an input guard. Note that in the second case $ J $ can only perform finitely many steps on receiver locks. Since $ \sumLock $ is instantiated by $ \false $, $ \EncodingMixAsyn{\tau.P_j} $ can reduce to $ \nullTerm $ only. Therefore, $ \EncodingMixAsyn{\tau.P_j} $ has to consume $ \Output{\sumLock}{\false} $ but the instantiation is always eventually restored. Because of that, we can ignore all $ \EncodingMixAsyn{\guard_i.P_i} $ for $ \guard_i = \tau $, i.e., they are junk. Let $ J' = \RestrictedTerm{\sumLock}{\left( \prod_{i \in \indexSet, \guard_i \neq \tau} \EncodingMixAsyn{\pi_i.P_i} \mid \Output{\sumLock}{\false} \right)} $. Moreover, obtain $ J'' $ from $ J' $ by reducing all free inputs on receiver locks.
	
	Then $ J'' $ can not perform any step, since either all remaining guards are output guards and so all unguarded inputs in the encoded subterms are on not instantiated sender locks, or all remaining guards are input guards and so the encoded subterms are input guarded, i.e., $ J'' \not\step $. In case the index set $ \Set{ i \mid i \in \indexSet \wedge \guard_i \neq \tau } $ is empty, we can apply Lemma \ref{lem:junkEmptySum}. Else, there are either free outputs or free inputs on translated source term names. So $ J'' $ can communicate with the context over these translated source term names. As a result of such a communication a test-statement is unguarded. In case this test-statement is within $ J'' $, the sum lock $ \sumLock $ is tested first. Since $ \sumLock $ is instantiated by $ \false $, the test-statement restores both, the consumed negative instantiation of $ \sumLock $ and the output on the translated source term name, which was consumed to unguard this test-statement. Else, if the unguarded test-statement is not within $ J'' $, then|besides the negative instantiation of the sum lock|$ J'' $ includes only encodings of output guarded source terms. In this case all these encodings of output guarded source terms provide exactly one output on a translated source term name, which sends as first value the sum lock $ \sumLock $. In the communication with the context one of these outputs was consumed and, because of that, the unguarded test-statement is either a nested test-statement testing $ \sumLock $ as second lock or it is a single test-statement testing only $ \sumLock $. In both cases all information necessary to unguard and to reduce the test-statement except for the output on the source term name is restored. Because of that, each such output can be used as most once. In both cases no step that results from an interaction with $ J'' $ influences the reachability of success or of translated observables. Moreover each such step is either a \pure \admin step, which by Lemma \ref{lem:pureAdminStepsTransBarbBisimSepAsyn} do not influence the state of a term modulo $ \transBarbBisimSepAsyn $, or it is an \impure \admin step. Since reachability of $ \success $ is not influenced at all and all translated observables of the context which are consumed to unguard the test-statement are eventually restored, i.e., reachability of translated observables is not affected, even these \impure \admin steps do not influence the state of the context modulo $ \transBarbBisimSepAsyn $ here.
	
	Thus for all contexts $ \Context{}{}{\hole} \in \piAsynProc \to \piAsynProc $, such that $ \Context{}{}{J} \in \pieceTargetTermsSepAsyn $, we have $ \left( \Context{}{}{J}\reachSuccess \text{ iff } \Context{}{}{\nullTerm}\reachSuccess \right) $ and $ \left( \Context{}{}{J}\WeakTransBarbSepAsyn{\mu} \text{ iff } \Context{}{}{\nullTerm}\WeakTransBarbSepAsyn{\mu} \right) $ for all $ \mu \in \names \cup \coNames $. Moreover no step of $ J $ on its own or that results from an interaction with $ J $ does influence the state of the context modulo $ \transBarbBisimSepAsyn $. So, $ \Context{}{}{J} \transBarbBisimSepAsyn \Context{}{}{\nullTerm} $.
	\qed
\end{proof}

\begin{lemma} \label{lem:junkRemainsOfSumsMixAsyn}
	For any name $ \sumLock $, any finite index set $ \indexSet $, all guards $ \pi_i $, and all processes $ P_i \in \piMixProc $ the term
	\begin{align*}
		\RestrictedTerm{\sumLock}{\left( \prod_{i \in \indexSet} \EncodingMixAsyn{\pi_i.P_i} \mid \Output{\sumLock}{\false} \right)}
	\end{align*}
	is junk of $ \encodingMixAsyn $.
\end{lemma}

\begin{proof}
	Let $ J = \RestrictedTerm{\sumLock}{\left( \prod_{i \in \indexSet} \EncodingMixAsyn{\pi_i.P_i} \mid \Output{\sumLock}{\false} \right)} $. By Definition \ref{def:junk}, we have to show that for all contexts $ \Context{}{}{\hole} \in \piAsynProc \to \piAsynProc $ such that $ \Context{}{}{J} \in \pieceTargetTermsMixAsyn $ we have $ \Context{}{}{J} \transBarbBisimMixAsyn \Context{}{}{\nullTerm} $.
	
	By Lemma \ref{lem:instantiationSumLocks}, the negative instantiation of the sum lock $ \sumLock $ is the only instantiation of that lock and, by Lemma \ref{lem:changeInstantiationSumLock}, it can not be changed by the context into a positive instantiation. Analysing the encoding function in Figure \ref{fig:encodingMixAsyn} we observe, that all unguarded requests in $ J $ are connected to $ \Output{\sumLock}{\false} $ (compare to Definition \ref{def:connectionSumLockSenderReceiver} and Lemma \ref{lem:sumLockOfReceiverOrSender}). Thus $ J $ has no translated observables, i.e., $ J\not\TransBarbSepAsyn{\mu} $ for all $ \mu \in \names \cup \coNames $. Moreover, because of the guards $ \guard_i $, $ J $ has no unguarded occurrence of $ \success $ and can not reach some on its own, i.e., $ J\not\reachSuccess $.
	
	$ J $ can perform a step on its own only if for some $ j \in \indexSet $ the guard $ \guard_j $ is equal to $ \tau $. Since $ \sumLock $ is instantiated by $ \false $, $ \EncodingMixAsyn{\tau.P_j} $ can reduce to $ \nullTerm $ only. Therefore, $ \EncodingMixAsyn{\tau.P_j} $ has to consume $ \Output{\sumLock}{\false} $ but the instantiation is always eventually restored. Because of that, we can ignore all $ \EncodingMixAsyn{\guard_i.P_i} $ for $ \guard_i = \tau $, i.e., they are junk. Let $ J' = \RestrictedTerm{\sumLock}{\left( \prod_{i \in \indexSet, \guard_i \neq \tau} \EncodingMixAsyn{\pi_i.P_i} \mid \Output{\sumLock}{\false} \right)} $.
	
	Then $ J' $ can not perform any step, i.e., $ J' \not\step $. There are no unguarded inputs on free names of $ J' $. In case the index set $ \Set{ i \mid i \in \indexSet \wedge \guard_i \neq \tau } $ is empty, we can apply Lemma \ref{lem:junkEmptySum}. Else, there are free requests, i.e., free outputs on requests channels. So $ J' $ can communicate with the context by transmitting its requests. By Lemma \ref{lem:junkRequestsOnFalseSumLocks} these requests are junk. Moreover by revisiting the argumentation in the proof of Lemma \ref{lem:junkRequestsOnFalseSumLocks} we observe, that the false instantiation of $ \sumLock $ reduces all tests on that lock to simple forwarders. As a consequence no such test can unguard the encoding of a guarded source term, i.e., no such test can unguard new requests or former unguarded occurrences of success. So again we observe that no interaction with $ J $ or $ J' $ can influence the ability of the context to reach success or translated observables and no step that results from an interaction with $ J $ or $ J' $ influences the state of the context modulo $ \transBarbBisimMixAsyn $.
	
	Thus for all contexts $ \Context{}{}{\hole} \in \piAsynProc \to \piAsynProc $, such that $ \Context{}{}{J} \in \pieceTargetTermsMixAsyn $, we have $ \left( \Context{}{}{J}\reachSuccess \text{ iff } \Context{}{}{\nullTerm}\reachSuccess \right) $ and $ \left( \Context{}{}{J}\WeakTransBarbMixAsyn{\mu} \text{ iff } \Context{}{}{\nullTerm}\WeakTransBarbMixAsyn{\mu} \right) $ for all $ \mu \in \names \cup \coNames $. Moreover no step of $ J $ on its own or that results from an interaction with $ J $ does influence the state of the context modulo $ \transBarbBisimMixAsyn $. So, $ \Context{}{}{J} \transBarbBisimMixAsyn \Context{}{}{\nullTerm} $.
	\qed
\end{proof}

Analysing the encoding function we observe, that the input on a receiver lock of an encoded input guarded source term, that guards a test-statement, is a replicated input. Of course, such a test-statement can only be used once to \simulate a source term step. After such an \simulation this replicated input becomes junk.

\begin{lemma} \label{lem:junkReceiverMixAsyn}
	For any $ \sumLock_1, \sumLock_2, \receiverLock, \senderLock, \RenamingPolicyMixAsyn{x} \in \names $ and any $ P \in \piMixProc $ the term
	\begin{align*}
		J = \ReplicateInput{\receiverLock}{\sumLock_1, \sumLock_2, -, \senderLock, \RenamingPolicyMixAsyn{x}}.\Test{\sumLock_1}{\Test{\sumLock_2}{\Output{\sumLock_1}{\false} \mid \Output{\sumLock_2}{\false} \mid \Out{\senderLock} \mid \EncodingMixAsyn{P}}{\Output{\sumLock_1}{\true} \mid \Output{\sumLock_2}{\false}}}{\Output{\sumLock_1}{\false}}
	\end{align*}
	in combination with a negative instantiation of the sum lock that belongs to $ \receiverLock $ is junk of $ \encodingMixAsyn $, i.e.,
	\begin{align*}
		& \forall \Context{}{}{\hole} \in \piAsynProc \to \piAsynProc \logdot \forall \sumLock_1, \sumLock_2, \receiverLock, \senderLock, \RenamingPolicyMixAsyn{x} \in \names \logdot \forall P \in \piMixProc \logdot \forall J \in \piAsynProc \logdot\\
		& \hspace*{1em} J = \ReplicateInput{\receiverLock}{\sumLock_1, \sumLock_2, -, \senderLock, \RenamingPolicyMixAsyn{x}}.\Test{\sumLock_1}{\Test{\sumLock_2}{\Output{\sumLock_1}{\false} \mid \Output{\sumLock_2}{\false} \mid \Out{\senderLock} \mid \EncodingMixAsyn{P}}{\Output{\sumLock_1}{\true} \mid \Output{\sumLock_2}{\false}}}{\Output{\sumLock_1}{\false}}\\
		& \hspace*{1em} \wedge \Context{}{}{J} \in \pieceTargetTermsMixAsyn \wedge \left( \exists T \in \piAsynProc \logdot \exists \tilde{x} \subset \names \logdot \exists \parallelChannelIn, y, \sumLock \in \names \logdot \Context{}{}{\nullTerm} \equiv \RestrictedTerm{\tilde{x}}{\left( T \mid \Output{\parallelChannelIn}{y, \sumLock, \receiverLock} \mid \Output{\sumLock}{\false} \right)} \right)\\
		& \hspace*{1em} \text{implies } \Context{}{}{J} \transBarbBisimMixAsyn \Context{}{}{\nullTerm}
	\end{align*}
\end{lemma}

\begin{proof}
	Analysing the encoding function in Figure \ref{fig:encodingMixAsyn} we observe that the request $ \Output{\parallelChannelIn}{y, \sumLock, \receiverLock} $ implies that each test induced on $ \receiverLock $ tests the sum lock $ \sumLock $ (compare to Lemma \ref{lem:sumLockOfReceiverOrSender}). Since that lock is already instantiated by false, we can revisit the argumentation of the proof of Lemma \ref{lem:junkRemainsOfSumsMixAsyn} to conclude that each test induced on $ \receiverLock $ reduces to a simple forwarder, which restores all consumed information. So $ J $ can be ignored, i.e., $ J $ is Junk.
	\qed
\end{proof}

Another remain of an \simulation, which we can simply consider as junk, is the preparation of a test on a negative instantiated sum lock.

\begin{lemma}  \label{lem:junkInduceTest}
	The preparation of a test on a negative instantiation of a sum lock is junk of $ \encodingMixAsyn $, i.e.,
	\begin{align*}
		& \forall \Context{}{}{\hole} \in \piAsyn \to \piAsyn \logdot \forall \matchingCoordinatorIn, \matchingUpIn, y, y', \sumLock, \sumLock_r, \receiverLock, \senderLock, z \in \names \logdot \forall J \in \piAsynProc \logdot\\
		& \hspace*{1em} J = \ReplicateInput{\matchingCoordinatorIn}{y', \sumLock_r, \receiverLock}.\left( \Match{y'}{y}\Output{\receiverLock}{\sumLock_r, \sumLock, \sumLock, \senderLock, z} \mid \Output{\matchingUpIn}{y', \sumLock_r, \receiverLock} \right) \wedge \Context{}{}{J} \in \pieceTargetTermsMixAsyn\\
		& \hspace*{1em} \wedge \left( \exists T \in \piAsynProc \logdot \exists \tilde{x} \subset \names \logdot \Context{}{}{\nullTerm} \equiv \RestrictedTerm{\tilde{x}}{\left( T \mid \Output{\sumLock}{\false} \right)} \right)\\
		& \hspace*{1em} \text{implies } \Context{}{}{J} \transBarbBisimMixAsyn \Context{}{}{\Forward{\matchingCoordinatorIn}{\matchingUpIn}}
	\end{align*}
	and
	\begin{align*}
		& \forall \Context{}{}{\hole} \in \piAsyn \to \piAsyn \logdot \forall \matchingCoordinatorOut, \matchingUpOut, y, y', \sumLock, \sumLock_s, \receiverLock, \senderLock, z \in \names \logdot \forall J \in \piAsynProc \logdot\\
		& \hspace*{1em} J = \ReplicateInput{\matchingCoordinatorOut}{y', \sumLock_s, \senderLock, z}.\left( \Match{y'}{y}\Output{\receiverLock}{\sumLock_s, \sumLock, \sumLock_s, \senderLock, z} \mid \Output{\matchingUpOut}{y', \sumLock_s, \senderLock, z} \right) \wedge \Context{}{}{J} \in \pieceTargetTermsMixAsyn\\
		& \hspace*{1em} \wedge \left( \exists T \in \piAsynProc \logdot \exists \tilde{x} \subset \names \logdot \Context{}{}{\nullTerm} \equiv \RestrictedTerm{\tilde{x}}{\left( T \mid \Output{\sumLock}{\false} \right)} \right)\\
		& \hspace*{1em} \text{implies } \Context{}{}{J} \transBarbBisimMixAsyn \Context{}{}{\Forward{\matchingCoordinatorOut}{\matchingUpOut}}
	\end{align*}
\end{lemma}

\begin{proof}
	First note that here not $ J $ but only $ \Match{y'}{y}\Output{\receiverLock}{\sumLock_r, \sumLock, \sumLock, \senderLock, z} $ and $ \Match{y'}{y}\Output{\receiverLock}{\sumLock_s, \sumLock, \sumLock_s, \senderLock, z} $ are considered as junk. If these terms are omitted then the respective $ J $ reduces to the forwarder $ \Forward{\matchingCoordinatorIn}{\matchingUpIn} $ or $ \Forward{\matchingCoordinatorOut}{\matchingUpOut} $.
	
	In the first case|regardless whether the receiver lock $ \receiverLock $ belongs to an encoded input guarded term or an encoded replicated input|the output $ \Output{\receiverLock}{\sumLock_s, \sumLock, \sumLock_s, \senderLock, z} $ induces a test on the sum lock $ \sumLock $. Since that lock is already instantiated by false we can revisit the argumentation of the proof of Lemma \ref{lem:junkRemainsOfSumsMixAsyn} to conclude that we can ignore this induced test. So we can also ignore the inducing output $ \Output{\receiverLock}{\sumLock_s, \sumLock, \sumLock_s, \senderLock, z} $. The remainig term $ \ReplicateInput{\matchingCoordinatorIn}{y', \sumLock_r, \receiverLock}. \Output{\matchingUpIn}{y', \sumLock_r, \receiverLock} $ is equal to the forwarder $ \Forward{\matchingCoordinatorIn}{\matchingUpIn} $.
	
	For the second case, since there is a negative instantiation of $ \sumLock $, the receiver lock $ \receiverLock $ was not created by the encoding of a replicated input. So again the output $ \Output{\receiverLock}{\sumLock_s, \sumLock, \sumLock_s, \senderLock, z} $ induces a test on the sum lock $ \sumLock $. The rest of that case is similar to the case before.
	\qed
\end{proof}

Unfortunately, we can not declare any remains of \simulations as junk, because we can not ignore the forwarding of left requests in the chains of right requests which is left over by former considered right requests. However, after extracting the junk, by Lemma \ref{lem:junkInduceTest}, there is indeed nothing more left, than a simple forwarder, which can not influence the state of the process modulo $ \transBarbBisimMixAsyn $. That suffice to prove operational completeness.

\subsection{Semantical Criteria}

Among the semantical criteria operational correspondence is the most elaborate to prove. Therefore we show the both its conditions, operational completeness and operational soundness, separately. In order to show operational completeness, we have to show how source terms steps are \simulated by the encodings.

\begin{lemma}[Operational Completeness] \label{lem:operationalCompletenessSepAsyn}
	The encoding $ \encodingSepAsyn $ fulfils operational completeness.
\end{lemma}

\begin{proof}
	By Definition \ref{def:operationalCorrespondence} it suffice to show that:
	\begin{align*}
		\forall S, S' \in \piSepProc \logdot S \step S' \text{ implies } \exists T \in \piAsynProc \logdot \EncodingSepAsyn{S} \steps T \wedge T \transBarbCongSepAsyn \EncodingSepAsyn{S'}
	\end{align*}
	The lemma then holds by induction over the number of steps in $ S \steps S' $. To prove the condition above, we perform an induction over the proof tree that leads to the step $ S \step S' $.
	\begin{description}
		\item[Base Case:] By the rules in Figure \ref{fig:concurrentReductionSemantics} each step on $ S $ is based either on Rule $ \textsc{Tau}_{\indexMix, \indexSep} $, $ \textsc{Com}_{\indexMix, \indexSep}  $, or $ \textsc{Rep}_{\indexMix, \indexSep} $.
			\begin{description}
				\item[Case of Rule] $ \textsc{Tau}_{\indexMix, \indexSep} $\textbf{:} In this case $ S $ is a single sum, of which one summand is guarded by $ \tau $, and $ S' $ is the continuation of this $ \tau $ guarded summand, i.e., there are some finite index set $ I $, some guards $ \pi_i $, and some processes $ P_i \in \piSepProc $ such that $ S = \sum_{i \in \indexSet} \guard_i.P_i $ with $ \guard_j = \tau $ for some $ j \in \indexSet $ and $ S' = P_j $. The corresponding encodings are given by the following terms:
					\begin{align*}
						\EncodingSepAsyn{S} & \equiv \RestrictedTerm{\sumLock}{\left( \Output{\sumLock}{\true} \mid \prod_{i \in \indexSet, i \neq j} \EncodingSepAsyn{\guard_i.P_i} \mid \Test{\sumLock}{\Output{\sumLock}{\false} \mid \EncodingSepAsyn{P_j}}{\Output{\sumLock}{\false}} \right)}\\
						\EncodingSepAsyn{S'} & = \EncodingSepAsyn{P_j}
					\end{align*}
					We observe that $ \EncodingSepAsyn{S} $ can \simulate the step $ S \step S' $ by reducing the test-statement in the encoding of the $ j $'s summand, i.e., by
					\begin{align*}
						\EncodingSepAsyn{S} \step^2 \RestrictedTerm{\sumLock}{\left( \prod_{i \in \indexSet, i \neq j} \EncodingSepAsyn{\guard_i.P_i} \mid \Output{\sumLock}{\false} \mid \EncodingSepAsyn{P_j} \right)} = T.
					\end{align*}
					Note that, since the test-statement and the implementation of booleans are no native \piCal-terms but abbreviations for \piCal-constructs, this reduction indeed requires two steps. Moreover note, that we silently omit junk that results from the reduction of test-statements (compare to Lemma \ref{lem:junkTestStatement}) here and in the following proofs. Further, we observe that $ T \equiv \RestrictedTerm{\sumLock}{\left( \prod_{i \in \indexSet, i \neq j} \EncodingSepAsyn{\guard_i.P_i} \mid \Output{\sumLock}{\false} \right)} \mid \EncodingSepAsyn{P_j} $, since, due to the renaming policy $ \renamingPolicySepAsyn $, the name $ \sumLock $ is not free in $ \EncodingSepAsyn{P_j} $. So the \simulation leaves over the term $ \RestrictedTerm{\sumLock}{\left( \prod_{i \in \indexSet, i \neq j} \EncodingSepAsyn{\guard_i.P_i} \mid \Output{\sumLock}{\false} \right)} $, which is by Lemma \ref{lem:junkRemainsOfSumsSepAsyn} junk. By Lemma \ref{lem:removeJunk}, we conclude that $ T \transBarbCongSepAsyn \EncodingSepAsyn{S'} $.
				\item[Case of Rule] $ \textsc{Com}_{\indexMix, \indexSep} $\textbf{:} Here $ S $ is a parallel composition of two sums and $ S' $ is the parallel composition of the continuations of an input guarded summand of the first and a matching output guarded summand of the second sum, i.e., there are two finite index sets $ \indexSet_1, \indexSet_2 $, some guards $ \guard_i $, and some processes $ P_i, Q_i \in \piSepProc $ such that $ S = \sum_{i \in \indexSet_1} \guard_i.P_i \mid \sum_{i \in \indexSet_2} \guard_i.Q_i $ with $ \guard_{j_1} = \Input{y}{x} $ and $ \guard_{j_2} = \Output{y}{z} $ for some $ j_1 \in \indexSet_1 $, some $ j_2 \in \indexSet_2 $, and $ x, y, z \in \names $ and $ S' = \Set{ \Subst{z}{x} } P_{j_1} \mid Q_{j_2} $. The encodings of $ S $ and $ S' $ are given by the following terms:
					\begin{align*}
						\EncodingSepAsyn{S} & \equiv \RestrictedTerm{\sumLock}{\Big( \begin{aligned}[t]
								& \Output{\sumLock}{\true} \mid \prod_{i \in \indexSet_1, i \neq j_1} \EncodingSepAsyn{\guard_i.P_i}\\
								& \mid \RestrictedTerm{\receiverLock}{\big( \Out{\receiverLock} \mid \ReplicateIn{\receiverLock}.\Input{\RenamingPolicySepAsyn{y}}{\sumLock', \senderLock, \RenamingPolicySepAsyn{x}}.\BigTest{\sumLock}{\BigTest{\sumLock'}{\Output{\sumLock}{\false} \mid \Output{\sumLock'}{\false} \mid \Out{\senderLock} \mid \EncodingSepAsyn{P_{j_1}}}{\Output{\sumLock}{\true} \mid \Output{\sumLock'}{\false} \mid \Out{\receiverLock}}}{\Output{\sumLock}{\true} \mid \Output{\RenamingPolicySepAsyn{y}}{\sumLock', \senderLock, \RenamingPolicySepAsyn{x}} \big) \Big)}}
							\end{aligned}}\\
							& \hspace*{1em} \mid \RestrictedTerm{\sumLock}{\left( \Output{\sumLock}{\true} \mid \prod_{i \in \indexSet_1, i \neq j_1} \EncodingSepAsyn{\guard_i.Q_i} \mid \RestrictedTerm{\senderLock}{\left( \Output{\RenamingPolicySepAsyn{y}}{\sumLock, \senderLock, \RenamingPolicySepAsyn{z}} \mid \In{\senderLock}.\EncodingSepAsyn{Q_{j_2}} \right)} \right)}\\
						\EncodingSepAsyn{S'} & = \EncodingSepAsyn{\Set{ \Subst{z}{x} }\left( P_{j_1} \right)} \mid \EncodingSepAsyn{Q_{j_2}}
					\end{align*}
					To \simulate the source term step $ S \step S' $ first the receiver lock has to be reduced to enable a communication over $ \RenamingPolicySepAsyn{y} $. Then the test-statement and the sender lock are reduced to complete the \simulation of the source term step.
					\begin{align*}
						\EncodingSepAsyn{S} & \step^2 \RestrictedTerm{\sumLock_1, \sumLock_2, \senderLock}{\Big( \begin{aligned}[t]
								& \Output{\sumLock_1}{\true} \mid \Set{ \Subst{\sumLock_1}{\sumLock} }\left( \prod_{i \in \indexSet_1, i \neq j_1} \EncodingSepAsyn{\guard_i.P_i} \right)\\
								& \mid \RestrictedTerm{\receiverLock}{\big( \begin{aligned}[t]
										& \BigTest{\sumLock_1}{\BigTest{\sumLock_2}{\Output{\sumLock_1}{\false} \mid \Output{\sumLock_2}{\false} \mid \Out{\senderLock} \mid \Set{ \Subst{\RenamingPolicySepAsyn{z}}{\RenamingPolicySepAsyn{x}} }\left( \EncodingSepAsyn{P_{j_1}} \right)}{\Output{\sumLock_1}{\true} \mid \Output{\sumLock_2}{\false} \mid \Out{\receiverLock}}}{\Output{\sumLock_1}{\false} \mid \Output{\RenamingPolicySepAsyn{y}}{\sumLock_2, \senderLock, \RenamingPolicySepAsyn{z}}}\\
										& \ReplicateIn{\receiverLock}.\Input{\RenamingPolicySepAsyn{y}}{\sumLock', \senderLock, \RenamingPolicySepAsyn{x}}.\BigTest{\sumLock_1}{\BigTest{\sumLock'}{\Output{\sumLock_1}{\false} \mid \Output{\sumLock'}{\false} \mid \Out{\senderLock} \mid \EncodingSepAsyn{P_{j_1}}}{\Output{\sumLock_1}{\true} \mid \Output{\sumLock'}{\false} \mid \Out{\receiverLock}}}{\Output{\sumLock_1}{\false} \mid \Output{\RenamingPolicySepAsyn{y}}{\sumLock', \senderLock, \RenamingPolicySepAsyn{x}} \big)}
									\end{aligned}\\
								& \mid \Output{\sumLock_2}{\true} \mid \Set{ \Subst{\sumLock_2}{\sumLock} }\left( \prod_{i \in \indexSet_1, i \neq j_2} \EncodingSepAsyn{\guard_i.Q_i} \right) \mid \In{\senderLock}.\EncodingSepAsyn{Q_{j_2}} \Big)}
							\end{aligned}}
					\end{align*}
					\begin{align*}
						\hspace*{1em} & \step^5 \RestrictedTerm{\sumLock_1, \sumLock_2, \senderLock}{\Big( \begin{aligned}[t]
								& \Set{ \Subst{\sumLock_1}{\sumLock} }\left( \prod_{i \in \indexSet_1, i \neq j_1} \EncodingSepAsyn{\guard_i.P_i} \right)\\
								& \mid \RestrictedTerm{\receiverLock}{\big( \begin{aligned}[t]
										& \Output{\sumLock_1}{\false} \mid \Output{\sumLock_2}{\false} \mid \Set{ \Subst{\RenamingPolicySepAsyn{z}}{\RenamingPolicySepAsyn{x}} }\left( \EncodingSepAsyn{P_{j_1}} \right)\\
										& \ReplicateIn{\receiverLock}.\Input{\RenamingPolicySepAsyn{y}}{\sumLock', \senderLock, \RenamingPolicySepAsyn{x}}.\BigTest{\sumLock_1}{\BigTest{\sumLock'}{\Output{\sumLock_1}{\false} \mid \Output{\sumLock'}{\false} \mid \Out{\senderLock} \mid \EncodingSepAsyn{P_{j_1}}}{\Output{\sumLock_1}{\true} \mid \Output{\sumLock'}{\false} \mid \Out{\receiverLock}}}{\Output{\sumLock_1}{\false} \mid \Output{\RenamingPolicySepAsyn{y}}{\sumLock', \senderLock, \RenamingPolicySepAsyn{x}} \big)}
									\end{aligned}\\
								& \mid \Set{ \Subst{\sumLock_2}{\sumLock} }\left( \prod_{i \in \indexSet_1, i \neq j_2} \EncodingSepAsyn{\guard_i.Q_i} \right) \mid \EncodingSepAsyn{Q_{j_2}} \Big) = T}
							\end{aligned}}
					\end{align*}
					By Corollary \ref{col:encodingSubstitutions}, $ \Set{ \Subst{\RenamingPolicySepAsyn{z}}{\RenamingPolicySepAsyn{x}} } \EncodingSepAsyn{P_{j_1}} \equivAlpha \EncodingSepAsyn{\Set{ \Subst{z}{x} } P_{j_1}}  $. To show that $ T \transBarbCongSepAsyn \EncodingMixAsyn{S'} $ we stepwise reduce $ T $ by ignoring junk. Since $ \sumLock_1, \sumLock_2, \receiverLock, \senderLock \notin \FreeNames{\EncodingMixAsyn{\Set{ \Subst{z}{x} } P_{j_1}}} \cup \FreeNames{\EncodingMixAsyn{Q_{j_2}}} $, we can reorder the term according to the restrictions on $ \sumLock_1, \sumLock_2, \receiverLock $ and the restriction on $ \senderLock $ can be omitted. The term
					\begin{align*}
						\RestrictedTerm{\receiverLock}{\big( \ReplicateIn{\receiverLock}.\Input{\RenamingPolicySepAsyn{y}}{\sumLock', \senderLock, \RenamingPolicySepAsyn{x}}.\BigTest{\sumLock_1}{\BigTest{\sumLock'}{\Output{\sumLock_1}{\false} \mid \Output{\sumLock'}{\false} \mid \Out{\senderLock} \mid \EncodingSepAsyn{P_{j_1}}}{\Output{\sumLock_1}{\true} \mid \Output{\sumLock'}{\false} \mid \Out{\receiverLock}}}{\Output{\sumLock_1}{\false} \mid \Output{\RenamingPolicySepAsyn{y}}{\sumLock', \senderLock, \RenamingPolicySepAsyn{x}} \big)}}
					\end{align*}
					is obviously junk, since it is closed and can not perform any step.	Moreover, by Lemma \ref{lem:junkRemainsOfSumsSepAsyn}, the terms $ \RestrictedTerm{\sumLock}{\left( \prod_{i \in \indexSet_1, i \neq j_1} \EncodingSepAsyn{\guard_i.P_i} \mid \Output{\sumLock}{\false} \right)} $ and $ \RestrictedTerm{\sumLock}{\left( \prod_{i \in \indexSet_1, i \neq j_2} \EncodingSepAsyn{\guard_i.Q_i} \mid \Output{\sumLock}{\false} \right)} $ are junk. So, by Lemma \ref{lem:removeJunk}, we conclude $ T \transBarbCongSepAsyn \EncodingSepAsyn{S'} $.
				\item[Case of Rule] $ \textsc{Rep}_{\indexMix, \indexSep} $\textbf{:} Here $ S $ is a parallel composition of a replicated input and a sum and $ S' $ is the parallel composition of the replicated input, the continuation of the replicated input, and the continuation of a matching output guarded summand, i.e., there is a finite index sets $ \indexSet $, some guards $ \guard_i $, and some processes $ P, Q_i \in \piSepProc $ such that $ S = \ReplicateInput{y}{x}.P \mid \sum_{i \in \indexSet} \guard_i.Q_i $ with $ \guard_j = \Output{y}{z} $ for some $ j \in \indexSet $ and $ x, y, z \in \names $, and $ S' = \Set{ \Subst{z}{x} } P \mid Q_j \mid \ReplicateInput{y}{x}.P $. The encodings of $ S $ and $ S' $ are given by the following terms:
					\begin{align*}
						\EncodingSepAsyn{S} & \equiv \ReplicateInput{\RenamingPolicySepAsyn{y}}{\sumLock, \senderLock, \RenamingPolicySepAsyn{x}}.\Test{\sumLock}{\Output{\sumLock}{\false} \mid \Out{\senderLock} \mid \EncodingSepAsyn{P}}{\Output{\sumLock}{\false}}\\
							& \hspace*{1em} \mid \RestrictedTerm{\sumLock}{\left( \Output{\sumLock}{\true} \mid \prod_{i \in \indexSet_1, i \neq j} \EncodingSepAsyn{\guard_i.Q_i} \mid \RestrictedTerm{\senderLock}{\left( \Output{\RenamingPolicySepAsyn{y}}{\sumLock, \senderLock, \RenamingPolicySepAsyn{z}} \mid \In{\senderLock}.\EncodingSepAsyn{Q_j} \right)} \right)}\\
						\EncodingSepAsyn{S'} & = \EncodingSepAsyn{\Set{ \Subst{z}{x} }\left( P \right)} \mid \EncodingSepAsyn{Q_j} \mid \EncodingSepAsyn{\ReplicateInput{y}{x}.P}
					\end{align*}
					To \simulate the source term step $ S \step S' $, first the two subprocesses of $ \EncodingSepAsyn{S} $ communicate over $ \RenamingPolicySepAsyn{y} $. Then the test-statement and the sender lock are reduced to complete the \simulation of the source term step.
					\begin{align*}
						\EncodingSepAsyn{S} & \step \RestrictedTerm{\sumLock, \senderLock}{\Big( \begin{aligned}[t]
								& \Test{\sumLock}{\Output{\sumLock}{\false} \mid \Out{\senderLock} \mid \Set{ \Subst{\RenamingPolicySepAsyn{z}}{\RenamingPolicySepAsyn{x}} }\left( \EncodingSepAsyn{P} \right)}{\Output{\sumLock}{\false}}\\
								& \mid \ReplicateInput{\RenamingPolicySepAsyn{y}}{\sumLock, \senderLock, \RenamingPolicySepAsyn{x}}.\Test{\sumLock}{\Output{\sumLock}{\false} \mid \Out{\senderLock} \mid \EncodingSepAsyn{P}}{\Output{\sumLock}{\false}}\\
								& \mid \Output{\sumLock}{\true} \mid \prod_{i \in \indexSet_1, i \neq j} \EncodingSepAsyn{\guard_i.Q_i} \mid \In{\senderLock}.\EncodingSepAsyn{Q_j} \Big)
							\end{aligned}}\\
						& \step^3 \RestrictedTerm{\sumLock, \senderLock}{\Big( \begin{aligned}[t]
								& \Output{\sumLock}{\false} \mid \Set{ \Subst{\RenamingPolicySepAsyn{z}}{\RenamingPolicySepAsyn{x}} }\left( \EncodingSepAsyn{P} \right)\\
								& \mid \ReplicateInput{\RenamingPolicySepAsyn{y}}{\sumLock, \senderLock, \RenamingPolicySepAsyn{x}}.\Test{\sumLock}{\Output{\sumLock}{\false} \mid \Out{\senderLock} \mid \EncodingSepAsyn{P}}{\Output{\sumLock}{\false}}\\
								& \mid \prod_{i \in \indexSet_1, i \neq j} \EncodingSepAsyn{\guard_i.Q_i} \mid \EncodingSepAsyn{Q_j} \Big) = T
							\end{aligned}}
					\end{align*}
					By Corollary \ref{col:encodingSubstitutions}, $ \Set{ \Subst{\RenamingPolicySepAsyn{z}}{\RenamingPolicySepAsyn{x}} } \EncodingSepAsyn{P} \equivAlpha \EncodingSepAsyn{\Set{ \Subst{z}{x} } P}  $. To show that $ T \transBarbCongSepAsyn \EncodingMixAsyn{S'} $ we stepwise reduce $ T $ by ignoring junk. Since $ \sumLock, \senderLock \notin \FreeNames{\EncodingMixAsyn{\Set{ \Subst{z}{x} } P}} \cup \FreeNames{\EncodingMixAsyn{Q_j}} $, we can reorder the term according to the restriction on $ \sumLock $ and the restriction on $ \senderLock $ can be omitted. By Lemma \ref{lem:junkRemainsOfSumsSepAsyn}, then $ \RestrictedTerm{\sumLock}{\left( \prod_{i \in \indexSet, i \neq j} \EncodingSepAsyn{\guard_i.Q_i} \mid \Output{\sumLock}{\false} \right)} $ is junk. Note that, by Lemma \ref{lem:transBarbCongIncludesSC}, the relation $ \transBarbCongSepAsyn $ includes structural congruence. Thus, by Lemma \ref{lem:removeJunk}, we conclude $ T \transBarbCongSepAsyn \EncodingSepAsyn{S'} $.
			\end{description}
		\item[Induction Hypothesis:] $ S_1 \step S_1' \text{ implies } \exists T_1 \in \piAsynProc \logdot \EncodingSepAsyn{S_1} \steps T_1 \wedge T_1 \transBarbCongSepAsyn \EncodingSepAsyn{S_1'} $
		\item[Induction Step:] We have to consider the remaining three rules \textsc{Par}, \textsc{Res}, and \textsc{Cong} of Figure \ref{fig:concurrentReductionSemantics}.
			\begin{description}
				\item[Case of Rule] \textsc{Par}\textbf{:} Then $ S = S_1 \mid S_2 $ for some $ S_1, S_2 \in \piSepProc $, $ S_1 \step S_1' $, and $ S' = S_1' \mid S_2 $. By the induction hypothesis there is some $ T_1 \in \piAsynProc $ such that $ \EncodingSepAsyn{S_1} \steps T_1 $ and $ T_1 \transBarbCongSepAsyn \EncodingSepAsyn{S_1'} $. Since the encoding of the parallel operator is \clean, i.e., $ \EncodingSepAsyn{S} = \EncodingSepAsyn{S_1} \mid \EncodingSepAsyn{S_2} $ and $ \EncodingSepAsyn{S'} = \EncodingSepAsyn{S_1'} \mid \EncodingSepAsyn{S_2} $, we can apply rule \textsc{Par} to conclude from $ \EncodingSepAsyn{S_1} \steps T_1 $ to $ \EncodingSepAsyn{S} \steps T_1 \mid \EncodingSepAsyn{S_2} = T $. By Definition \ref{def:transBarbCong}, $ T_1 \transBarbCongSepAsyn \EncodingSepAsyn{S_1'} $ implies $ \Context{}{}{T_1} \transBarbBisimSepAsyn \Context{}{}{\EncodingSepAsyn{S_1'}} $ for all contexts $ \Context{}{}{\hole} \in \piAsynProc \to \piAsynProc $ such that $ \Context{}{}{\EncodingSepAsyn{\nullTerm}} \in \targetTermsSepAsyn $. Since $ \EncodingSepAsyn{\nullTerm} \mid \EncodingSepAsyn{S_2} \in \targetTermsSepAsyn $, the quantification over $ \context $ includes all contexts $ \context $ such that $ \Context{}{}{\hole} = \Context{'}{}{\hole \mid \EncodingSepAsyn{S_2}} $. Because of that, we have $ \Context{'}{}{T_1 \mid \EncodingSepAsyn{S_2}} \transBarbBisimSepAsyn \Context{'}{}{\EncodingSepAsyn{S_1'} \mid \EncodingSepAsyn{S_2}} $ for all contexts $ \Context{'}{}{\hole} \in \piAsynProc \to \piAsynProc $ such that $ \Context{'}{}{\EncodingSepAsyn{\nullTerm}} \in \targetTermsSepAsyn $. By Definition \ref{def:transBarbCong}, we conclude $ T \transBarbCongSepAsyn \EncodingSepAsyn{S'} $.
				\item[Case of Rule] \textsc{Res}\textbf{:} Then $ S = \RestrictedTerm{x}{S_1} $ for some $ x \in \names $ and some $ S_1 \in \piSepProc $, $ S_1 \step S_1' $, and $ S' = \RestrictedTerm{x}{S_1'} $. By the induction hypothesis there is some $ T_1 \in \piAsynProc $ such that $ \EncodingSepAsyn{S_1} \steps T_1 $ and $ T_1 \transBarbCongSepAsyn \EncodingSepAsyn{S_1'} $. Since the encoding of restriction is \clean, i.e., $ \EncodingSepAsyn{S} = \RestrictedTerm{\RenamingPolicySepAsyn{x}}{\EncodingSepAsyn{S_1}} $ and $ \EncodingSepAsyn{S'} = \RestrictedTerm{\RenamingPolicySepAsyn{x}}{\EncodingSepAsyn{S_1'}} $, we can apply rule \textsc{Res} to conclude from $ \EncodingSepAsyn{S_1} \steps T_1 $ to $ \EncodingSepAsyn{S} \steps \RestrictedTerm{\RenamingPolicySepAsyn{x}}{T_1} = T $. By Definition \ref{def:transBarbCong}, $ T_1 \transBarbCongSepAsyn \EncodingSepAsyn{S_1'} $ implies $ \Context{}{}{T_1} \transBarbBisimSepAsyn \Context{}{}{\EncodingSepAsyn{S_1'}} $ for all contexts $ \Context{}{}{\hole} \in \piAsynProc \to \piAsynProc $ such that $ \Context{}{}{\EncodingSepAsyn{\nullTerm}} \in \targetTermsSepAsyn $. Since $ \RestrictedTerm{\RenamingPolicySepAsyn{x}}{\EncodingSepAsyn{\nullTerm}} \in \targetTermsSepAsyn $, the quantification over $ \context $ includes all contexts $ \context $ such that $ \Context{}{}{\hole} = \Context{'}{}{\RestrictedTerm{\RenamingPolicySepAsyn{x}}{\hole}} $. Because of that, we have $ \Context{'}{}{\RestrictedTerm{\RenamingPolicySepAsyn{x}}{T_1}} \transBarbBisimSepAsyn \Context{'}{}{\RestrictedTerm{\RenamingPolicySepAsyn{x}}{\EncodingSepAsyn{S_1'}}} $ for all contexts $ \Context{'}{}{\hole} \in \piAsynProc \to \piAsynProc $ such that $ \Context{'}{}{\EncodingSepAsyn{\nullTerm}} \in \targetTermsSepAsyn $. By Definition \ref{def:transBarbCong}, we conclude $ T \transBarbCongSepAsyn \EncodingSepAsyn{S'} $.
				\item[Case of Rule] \textsc{Cong}\textbf{:} Then $ S \equiv S_1 $ for some $ S_1 \in \piSepProc $, $ S_1 \step S_1' $, and $ S_1' \equiv S' $. By Lemma \ref{lem:preservesSCModuloTransBarbBisimSepAsyn}, the encoding $ \encodingSepAsyn $ preserves structural congruence of source terms modulo $ \transBarbCongSepAsyn $. So $ S \equiv S_1 $ and $ S_1' \equiv S' $ implies $ \EncodingSepAsyn{S} \transBarbCongSepAsyn \EncodingSepAsyn{S_1} $ and $ \EncodingSepAsyn{S_1'} \transBarbCongSepAsyn \EncodingSepAsyn{S'} $. By Definition \ref{def:transBarbCong}, for all contexts $ \Context{}{}{\hole} \in \piAsynProc \to \piAsynProc $ such that $ \Context{}{}{\EncodingSepAsyn{\nullTerm}} \in \targetTermsSepAsyn $ we have $ \Context{}{}{\EncodingSepAsyn{S}} \transBarbBisimSepAsyn \Context{}{}{\EncodingSepAsyn{S_1}} $, i.e., especially $ \EncodingSepAsyn{S} \transBarbBisimSepAsyn \EncodingSepAsyn{S_1} $. Thus, by Definition \ref{def:transBarbBisim}, for each sequence $ \EncodingSepAsyn{S} \steps T $ there is a sequence $ \EncodingSepAsyn{S_1} \steps T_1 $ for some $ T_1 \in \piAsynProc $ such that $ T \transBarbBisimSepAsyn T_1 $. The same holds for all Contexts $ \context $, i.e., since $ \Context{}{}{\EncodingSepAsyn{S}} \transBarbBisimSepAsyn \Context{}{}{\EncodingSepAsyn{S_1}} $, for each sequence $ \Context{}{}{\EncodingSepAsyn{S}} \steps \Context{}{}{T} $ there is a sequence $ \Context{}{}{\EncodingSepAsyn{S_1}} \steps \Context{}{}{T_1} $ for some $ T_1 \in \piAsynProc $ such that $ \Context{}{}{T} \transBarbBisimSepAsyn \Context{}{}{T_1} $. So, by Definition \ref{def:transBarbCong}, $ T \transBarbCongSepAsyn T_1 $. By the induction hypothesis $ T_1 \transBarbCongSepAsyn \EncodingSepAsyn{S_1'} $. Since, by Lemma \ref{lem:transBarbCongIsEquivalence}, $ \transBarbCongSepAsyn $ is an equivalence, $ T \transBarbCongSepAsyn T_1 $, $ T_1 \transBarbCongSepAsyn \EncodingSepAsyn{S_1'} $, and $ \EncodingSepAsyn{S_1'} \transBarbCongSepAsyn \EncodingSepAsyn{S'} $ implies $ T \transBarbCongSepAsyn \EncodingSepAsyn{S'} $.
			\end{description}
	\end{description}
	\qed
\end{proof}

\begin{lemma}[Operational Completeness] \label{lem:operationalCompletenessMixAsyn}
	The encoding $ \encodingMixAsyn $ fulfils operational completeness.
\end{lemma}

\begin{proof}
	By Definition \ref{def:operationalCorrespondence} it suffice to show that:
	\begin{align*}
		\forall S, S' \in \piMixProc \logdot S \step S' \text{ implies } \exists T \in \piAsynProc \logdot \EncodingMixAsyn{S} \steps T \wedge T \transBarbCongMixAsynB \EncodingMixAsyn{S'}
	\end{align*}
	The lemma then holds by induction over the number of steps in $ S \steps S' $. To prove the condition above, we perform an induction over the proof tree that leads to the step $ S \step S' $.
	\begin{description}
		\item[Base Case:] By the rules in Figure \ref{fig:concurrentReductionSemantics} each step on $ S $ is based either on Rule $ \textsc{Tau}_{\indexMix, \indexSep} $, $ \textsc{Com}_{\indexMix, \indexSep}  $, or $ \textsc{Rep}_{\indexMix, \indexSep}  $.
			\begin{description}
				\item[Case of Rule] $ \textsc{Tau}_{\indexMix, \indexSep} $\textbf{:} In this case $ S $ is a single sum, of which one summand is guarded by $ \tau $, and $ S' $ is the continuation of this $ \tau $ guarded summand, i.e., there are some finite index set $ I $, some guards $ \pi_i $, and some processes $ P_i \in \piMixProc $ such that $ S = \sum_{i \in \indexSet} \guard_i.P_i $ with $ \guard_j = \tau $ for some $ j \in \indexSet $ and $ S' = P_j $. The corresponding encodings are given by the following terms:
					\begin{align*}
						\EncodingMixAsyn{S} & \equiv \RestrictedTerm{\sumLock}{\left( \Output{\sumLock}{\true} \mid \prod_{i \in \indexSet, i \neq j} \EncodingMixAsyn{\guard_i.P_i} \mid \Test{\sumLock}{\Output{\sumLock}{\false} \mid \EncodingMixAsyn{P_j}}{\Output{\sumLock}{\false}} \right)}\\
						\EncodingMixAsyn{S'} & = \EncodingMixAsyn{P_j}
					\end{align*}
					We observe that $ \EncodingMixAsyn{S} $ can \simulate the step $ S \step S' $ by reducing the test-statement in the encoding of the $ j $'s summand, i.e., by
					\begin{align*}
						\EncodingMixAsyn{S} \step^2 \RestrictedTerm{\sumLock}{\left( \prod_{i \in \indexSet, i \neq j} \EncodingMixAsyn{\guard_i.P_i} \mid \Output{\sumLock}{\false} \mid \EncodingMixAsyn{P_j} \right)} = T.
					\end{align*}
					We observe that $ T \equiv \RestrictedTerm{\sumLock}{\left( \prod_{i \in \indexSet, i \neq j} \EncodingMixAsyn{\guard_i.P_i} \mid \Output{\sumLock}{\false} \right)} \mid \EncodingMixAsyn{P_j} $, since, due to the renaming policy $ \renamingPolicyMixAsyn $, the name $ \sumLock $ is not free in $ \EncodingMixAsyn{P_j} $. So the term $ \RestrictedTerm{\sumLock}{\left( \prod_{i \in \indexSet, i \neq j} \EncodingMixAsyn{\guard_i.P_i} \mid \Output{\sumLock}{\false} \right)} $ is leftover, which is by Lemma \ref{lem:junkRemainsOfSumsMixAsyn} junk. By Lemma \ref{lem:removeJunk}, we conclude that $ T \transBarbCongMixAsynB \EncodingMixAsyn{S'} $.
				\item[Case of Rule] $ \textsc{Com}_{\indexMix, \indexSep} $\textbf{:} Here $ S $ is a parallel composition of two sums and $ S' $ is the parallel composition of the continuations of an input guarded summand of the first and a matching output guarded summand of the second sum, i.e., there are two finite index sets $ \indexSet_1, \indexSet_2 $, some guards $ \guard_i $, and some processes $ P_i, Q_i \in \piMixProc $ such that $ S = \sum_{i \in \indexSet_1} \guard_i.P_i \mid \sum_{i \in \indexSet_2} \guard_i.Q_i $ with $ \guard_{j_1} = \Input{y}{x} $ and $ \guard_{j_2} = \Output{y}{z} $ for some $ j_1 \in \indexSet_1 $, some $ j_2 \in \indexSet_2 $, and $ x, y, z \in \names $, and $ S' = \Set{ \Subst{z}{x} } P_{j_1} \mid Q_{j_2} $. Unfortunately the encodings of these terms are rather long:
					\begin{align*}
						\EncodingMixAsyn{S} & \equiv \begin{aligned}[t]
								& \RestrictedTerm{\matchingCoordinatorOut, \matchingCoordinatorIn, \coordinatorUpOut, \coordinatorUpIn, \coordinatorMatchingOut, \coordinatorMatchingIn, \matchingUpOut, \matchingUpIn}{\big(\\
								& \hspace*{1em} \RestrictedTerm{\parallelChannelOut, \parallelChannelIn}{\big( \begin{aligned}[t]
										& \RestrictedTerm{\sumLock}{\big( \begin{aligned}[t]
												& \Output{\sumLock}{\true} \mid \prod_{i \in \indexSet_1, i \neq j_1} \EncodingMixAsyn{\guard_i.P_i}\\
												& \mid \RestrictedTerm{\receiverLock}{\big( \begin{aligned}[t]
														& \Output{\parallelChannelIn}{\RenamingPolicyMixAsyn{y}, \sumLock, \receiverLock} \mid \ReplicateInput{\receiverLock}{\sumLock_1, \sumLock_2, -, \senderLock, \RenamingPolicyMixAsyn{x}}.\\
														& \hspace*{1em} \BigTest{\sumLock_1}{\BigTest{\sumLock_2}{\Output{\sumLock_1}{\false} \mid \Output{\sumLock_2}{\false} \mid \Out{\senderLock} \mid \EncodingMixAsyn{P_{j_1}}}{\Output{\sumLock_1}{\true} \mid \Output{\sumLock_2}{\false}}}{\Output{\sumLock_1}{\false} \big) \big)}
													\end{aligned}}
											\end{aligned}}\\
										& \mid \processLeftOutputRequests \mid \processLeftInputRequests \big)
									\end{aligned}}\\
								& \hspace*{1em} \mid \RestrictedTerm{\parallelChannelOut, \parallelChannelIn}{\big( \begin{aligned}[t]
										& \RestrictedTerm{\sumLock}{\big( \begin{aligned}[t]
												& \Output{\sumLock}{\true} \mid \prod_{i \in \indexSet_2, i \neq j_2} \EncodingMixAsyn{\guard_i.Q_i} \mid \RestrictedTerm{\senderLock}{\left( \Output{\parallelChannelOut}{\RenamingPolicyMixAsyn{y}, \sumLock, \senderLock, \RenamingPolicyMixAsyn{z}} \mid \In{\senderLock}.\EncodingMixAsyn{Q_{j_2}} \right)} \big)
											\end{aligned}}\\
										& \mid \processRightOutputRequests \mid \processRightInputRequests \big)
									\end{aligned}}\\
								& \hspace*{1em} \mid \pushRequests \big)}
							\end{aligned}\\
						\EncodingMixAsyn{S'} & = \begin{aligned}[t]
								& \RestrictedTerm{\matchingCoordinatorOut, \matchingCoordinatorIn, \coordinatorUpOut, \coordinatorUpIn, \coordinatorMatchingOut, \coordinatorMatchingIn, \matchingUpOut, \matchingUpIn}{\big(\\
								& \hspace*{1em} \RestrictedTerm{\parallelChannelOut, \parallelChannelIn}{\left( \EncodingMixAsyn{\Set{ \Subst{z}{x} } P_{j_1}} \mid \processLeftOutputRequests \mid \processLeftInputRequests \right)}\\
								& \hspace*{1em} \mid \RestrictedTerm{\parallelChannelOut, \parallelChannelIn}{\left( \EncodingMixAsyn{Q_{j_2}} \mid \processRightOutputRequests \mid \processRightInputRequests \right)}\\
								& \hspace*{1em} \mid \pushRequests \big)}
							\end{aligned}
					\end{align*}
					To \simulate the source term step $ S \step S' $, the endings of the two sums in $ S $ have to interact with the encoding of the parallel operator between them. At first the input and output register themselves to the encoding of the parallel operator by pushing requests. These requests are then combined and a test on the respective sum locks\footnote{In order to avoid a deadlock caused by multiple simultaneous such tests on sum locks, the sum locks are ordered by ensuring that always the left one is checked first.} is induced by providing an output on the receiver lock. At least the test-statement is reduced to complete the \simulation of the source term step.
					\begin{align*}
						\EncodingMixAsyn{S} & \step^3 \begin{aligned}[t]
								& \RestrictedTerm{\matchingCoordinatorOut, \matchingCoordinatorIn, \coordinatorUpOut, \coordinatorUpIn, \coordinatorMatchingOut, \coordinatorMatchingIn, \matchingUpOut, \matchingUpIn, \sumLock_a, \sumLock_b, \receiverLock, \senderLock}{\big(\\
								& \hspace*{1em} \RestrictedTerm{\parallelChannelOut, \parallelChannelIn}{\big( \begin{aligned}[t]
										& \Output{\sumLock_a}{\true} \mid \Set{ \Subst{\sumLock_a}{\sumLock} }\left( \prod_{i \in \indexSet_1, i \neq j_1} \EncodingMixAsyn{\guard_i.P_i} \right) \mid \ReplicateInput{\receiverLock}{\sumLock_1, \sumLock_2, -, \senderLock, \RenamingPolicyMixAsyn{x}}.\\
										& \hspace*{1em} \Test{\sumLock_1}{\Test{\sumLock_2}{\Output{\sumLock_1}{\false} \mid \Output{\sumLock_2}{\false} \mid \Out{\senderLock} \mid \EncodingMixAsyn{P_{j_1}}}{\Output{\sumLock_1}{\true} \mid \Output{\sumLock_2}{\false}}}{\Output{\sumLock_1}{\false}}\\
										& \mid \processLeftOutputRequests \mid \processLeftInputRequests  \mid \Output{\matchingCoordinatorIn}{\RenamingPolicyMixAsyn{y}, \sumLock_a, \receiverLock} \mid \Output{\coordinatorUpIn}{\RenamingPolicyMixAsyn{y}, \sumLock_a, \receiverLock} \big)
									\end{aligned}}\\
								& \hspace*{1em} \mid \RestrictedTerm{\parallelChannelOut, \parallelChannelIn}{\big( \begin{aligned}[t]
										& \Output{\sumLock_b}{\true} \mid \Set{ \Subst{\sumLock_b}{\sumLock} }\left(\prod_{i \in \indexSet_2, i \neq j_2} \EncodingMixAsyn{\guard_i.Q_i} \right) \mid \In{\senderLock}.\EncodingMixAsyn{Q_{j_2}}\\
										& \mid \RestrictedTerm{\matchingUpIn}{\big( \begin{aligned}[t]
												& \ReplicateInput{\matchingCoordinatorIn}{y', \sumLock_r, \receiverLock}.\left( \Match{y'}{\RenamingPolicyMixAsyn{y}}\Output{\receiverLock}{\sumLock_r, \sumLock_b, \sumLock_b, \senderLock, \RenamingPolicyMixAsyn{z}} \mid \Output{\matchingUpIn}{y', \sumLock_r, \receiverLock} \right)\\
												& \mid \RestrictedTerm{\matchingCoordinatorIn}{\left( \Forward{\matchingUpIn}{\matchingCoordinatorIn} \mid \processRightOutputRequests \right)} \big)
											\end{aligned}}\\
										 & \mid \Output{\coordinatorUpOut}{\RenamingPolicyMixAsyn{y}, \sumLock_b, \senderLock, \RenamingPolicyMixAsyn{z}} \mid \processRightInputRequests \big)
									\end{aligned}}\\
								& \hspace*{1em} \mid \pushRequests \big)}
							\end{aligned}
					\end{align*}
					\begin{align*}
						\hspace*{1em} & \step \begin{aligned}[t]
								& \RestrictedTerm{\matchingCoordinatorOut, \matchingCoordinatorIn, \coordinatorUpOut, \coordinatorUpIn, \coordinatorMatchingOut, \coordinatorMatchingIn, \matchingUpOut, \matchingUpIn, \sumLock_a, \sumLock_b, \receiverLock, \senderLock}{\big(\\
								& \hspace*{1em} \RestrictedTerm{\parallelChannelOut, \parallelChannelIn}{\big( \begin{aligned}[t]
										& \Output{\sumLock_a}{\true} \mid \Set{ \Subst{\sumLock_a}{\sumLock} }\left( \prod_{i \in \indexSet_1, i \neq j_1} \EncodingMixAsyn{\guard_i.P_i} \right) \mid \ReplicateInput{\receiverLock}{\sumLock_1, \sumLock_2, -, \senderLock, \RenamingPolicyMixAsyn{x}}.\\
										& \hspace*{1em} \Test{\sumLock_1}{\Test{\sumLock_2}{\Output{\sumLock_1}{\false} \mid \Output{\sumLock_2}{\false} \mid \Out{\senderLock} \mid \EncodingMixAsyn{P_{j_1}}}{\Output{\sumLock_1}{\true} \mid \Output{\sumLock_2}{\false}}}{\Output{\sumLock_1}{\false}}\\
										& \mid \processLeftOutputRequests \mid \processLeftInputRequests \mid \Output{\coordinatorUpIn}{\RenamingPolicyMixAsyn{y}, \sumLock_a, \receiverLock} \big)
									\end{aligned}}\\
								& \hspace*{1em} \mid \RestrictedTerm{\parallelChannelOut, \parallelChannelIn}{\big( \begin{aligned}[t]
										& \Output{\sumLock_b}{\true} \mid \Set{ \Subst{\sumLock_b}{\sumLock} }\left( \prod_{i \in \indexSet_2, i \neq j_2} \EncodingMixAsyn{\guard_i.Q_i} \right) \mid \In{\senderLock}.\EncodingMixAsyn{Q_{j_2}}\\
										& \mid \RestrictedTerm{\matchingUpIn}{\big( \begin{aligned}[t]
												& \ReplicateInput{\matchingCoordinatorIn}{y', \sumLock_r, \receiverLock}.\left( \Match{y'}{\RenamingPolicyMixAsyn{y}}\Output{\receiverLock}{\sumLock_r, \sumLock_b, \sumLock_b, \senderLock, \RenamingPolicyMixAsyn{z}} \mid \Output{\matchingUpIn}{y', \sumLock_r, \receiverLock} \right)\\
												& \mid \Output{\receiverLock}{\sumLock_a, \sumLock_b, \sumLock_b, \senderLock, \RenamingPolicyMixAsyn{z}} \mid \Output{\matchingUpIn}{\RenamingPolicyMixAsyn{y}, \sumLock_a, \receiverLock}\\
												& \mid \RestrictedTerm{\matchingCoordinatorIn}{\left( \Forward{\matchingUpIn}{\matchingCoordinatorIn} \mid \processRightOutputRequests \right)} \big)
											\end{aligned}}\\
										 & \mid \Output{\coordinatorUpOut}{\RenamingPolicyMixAsyn{y}, \sumLock_b, \senderLock, \RenamingPolicyMixAsyn{z}} \mid \processRightInputRequests \big)
									\end{aligned}}\\
								& \hspace*{1em} \mid \pushRequests \big)}
							\end{aligned}\\
						& \step^6 \begin{aligned}[t]
								& \RestrictedTerm{\matchingCoordinatorOut, \matchingCoordinatorIn, \coordinatorUpOut, \coordinatorUpIn, \coordinatorMatchingOut, \coordinatorMatchingIn, \matchingUpOut, \matchingUpIn, \sumLock_a, \sumLock_b, \receiverLock, \senderLock}{\big(\\
								& \hspace*{1em} \RestrictedTerm{\parallelChannelOut, \parallelChannelIn}{\big( \begin{aligned}[t]
										& \Set{ \Subst{\sumLock_a}{\sumLock} }\left( \prod_{i \in \indexSet_1, i \neq j_1} \EncodingMixAsyn{\guard_i.P_i} \right) \mid \Output{\sumLock_a}{\false} \mid \Output{\sumLock_b}{\false} \mid \Set{ \Subst{\RenamingPolicyMixAsyn{z}}{\RenamingPolicyMixAsyn{x}} } \EncodingMixAsyn{P_{j_1}}\\
										& \mid \ReplicateInput{\receiverLock}{\sumLock_1, \sumLock_2, -, \senderLock, \RenamingPolicyMixAsyn{x}}.\BigTest{\sumLock_1}{\BigTest{\sumLock_2}{\Output{\sumLock_1}{\false} \mid \Output{\sumLock_2}{\false} \mid \Out{\senderLock} \mid \EncodingMixAsyn{P_{j_1}}}{\Output{\sumLock_1}{\true} \mid \Output{\sumLock_2}{\false}}}{\Output{\sumLock_1}{\false}}\\
										& \mid \processLeftOutputRequests \mid \processLeftInputRequests \mid \Output{\coordinatorUpIn}{\RenamingPolicyMixAsyn{y}, \sumLock_a, \receiverLock} \big)
									\end{aligned}}\\
								& \hspace*{1em} \mid \RestrictedTerm{\parallelChannelOut, \parallelChannelIn}{\big( \begin{aligned}[t]
										& \Set{ \Subst{\sumLock_b}{\sumLock} }\left( \prod_{i \in \indexSet_2, i \neq j_2} \EncodingMixAsyn{\guard_i.Q_i} \right) \mid \EncodingMixAsyn{Q_{j_2}}\\
										& \mid \RestrictedTerm{\matchingUpIn}{\big( \begin{aligned}[t]
												& \ReplicateInput{\matchingCoordinatorIn}{y', \sumLock_r, \receiverLock}.\left( \Match{y'}{\RenamingPolicyMixAsyn{y}}\Output{\receiverLock}{\sumLock_r, \sumLock_b, \sumLock_b, \senderLock, \RenamingPolicyMixAsyn{z}} \mid \Output{\matchingUpIn}{y', \sumLock_r, \receiverLock} \right)\\
												& \mid \Output{\matchingUpIn}{\RenamingPolicyMixAsyn{y}, \sumLock_a, \receiverLock} \mid \RestrictedTerm{\matchingCoordinatorIn}{\left( \Forward{\matchingUpIn}{\matchingCoordinatorIn} \mid \processRightOutputRequests \right)} \big)
											\end{aligned}}\\
										 & \mid \Output{\coordinatorUpOut}{\RenamingPolicyMixAsyn{y}, \sumLock_b, \senderLock, \RenamingPolicyMixAsyn{z}} \mid \processRightInputRequests \big)
									\end{aligned}}\\
								& \hspace*{1em} \mid \pushRequests \big) = T}
							\end{aligned}
					\end{align*}
					By Corollary \ref{col:encodingSubstitutions}, $ \Set{ \Subst{\RenamingPolicyMixAsyn{z}}{\RenamingPolicyMixAsyn{x}} } \EncodingMixAsyn{P_{j_1}} \equivAlpha \EncodingMixAsyn{\Set{ \Subst{z}{x} } P_{j_1}}  $. To show that $ T \transBarbCongMixAsynB \EncodingMixAsyn{S'} $, we stepwise reduce $ T $ by ignoring junk. By Lemma \ref{lem:junkRequestsOnFalseSumLocks}, we can ignore the requests $ \Output{\coordinatorUpIn}{\RenamingPolicyMixAsyn{y}, \sumLock_a, \receiverLock} $, $ \Output{\matchingUpIn}{\RenamingPolicyMixAsyn{y}, \sumLock_a, \receiverLock} $, and $ \Output{\coordinatorUpOut}{\RenamingPolicyMixAsyn{y}, \sumLock_b, \senderLock, \RenamingPolicyMixAsyn{z}} $. Next, by Lemma \ref{lem:junkReceiverMixAsyn}, we can ignore the term
					\begin{align*}
						\ReplicateInput{\receiverLock}{\sumLock_1, \sumLock_2, -, \senderLock, \RenamingPolicyMixAsyn{x}}.\BigTest{\sumLock_1}{\Test{\sumLock_2}{\Output{\sumLock_1}{\false} \mid \Output{\sumLock_2}{\false} \mid \Out{\senderLock} \mid \EncodingMixAsyn{P_{j_1}}}{\Output{\sumLock_1}{\true} \mid \Output{\sumLock_2}{\false}}}{\Output{\sumLock_1}{\false}.}
					\end{align*}
					And, by Lemma \ref{lem:junkInduceTest}, we can also ignore $ \Match{y'}{\RenamingPolicyMixAsyn{y}}\Output{\receiverLock}{\sumLock_r, \sumLock_b, \sumLock_b, \senderLock, \RenamingPolicyMixAsyn{z}} $, so
					\begin{align*}
						\ReplicateInput{\matchingCoordinatorIn}{y', \sumLock_r, \receiverLock}.\left( \Match{y'}{\RenamingPolicyMixAsyn{y}}\Output{\receiverLock}{\sumLock_r, \sumLock_b, \sumLock_b, \senderLock, \RenamingPolicyMixAsyn{z}} \mid \Output{\matchingUpIn}{y', \sumLock_r, \receiverLock} \right)
					\end{align*}
					becomes $ \Forward{\matchingCoordinatorIn}{\matchingUpIn} $. Note that this forwarder and the following forwarder $ \Forward{\matchingUpIn}{\matchingCoordinatorIn} $ for an other instance of $ \matchingCoordinatorIn $ may be necessary to \simulate further source term steps, but, since they perform only invisible steps, they do not influence the state of $ T $ modulo $ \transBarbCongMixAsynB $ in comparison to a fresh chain of right requests as in $ \EncodingMixAsyn{S'} $. At least, since $ \sumLock_a, \sumLock_b, \receiverLock, \senderLock \notin \FreeNames{\EncodingMixAsyn{P_{j_1}}} \cup \FreeNames{\EncodingMixAsyn{P_{j_2}}} $, we can reorder the term according to the restrictions on $ \sumLock_a, \sumLock_b, \receiverLock $ and the restriction on $ \senderLock $ can be omitted. By Lemma \ref{lem:junkRemainsOfSumsMixAsyn}, then $ \RestrictedTerm{\sumLock}{\left( \prod_{i \in \indexSet_1, i \neq j_1} \EncodingMixAsyn{\guard_i.P_i} \mid \Output{\sumLock}{\false} \right)} $ and $ \RestrictedTerm{\sumLock}{\left( \prod_{i \in \indexSet_1, i \neq j_2} \EncodingMixAsyn{\guard_i.Q_i} \mid \Output{\sumLock}{\false} \right)} $ are junk. So, by Lemma \ref{lem:removeJunk}, we conclude $ T \transBarbCongMixAsynB \EncodingMixAsyn{S'} $.
				\item[Case of Rule] $ \textsc{Rep}_{\indexMix, \indexSep} $\textbf{:} Here $ S $ is a parallel composition of a replicated input and a sum and $ S' $ is the parallel composition of the replicated input, the continuation of the replicated input, and the continuation of a matching output guarded summand, i.e., there is a finite index sets $ \indexSet $, some guards $ \guard_i $, and some processes $ P, Q_i \in \piMixProc $ such that $ S = \ReplicateInput{y}{x}.P \mid \sum_{i \in \indexSet} \guard_i.Q_i $ with $ \guard_j = \Output{y}{z} $ for some $ j \in \indexSet $ and $ x, y, z \in \names $, and $ S' = \Set{ \Subst{z}{x} } P \mid Q_j \mid \ReplicateInput{y}{x}.P $. Unfortunately, the encodings of $ S $ and $ S' $ are again long:
					\begin{align*}
						\EncodingMixAsyn{S} & \equiv \begin{aligned}[t]
								& \RestrictedTerm{\matchingCoordinatorOut, \matchingCoordinatorIn, \coordinatorUpOut, \coordinatorUpIn, \coordinatorMatchingOut, \coordinatorMatchingIn, \matchingUpOut, \matchingUpIn}{\big(\\
								& \hspace*{1em} \RestrictedTerm{\parallelChannelOut, \parallelChannelIn}{\big( \begin{aligned}[t]
										& \RestrictedTerm{\sumLock, \receiverLock, \coordinatorRepA, \coordinatorRepB, \matchingReceiverOut, \matchingReceiverIn}{\big(\\
										& \hspace*{1em} \ReplicateInput{\receiverLock}{-, -, \sumLock_s, \senderLock, \RenamingPolicyMixAsyn{x}}.\Test{\sumLock_s}{\Output{\sumLock_s}{\false} \mid \Out{\senderLock} \mid \Output{\coordinatorRepA}{\RenamingPolicyMixAsyn{x}}}{\Output{\sumLock_s}{\false}}\\
										& \hspace*{1em} \Output{\parallelChannelIn}{\RenamingPolicyMixAsyn{y}, \sumLock, \receiverLock} \mid \Output{\matchingReceiverIn}{\RenamingPolicyMixAsyn{y}, \sumLock, \receiverLock} \mid \Output{\sumLock}{\true} \Output{\coordinatorRepB}{\matchingReceiverOut, \matchingReceiverIn}\\
										& \hspace*{1em} \mid \ReplicateInput{\coordinatorRepA}{\RenamingPolicyMixAsyn{x}}.\Input{\coordinatorRepB}{\matchingReceiverOut, \matchingReceiverIn}.\RestrictedTerm{\matchingCoordinatorOut, \matchingCoordinatorIn, \coordinatorUpOut, \coordinatorUpIn, \matchingReceiverUpOut, \matchingReceiverUpIn, \coordinatorMatchingOut, \coordinatorMatchingIn, \matchingUpOut, \matchingUpIn}{\big(\\
										& \hspace*{2em} \pushRequestsIn\\
										& \hspace*{2em} \mid \RestrictedTerm{\parallelChannelOut, \parallelChannelIn}{\left( \EncodingMixAsyn{P} \mid \processRightOutputRequests \mid \processRightInputRequests \right)}\\
										& \hspace*{2em} \mid \RestrictedTerm{\matchingReceiverOut, \matchingReceiverIn}{\left( \Output{\coordinatorRepB}{\matchingReceiverOut \mid \matchingReceiverIn} \mid \pushRequestsOut \right)} \big)} \big)}\\
										& \mid \processLeftOutputRequests \mid \processLeftInputRequests \big)
									\end{aligned}}\\
								& \hspace*{1em} \RestrictedTerm{\parallelChannelOut, \parallelChannelIn}{\big( \begin{aligned}[t]
										& \RestrictedTerm{\sumLock}{\left( \Output{\sumLock}{\true} \mid \prod_{i \in \indexSet_2, i \neq j} \EncodingMixAsyn{\guard_i.Q_i} \mid \RestrictedTerm{\senderLock}{\left( \Output{\parallelChannelOut}{\RenamingPolicyMixAsyn{y}, \sumLock, \senderLock, \RenamingPolicyMixAsyn{z}} \mid \In{\senderLock}.\EncodingMixAsyn{Q_j} \right)} \right)}\\
										& \mid \processRightOutputRequests \mid \processRightInputRequests \big)
									\end{aligned}}\\
								& \hspace*{1em} \pushRequests \big)}
							\end{aligned}\\
						\EncodingMixAsyn{S'} & = \begin{aligned}[t]
								& \RestrictedTerm{\matchingCoordinatorOut, \matchingCoordinatorIn, \coordinatorUpOut, \coordinatorUpIn, \coordinatorMatchingOut, \coordinatorMatchingIn, \matchingUpOut, \matchingUpIn}{\big(\\
								& \hspace*{1em} \RestrictedTerm{\parallelChannelOut, \parallelChannelIn}{\big( \begin{aligned}[t]
										& \RestrictedTerm{\matchingCoordinatorOut, \matchingCoordinatorIn, \coordinatorUpOut, \coordinatorUpIn, \coordinatorMatchingOut, \coordinatorMatchingIn, \matchingUpOut, \matchingUpIn}{\big(\\
										& \hspace*{1em} \RestrictedTerm{\parallelChannelOut, \parallelChannelIn}{\left( \EncodingMixAsyn{\Set{ \Subst{z}{x} }\left( P \right)} \mid \processLeftOutputRequests \mid \processLeftInputRequests \right)}\\
										& \hspace*{1em} \RestrictedTerm{\parallelChannelOut, \parallelChannelIn}{\left( \EncodingMixAsyn{Q_j} \mid \processRightOutputRequests \mid \processRightInputRequests \right)}\\
										& \hspace*{1em} \pushRequests \big)\\
										& \mid \processLeftOutputRequests \mid \processLeftInputRequests \big)
									\end{aligned}}}\\
								& \hspace*{1em} \RestrictedTerm{\parallelChannelOut, \parallelChannelIn}{\left( \EncodingMixAsyn{\ReplicateInput{y}{x}.P} \mid \processRightOutputRequests \mid \processRightInputRequests \right)}\\
								& \hspace*{1em} \pushRequests \big)}
							\end{aligned}
					\end{align*}
					To \simulate the source term step $ S \step S' $, the two subterms of $ \EncodingMixAsyn{S} $ have to interact with the encoding of the parallel operator between them. At first the replicated input and the output register themselves to the encoding of the parallel operator by pushing requests. There the requests are combined and a test on the sum lock of the sender is induced by providing an output on the receiver lock. Next the test-statement is reduced. To complete the \simulation of the source term step, at least the continuation of the replicated input encoding is unguarded and placed within an adoption of the parallel operator encoding.
					{\allowdisplaybreaks
					\begin{align*}
						\EncodingMixAsyn{S} & \step^3 \begin{aligned}[t]
								& \RestrictedTerm{\matchingCoordinatorOut, \matchingCoordinatorIn, \coordinatorUpOut, \coordinatorUpIn, \coordinatorMatchingOut, \coordinatorMatchingIn, \matchingUpOut, \matchingUpIn, \sumLock_a, \sumLock_b, \receiverLock, \senderLock}{\big(\\
								& \hspace*{1em} \RestrictedTerm{\parallelChannelOut, \parallelChannelIn}{\big( \begin{aligned}[t]
										& \RestrictedTerm{\coordinatorRepA, \coordinatorRepB, \matchingReceiverOut, \matchingReceiverIn}{\big(\\
										& \hspace*{1em} \ReplicateInput{\receiverLock}{-, -, \sumLock_s, \senderLock, \RenamingPolicyMixAsyn{x}}.\Test{\sumLock_s}{\Output{\sumLock_s}{\false} \mid \Out{\senderLock} \mid \Output{\coordinatorRepA}{\RenamingPolicyMixAsyn{x}}}{\Output{\sumLock_s}{\false}}\\
										& \hspace*{1em} \mid \Output{\matchingReceiverIn}{\RenamingPolicyMixAsyn{y}, \sumLock_a, \receiverLock} \mid \Output{\sumLock_a}{\true} \mid \Output{\coordinatorRepB}{\matchingReceiverOut, \matchingReceiverIn} \mid \ReplicateInput{\coordinatorRepA}{\RenamingPolicyMixAsyn{x}}.\Input{\coordinatorRepB}{\matchingReceiverOut, \matchingReceiverIn}.\\
										& \hspace*{1.7em} \RestrictedTerm{\matchingCoordinatorOut, \matchingCoordinatorIn, \coordinatorUpOut, \coordinatorUpIn, \matchingReceiverUpOut, \matchingReceiverUpIn, \coordinatorMatchingOut, \coordinatorMatchingIn, \matchingUpOut, \matchingUpIn}{\big( \pushRequestsIn\\
										& \hspace*{2em} \mid \RestrictedTerm{\parallelChannelOut, \parallelChannelIn}{\left( \EncodingMixAsyn{P} \mid \processRightOutputRequests \mid \processRightInputRequests \right)}\\
										& \hspace*{2em} \mid \RestrictedTerm{\matchingReceiverOut, \matchingReceiverIn}{\left( \Output{\coordinatorRepB}{\matchingReceiverOut \mid \matchingReceiverIn} \mid \pushRequestsOut \right)} \big)} \big)}\\
										& \mid \processLeftOutputRequests \mid \processLeftInputRequests \mid \Output{\matchingCoordinatorIn}{\RenamingPolicyMixAsyn{y}, \sumLock_a, \receiverLock} \mid \Output{\coordinatorUpIn}{\RenamingPolicyMixAsyn{y}, \sumLock_a, \receiverLock} \big)
									\end{aligned}}\\
								& \hspace*{1em} \mid \RestrictedTerm{\parallelChannelOut, \parallelChannelIn}{\big( \begin{aligned}[t]
										& \Output{\sumLock_b}{\true} \mid \Set{ \Subst{\sumLock_b}{\sumLock} }\left( \prod_{i \in \indexSet_2, i \neq j} \EncodingMixAsyn{\guard_i.Q_i} \right) \mid \In{\senderLock}.\EncodingMixAsyn{Q_j}\\
										& \mid \RestrictedTerm{\matchingUpIn}{\big( \begin{aligned}[t]
												& \ReplicateInput{\matchingCoordinatorIn}{y', \sumLock_r, \receiverLock}.\left( \Match{y'}{\RenamingPolicyMixAsyn{y}}\Output{\receiverLock}{\sumLock_r, \sumLock_b, \sumLock_b, \senderLock, \RenamingPolicyMixAsyn{z}} \mid \Output{\matchingUpIn}{y', \sumLock_r, \receiverLock} \right)\\
												& \mid \RestrictedTerm{\matchingCoordinatorIn}{\left( \Forward{\matchingUpIn}{\matchingCoordinatorIn} \mid \processRightOutputRequests \right)} \big)
											\end{aligned}}\\
										 & \mid \Output{\coordinatorUpOut}{\RenamingPolicyMixAsyn{y}, \sumLock_b, \senderLock, \RenamingPolicyMixAsyn{z}} \mid \processRightInputRequests \big) \mid \pushRequests \big)
									\end{aligned}}}
							\end{aligned}\\
						& \step \begin{aligned}[t]
								& \RestrictedTerm{\matchingCoordinatorOut, \matchingCoordinatorIn, \coordinatorUpOut, \coordinatorUpIn, \coordinatorMatchingOut, \coordinatorMatchingIn, \matchingUpOut, \matchingUpIn, \sumLock_a, \sumLock_b, \receiverLock, \senderLock}{\big(\\
								& \hspace*{1em} \RestrictedTerm{\parallelChannelOut, \parallelChannelIn}{\big( \begin{aligned}[t]
										& \RestrictedTerm{\coordinatorRepA, \coordinatorRepB, \matchingReceiverOut, \matchingReceiverIn}{\big(\\
										& \hspace*{1em} \ReplicateInput{\receiverLock}{-, -, \sumLock_s, \senderLock, \RenamingPolicyMixAsyn{x}}.\Test{\sumLock_s}{\Output{\sumLock_s}{\false} \mid \Out{\senderLock} \mid \Output{\coordinatorRepA}{\RenamingPolicyMixAsyn{x}}}{\Output{\sumLock_s}{\false}}\\
										& \hspace*{1em} \mid \Output{\matchingReceiverIn}{\RenamingPolicyMixAsyn{y}, \sumLock_a, \receiverLock} \mid \Output{\sumLock_a}{\true} \mid \Output{\coordinatorRepB}{\matchingReceiverOut, \matchingReceiverIn} \mid \ReplicateInput{\coordinatorRepA}{\RenamingPolicyMixAsyn{x}}.\Input{\coordinatorRepB}{\matchingReceiverOut, \matchingReceiverIn}.\\
										& \hspace*{1.7em} \RestrictedTerm{\matchingCoordinatorOut, \matchingCoordinatorIn, \coordinatorUpOut, \coordinatorUpIn, \matchingReceiverUpOut, \matchingReceiverUpIn, \coordinatorMatchingOut, \coordinatorMatchingIn, \matchingUpOut, \matchingUpIn}{\big( \pushRequestsIn\\
										& \hspace*{2em} \mid \RestrictedTerm{\parallelChannelOut, \parallelChannelIn}{\left( \EncodingMixAsyn{P} \mid \processRightOutputRequests \mid \processRightInputRequests \right)}\\
										& \hspace*{2em} \mid \RestrictedTerm{\matchingReceiverOut, \matchingReceiverIn}{\left( \Output{\coordinatorRepB}{\matchingReceiverOut \mid \matchingReceiverIn} \mid \pushRequestsOut \right)} \big)} \big)}\\
										& \mid \processLeftOutputRequests \mid \processLeftInputRequests \mid \Output{\coordinatorUpIn}{\RenamingPolicyMixAsyn{y}, \sumLock_a, \receiverLock} \big)
									\end{aligned}}\\
								& \hspace*{1em} \mid \RestrictedTerm{\parallelChannelOut, \parallelChannelIn}{\big( \begin{aligned}[t]
										& \Output{\sumLock_b}{\true} \mid \Set{ \Subst{\sumLock_b}{\sumLock} }\left( \prod_{i \in \indexSet_2, i \neq j} \EncodingMixAsyn{\guard_i.Q_i} \right) \mid \In{\senderLock}.\EncodingMixAsyn{Q_j}\\
										& \mid \RestrictedTerm{\matchingUpIn}{\big( \begin{aligned}[t]
												& \ReplicateInput{\matchingCoordinatorIn}{y', \sumLock_r, \receiverLock}.\left( \Match{y'}{\RenamingPolicyMixAsyn{y}}\Output{\receiverLock}{\sumLock_r, \sumLock_b, \sumLock_b, \senderLock, \RenamingPolicyMixAsyn{z}} \mid \Output{\matchingUpIn}{y', \sumLock_r, \receiverLock} \right)\\
												& \mid \Output{\receiverLock}{\sumLock_a, \sumLock_b, \sumLock_b, \senderLock, \RenamingPolicyMixAsyn{z}} \mid \Output{\matchingUpIn}{\RenamingPolicyMixAsyn{y}, \sumLock_a, \receiverLock}\\
												& \mid \RestrictedTerm{\matchingCoordinatorIn}{\left( \Forward{\matchingUpIn}{\matchingCoordinatorIn} \mid \processRightOutputRequests \right)} \big)
											\end{aligned}}\\
										 & \mid \Output{\coordinatorUpOut}{\RenamingPolicyMixAsyn{y}, \sumLock_b, \senderLock, \RenamingPolicyMixAsyn{z}} \mid \processRightInputRequests \big) \mid \pushRequests \big)
									\end{aligned}}}
							\end{aligned}\\
						& \step^4 \begin{aligned}[t]
								& \RestrictedTerm{\matchingCoordinatorOut, \matchingCoordinatorIn, \coordinatorUpOut, \coordinatorUpIn, \coordinatorMatchingOut, \coordinatorMatchingIn, \matchingUpOut, \matchingUpIn, \sumLock_a, \sumLock_b, \receiverLock, \senderLock}{\big(\\
								& \hspace*{1em} \RestrictedTerm{\parallelChannelOut, \parallelChannelIn}{\big( \begin{aligned}[t]
										& \RestrictedTerm{\coordinatorRepA, \coordinatorRepB, \matchingReceiverOut, \matchingReceiverIn}{\big( \Output{\sumLock_b}{\false} \mid \Output{\coordinatorRepA}{\RenamingPolicyMixAsyn{z}}\\
										& \hspace*{1em} \mid \ReplicateInput{\receiverLock}{-, -, \sumLock_s, \senderLock, \RenamingPolicyMixAsyn{x}}.\Test{\sumLock_s}{\Output{\sumLock_s}{\false} \mid \Out{\senderLock} \mid \Output{\coordinatorRepA}{\RenamingPolicyMixAsyn{x}}}{\Output{\sumLock_s}{\false}}\\
										& \hspace*{1em} \mid \Output{\matchingReceiverIn}{\RenamingPolicyMixAsyn{y}, \sumLock_a, \receiverLock} \mid \Output{\sumLock_a}{\true} \mid \Output{\coordinatorRepB}{\matchingReceiverOut, \matchingReceiverIn} \mid \ReplicateInput{\coordinatorRepA}{\RenamingPolicyMixAsyn{x}}.\Input{\coordinatorRepB}{\matchingReceiverOut, \matchingReceiverIn}.\\
										& \hspace*{1.7em} \RestrictedTerm{\matchingCoordinatorOut, \matchingCoordinatorIn, \coordinatorUpOut, \coordinatorUpIn, \matchingReceiverUpOut, \matchingReceiverUpIn, \coordinatorMatchingOut, \coordinatorMatchingIn, \matchingUpOut, \matchingUpIn}{\big( \pushRequestsIn\\
										& \hspace*{2em} \mid \RestrictedTerm{\parallelChannelOut, \parallelChannelIn}{\left( \EncodingMixAsyn{P} \mid \processRightOutputRequests \mid \processRightInputRequests \right)}\\
										& \hspace*{2em} \mid \RestrictedTerm{\matchingReceiverOut, \matchingReceiverIn}{\left( \Output{\coordinatorRepB}{\matchingReceiverOut \mid \matchingReceiverIn} \mid \pushRequestsOut \right)} \big)} \big)}\\
										& \mid \processLeftOutputRequests \mid \processLeftInputRequests \mid \Output{\coordinatorUpIn}{\RenamingPolicyMixAsyn{y}, \sumLock_a, \receiverLock} \big)
									\end{aligned}}\\
								& \hspace*{1em} \mid \RestrictedTerm{\parallelChannelOut, \parallelChannelIn}{\big( \begin{aligned}[t]
										& \Set{ \Subst{\sumLock_b}{\sumLock} }\left( \prod_{i \in \indexSet_2, i \neq j} \EncodingMixAsyn{\guard_i.Q_i} \right) \mid \EncodingMixAsyn{Q_j}\\
										& \mid \RestrictedTerm{\matchingUpIn}{\big( \begin{aligned}[t]
												& \ReplicateInput{\matchingCoordinatorIn}{y', \sumLock_r, \receiverLock}.\left( \Match{y'}{\RenamingPolicyMixAsyn{y}}\Output{\receiverLock}{\sumLock_r, \sumLock_b, \sumLock_b, \senderLock, \RenamingPolicyMixAsyn{z}} \mid \Output{\matchingUpIn}{y', \sumLock_r, \receiverLock} \right)\\
												& \mid \Output{\matchingUpIn}{\RenamingPolicyMixAsyn{y}, \sumLock_a, \receiverLock} \mid \RestrictedTerm{\matchingCoordinatorIn}{\left( \Forward{\matchingUpIn}{\matchingCoordinatorIn} \mid \processRightOutputRequests \right)} \big)
											\end{aligned}}\\
										 & \mid \Output{\coordinatorUpOut}{\RenamingPolicyMixAsyn{y}, \sumLock_b, \senderLock, \RenamingPolicyMixAsyn{z}} \mid \processRightInputRequests \big) \mid \pushRequests \big)
									\end{aligned}}}
							\end{aligned}\\
						& \step^2 \begin{aligned}[t]
								& \RestrictedTerm{\matchingCoordinatorOut, \matchingCoordinatorIn, \coordinatorUpOut, \coordinatorUpIn, \coordinatorMatchingOut, \coordinatorMatchingIn, \matchingUpOut, \matchingUpIn, \sumLock_a, \sumLock_b, \receiverLock, \senderLock}{\big(\\
								& \hspace*{1em} \RestrictedTerm{\parallelChannelOut, \parallelChannelIn}{\big( \begin{aligned}[t]
										& \RestrictedTerm{\coordinatorRepA, \coordinatorRepB, \matchingReceiverOut, \matchingReceiverIn}{\big(\\
										& \hspace*{1em} \Output{\sumLock_b}{\false} \mid \ReplicateInput{\receiverLock}{-, -, \sumLock_s, \senderLock, \RenamingPolicyMixAsyn{x}}.\Test{\sumLock_s}{\Output{\sumLock_s}{\false} \mid \Out{\senderLock} \mid \Output{\coordinatorRepA}{\RenamingPolicyMixAsyn{x}}}{\Output{\sumLock_s}{\false}}\\
										& \hspace*{1em} \mid \Output{\matchingReceiverIn}{\RenamingPolicyMixAsyn{y}, \sumLock_a, \receiverLock} \mid \Output{\sumLock_a}{\true} \mid \ReplicateInput{\coordinatorRepA}{\RenamingPolicyMixAsyn{x}}.\Input{\coordinatorRepB}{\matchingReceiverOut, \matchingReceiverIn}.\\
										& \hspace*{1.7em} \RestrictedTerm{\matchingCoordinatorOut, \matchingCoordinatorIn, \coordinatorUpOut, \coordinatorUpIn, \matchingReceiverUpOut, \matchingReceiverUpIn, \coordinatorMatchingOut, \coordinatorMatchingIn, \matchingUpOut, \matchingUpIn}{\big( \pushRequestsIn\\
										& \hspace*{2em} \mid \RestrictedTerm{\parallelChannelOut, \parallelChannelIn}{\left( \EncodingMixAsyn{P} \mid \processRightOutputRequests \mid \processRightInputRequests \right)}\\
										& \hspace*{2em} \mid \RestrictedTerm{\matchingReceiverOut, \matchingReceiverIn}{\left( \Output{\coordinatorRepB}{\matchingReceiverOut \mid \matchingReceiverIn} \mid \pushRequestsOut \right)} \big)}\\
										& \hspace*{1em} \mid \RestrictedTerm{\matchingCoordinatorOut, \matchingCoordinatorIn, \coordinatorUpOut, \coordinatorUpIn, \matchingReceiverUpOut, \matchingReceiverUpIn, \coordinatorMatchingOut, \coordinatorMatchingIn, \matchingUpOut, \matchingUpIn}{\big( \pushRequestsIn\\
										& \hspace*{2em} \mid \RestrictedTerm{\parallelChannelOut, \parallelChannelIn}{\big( \begin{aligned}[t]
												& \Set{ \Subst{\RenamingPolicyMixAsyn{z}}{\RenamingPolicyMixAsyn{x}} }\left( \EncodingMixAsyn{P} \right)\\
												& \mid \processRightOutputRequests \mid \processRightInputRequests \big)
											\end{aligned}}\\
										& \hspace*{2em} \mid \RestrictedTerm{\matchingReceiverOut, \matchingReceiverIn}{\left( \Output{\coordinatorRepB}{\matchingReceiverOut \mid \matchingReceiverIn} \mid \pushRequestsOut \right)} \big)} \big)}\\
										& \mid \processLeftOutputRequests \mid \processLeftInputRequests \mid \Output{\coordinatorUpIn}{\RenamingPolicyMixAsyn{y}, \sumLock_a, \receiverLock} \big)
									\end{aligned}}\\
								& \hspace*{1em} \mid \RestrictedTerm{\parallelChannelOut, \parallelChannelIn}{\big( \begin{aligned}[t]
										& \Set{ \Subst{\sumLock_b}{\sumLock} }\left( \prod_{i \in \indexSet_2, i \neq j} \EncodingMixAsyn{\guard_i.Q_i} \right) \mid \EncodingMixAsyn{Q_j}\\
										& \mid \RestrictedTerm{\matchingUpIn}{\big( \begin{aligned}[t]
												& \ReplicateInput{\matchingCoordinatorIn}{y', \sumLock_r, \receiverLock}.\left( \Match{y'}{\RenamingPolicyMixAsyn{y}}\Output{\receiverLock}{\sumLock_r, \sumLock_b, \sumLock_b, \senderLock, \RenamingPolicyMixAsyn{z}} \mid \Output{\matchingUpIn}{y', \sumLock_r, \receiverLock} \right)\\
												& \mid \Output{\matchingUpIn}{\RenamingPolicyMixAsyn{y}, \sumLock_a, \receiverLock} \mid \RestrictedTerm{\matchingCoordinatorIn}{\left( \Forward{\matchingUpIn}{\matchingCoordinatorIn} \mid \processRightOutputRequests \right)} \big)
											\end{aligned}}\\
										 & \mid \Output{\coordinatorUpOut}{\RenamingPolicyMixAsyn{y}, \sumLock_b, \senderLock, \RenamingPolicyMixAsyn{z}} \mid \processRightInputRequests \big)
									\end{aligned}}\\
								& \hspace*{1em} \mid \pushRequests \big)} = T
							\end{aligned}
					\end{align*}}
					By Corollary \ref{col:encodingSubstitutions}, $ \Set{ \Subst{\RenamingPolicyMixAsyn{z}}{\RenamingPolicyMixAsyn{x}} } \EncodingMixAsyn{P} \equivAlpha \EncodingMixAsyn{\Set{ \Subst{z}{x} } P}  $. Unfortunately, this time it does not suffer to ignore junk to prove that $ T \transBarbCongMixAsynB \EncodingMixAsyn{S'} $, because in $ \EncodingMixAsyn{S'} $ there are two encoded parallel operators whereas in $ T $ there is only one. Nevertheless, we start reducing $ T $ by omitting junk. Since the sum lock $ \sumLock_b $ is instantiated by false, by Lemma \ref{lem:junkRequestsOnFalseSumLocks}, the request $ \Output{\coordinatorUpOut}{\RenamingPolicyMixAsyn{y}, \sumLock_b, \senderLock, \RenamingPolicyMixAsyn{z}} $ is junk. Moreover, by Lemma \ref{lem:junkInduceTest}, the term
					\begin{align*}
						\ReplicateInput{\matchingCoordinatorIn}{y', \sumLock_r, \receiverLock}.\left( \Match{y'}{\RenamingPolicyMixAsyn{y}}\Output{\receiverLock}{\sumLock_r, \sumLock_b, \sumLock_b, \senderLock, \RenamingPolicyMixAsyn{z}} \mid \Output{\matchingUpIn}{y', \sumLock_r, \receiverLock} \right)
					\end{align*}
					reduces to the forwarder $ \Forward{\matchingCoordinatorIn}{\matchingUpIn} $. Since $ \sumLock_a, \sumLock_b, \receiverLock, \senderLock \notin \FreeNames{\EncodingMixAsyn{\Set{ \Subst{z}{x} } P}} \cup \FreeNames{\EncodingMixAsyn{Q_j}} $, we can reorder the term according to the restrictions on $ \sumLock_a $, $ \sumLock_b $, and $ \receiverLock $ and the restriction on $ \senderLock $ can be omitted. By Lemma \ref{lem:junkRemainsOfSumsMixAsyn}, then $ \RestrictedTerm{\sumLock}{\left( \prod_{i \in \indexSet, i \neq j} \EncodingMixAsyn{\guard_i.Q_i} \mid \Output{\sumLock}{\false} \right)} $ is junk. By Lemma \ref{lem:removeJunk}, we deduce $ T \transBarbCongMixAsynB T' $, where $ T' $ is:
					\begin{align*}
						T & \transBarbCongMixAsynB \begin{aligned}[t]
								& \RestrictedTerm{\matchingCoordinatorOut, \matchingCoordinatorIn, \coordinatorUpOut, \coordinatorUpIn, \coordinatorMatchingOut, \coordinatorMatchingIn, \matchingUpOut, \matchingUpIn}{\big(\\
								& \hspace*{1em} \RestrictedTerm{\parallelChannelOut, \parallelChannelIn}{\big( \begin{aligned}[t]
										& \RestrictedTerm{\sumLock, \receiverLock, \coordinatorRepA, \coordinatorRepB, \matchingReceiverOut, \matchingReceiverIn}{\big(\\
										& \hspace*{1em} \ReplicateInput{\receiverLock}{-, -, \sumLock_s, \senderLock, \RenamingPolicyMixAsyn{x}}.\Test{\sumLock_s}{\Output{\sumLock_s}{\false} \mid \Out{\senderLock} \mid \Output{\coordinatorRepA}{\RenamingPolicyMixAsyn{x}}}{\Output{\sumLock_s}{\false}}\\
										& \hspace*{1em} \mid \Output{\matchingReceiverIn}{\RenamingPolicyMixAsyn{y}, \sumLock, \receiverLock} \mid \Output{\sumLock}{\true}\\
										& \hspace*{1em} \mid \ReplicateInput{\coordinatorRepA}{\RenamingPolicyMixAsyn{x}}.\Input{\coordinatorRepB}{\matchingReceiverOut, \matchingReceiverIn}.\RestrictedTerm{\matchingCoordinatorOut, \matchingCoordinatorIn, \coordinatorUpOut, \coordinatorUpIn, \matchingReceiverUpOut, \matchingReceiverUpIn, \coordinatorMatchingOut, \coordinatorMatchingIn, \matchingUpOut, \matchingUpIn}{\big(\\
										& \hspace*{2em} \pushRequestsIn\\
										& \hspace*{2em} \mid \RestrictedTerm{\parallelChannelOut, \parallelChannelIn}{\left( \EncodingMixAsyn{P} \mid \processRightOutputRequests \mid \processRightInputRequests \right)}\\
										& \hspace*{2em} \mid \RestrictedTerm{\matchingReceiverOut, \matchingReceiverIn}{\left( \Output{\coordinatorRepB}{\matchingReceiverOut \mid \matchingReceiverIn} \mid \pushRequestsOut \right)} \big)}\\
										& \hspace*{1em} \mid \RestrictedTerm{\matchingCoordinatorOut, \matchingCoordinatorIn, \coordinatorUpOut, \coordinatorUpIn, \matchingReceiverUpOut, \matchingReceiverUpIn, \coordinatorMatchingOut, \coordinatorMatchingIn, \matchingUpOut, \matchingUpIn}{\big( \pushRequestsIn\\
										& \hspace*{2em} \mid \RestrictedTerm{\parallelChannelOut, \parallelChannelIn}{\big( \begin{aligned}[t]
												& \EncodingMixAsyn{\Set{ \Subst{z}{x} } P} \mid \processRightOutputRequests \mid \processRightInputRequests \big)
											\end{aligned}}\\
										& \hspace*{2em} \mid \RestrictedTerm{\matchingReceiverOut, \matchingReceiverIn}{\left( \Output{\coordinatorRepB}{\matchingReceiverOut \mid \matchingReceiverIn} \mid \pushRequestsOut \right)} \big)} \big)}\\
										& \mid \processLeftOutputRequests \mid \processLeftInputRequests \mid \Output{\coordinatorUpIn}{\RenamingPolicyMixAsyn{y}, \sumLock, \receiverLock} \big)
									\end{aligned}}\\
								& \hspace*{1em} \mid \RestrictedTerm{\parallelChannelOut, \parallelChannelIn}{\big( \begin{aligned}[t]
										& \EncodingMixAsyn{Q_j} \mid \processRightInputRequests\\
										& \hspace*{-1em} \mid \RestrictedTerm{\matchingUpIn}{\left(  \Forward{\matchingCoordinatorIn}{\matchingUpIn} \mid \Output{\matchingUpIn}{\RenamingPolicyMixAsyn{y}, \sumLock, \receiverLock} \mid \RestrictedTerm{\matchingCoordinatorIn}{\left( \Forward{\matchingUpIn}{\matchingCoordinatorIn} \mid \processRightOutputRequests \right)} \right)} \big)
									\end{aligned}}\\
								& \hspace*{1em} \mid \pushRequests \big)} = T'
							\end{aligned}
					\end{align*}
					Analysing $ T' $ we observe that in comparison to $ \EncodingMixAsyn{S'} $ the encoded subterms $ \EncodingMixAsyn{\Set{ \Subst{z}{x} } P} $, $ \EncodingMixAsyn{Q_j} $, and the term representing $ \EncodingMixAsyn{\ReplicateInput{y}{x}.P} $ appear in the wrong order. However, since $ S' \equiv S'' = \left( \ReplicateInput{y}{x}.P \mid \Set{ \Subst{z}{x} } P \right) \mid Q_j $ and $ \transBarbCongMixAsynB $, by Lemma \ref{lem:preservesSCModuloTransBarbBisimMixAsyn}, preserves structural congruence of source terms, we have $ \EncodingMixAsyn{S'} \transBarbCongMixAsynB \EncodingMixAsyn{S''} $, i.e., the order of these subterms does not matter. As in the case before, on the right side of the parallel operator encoding there are the two forwarders $ \Forward{\matchingCoordinatorIn}{\matchingUpIn} $ and $ \Forward{\matchingUpIn}{\matchingCoordinatorIn} $ (for different instances of $ \matchingCoordinatorIn $). Again they are necessary to \simulate further source term steps on the continuation $ \EncodingMixAsyn{Q_j} $, but, since they perform only invisible steps, they do not influence the state of $ T' $ modulo $ \transBarbCongMixAsynB $.
					
					Moreover there is the request $ \Output{\matchingUpIn}{\RenamingPolicyMixAsyn{y}, \sumLock, \receiverLock} $, to enable an \simulation of a communication of $ Q_j $ and $ \ReplicateInput{y}{x}.P $. Note that there is also the request $ \Output{\coordinatorUpIn}{\RenamingPolicyMixAsyn{y}, \sumLock, \receiverLock} $ at the right side of the parallel operator encoding, but the request $ \Output{\parallelChannelIn}{\RenamingPolicyMixAsyn{y}, \sumLock, \receiverLock} $, which belongs to $ \EncodingMixAsyn{\ReplicateInput{y}{x}.P} $, is missing. However, since by $ \Forward{\matchingUpIn}{\matchingCoordinatorIn} $ the request $ \Output{\coordinatorUpIn}{\RenamingPolicyMixAsyn{y}, \sumLock, \receiverLock} $ is forwarded to $ \Output{\matchingCoordinatorIn}{\RenamingPolicyMixAsyn{y}, \sumLock, \receiverLock} $ within a \pure \admin step and since this configuration is equal to one application of \processLeftInputRequests \ on $ \Output{\parallelChannelIn}{\RenamingPolicyMixAsyn{y}, \sumLock, \receiverLock} $, which is again a \pure \admin step, these two requests in comparison to $ \Output{\parallelChannelIn}{\RenamingPolicyMixAsyn{y}, \sumLock, \receiverLock} $ do not influence the state of $ T' $ modulo $ \transBarbCongMixAsynB $.
					
					What remains as difference of $ T' $ and $ \EncodingMixAsyn{S''} $ is the fact, that in $ T' $ the encoding of $ \Set{ \Subst{z}{x} } P $ appears within a branch of the replicated input encoding whereas in $ \EncodingMixAsyn{S''} $ it appears as right branch of a parallel operator encoding, i.e., it remains to show that $ T'' \transBarbCongMixAsynB \EncodingMixAsyn{\ReplicateInput{y}{x}.P \mid \Set{ \Subst{z}{x} } P} $, where
					\begin{align*}
						T'' & = \RestrictedTerm{\sumLock, \receiverLock, \coordinatorRepA, \coordinatorRepB, \matchingReceiverOut, \matchingReceiverIn}{\big(\\
								& \hspace*{2em} \ReplicateInput{\receiverLock}{-, -, \sumLock_s, \senderLock, \RenamingPolicyMixAsyn{x}}.\Test{\sumLock_s}{\Output{\sumLock_s}{\false} \mid \Out{\senderLock} \mid \Output{\coordinatorRepA}{\RenamingPolicyMixAsyn{x}}}{\Output{\sumLock_s}{\false}}\\
								& \hspace*{2em} \mid \Output{\matchingReceiverIn}{\RenamingPolicyMixAsyn{y}, \sumLock, \receiverLock} \mid \Output{\sumLock}{\true}\\
								& \hspace*{2em} \mid \ReplicateInput{\coordinatorRepA}{\RenamingPolicyMixAsyn{x}}.\Input{\coordinatorRepB}{\matchingReceiverOut, \matchingReceiverIn}.\RestrictedTerm{\matchingCoordinatorOut, \matchingCoordinatorIn, \coordinatorUpOut, \coordinatorUpIn, \matchingReceiverUpOut, \matchingReceiverUpIn, \coordinatorMatchingOut, \coordinatorMatchingIn, \matchingUpOut, \matchingUpIn}{\big( \pushRequestsIn\\
								& \hspace*{4em} \mid \RestrictedTerm{\parallelChannelOut, \parallelChannelIn}{\left( \EncodingMixAsyn{P} \mid \processRightOutputRequests \mid \processRightInputRequests \right)}\\
								& \hspace*{4em} \mid \RestrictedTerm{\matchingReceiverOut, \matchingReceiverIn}{\left( \Output{\coordinatorRepB}{\matchingReceiverOut \mid \matchingReceiverIn} \mid \pushRequestsOut \right)} \big)}\\
								& \hspace*{2em} \mid \RestrictedTerm{\matchingCoordinatorOut, \matchingCoordinatorIn, \coordinatorUpOut, \coordinatorUpIn, \matchingReceiverUpOut, \matchingReceiverUpIn, \coordinatorMatchingOut, \coordinatorMatchingIn, \matchingUpOut, \matchingUpIn}{\big( \pushRequestsIn\\
								& \hspace*{4em} \mid \RestrictedTerm{\parallelChannelOut, \parallelChannelIn}{\left( \EncodingMixAsyn{\Set{ \Subst{z}{x} } P} \mid \processRightOutputRequests \mid \processRightInputRequests \right)}\\
								& \hspace*{4em} \mid \RestrictedTerm{\matchingReceiverOut, \matchingReceiverIn}{\left( \Output{\coordinatorRepB}{\matchingReceiverOut \mid \matchingReceiverIn} \mid \pushRequestsOut \right)} \big)} \big)}
					\end{align*}
					First, note that the term
					\begin{align*}
						\RestrictedTerm{\parallelChannelOut, \parallelChannelIn}{\left( \EncodingMixAsyn{\Set{ \Subst{z}{x} } P} \mid \processRightOutputRequests \mid \processRightInputRequests \right)}
					\end{align*}
					exactly corresponds to the right branch of $ \EncodingMixAsyn{\ReplicateInput{y}{x}.P \mid \Set{ \Subst{z}{x} } P} $. If we compare \pushRequestsIn \ with \processLeftOutputRequests \ and \processLeftInputRequests, then we observe that the former includes exactly the same forwarders as the later but also some additional forwarders. The same holds for \pushRequestsOut \ and \pushRequests. Note that the additional forwarders ensures that each requests of each branch of the replicated input encoding is forwarded to each next right branch, and so these additional forwarders are necessary in case there is more than one branch. Also note that the given forwarders guarantee that each pair of requests, such that one is an input and the other one an output request and both requests do not origin from the same sum, can be combined. Moreover, note that the only request from the left side, i.e., of the encoding of the replicated input, is transmitted to the right side, i.e., the only branch of the replicated input, by the request $ \Output{\matchingReceiverIn}{\RenamingPolicyMixAsyn{y}, \sumLock, \receiverLock} $ and \pushRequestsIn. So these forwarders do not distinguish $ T' $ and $ \EncodingMixAsyn{S''} $ modulo $ \transBarbCongMixAsynB $.
					
					Since $ T'' $ and $ \EncodingMixAsyn{\ReplicateInput{y}{x}.P \mid \Set{ \Subst{z}{x} } P} $ do only differ by the forwarding of requests but nevertheless allow for the same combinations, we deduce that $ T'' \transBarbCongMixAsynB \EncodingMixAsyn{\ReplicateInput{y}{x}.P \mid \Set{ \Subst{z}{x} } P} $. Thus, by Lemma \ref{lem:transBarbCongIsEquivalence}, we conclude $ T \transBarbCongMixAsynB \EncodingMixAsyn{S'} $.
			\end{description}
		\item[Induction Hypothesis:] $ S_1 \step S_1' $ implies $ \exists T_1 \in \piAsynProc \logdot \EncodingMixAsyn{S_1} \steps T_1 \wedge T_1 \transBarbCongMixAsynB \EncodingMixAsyn{S_1'} $
		\item[Induction Step:] We have to consider the remaining three rules \textsc{Par}, \textsc{Res}, and \textsc{Cong} of Figure \ref{fig:concurrentReductionSemantics}.
			\begin{description}
				\item[Case of Rule] \textsc{Par}\textbf{:} Then $ S = S_1 \mid S_2 $ for some $ S_1, S_2 \in \piMixProc $, $ S_1 \step S_1' $, and $ S' = S_1' \mid S_2 $. By the induction hypothesis there is some $ T_1 \in \piAsynProc $ such that $ \EncodingMixAsyn{S_1} \steps T_1 $ and $ T_1 \transBarbCongMixAsynB \EncodingMixAsyn{S_1'} $. The corresponding encodings are given by the following terms:
					\begin{align*}
						\EncodingMixAsyn{S} & = \begin{aligned}[t]
								& \RestrictedTerm{\matchingCoordinatorOut, \matchingCoordinatorIn, \coordinatorUpOut, \coordinatorUpIn, \coordinatorMatchingOut, \coordinatorMatchingIn, \matchingUpOut, \matchingUpIn}{\big(\\
								& \hspace*{1em} \RestrictedTerm{\parallelChannelOut, \parallelChannelIn}{\left( \EncodingMixAsyn{S_1} \mid \processLeftOutputRequests \mid \processLeftInputRequests \right)}\\
								& \hspace*{1em} \mid \RestrictedTerm{\parallelChannelOut, \parallelChannelIn}{\left( \EncodingMixAsyn{S_2} \mid \processRightOutputRequests \mid \processRightInputRequests \right)}\\
								& \hspace*{1em} \mid \pushRequests \big)}
							\end{aligned}\\
						\EncodingMixAsyn{S'} & = \begin{aligned}[t]
								& \RestrictedTerm{\matchingCoordinatorOut, \matchingCoordinatorIn, \coordinatorUpOut, \coordinatorUpIn, \coordinatorMatchingOut, \coordinatorMatchingIn, \matchingUpOut, \matchingUpIn}{\big(\\
								& \hspace*{1em} \RestrictedTerm{\parallelChannelOut, \parallelChannelIn}{\left( \EncodingMixAsyn{S_1'} \mid \processLeftOutputRequests \mid \processLeftInputRequests \right)}\\
								& \hspace*{1em} \mid \RestrictedTerm{\parallelChannelOut, \parallelChannelIn}{\left( \EncodingMixAsyn{S_2} \mid \processRightOutputRequests \mid \processRightInputRequests \right)}\\
								& \hspace*{1em} \mid \pushRequests \big)}
							\end{aligned}
					\end{align*}
					Since $ \EncodingMixAsyn{S_1} \steps T_1 $ and since $ \EncodingMixAsyn{S_1} $ is not guarded in $ \EncodingMixAsyn{S} $, we can use the rules \textsc{Par}, \textsc{Res}, and \textsc{Cong} in the asynchronous calculus to show that:
					\begin{align*}
						\EncodingMixAsyn{S} \steps \begin{aligned}[t]
								& \RestrictedTerm{\matchingCoordinatorOut, \matchingCoordinatorIn, \coordinatorUpOut, \coordinatorUpIn, \coordinatorMatchingOut, \coordinatorMatchingIn, \matchingUpOut, \matchingUpIn}{\big(\\
								& \hspace*{1em} \RestrictedTerm{\parallelChannelOut, \parallelChannelIn}{\left( T_1 \mid \processLeftOutputRequests \mid \processLeftInputRequests \right)}\\
								& \hspace*{1em} \mid \RestrictedTerm{\parallelChannelOut, \parallelChannelIn}{\left( \EncodingMixAsyn{S_2} \mid \processRightOutputRequests \mid \processRightInputRequests \right)}\\
								& \hspace*{1em} \mid \pushRequests \big) = T}
							\end{aligned}
					\end{align*}
					By Definition \ref{def:transBarbCong}, $ T_1 \transBarbCongMixAsynB \EncodingMixAsyn{S_1'} $ implies $ \Context{}{}{T_1} \transBarbBisimMixAsynB \Context{}{}{\EncodingMixAsyn{S_1'}} $ for all contexts $ \Context{}{}{\hole} \in \piAsynProc \to \piAsynProc $ such that $ \Context{}{}{\EncodingMixAsyn{\nullTerm}} \in \targetTermsMixAsyn $. Since $ \EncodingMixAsyn{\nullTerm \mid S_2 } \in \targetTermsMixAsyn $, the quantification over $ \context $ includes all contexts $ \context $ such that:
					\begin{align*}
						\Context{}{}{\hole} & = \context'\big( \begin{aligned}[t]
								& \RestrictedTerm{\matchingCoordinatorOut, \matchingCoordinatorIn, \coordinatorUpOut, \coordinatorUpIn, \coordinatorMatchingOut, \coordinatorMatchingIn, \matchingUpOut, \matchingUpIn}{\big(\\
								& \hspace*{1em} \RestrictedTerm{\parallelChannelOut, \parallelChannelIn}{\left( \hole \mid \processLeftOutputRequests \mid \processLeftInputRequests \right)}\\
								& \hspace*{1em} \mid \RestrictedTerm{\parallelChannelOut, \parallelChannelIn}{\left( \EncodingMixAsyn{S_2} \mid \processRightOutputRequests \mid \processRightInputRequests \right)}\\
								& \hspace*{1em} \mid \pushRequests \big) \big)}
							\end{aligned}\\
							& = \Context{'}{}{\Context{''}{}{\hole}}
					\end{align*}
					Because of that, we have $ \Context{'}{}{\Context{''}{}{T_1}} \transBarbBisimMixAsynB \Context{'}{}{\Context{''}{}{\EncodingMixAsyn{S_1'}}} $ for all contexts $ \Context{'}{}{\hole} \in \piAsynProc \to \piAsynProc $ such that $ \Context{'}{}{\EncodingMixAsyn{\nullTerm}} \in \targetTermsMixAsyn $. By Definition \ref{def:transBarbCong}, we conclude $ T \transBarbCongMixAsynB \EncodingMixAsyn{S'} $.
				\item[Case of Rule] \textsc{Res}\textbf{:} Then $ S = \RestrictedTerm{x}{S_1} $ for some $ x \in \names $ and some $ S_1 \in \piMixProc $, $ S_1 \step S_1' $, and $ S' = \RestrictedTerm{x}{S_1'} $. By the induction hypothesis there is some $ T_1 \in \piAsynProc $ such that $ \EncodingMixAsyn{S_1} \steps T_1 $ and $ T_1 \transBarbCongMixAsynB \EncodingMixAsyn{S_1'} $. Since the encoding of restriction is \clean, i.e., $ \EncodingMixAsyn{S} = \RestrictedTerm{\RenamingPolicyMixAsyn{x}}{\EncodingMixAsyn{S_1}} $ and $ \EncodingMixAsyn{S'} = \RestrictedTerm{\RenamingPolicyMixAsyn{x}}{\EncodingMixAsyn{S_1'}} $, we can apply rule \textsc{Res} to conclude from $ \EncodingMixAsyn{S_1} \steps T_1 $ to $ \EncodingMixAsyn{S} \steps \RestrictedTerm{\RenamingPolicySepAsyn{x}}{T_1} = T $. By Definition \ref{def:transBarbCong}, $ T_1 \transBarbCongMixAsynB \EncodingMixAsyn{S_1'} $ implies $ \Context{}{}{T_1} \transBarbBisimMixAsynB \Context{}{}{\EncodingMixAsyn{S_1'}} $ for all contexts $ \Context{}{}{\hole} \in \piAsynProc \to \piAsynProc $ such that $ \Context{}{}{\EncodingMixAsyn{\nullTerm}} \in \targetTermsMixAsyn $. Since $ \RestrictedTerm{\RenamingPolicyMixAsyn{x}}{\EncodingMixAsyn{\nullTerm}} \in \targetTermsMixAsyn $, the quantification over $ \context $ includes all contexts $ \context $ such that $ \Context{}{}{\hole} = \Context{'}{}{\RestrictedTerm{\RenamingPolicyMixAsyn{x}}{\hole}} $. Because of that, we have $ \Context{'}{}{\RestrictedTerm{\RenamingPolicyMixAsyn{x}}{T_1}} \transBarbBisimMixAsynB \Context{'}{}{\RestrictedTerm{\RenamingPolicyMixAsyn{x}}{\EncodingMixAsyn{S_1'}}} $ for all contexts $ \Context{'}{}{\hole} \in \piAsynProc \to \piAsynProc $ such that $ \Context{'}{}{\EncodingMixAsyn{\nullTerm}} \in \targetTermsMixAsyn $. By Definition \ref{def:transBarbCong}, we conclude $ T \transBarbCongMixAsynB \EncodingMixAsyn{S'} $.
				\item[Case of Rule] \textsc{Cong}\textbf{:} Then $ S \equiv S_1 $ for some $ S_1 \in \piMixProc $, $ S_1 \step S_1' $, and $ S_1' \equiv S' $. By Lemma \ref{lem:preservesSCModuloTransBarbBisimMixAsyn}, the encoding $ \encodingMixAsyn $ preserves structural congruence of source terms modulo $ \transBarbCongMixAsynB $. So $ S \equiv S_1 $ and $ S_1' \equiv S' $ implies $ \EncodingMixAsyn{S} \transBarbCongMixAsynB \EncodingMixAsyn{S_1} $ and $ \EncodingMixAsyn{S_1'} \transBarbCongMixAsynB \EncodingMixAsyn{S'} $. By Definition \ref{def:transBarbCong}, for all contexts $ \Context{}{}{\hole} \in \piAsynProc \to \piAsynProc $ such that $ \Context{}{}{\EncodingMixAsyn{\nullTerm}} \in \targetTermsMixAsyn $ we have $ \Context{}{}{\EncodingMixAsyn{S}} \transBarbBisimMixAsynB \Context{}{}{\EncodingMixAsyn{S_1}} $, i.e., especially $ \EncodingMixAsyn{S} \transBarbBisimMixAsynB \EncodingMixAsyn{S_1} $. Thus, by Definition \ref{def:transBarbBisim}, for each sequence $ \EncodingMixAsyn{S} \steps T $ there is a sequence $ \EncodingMixAsyn{S_1} \steps T_1 $ for some $ T_1 \in \piAsynProc $ such that $ T \transBarbBisimMixAsynB T_1 $. The same holds for all Contexts $ \context $, i.e., since $ \Context{}{}{\EncodingMixAsyn{S}} \transBarbBisimMixAsynB \Context{}{}{\EncodingMixAsyn{S_1}} $, for each sequence $ \Context{}{}{\EncodingMixAsyn{S}} \steps \Context{}{}{T} $ there is a sequence $ \Context{}{}{\EncodingMixAsyn{S_1}} \steps \Context{}{}{T_1} $ for some $ T_1 \in \piAsynProc $ such that $ \Context{}{}{T} \transBarbBisimMixAsynB \Context{}{}{T_1} $. So, by Definition \ref{def:transBarbCong}, $ T \transBarbCongMixAsynB T_1 $. By the induction hypothesis $ T_1 \transBarbCongMixAsynB \EncodingMixAsyn{S_1'} $. Since, by Lemma \ref{lem:transBarbCongIsEquivalence}, $ \transBarbCongMixAsynB $ is an equivalence, $ T \transBarbCongMixAsynB T_1 $, $ T_1 \transBarbCongMixAsynB \EncodingMixAsyn{S_1'} $, and $ \EncodingMixAsyn{S_1'} \transBarbCongMixAsynB \EncodingMixAsyn{S'} $ implies $ T \transBarbCongMixAsynB \EncodingMixAsyn{S'} $.
			\end{description}
	\end{description}
	\qed
\end{proof}

If we analyse the proofs of the Lemmata \ref{lem:operationalCompletenessSepAsyn} and \ref{lem:operationalCompletenessMixAsyn}, we observe that for each \simulation there is exactly one \nonAdmin step, i.e., there is exactly one \nonAdmin step for each of the rules $ \textsc{Tau}_{\indexMix, \indexSep} $, $ \textsc{Com}_{\indexMix, \indexSep}  $, and $ \textsc{Rep}_{\indexMix, \indexSep}  $ and the \simulation of the remaining rules do not introduce additional \nonAdmin steps. This underpins our intuition of \nonAdmin steps (compare to Section \ref{sec:stepsSimulation}) that any \simulation of a source term step is connected to exactly one \nonAdmin step. Moreover, any \nonAdmin step marks exactly one \simulation of a source term step by steering the \simulation to the ``point of no return'', i.e., to a point, from where no other sequence of steps can disable the completion of that \simulation and from where any possibility to \simulate a conflicting source term step is ultimately withdrawn.

\begin{lemma} \label{lem:simulationVSNonAdminStep}
	Any \simulation of a source term step includes exactly one \nonAdmin step and any \nonAdmin step steers the \simulation of a source term step to a point, from where it eventually has to be completed.
\end{lemma}

\begin{proof}
	If we analyse the reduction rules in Figure \ref{fig:concurrentReductionSemantics}, we observe that the result of a source term step is the unguarding of one or two former guarded subterms. If we rely on this fact to mark the essence of a step, then an \simulation is characterised by the unguarding of the respective encoded continuations. Analysing the encoding functions in Figure \ref{fig:encodingSepAsyn} and Figure \ref{fig:encodingMixAsyn} we observe that such encoded continuations are guarded either by a sender lock or a test-statement. Note that, as proved in \cite{nestmann00} for $ \encodingSepAsyn $ and by the Lemma \ref{lem:pureAdminStepsNoDeadlock} and Lemma \ref{lem:impureAdminStepsNoDeadlock}, neither \pure nor \impure \admin steps can introduce deadlock.
	
	The reduction of a summand guarded by $ \tau $ is the only case of a reduction step that does not require the interaction of a receiver and a sender (compare to rule $ \textsc{Tau}_{\indexMix, \indexSep} $). Both encodings introduce a sum lock to encode a sum and translate summands guarded by $ \tau $ into a single test-statement, that tests the corresponding sum lock and in case of success unguards the encoded continuation and provides a negative instantiation of the sum lock. Note that to \simulate such a source term for both encodings only two steps are necessary to reduce the test-statement. The first consumes the instantiation of the sum lock. If this instantiation is positive we call that step a \nonAdmin step. Note that before this step, an \simulation of a conflicting source term step may change the instantiation of the sum lock into a negative instantiation and thus withdraw the possibility to \simulate the step. On the other side as soon as the positive instantiation of the sum lock is consumed, there is, by Lemma \ref{lem:nonConflictingStepsSepAsyn} and Lemma \ref{lem:nonConflictingStepsMixAsyn}, no possibility to prevent that the second step necessary to complete the \simulation eventually happens. So, whenever such \nonAdmin step is performed, eventually the encoded continuation is unguarded. Moreover, the only way to instantiate the consumed sum lock is to complete the \simulation, which leads to a negative instantiation of that sum lock. By Lemma \ref{lem:changeInstantiationSumLock}, there is no chance to unguard a positive instantiation of that sum lock afterwards. Note that any source term step, that is in conflict to the considered step on the $ \tau $ guarded summand, have to reduce another summand of the same sum. Remember that, by Lemma \ref{lem:junkRemainsOfSumsSepAsyn} and \ref{lem:junkRemainsOfSumsMixAsyn}, encoded summands connected to a negative instantiation of a sum lock are junk. So the consumption of the positive instantiation of the sum lock immediately rules out any \simulation of a conflicting source term step. Moreover, since there is no possibility to reach a positive instantiation of that sum lock again, there is no possibility to \simulate this source term step twice and any translated observable connected to this sum lock, i.e., to this sum, is immediately withdrawn by the consumption of the positive instantiation. Because of that, for any \simulation of a source term step based on rule $ \textsc{Tau}_{\indexMix, \indexSep} $ there is exactly one \nonAdmin step and for any such \nonAdmin step exactly one source term step is \simulated.
	
	In case of rule $ \textsc{Com}_{\indexMix, \indexSep} $ the source term step is on an interaction of an input guarded and an output guarded summand. The encoding of the continuation of the receiver is guarded by the second of a nested test-statement, while the encoded continuation of the sender is guarded by a sender lock, which in turn is guarded by the nested test-statement. In Example \ref{exa:intermediateStates} we explain that reducing the first test-statement does not ensure, that the \simulation can be completed. However, as soon as the positive instantiation of the second sum lock is consumed we can repeat the argumentation of the case above to show that the point of no return of that \simulation is reached. We observe that the consumption of the second positive sum lock instantiation is indeed the only \nonAdmin step necessary to \simulate this source term step. It immediately withdraw any possibility to complete the \simulation of a conflicting source term step, because it ensures that both consumed positive instantiations of sum lock are never restored. With that it also immediately withdraws all translated observables on the summands of these to sums. Moreover, by Lemma \ref{lem:nonConflictingStepsSepAsyn} and Lemma \ref{lem:nonConflictingStepsMixAsyn}, the encoded continuations of the respective sender and receiver have eventually to become unguarded. Thus again for any \simulation of a source term step based on rule $ \textsc{Com}_{\indexMix, \indexSep} $ there is exactly one \nonAdmin step and for any such \nonAdmin step exactly one source term step is \simulated.
	
	In case of rule $ \textsc{Rep}_{\indexMix, \indexSep} $ the source term step is on an interaction of a replicated input and an output guarded summand. The encoding of the continuation of the receiver is guarded by a single test-statement, while the encoded continuation of the sender is guarded by a sender lock, which in turn is guarded by the single test-statement. As soon as the positive instantiation of the second sum lock is consumed we can repeat the argumentation of the first case to show that the point of no return of that \simulation is reached. We observe that the consumption of the positive sum lock instantiation is indeed the only \nonAdmin step necessary to \simulate this source term step. It immediately withdraw any possibility to complete the \simulation of a conflicting source term step, because it ensures that the consumed positive instantiation of the sum lock is never restored. With that it also immediately withdraws all translated observables on the summands of that sum. Moreover, by Lemma \ref{lem:nonConflictingStepsSepAsyn} and Lemma \ref{lem:nonConflictingStepsMixAsyn}, the encoded continuations of the respective sender and receiver have eventually to become unguarded. Thus again for any \simulation of a source term step based on rule $ \textsc{Rep}_{\indexMix, \indexSep} $ there is exactly one \nonAdmin step and for any such \nonAdmin step exactly one source term step is \simulated.
	\qed
\end{proof}

Based on this Lemma we prove operational soundness, by showing that each target term is part of an \simulation of some source term step.

\begin{lemma}[Operational Soundness] \label{lem:operationalSoundness}
	The encodings $ \encodingSepAsyn $ and $ \encodingMixAsyn $ fulfil operational soundness.
\end{lemma}

\begin{proof}
	We start with $ \encodingSepAsyn $. By Definition \ref{def:operationalCorrespondence}, we have to show that:
	\begin{align*}
		\forall S \in \piSepProc \logdot \forall T \in \piAsynProc \logdot \EncodingSepAsyn{S} \steps T \text{ implies } \exists S' \in \piSepProc \logdot \exists T' \in \piAsynProc \logdot S \steps S' \wedge T \steps T' \wedge T' \transBarbCongSepAsyn \EncodingSepAsyn{S'}
	\end{align*}
	Note that $ T $ is a target term, i.e., $ T \in \targetTermsSepAsyn $. By Lemma \ref{lem:pureAdminStepsTransBarbBisimSepAsyn}, \pure \admin steps do not influence the state of a target term modulo $ \transBarbCongSepAsyn $, i.e., $ \forall T, T' \in \targetTermsSepAsyn \logdot T \pureAdminSteps T' $ implies $ T \transBarbCongSepAsyn T' $. Because of that, it suffice to consider \impure \admin and \nonAdmin steps, i.e., steps on translated source term names or sum locks. Moreover note, that steps on negative instantiations of sum locks reduce the corresponding test-statements to simple forwarders, that immediately restore the information consumed to resolve this test-statement. Thus \impure \admin steps on negative instantiations of sum locks do not change the state of a target term modulo $ \transBarbCongSepAsyn $.
	
	By Lemma \ref{lem:simulationVSNonAdminStep}, \nonAdmin steps indicate the border between the encoding of one source term and the encoding of its reduction. While, \pure \admin steps do not influence the state of a source term modulo $ \transBarbCongSepAsyn $, the \nonAdmin steps do. \NonAdmin steps finally rule out each way to reach one of the consumed translated observables and all translated observables connected to the same sums. With that, they also rule out each \simulation of conflicting source term steps and accordingly the reachability of the corresponding translated observables and occurrences of $ \success $. On the other hand, since all post-processing steps of all \simulations are \pure \admin steps, they ensure that eventually the respective encoded continuations are unguarded. Because of that immediately after the \nonAdmin step all translated observables and all reachable occurrences of $ \success $ of the encoded continuations are reachable. Because of that, source term steps and their \simulations handle source term observables and their translated observables in exactly the same way, i.e., they remove old and unguard new in the same way. So, if we consider only \pure \admin steps and a single \nonAdmin step, i.e., the sequence $ \EncodingSepAsyn{S} \pureAdminSteps T \nonAdminStep T' \pureAdminSteps T'' $ for a source term $ S \in \piSepProc $, then each target term in the sequence $ \EncodingSepAsyn{S} \pureAdminSteps T $, including $ T $, is congruent to $ \EncodingSepAsyn{S} $ modulo $ \transBarbCongSepAsyn $ and each target term in the sequence $ T' \pureAdminSteps T'' $, including $ T' $, is congruent to $ T'' $ modulo $ \transBarbCongSepAsyn $. By Lemma \ref{lem:simulationVSNonAdminStep}, the step $ T \nonAdminStep T' $ marks the \simulation of a source term step. Thus, $ S $ must be able to perform a step, i.e., $ \exists S' \in \piSepProc $ such that $ S \step S' $, and the sequence $ \EncodingSepAsyn{S} \pureAdminSteps T \nonAdminStep T' \pureAdminSteps T'' $ \simulates this step by unguarding the encodings of the subterms of $ S' $ and removing the translated observables that refer to the observables removed from $ S $ during the step to $ S' $. Moreover, note that, since all post-processing steps are \pure \admin steps, it is always possible to complete the \simulation after the \nonAdmin step. Thus $ \exists T''' \in \piAsynProc $ such that $ T''' $ is the result of the completion of the \simulation of $ S \step S' $, i.e., $ T' \pureAdminSteps T''' $ and $ T''' \transBarbCongSepAsyn \EncodingSepAsyn{S'} $.
	
	Unfortunately, \impure \admin steps complicate the situation. As described in Example \ref{exa:intermediateStates} they can already rule out some \simulations on conflicting source term steps by disabling the reachability of some translated observables (e.g. by consuming their respective positive instantiations of sum locks). If they rule out all \simulations on conflicting source term steps, then they behave like \nonAdmin steps and we are back to the situation descripted above. But, if they only rule out some \simulations on conflicting source term steps, we result in what we denote as intermediate state (compare to Definition \ref{def:intermediateState}). Intermediate states are in general not congruent to any of the surrounding encoded source terms modulo $ \transBarbCongSepAsyn $; neither to the encoded source term there we start our \simulation attempt nor to any encoding of the possible reductions. Note that the Definition \ref{def:operationalCorrespondence} of operational soundness explicitly allows for the presence of such intermediate states as long as we can ensure, that from each intermediate state a state, that is congruent modulo $ \transBarbCongSepAsyn $ to an encoded source term, is reachable. So let us have a closer look on these \impure \admin steps. By Definition \ref{def:pureImpureAdminStep}, \impure \admin steps are steps on translated source term names or on sum locks. They prepare an \simulation by unguarding and/or partially reducing a test-statement. Since for all sum locks eventually an instantiation is restored, each of these test-stements is eventually resolved (compare to the proof in \cite{nestmann00} that the encoding $ \encodingSepAsyn $ does not introduce deadlock).
	
	So let us consider a sequence $ \EncodingSepAsyn{S} \adminSteps T \nonAdminStep T' \pureAdminSteps T'' $. If the \nonAdmin step $ T \nonAdminStep T' $ rules out any \simulation attempt except for the one it steers to the point of no return, then we can use the same argumentation as in the case without \impure \admin steps to show that $ \exists S' \in \piSepProc $ and $ \exists T''' \in \piAsynProc $ such that $ S \step S' $, $ T' \pureAdminSteps T''' $, i.e., $ T \steps T''' $, and $ T''' \transBarbCongSepAsyn \EncodingSepAsyn{S'} $.
	
	However, by performing several \impure \admin steps in $ \EncodingSepAsyn{S} \adminSteps T $, several \simulations can be started by unguarding and/or partially reducing test-statements before a single of them is resolved. In this case, the resolution of one of these test-statements may not directly lead to a term congruent modulo $ \transBarbCongSepAsyn $ to an encoded source term. This can happen, if the source term can perform several sets of conflicting steps, i.e., there are at least two non conflicting but parallel steps and for each of these steps there can be some conflicting alternatives. To reach again a term, that is congruent modulo $ \transBarbCongSepAsyn $ to an encoded source term, $ T $ has to finally decide on which of the started \simulations are completed and rule out all conflicting \simulations. Note that, therefore in general it is not necessary to resolve all unguarded or partially reduced test-statements. However, since the $ \encodingSepAsyn $ does not introduce deadlock, all unguarded and partially reduced test-statements in $ T $ can be resolved, which definitely completes some of the started \simulations by \nonAdmin steps and rules out the completion of any remaining started \simulation\footnote{Note that the completion of \simulations, that are up to now only started by performing some \pure \admin steps, possibly might not be ruled out by resolving these test-statements. However, since \pure \admin steps do not influence the state of a target term modulo $ \transBarbCongSepAsyn $, such \simulations can be ignored.}. Thus $ \exists T' \in \piAsynProc $ such that $ T' $ is the target term that is reached after all unguarded or started test-statements of $ T $ are resolved and all necessary post-processing steps to complete the respective \simulations are performed, i.e., $ T \steps T' $. Then for each \nonAdmin step in the sequence $ T \steps T' $ there is one source term step of $ S $ or a reduction of $ S $. Note that in $ \EncodingSepAsyn{S} \adminSteps T $ several \simulations are started, but, since there are no \nonAdmin steps, none of these \simulations is completed. So $ S $ must be able to perform all source term step, that correspond to an \simulation completed in $ T \steps T' $, in parallel. Because of that, we can split up $ S $ into $ n $ parallel subterms $ S_1 $, \ldots, $ S_n $, where $ n $ is the number of \simulations completed in $ T \steps T' $. Then we can prove the lemma for the subterms using the cases above, i.e., we have $ \EncodingSepAsyn{S_i} \adminSteps T_i $ implies $ \exists S_i' \in \piSepProc $ and $ \exists T_i' \in \piAsynProc $ such that $ S_i \step S_i' $, $ T_i \steps T_i' $, and $ T_i' \transBarbCongSepAsyn \EncodingSepAsyn{S_i'} $. Since $ S $ can perform all these source term steps in parallel, it can also perform them arbitrarily ordered in a sequence, i.e., $ \exists S' \in \piSepProc $ such that $ S \steps S' $ and $ S' $ is structural congruent to the parallel composition of the reductions $ S_1' $, \ldots, $ S_n' $. Then we use an argumentation similar to the case of the \textsc{Par} rule in the proof of Lemma \ref{lem:operationalCompletenessSepAsyn} to conclude that $ \exists T' \in \piAsynProc $ such that $ T \steps T' $ and $ T_i' \transBarbCongSepAsyn \EncodingSepAsyn{S_i'} $ for all $ i $ with $ 1 \leq i \leq n $ implies $ T' \transBarbCongSepAsyn \EncodingSepAsyn{S'} $.
	
	Finally, let us consider an arbitrary sequence $ \EncodingSepAsyn{S} \steps T $, which may even include \nonAdmin steps. Note that, in opposite two the cases before, this sequence covers the case there already some \simulations are completed while for other \simulations\!\!|possibly of parallel source term steps|up to now only test-statements are unguarded and/or partially reduced. However, revisiting the argumentation of the case before, $ \exists T' \in \piAsynProc $ such that $ T \steps T' $ and $ T' $ is the result from resolving all unguarded and partially reduced test-statements in $ T $ and performing all post-precessing steps necessary to complete all \simulations in $ \EncodingSepAsyn{S} \steps T \steps T' $, that are already driven beyond their point of no return by the respective \nonAdmin step. By concentrating again only on the \impure \admin and \nonAdmin steps|and ignoring \impure \admin steps that do not lead to a completed \simulation in $ \EncodingSepAsyn{S} \steps T' $|we can split up this sequence into subsequences of subsequent bundles of completed \simulations, such that each such bundle can not be further subdivided into such bundles. Then, if necessary, we split up these bundles into the parallel branches of the corresponding source terms as in the case before. Repeating this, we result in a tree of completed \simulation of subsequently and parallel source term steps, where for each parallel source term step there is a branch in the respective subtree. Now each line between a parent and its direct child node within this tree represents one of the above considered cases. So we can conclude the proof by an induction over this tree. Thus, $ \exists S' \in \piSepProc $ such that $ S \steps S' $ and $ T \transBarbCongSepAsyn \EncodingSepAsyn{S'} $.
	
	\vspace*{1em}
	\noindent The argumentation for $ \encodingMixAsyn $ is similar to the argumentation for $ \encodingSepAsyn $ above. Of course, the Lemmata \ref{lem:pureAdminStepsTransBarbBisimSepAsyn} and \ref{lem:operationalCompletenessSepAsyn} have to be replaced in the argumentation by \ref{lem:pureAdminStepsTransBarbBisimMixAsyn} and \ref{lem:operationalCompletenessMixAsyn}, respectively. Moreover note, that the combination of the Lemmata \ref{lem:pureAdminStepsNoDeadlock}, \ref{lem:impureAdminStepsNoDeadlock}, and \ref{lem:simulationVSNonAdminStep} proves that $ \encodingMixAsyn $ does not introduce deadlock. Finally, in $ \encodingMixAsyn $ there are no steps on translated source term names, so there are less \impure \admin steps.
	\qed
\end{proof}

\begin{lemma}[Divergence Reflection] \label{lem:divergenceReflectionSepAsyn}
	The encoding $ \encodingSepAsyn $ reflects divergence.
\end{lemma}

\begin{proof}
	By Lemma \ref{lem:simulationVSNonAdminStep}, \simulations and \nonAdmin steps corresponds. Because of that, an infinite sequence $ \EncodingSepAsyn{S} \step^{\omega} $ with infinite many \nonAdmin steps implies that infinitely many source term steps are \simulated, i.e., that $ S \step^{\omega} $. Because of that, it suffice to show that each sequence $ T \adminSteps T' $ on target terms $ T, T' \in \targetTermsSepAsyn $ between two \nonAdmin steps is finite. The argumentation for the sequence of \admin steps $ \EncodingSepAsyn{S} \adminSteps T' $ before the first \nonAdmin step and for the sequence of \admin steps after the last \nonAdmin step|in case of a terminating process|is then similar.
	
	A look at the definition of the corresponding renaming policy in Figure \ref{fig:encodingSepAsyn} suggests the following case split\footnote{Note that in most cases the considered names are restricted, so a simple alpha conversion may change them. Because of that, as already in the proof of Lemma \ref{lem:nonConflictingStepsSepAsyn}, the use of concrete names in the following case split should not imply that we consider steps on these specific names. Instead the names refer to the meaning which is related to them by the encoding function.} on the links of the steps in $ T \adminSteps T' $.
	\begin{description}
		\item[Case of $ \senderLock $:] The name $ \senderLock $ is used by the encoding function to denote sender locks. By Definition \ref{def:senderLock}, sender locks are channels of multiplicity one, that are used as input but not as replicated input channels and within the continuation of such an input there is no unguarded and unrestricted instantiation of a sum lock. Their purpose is to guard the encoded continuations within the encoding of output guarded source terms. Since we consider only terms with finite representation, there are only finitely many unguarded instantiations on sender locks in $ T $. Analysing the encoding function in Figure \ref{fig:encodingMixAsyn} we observe, that new instantiations on sender locks appear only in the then-case of the (nested) test-statement in the encoding of input guarded source terms and in the encoding of a replicated input. Because of that, new instantiations on sender locks can be unguarded only by a \nonAdmin step followed by a \pure \admin step to reduce the test-statement to the then-case. Some of these \nonAdmin steps may have happened before $ T $, but of course only finitely many. So within the sequence $ T \adminSteps T' $ only finitely many new instantiations on sender locks can be unguarded. We conclude that within the sequence $ T \adminSteps T' $ there are only finitely many steps on sender locks.
			
			Note that the encoded continuation of input guarded source terms or a replicated input appears in parallel to such an instantiation of a sender lock within the then-case of a (nested) test-statement. Because of that, within the sequence $ T \adminSteps T' $ only finitely many encoded source terms, i.e., continuations, can be unguarded.
		\item[Case of Outputs on Translated Source Term Names:] Since we consider only terms with finite representations, there are only finitely many unguarded outputs on translated source term names in $ T $. Analysing the encoding function in Figure \ref{fig:encodingSepAsyn} we observer that there are only two ways to unguard a new output on a translated source term name. First, they can be unguarded by unguarding an encoded source term as continuation of an \simulation. By the case before, within the sequence $ T \adminSteps T' $ only finitely many encoded source terms can be unguarded, so there can only finitely many new outputs on translated source term names can be unguarded that way. Second, in the else-case of the first test-statement in the encoding of an input guarded source term there is such an unguarded output. To unguard this test-statement an instantiation of a receiver lock is consumed, but no new is unguarded by the else-case of the first test-statement. Analysing the encoding function we observer that all receiver locks are generated under restriction and initially there is exactly one instantiation of each receiver lock. Moreover, receiver locks are never transmitted over these restriction and to unguard a new instantiation (in the else-case of the second test-statement) a former instantiation of that receiver lock has to be consumed. Because of that, in each target term there is at most one instantiation of each receiver lock. Thus, reducing the nested test-statement to the else-case of the first test-statements, prevents any further unguarding of this test-statement. Since in $ T $ there are only finitely many different guarded or unguarded replicated inputs on receiver locks guarding such a nested test-statement, it is not possible to unguard initially many new outputs on translated source term names that way. We conclude that within the sequence $ T \adminSteps T' $ only finitely many outputs on translated source term names can be unguarded.
		\item[Case of $ \receiverLock $:] The name $ \receiverLock $ is used by $ \encodingSepAsyn $ to introduce receiver locks. Since we consider only terms with finite representations, there are only finitely many unguarded instantiations on receiver locks in $ T $. Analysing the encoding in Figure \ref{fig:encodingSepAsyn} we observer that there are only two ways to unguard new instantiations on receiver locks. First, they can be unguarded by unguarding an encoded source term as continuation of an \simulation. By the first case, within the sequence $ T \adminSteps T' $ only finitely many encoded source terms can be unguarded, so there can only finitely many new instantiations on receiver locks be unguarded that way. Second, in the else-case of the second test-statement in the encoding of an input guarded source term there is such an unguarded instantiation on a receiver lock. However, to unguard such a nested test-statement an output on a translated source term name has to be consumed. Since by the case before there are only finitely many such outputs, there is no possibility to unguard infinity many instantiations on receiver lock this way. We conclude that within the sequence $ T \adminSteps T' $ there are only finitely many steps on receiver locks.
			
			Note that in $ T $ there are only finitely many unguarded test-statements and new test-statements can only be unguarded by unguarding an encoded source term or by a step on a receiver lock. Thus, since by the first two cases only finitely many encoded source terms can be unguarded and by the current case there are only finitely many steps on receiver locks, within the sequence $ T \adminSteps T' $ only finitely many new test-statements can be unguarded.
		\item[Case of Inputs on Translated Source Term Names:] Since we consider only terms with finite representations, there are only finitely many unguarded inputs on translated source term names in $ T $. New inputs on source terms names can only be unguarded by a step on a receiver lock. Since by the case before the number of steps on receiver locks is finite, within the sequence $ T \adminSteps T' $ only finitely many inputs on translated source term names can be unguarded. Moreover, since by the second case the number of outputs on translated source term names is finite as well, there are only finitely many steps on translated source term names within the sequence $ T \adminSteps T' $.
		\item[Case of $ \sumLock, \sumLock' $:] Both names are used by the encoding function to refer to sum locks. The only steps on sum locks are to reduce a test-statement. Note that in case of a nested test-statement two sum locks are reduced, i.e., there are two steps on sum locks. However, since by the third case there are only finitely many test-statements, within the sequence $ T \adminSteps T' $ there are only finitely many steps on sum locks.
		\item[Case of $ t, f $:] $ t $ and $ f $ are used by the encoding function to implement booleans. As in the case of sum locks, all steps on these names are used to reduce a test-statement. So again, since there are only finitely many test-statements, within the sequence $ T \adminSteps T' $ there are only finitely many steps on $ t $ and $ f $.
	\end{description}
	\qed
\end{proof}

\begin{lemma}[Divergence Reflection] \label{lem:divergenceReflectionMixAsyn}
	The encoding $ \encodingMixAsyn $ reflects divergence.
\end{lemma}

\begin{proof}
	By Lemma \ref{lem:simulationVSNonAdminStep}, \simulations and \nonAdmin steps corresponds. Because of that, an infinite sequence of target term steps $ \EncodingMixAsyn{S} \step^{\omega} $ with infinite many \nonAdmin steps implies that infinitely many source term steps are \simulated, i.e., that $ S \step^{\omega} $. Because of that, it suffice to show that each sequence $ T \adminSteps T' $ on target terms $ T, T' \in \targetTermsMixAsyn $ between two \nonAdmin steps is finite. The argumentation for the sequence of \admin steps $ \EncodingMixAsyn{S} \adminSteps T' $ before the first \nonAdmin step and for the sequence of \admin steps after the last \nonAdmin step|in case of a terminating process|is then similar.
	
	Since source term names are translated into values, never used as links, it suffice to consider steps on names introduced by the encoding function. A look at the definition of the corresponding renaming policy in Figure \ref{fig:encodingMixAsyn} suggests the following case split\footnote{Note that in most cases the considered names are restricted, so a simple alpha conversion may change them. Because of that, as already in the proof of Lemma \ref{lem:nonConflictingStepsMixAsyn}, the use of concrete names in the following case split should not imply that we consider steps on these specific names. Instead the names refer to the meaning which is related to them by the encoding function.} on the links of the steps in $ T \adminSteps T' $.
	\begin{description}
		\item[Case of $ \senderLock $:] The name $ \senderLock $ is used by the encoding function to denote sender locks. By Definition \ref{def:senderLock}, sender locks are carried as third value of output requests. Their purpose is to guard the encoded continuations within the encoding of output guarded source terms. Since we consider only terms with finite representation, there are only finitely many unguarded instantiations on sender locks in $ T $. Analysing the encoding function in Figure \ref{fig:encodingMixAsyn} we observe, that new instantiations on sender locks appear only in the then-case of the (nested) test-statement in the encoding of input guarded source terms and in the encoding of a replicated input. Because of that, new instantiations on sender locks can be unguarded only by a \nonAdmin step followed by a \pure \admin step to reduce the test-statement to the then-case. Some of these \nonAdmin steps may have happened before $ T $, but of course only finitely many. So within the sequence $ T \adminSteps T' $ only finitely many new instantiations on sender locks can be unguarded. We conclude that within the sequence $ T \adminSteps T' $ there are only finitely many steps on sender locks.
			
			Note that the encoded continuation of input guarded source terms appears in parallel to such an instantiation of a sender lock within the then-case of a nested test-statement. Because of that, within the sequence $ T \adminSteps T' $ only finitely many encoded continuations of encodings of input or output guarded source terms can be unguarded.
		\item[Case of $ \coordinatorRepA, \coordinatorRepB $:] These two names are used to link the encoded continuations of several reductions of the same replicated input within a chain. Note that the first name denotes a chain lock carrying a translated source term name (compare to Definition \ref{def:chainLock}). Again, since we consider only terms with finite representations, there are only finitely many instances, i.e., outputs, on these names. For each step on the chain lock a new output on an instance of $ \coordinatorRepB $ is unguarded. So the number of steps on instances of $ \coordinatorRepB $  is bounded by the number of steps on chain lock carrying a translated source term name. New instantiations of such chain locks are unguarded as instantiations of sender locks only in the then-case of the test-statement in the encoding of a replicated input. Because of that, new instantiations on such chain locks can be unguarded only by a \nonAdmin step followed by a \pure \admin step to reduce the test-statement to the then-case. Some of these \nonAdmin steps may have happened before $ T $, but of course only finitely many. So within the sequence $ T \adminSteps T' $ only finitely many new instantiations on such chain locks can be unguarded. We conclude that within the sequence $ T \adminSteps T' $ there are only finitely many steps on such chain locks and instances of $ \coordinatorRepB $.
			
			Because of that, within the sequence $ T \adminSteps T' $ only finitely many encoded continuations of encodings of replicated inputs can be unguarded.
		\item[Case of $ \parallelChannelOut, \parallelChannelIn, \coordinatorUpOut, \coordinatorUpIn, \matchingCoordinatorOut, \matchingCoordinatorIn, \matchingUpOut, \matchingUpIn, \matchingReceiverIn, \matchingReceiverOut, \matchingReceiverUpIn, \matchingReceiverUpOut $:] All these names are request channels, i.e., there are introduced by $ \encodingMixAsyn $ to transport requests. Since we consider only terms with finite representations, there are only finitely many unguarded requests in $ T $. Moreover, by the two cases above, only finitely new requests can be unguarded within the sequence $ T \adminSteps T' $ by unguarding source term encodings. Thus there are only finitely many different unguarded requests within the sequence $ T \adminSteps T' $. We have already argued that the encoding $ \encodingMixAsyn $ puts much effort in restricting the way a request may take. We underpin this fact here, by proving that these ways do not include cycles. Note that inputs on requests channels do only appear within the encoding of the parallel operator or a replicated input (compare to Figure \ref{fig:encodingMixAsyn}).
			
			Let us consider the encoding of a parallel operator first. The terms \processLeftOutputRequests \ and \processLeftInputRequests \ transmit all requests of the left side of a parallel operator encoding to the right side and upwards over the restriction on the request channels. In \processRightOutputRequests \ and \processRightInputRequests \ the requests of the right side of the parallel operator encoding are linked in two chains, one for input and one for output requests. Therefore the requests of the right side are consumed, but a single copy of them is pushed upwards over the restriction on the request channels again. Then all left requests can be received by the respective first member of the chain of requests of opposite kind to the left request. Within these first members all left requests are consumed to combine them with the respective request behind this member and a copy of this left requests is further pushed along the chain. Now each member of that chain subsequently receives the left requests from its predecessor and sends a copy to its successor. Since the chains have no cycles and can not be infinitely long, each left request is only finitely often transmitted within each of these chains. In \pushRequests for each left and right requests, that was pushed over the restriction on request channels at the left and right side of a parallel operator, a single copy either under the restriction of a surrounding parallel operator encoding or on free source term names is generated. Note that the requests on free source term names can never be consumed. Moreover, for each parallel operator encoding there is a parallel operator in the corresponding source term. Thus the parallel operator encodings cause the structure of a finite binary tree, i.e., there are only finitely many instances of the terms \processLeftOutputRequests, \processLeftInputRequests, \processRightOutputRequests, \processRightInputRequests, and \pushRequests. Note that because of the two first cases only finitely many new encodings of parallel operators can be unguarded. Because of that and since a tree is free of cycles, within the sequence $ T \adminSteps T' $ there are only finitely many steps on requests, that are induced by a parallel operator encoding.
			
			In the last case we prove, that only within $ T \adminSteps T' $ only finitely many encoded continuations of encodings of replicated inputs can be unguarded. Because of that, there are only finitely many instances of \encodedContinuation. Within each of these instances all requests originate from the respective replicated input or a more left instance of \encodedContinuation of the same replicated input encoding are copied twice by \pushRequestsIn; one copy is generated under a restriction of this instance of \encodedContinuation \ to be combined with the requests within this branch and one copy is prepared to be transmitted to the respective next right branch of that encoded replicated input. All these first copies are proceed by \processRightOutputRequests \ and \processRightInputRequests \ as a request of the left side of a parallel operator encoding, while each requests of the respective encoded continuation is proceed by these terms as a right request. For each such right request in \pushRequestsOut a single copy is generated and the restriction of the surrounding parallel operator encoding. Moreover for each such left and each such right request a single copy is generated under the binding of the \encodedContinuation, that is the next one in the chain. Since again there are no cycles, there are only finitely many steps on requests, that are induced by a \encodedContinuation. We conclude, that within the sequence $ T \adminSteps T' $ there are only finitely many steps on requests.
		\item[Case of $ \coordinatorMatchingOut, \coordinatorMatchingIn $:] These names implement chain locks carrying a request channel. They are used to allow for a new request at the right side of a parallel operator encoding to be linked in the corresponding chain. By the third case there are only finitely many such right requests and chains, so within the sequence $ T \adminSteps T' $ there are only finitely many steps on such chain locks.
		\item[Case of $ \receiverLock $:] The name $ \receiverLock $ is used by $ \encodingMixAsyn $ to introduce receiver locks. Since we consider only terms with finite representations, there are only finitely many unguarded instantiations on receiver locks in $ T $. Anything the encoding in Figure \ref{fig:encodingMixAsyn} we observer that new instantiations on receiver locks are only unguarded due to a a matching pair of requests. By the third case, there are only finitely many different requests and steps on request channels within the sequence $ T \adminSteps T' $. Moreover we observe that no request is ever combined with itself and no pair is combined twice. Because of that, within the sequence $ T \adminSteps T' $ only finitely many new instantiations on receiver locks can be unguarded, i.e., there are only finitely many steps on receiver locks.
			
			Note that in $ T $ there are only finitely many unguarded test-statements and new test-statements can only be unguarded by unguarding an encoded source term or by a step on a receiver lock. Thus, since by the first two cases only finitely many encoded source terms can be unguarded and by the current case there are only finitely many steps on receiver locks, within the sequence $ T \adminSteps T' $ only finitely many new test-statements can be unguarded.
		\item[Case of $ \sumLock, \sumLock_s, \sumLock_r, \sumLock_1, \sumLock_2 $:] All these names are used by the encoding function to refer to sum locks. The only steps on sum locks are to reduce a test-statement. Note that in case of a nested test-statement two sum locks are reduced, i.e., there are two steps on sum locks. However, since by the case before there are only finitely many test-statements, within the sequence $ T \adminSteps T' $ there are only finitely many steps on sum locks.
		\item[Case of $ t, f $:] $ t $ and $ f $ are used by the encoding function to implement booleans. As in the case of sum locks, all steps on these names are used to reduce a test-statement. So again, since there are only finitely many test-statements, within the sequence $ T \adminSteps T' $ there are only finitely many steps on $ t $ and $ f $.
		\item[Case of $ y, y', z $:] These names are used by the encoding function as values only, but never as links. So there are no (\admin) steps on these names.
	\end{description}
	\qed
\end{proof}

\begin{lemma}[Success Sensitiveness] \label{lem:successSensitiveness}
	The encodings $ \encodingSepAsyn $ and $ \encodingMixAsyn $ are success sensitive.
\end{lemma}

\begin{proof}
	We start with $ \encodingSepAsyn $. By Definition \ref{def:succesSensitiveness}, we have to show that $ \forall S \in \piSepProc \logdot S\reachSuccess \text{ iff } \EncodingSepAsyn{S}\reachSuccess $. Let $ S \in \piSepProc $ be an arbitrary source term. We prove both directions separately.
	\begin{description}
		\item[Case of] $ S\reachSuccess $ \textbf{implies} $ \EncodingSepAsyn{S}\reachSuccess $\textbf{:} By Definition \ref{def:success}, $ S\reachSuccess $ implies that there are some $ S', S'' \in \piSepProc $ such that $ S \steps S' $ and $ S' \equiv S'' \mid \success $. By operational completeness (Lemma \ref{lem:operationalCompletenessSepAsyn}), the sequence $ S \steps S' $ can be \simulated by $ \EncodingSepAsyn{S} $, i.e., there exists some $ T \in \targetTermsSepAsyn $ such that $ \EncodingSepAsyn{S} \steps T $ and $ T \transBarbCongSepAsyn \EncodingSepAsyn{S'} $. By Figure \ref{fig:encodingSepAsyn}, $ \EncodingSepAsyn{S'' \mid \success} = \EncodingSepAsyn{S''} \mid \success $. Thus $ \EncodingSepAsyn{S'' \mid \success}\reachSuccess $. Since by Lemma \ref{lem:preservesSCModuloTransBarbBisimSepAsyn} the encoding $ \encodingSepAsyn $ preserves structural congruence of source terms modulo $ \transBarbCongSepAsyn $, $ S' \equiv S'' \mid \success $ implies $ \EncodingSepAsyn{S'} \transBarbCongSepAsyn \EncodingSepAsyn{S'' \mid \success} $. By Definition \ref{def:transBarbCong}, then for all contexts $ \Context{}{}{\hole} \in \piAsynProc \to \piAsynProc $ such that $ \Context{}{}{\EncodingSepAsyn{\nullTerm}} \in \targetTermsSepAsyn $ we have $ \Context{}{}{\EncodingSepAsyn{S'}} \transBarbBisimSepAsyn \Context{}{}{\EncodingSepAsyn{S'' \mid \success}} $. Then, by the first condition of Definition \ref{def:transBarbBisim}, $ \Context{}{}{\EncodingSepAsyn{S'}}\reachSuccess $ iff $ \Context{}{}{\EncodingSepAsyn{S'' \mid \success}}\reachSuccess $ for all such contexts $ \context $. Thus $ \EncodingSepAsyn{S'' \mid \success}\reachSuccess $ implies $ \EncodingSepAsyn{S'}\reachSuccess $. By the same argumentation we conclude that $ \EncodingSepAsyn{S'}\reachSuccess $ implies $ T\reachSuccess $. Since $ \EncodingSepAsyn{S} \steps T $, we conclude $ \EncodingSepAsyn{S}\reachSuccess $.
		\item[Case of] $ \EncodingSepAsyn{S}\reachSuccess $ \textbf{implies} $ S\reachSuccess $\textbf{:} By Definition \ref{def:success}, $ \EncodingSepAsyn{S}\reachSuccess $ implies that either there is an unguarded occurrence of $ \success $ in $ \EncodingSepAsyn{S} $ or there are some $ T, T'' \in \piAsynProc $ such that $ \EncodingSepAsyn{S} \steps T $ and $ T \equiv T'' \mid \success $. In the first case, since the only way to obtain an occurrence of success in a target term is by $ \EncodingSepAsyn{\success} = \success $, we conclude that there is an unguarded occurrence of success in $ \EncodingSepAsyn{S} $. In the other case, by operational soundness (Lemma \ref{lem:operationalSoundness}), then there are some $ S' \in \piSepProc $ and $ T' \in \piAsynProc $ such that $ S \steps S' $, $ T \steps T' $, and $ T' \transBarbCongSepAsyn \EncodingSepAsyn{S'} $. Since $ T\reachSuccess $ and $ \success $ can not be reduced, we have $ T'\reachSuccess $. By Definition \ref{def:transBarbCong}, $ T' \transBarbCongSepAsyn \EncodingSepAsyn{S'} $ implies $ \Context{}{}{T'} \transBarbBisimSepAsyn \Context{}{}{\EncodingSepAsyn{S'}} $ for all contexts $ \Context{}{}{\hole} \in \piAsynProc \to \piAsynProc $ such that $ \Context{}{}{\EncodingSepAsyn{\nullTerm}} \in \targetTermsSepAsyn $. By the first condition of Definition \ref{def:transBarbBisim}, $ \Context{}{}{T}\reachSuccess $ iff $ \Context{}{}{\EncodingSepAsyn{S'}}\reachSuccess $ for all such contexts $ \context $. Thus $ T'\reachSuccess $ implies $ \EncodingSepAsyn{S'}\reachSuccess $. Moreover, since $ T' $ includes an unguarded occurrence of success and since all occurrences of $ \success $ in target terms are due to the encoding $ \EncodingSepAsyn{\success} = \success $, the sequence $ \EncodingSepAsyn{S} \steps T' $ \simulates a sequence $ S \steps S' $ of target terms unguarding an occurrence of $ \success $. So $ S' $ has an unguarded occurrence of $ \success $, i.e., $ S'\reachSuccess $. We conclude $ S\reachSuccess $.
	\end{description}
	The argumentation in case of $ \encodingMixAsyn $ is similar to the argumentation of $ \encodingSepAsyn $ above. For the first case note that, by Figure \ref{fig:encodingMixAsyn}, the parallel operator is not translated \cleanly but since $ \EncodingMixAsyn{\success} = \success $ and the encodings of the parameters of a parallel operator appear unguarded within the encoding of the parallel operator, the unguarded occurrence of $ \success $ in $ S'' \mid \success $ implies that there is an unguarded occurrence of $ \success $ in $ \EncodingMixAsyn{S'' \mid \success} $.
	\qed
\end{proof}

\begin{theorem}
	The encodings $ \encodingSepAsyn $ and $ \encodingMixAsyn $ are good.
\end{theorem}

\begin{proof}
	By Figure \ref{fig:encodingSepAsyn} and Figure \ref{fig:encodingMixAsyn}, both encodings are compositional. Name invariance is proved in the Lemmata \ref{lem:nameInvarianceSepAsyn} and \ref{lem:nameInvarianceMixAsyn}. Operational correspondence follows by the Lemmata \ref{lem:operationalCompletenessSepAsyn}, \ref{lem:operationalCompletenessMixAsyn}, and \ref{lem:operationalSoundness}. Lemma \ref{lem:divergenceReflectionSepAsyn} proves that $ \encodingSepAsyn $ do not introduce divergence, while Lemma \ref{lem:divergenceReflectionMixAsyn} proves the same condition for $ \encodingMixAsyn $. Finally, by Lemma \ref{lem:successSensitiveness}, both encodings are success sensitive. Thus both encodings are good with respect to the five criteria introduced by Gorla.
	\qed
\end{proof}

\section{Degree of Distribution} \label{sec:proofsDistributability}

In the following we prove that $ \encodingSepAsyn $ preserves the degree of distribution of its source terms, while $ \encodingMixAsyn $ does not. For a definition what it means to preserve the degree of distribution have a look at Definition 5 in \cite{petersNestmann12}. Moreover, we present an encoding function from \piMix (without replicated input) into \piAsynTwo, the asynchronous \piCal-calculus augmented with a two-level polyadic synchronisation.

Note that for simplicity we omit in Figure \ref{fig:SC} a structural congruence rule for replicated input. However, of course a replicated input represents a not determined number of instances of the corresponding input guarded term in parallel. In order to allow for the distribution of replicated inputs, we augment structural congruence by an additional rule, i.e., $ \equivRep $ is structural congruence $ \equiv $ augmented with the additional rule $ \ReplicateInput{y}{x}.P \equivRep \ReplicateInput{y}{x}.P \mid \ReplicateInput{y}{x}.P $ to split up replicated inputs.

\begin{lemma}
	The encoding $ \encodingSepAsyn $ preserves the degree of distribution.
\end{lemma}

\begin{proof}
	Let $ S, S_1, \ldots, S_n, S_1', \ldots, S_n' \in \piSepProc $ such that $ S \equivRep S_1 \mid \ldots \mid S_n $ and $ S_i \steps S_i' $ for all $ i $ with $ 1 \leq i \leq n $. Then, by Definition 5 in \cite{petersNestmann12}, we have to show that:
	\begin{align*}
		\exists T_1, \ldots, T_n \in \piAsynProc \logdot \exists \Context{}{}{\hole_1, \ldots, \hole_n} \in \piAsynProc^n \to \piAsynProc \logdot \EncodingSepAsyn{S} \equivRep \Context{}{}{T_1, \ldots, T_n} \wedge \left( \forall i \in \left[ 1, \ldots, n \right] \logdot T_i \steps \transBarbCongSepAsyn \EncodingSepAsyn{S_i'} \right)
	\end{align*}
	Revisiting the argumentation in the proof of Lemma \ref{lem:preservesSCModuloTransBarbBisimSepAsyn}, the encoding $ \encodingSepAsyn $ preserves structural congruence of source terms except for the rule $ P \mid \nullTerm \equiv P $, because it translates restriction, the parallel operator, and the match operator \cleanly. Moreover, because of the \clean translation of the parallel operator, even the additional rule $ \ReplicateInput{y}{x}.P \equivRep \ReplicateInput{y}{x}.P \mid \ReplicateInput{y}{x}.P $ is preserved. Thus $ S \equivRep S_1 \mid \ldots \mid S_n $ implies $ \EncodingSepAsyn{S} \equivRep \EncodingSepAsyn{S_1 \mid \ldots \mid S_n} = \EncodingSepAsyn{S_1} \mid \ldots \mid \EncodingSepAsyn{S_n} $, i.e., we can choose $ T_i = \EncodingSepAsyn{S_i} $ for all $ i $ with $ 1 \leq i \leq n $ and $ \Context{}{}{\hole_1, \ldots, \hole_n} = \hole_1 \mid \ldots \mid \hole_n $. By operational completeness (compare to Lemma \ref{lem:operationalCompletenessSepAsyn}), then $ S_i \steps S_i' $ implies $ \EncodingSepAsyn{S_i} \steps \transBarbCongSepAsyn \EncodingSepAsyn{S_i'} $ for all $ i $ with $ 1 \leq i \leq n $, i.e., $ T_i \steps \transBarbCongSepAsyn \EncodingSepAsyn{S_i'} $.
	\qed
\end{proof}

\begin{lemma}
	Any good encoding, that translates the parallel operator \cleanly\ and preserves enough of the structural congruence on source terms to ensure that $ S \equiv_1 S_1 \mid \ldots \mid S_n $ implies $ \ArbitraryEncoding{S} \equiv_2 \ArbitraryEncoding{S_1 \mid \ldots \mid S_n} $, preserves the degree of distribution.
\end{lemma}

\begin{proof}
	Let $ S, S_1, \ldots, S_n, S_1', \ldots, S_n' $ be terms of the source language such that $ S \equiv_1 S_1 \mid \ldots \mid S_n $ and $ S_i \steps S_i' $ for all $ i $ with $ 1 \leq i \leq n $. Then, by Definition 5 in \cite{petersNestmann12}, we have to show that:
	\begin{align*}
		\exists T_1, \ldots, T_n \in \piAsynProc \logdot \exists \Context{}{}{\hole_1, \ldots, \hole_n} \in \piAsynProc^n \to \piAsynProc \logdot \ArbitraryEncoding{S} \equiv_2 \Context{}{}{T_1, \ldots, T_n} \wedge \left( \forall i \in \left[ 1, \ldots, n \right] \logdot T_i \steps \asymp \ArbitraryEncoding{S_i'} \right)
	\end{align*}
	Since the encoding translates the parallel operator \cleanly\ and preserves enough of the structural congruence on source terms to ensure that $ S \equiv_1 S_1 \mid \ldots \mid S_n $ implies $ \ArbitraryEncoding{S} \equiv_2 \ArbitraryEncoding{S_1 \mid \ldots \mid S_n} $, we have $ \ArbitraryEncoding{S} \equiv_2 \ArbitraryEncoding{S_1 \mid \ldots \mid S_n} = \ArbitraryEncoding{S_1} \mid \ldots \mid \ArbitraryEncoding{S_n} $, i.e., we can choose $ T_i = \ArbitraryEncoding{S_i} $ for all $ i $ with $ 1 \leq i \leq n $ and $ \Context{}{}{\hole_1, \ldots, \hole_n} = \hole_1 \mid \ldots \mid \hole_n $. By operational completeness (compare to Definition \ref{def:operationalCorrespondence}), then $ S_i \steps S_i' $ implies $ \ArbitraryEncoding{S_i} \steps \transBarbCongSepAsyn \ArbitraryEncoding{S_i'} $ for all $ i $ with $ 1 \leq i \leq n $, i.e., $ T_i \steps \asymp \ArbitraryEncoding{S_i'} $.
	\qed
\end{proof}

\begin{lemma}
\label{lem:EncodingMixAsynNotPreservesDegreeOfDistribution}
	The encoding $ \encodingMixAsyn $ does \emph{not} preserve the degree of distribution.
\end{lemma}

\begin{proof}
	Consider the counter example $ S = \left( \Out{a} \mid \Out{b} \right) \mid \left( \In{a} \mid \In{b} \right) $. $ S $ can perform two different reductions, which are independent of each other. So $ S \equiv S_1 \mid S_2 $ with $ S_1 = \left( \Out{a} \mid \In{a} \right) $ and $ S_2 = \left( \Out{b} \mid \In{b} \right) $ such that $ S_1 \step \nullTerm $ and $ S_2 \step \nullTerm $. Of course, since $ \encodingMixAsyn $ is good, $ \EncodingMixAsyn{S} $ can \simulate both steps in either order. But it can not \simulate both truly in parallel.
	
	To prove this fact, let us assume the contrary, i.e., let us assume the encoding $ \encodingMixAsyn $ preserves the degree of distribution of $ S $. Then, by Definition 5 in \cite{petersNestmann12}, we have:
	\begin{align*}
		\exists T_1, T_2 \in \piAsynProc \logdot \exists \Context{}{}{\hole_1, \hole_2} \in \piAsynProc \times \piAsynProc \to \piAsynProc \logdot \EncodingMixAsyn{S} \equivRep \Context{}{}{T_1, T_2} \wedge T_1 \steps \transBarbCongMixAsyn \EncodingMixAsyn{S_1'} \wedge T_2 \steps \transBarbCongMixAsyn \EncodingMixAsyn{S_2'}
	\end{align*}
	The encoding of $ S $ is given by:
	\begin{align*}
		\EncodingMixAsyn{S} & = \begin{aligned}[t]
				& \RestrictedTerm{\matchingCoordinatorOut, \matchingCoordinatorIn, \coordinatorUpOut, \coordinatorUpIn, \coordinatorMatchingOut, \coordinatorMatchingIn, \matchingUpOut, \matchingUpIn}{\big(\\
				& \hspace*{1em} \RestrictedTerm{\parallelChannelOut, \parallelChannelIn}{\left( \EncodingMixAsyn{\Out{a} \mid \Out{b}} \mid \processLeftOutputRequests \mid \processLeftInputRequests \right)}\\
				& \hspace*{1em} \mid \RestrictedTerm{\parallelChannelOut, \parallelChannelIn}{\left( \EncodingMixAsyn{\In{a} \mid \In{b}} \mid \processRightOutputRequests \mid \processRightInputRequests \right)}\\
				& \hspace*{1em} \mid \pushRequests \big)}
			\end{aligned}
	\end{align*}
	We observe that for both source term steps their \simulations require that the corresponding requests are combined at the outermost parallel operator encoding, i.e., the one given above. Moreover, in both \simulations the output requests arrive at the left and the input requests arrive at the right side of this parallel operator encoding. Thus, for both \simulations the requests are proceed by \processRightInputRequests.
	\begin{align*}
		\processRightInputRequests & = \Output{\coordinatorMatchingIn}{\matchingCoordinatorOut} \mid \ReplicateInput{\coordinatorMatchingIn}{\matchingCoordinatorOut}.\Input{\parallelChannelIn}{y, \sumLock_r, \receiverLock}.\big(\\
			& \hspace*{2em} \RestrictedTerm{\matchingUpOut}{\big( \begin{aligned}[t]
					& \ReplicateInput{\matchingCoordinatorOut}{y', \sumLock_s, \senderLock, z}.\left( \Match{y'}{y}\Output{\receiverLock}{\sumLock_s, \sumLock_r, \sumLock_s, \senderLock, z} \mid \Output{\matchingUpOut}{y', \sumLock_s, \senderLock, z} \right)\\
					& \mid \RestrictedTerm{\matchingCoordinatorOut}{\left( \Forward{\matchingUpOut}{\matchingCoordinatorOut} \mid \Output{\coordinatorMatchingIn}{\matchingCoordinatorOut} \right)} \big)
				\end{aligned}\\
			& \hspace*{2em} \mid \Output{\coordinatorUpIn}{y, \sumLock_r, \receiverLock}} \big)
	\end{align*}
	Note that, since $ \EncodingMixAsyn{\Out{a}} $ and $ \EncodingMixAsyn{\Out{b}} $ appear unguarded in $ \EncodingMixAsyn{\Out{a} \mid \Out{b}} $, it is not difficult to distribute $ \EncodingMixAsyn{\Out{a} \mid \Out{b}} $ over $ T_1 $ and $ T_2 $, because the respective context introduced by the corresponding parallel operator encoding is not used within the considered \simulations, i.e., they can be proved to be junk. The same holds for $ \EncodingMixAsyn{\In{a} \mid \In{b}} $. Moreover, since none of the terms \processLeftInputRequests, \processRightOutputRequests, and \pushRequests \ of the outermost parallel operator encoding is of any importance for the considered \simulations, it does not matter whether we put them within $ T_1 $ or $ T_2 $. Note that for each of the considered \simulations it is necessary to transmit one left output request to the right side of the outermost parallel operator encoding by the forwarder in \processLeftOutputRequests. However, since this forwarder is guarded by a replicated input we can use the rule $ \ReplicateInput{y}{x}.P \equivRep \Input{y}{x}.P \mid \ReplicateInput{y}{x}.P $ to distribute this forwarder, e.g. by placing $ \Input{\parallelChannelOut}{y, \sumLock, \senderLock, z}.\left( \Output{\matchingCoordinatorOut}{y, \sumLock, \senderLock, z} \mid \Output{\coordinatorUpOut}{y, \sumLock, \senderLock, z} \right) $ within $ T_1 $ and the rest of the forwarder within $ T_2 $. The only remaining term to distribute is \processRightInputRequests.
	
	Since $ T_1 \steps \transBarbCongMixAsyn \EncodingMixAsyn{S_1'} $ and $ T_2 \steps \transBarbCongMixAsyn \EncodingMixAsyn{S_2'} $, both $ T_1 $ and $ T_2 $ need the ability to proceed a request from $ \EncodingMixAsyn{\In{a} \mid \In{b}} $. Since the respective right requests are restricted the only possibility to process them is \processRightInputRequests. We observe that the main part of \processRightInputRequests \ is guarded by a replicated input and can thus be distributed as \processLeftOutputRequests \ before. But there is only one unguarded instantiation of the corresponding chain lock $ \Output{\coordinatorMatchingIn}{\matchingCoordinatorOut} $ and without it the remaining term guarded by the replicated input on this chain lock is useless. Because of that \processRightInputRequests \ can not be distributed, i.e., it is not possible to distribute the encoding of $ S $ such that $ \EncodingMixAsyn{S} \equivRep \RestrictedTerm{\tilde{x}}{\left( T_1 \mid T_2 \right)} $, $ T_1 \steps \transBarbCongMixAsyn \EncodingMixAsyn{S_1'} $, and $ T_2 \steps \transBarbCongMixAsyn \EncodingMixAsyn{S_2'} $.
	
	Because of that $ \EncodingMixAsyn{S} $ can not \simulate both source term steps without sequentialising the linking of the respective right requests within the required chain. So $ \EncodingMixAsyn{S} $ can not completely \simulate the independent source terms steps in parallel.
	\qed
\end{proof}

\begin{lemma}
	Any good encoding $ \arbitraryEncoding $ preserves the criterion: For every $ S $ such that $ S = S_1 \mid \ldots \mid S_n $ and $ S_i \steps S_i' $ for all $ i $ with $ 1 \leq i \leq n $, there exists $ T_1, \ldots, T_n $ and a context $ \context $ with $ n $ holes such that $ \ArbitraryEncoding{S} \equiv_2 \Context{}{}{T_1, \ldots, T_n} $ and $ T_i \steps \asymp \ArbitraryEncoding{S_i'} $ for all $ i $ with $ 1 \leq i \leq n $.
\end{lemma}

\begin{proof}
	Assume $ n \in \nat $ and $ S, S_1, \ldots, S_n, S_1', \ldots, S_n' $ are terms of the source language such that $ S = S_1 \mid \ldots \mid S_n $ and $ \left( \forall i \in \left[ 1, \ldots, n \right] \logdot S_i \steps S_i' \right) $. Moreover, for simplicity, let us assume, that $ S_1 \mid \ldots \mid S_n = ( S_1 \mid ( \ldots ( S_{n - 1} \mid S_n ) \ldots ) $, but note that the actually orientation of the brackets is not import for the rest of the proof. Let $ \Context{}{\mid}{\hole_1, \hole_2} $ be the context introduced by the encoding of the parallel operator, i.e., $ \ArbitraryEncoding{P \mid Q } = \Context{}{\mid}{\ArbitraryEncoding{P}, \ArbitraryEncoding{Q}} $ for all $ P, Q $ in the source language. Then $ \ArbitraryEncoding{S} = \Context{}{\mid}{\ArbitraryEncoding{S_1}, \Context{}{\mid}{\ldots \Context{}{\mid}{\ArbitraryEncoding{S_{n - 1}}, \ArbitraryEncoding{S_n}}} \ldots } $. So we can choose $ T_i = \ArbitraryEncoding{S_i} $ for all $ i \in \left[ 1, \ldots, n \right] $ and $ \Context{}{}{T_1, \ldots, T_n} = \ArbitraryEncoding{S} $. By operational completeness (compare to Definition \ref{def:operationalCorrespondence}), then $ T_i \steps \asymp \ArbitraryEncoding{S_i'} $ for all $ i $ with $ 1 \leq i \leq n $.
	\qed
\end{proof}

\begin{theorem}
	There is no good encoding from \piMix into \piAsyn that preserves the degree of distribution.
\end{theorem}

\begin{proof}
	By the Theorem in \cite{petersSchickeNestmann11}, any good encoding from \piMix into \piAsyn introduces additional causal dependencies. Note that in \cite{petersSchickeNestmann11} slightly different definitions of \piMix and \piAsyn are used, i.e., they use $ !P $ instead of replicated input, have no match in the source language, and no $ \tau $-prefix. However, revisiting the argumentation on the proof of the Theorem in \cite{petersSchickeNestmann11} we observe that these details are not crucial.
	
	A closer look at the definition of additional causal dependencies used in \cite{petersSchickeNestmann11} reveals, that they show that any good encoding from \piMix into \piAsyn introduces causal dependencies between some \simulations of independent source term steps. As already stated in the discussion of \cite{petersSchickeNestmann11}, because of this causal dependency some \simulations of independent source terms steps have to be sequentialised, i.e., the \simulations can still be performed in either order and even somehow overlapping but not completely independent. Note that this is exactly what we observe in Lemma \ref{lem:EncodingMixAsynNotPreservesDegreeOfDistribution}. Because of that, there is no good encoding from \piMix into \piAsyn that preserves the degree of distribution.
	\qed
\end{proof}

Let $ \piAsyn^2 $ be the asynchronous \piCal-culus augmented with poliadic synchronisation on channels composed by up to two names as introduced by \cite{carboneMaffeis03}.

\begin{definition}[$ \piAsyn^2 $]
  The set of process terms of the \emph{asynchronous \piCal-calculus with two-level poliadic synchronisation}, denoted by $ \piAsynTwoProc $, is given by
	\begin{align*}
		P & \;\mathop{::=}\;
				\nullTerm
                \sep \RestrictedTerm{n}{P}
                \sep P_1 \mid P_2
                \sep \Match{a}{b}P
                \sep \Output{y}{\tilde{z}}.P
                \sep \Input{y}{\tilde{x}}.P
                \sep \ReplicateInput{y}{\tilde{x}}.P
	\end{align*}
	for some names $ n, a, b \in \names $, some finite sequences of names $ \tilde{x}, \tilde{z} \subset \names $, and either $ y \in \names $ or $ y = y_1 \cdot y_2 $ for some $ y_1, y_2 \in \names $.
\end{definition}

\begin{figure}[htp]
	\begin{align*}
		\EncodingMixAsynTwo{P \mid Q} & \deff \RestrictedTerm{\coordinatorUpOut, \coordinatorUpIn, \itag, \otag}{\big(\\
			& \hspace*{2em} \begin{aligned}[t]
					&  \RestrictedTerm{\parallelChannelOut, \parallelChannelIn}{\big( \begin{aligned}[t] 
							& \EncodingMixAsynTwo{P}\\
							& \mid \ReplicateInput{\parallelChannelOut}{y, \sumLock, \senderLock_1, \senderLock_2, z}.\left( \Output{y \cdot \otag}{\sumLock, \senderLock_1, \senderLock_2, z} \mid \Output{\coordinatorUpOut}{y, \sumLock, \senderLock_1, \senderLock_2, z} \right)\\
							& \mid \ReplicateInput{\parallelChannelIn}{y, \sumLock, \receiverLock}.\left( \Output{y \cdot \itag}{\sumLock, \receiverLock} \mid \Output{\coordinatorUpIn}{y, \sumLock, \receiverLock} \right) \big)
						\end{aligned}}\\
					& \mid \RestrictedTerm{\parallelChannelOut, \parallelChannelIn}{\big( \begin{aligned}[t]
							& \EncodingMixAsynTwo{Q}\\
							& \mid \ReplicateInput{\parallelChannelOut}{y, \sumLock_s, \senderLock_1, \senderLock_2, z}.\left( \Input{y \cdot \itag}{\sumLock_r, \receiverLock}.\Output{\receiverLock}{\sumLock_r, \sumLock_s, \sumLock_s, \senderLock_1, \senderLock_2, z, \true} \mid \Output{\coordinatorUpOut}{y, \sumLock_s, \senderLock_1, \senderLock_2, z} \right)\\
							& \mid \ReplicateInput{\parallelChannelIn}{y, \sumLock_r, \receiverLock}.\left( \Input{y \cdot \otag}{\sumLock_s, \senderLock_1, \senderLock_2, z}.\Output{\receiverLock}{\sumLock_s, \sumLock_r, \sumLock_s, \senderLock_1, \senderLock_2, z, \false} \mid \Output{\coordinatorUpIn}{y, \sumLock_r, \receiverLock} \right) \big)
						\end{aligned}}\\
					& \mid \Forward{\coordinatorUpOut}{\parallelChannelOut} \mid \Forward{\coordinatorUpIn}{\parallelChannelIn} \big)
				\end{aligned}}\\
		\EncodingMixAsynTwo{\Output{y}{z}.P} & \deff \RestrictedTerm{\senderLock_1, \senderLock_2}{\left( \Out{\senderLock_1} \mid \ReplicateIn{\senderLock_1}.\Output{\parallelChannelOut}{y, \sumLock, \senderLock_1, \senderLock_2, z} \mid \In{\senderLock_2}.\EncodingMixAsynTwo{P} \right)}\\
		\EncodingMixAsynTwo{\Input{y}{x}.P} & \deff \RestrictedTerm{\receiverLock}{\big( \begin{aligned}[t]
					& \Output{\parallelChannelIn}{y, \sumLock, \receiverLock} \mid \ReplicateInput{\receiverLock}{\sumLock_1, \sumLock_2, -, \senderLock_1, \senderLock_2, x, b}.\\
					& \hspace*{2em} \BigTest{\sumLock_1}{\BigTest{\sumLock_2}{\Output{\sumLock_1}{\false} \mid \Output{\sumLock_2}{\false} \mid \Out{\senderLock_2} \mid \EncodingMixAsynTwo{P}}{\Output{\sumLock_1}{\true} \mid \Output{\sumLock_2}{\false} \mid \Test{b}{\Output{\parallelChannelIn}{y, \sumLock, \receiverLock}}{\Out{\senderLock_1}}}}{\Output{\sumLock_1}{\false} \mid \Test{b}{\Out{\senderLock_1}}{\Output{\parallelChannelIn}{y, \sumLock, \receiverLock}} \big)}
				\end{aligned}}
	\end{align*}
	\begin{center}
		Here $ \renamingPolicyMixAsynTwo $ is some arbitrary injective substitution such that $ \forall n \in \names \logdot \RenamingPolicyMixAsynTwo{n} \notin N $, where $ N = \Set{ \parallelChannelOut, \parallelChannelIn, \coordinatorUpOut, \coordinatorUpIn, \sumLock, \sumLock_s, \sumLock_r, \sumLock_1, \sumLock_2, \senderLock_1, \senderLock_2, \receiverLock, \otag, \itag, y, y, z, t, f, b } $. The remaining operators restriction, sum, $ \tau $ guarded terms, and $ \success $ are translated as in $ \encodingSepAsyn $ or $ \encodingMixAsyn $ again.
	\end{center}
	\caption{Encoding $ \encodingMixAsynTwo $ from \piMix without replicated input into $ \piAsyn^2 $.} \label{fig:encodingMixAsynTwo}
\end{figure}

For simplicity, we do not consider replicated input on the source language. Then the encoding $ \encodingMixAsynTwo $, given in Figure \ref{fig:encodingMixAsynTwo}, is a good encoding from \piMix without replicated input into \piAsynTwo. Note that, to prove that $ \encodingMixAsynTwo $ is correct with respect to the criteria presented in Section \ref{sec:qualityCriteria}, we can argument very similar as for $ \encodingMixAsyn $, because the main feature, i.e., the way in which sum locks are ordered, is the same for both encodings. On the other side, as already stated by \cite{carboneMaffeis03}, there is no good encoding from \piMix into \piAsynTwo, that translates the parallel operator \cleanly. Note that the separation result in \cite{carboneMaffeis03} does not rely on replication, i.e., it also implies that there is no such encoding from \piMix without replicated input into \piAsynTwo.

\addcontentsline{toc}{section}{References}
\bibliographystyle{alpha}
\bibliography{encodingSynPi}

\end{document}